%% file: main.tex
\documentclass[
11pt, % The default document font size, options: 10pt, 11pt, 12pt
english, % ngerman for German
singlespacing, % Single line spacing, alternatives: onehalfspacing or doublespacing
headsepline, % Uncomment to get a line under the header
]{DoctoralThesis} % The class file specifying the document structure
\usepackage{amsmath,amssymb,amsfonts}
\usepackage{amsthm}
\usepackage{tabularx} 
\usepackage{graphicx}
\usepackage{subcaption}
\usepackage{comment}
\usepackage{xcolor}
\usepackage{multirow}
\newtheorem{problem}{Problem}
\newtheorem{definition}{Definition}
\newtheorem{theorem}{Theorem}

\newtheorem{lemma}[theorem]{Lemma}
\newtheorem{assumption}{Assumption}
\newtheorem{proposition}{Proposition}
\usepackage{rotating}
\usepackage[ruled,linesnumbered]{algorithm2e}
\usepackage[utf8]{inputenc} % Required for inputting international characters
\usepackage{pdflscape}
\usepackage[T1]{fontenc} % Output font encoding for international characters
\usepackage{booktabs}% http://ctan.org/pkg/booktabs
\usepackage{booktabs} % For formal tables
\usepackage{url}
\usepackage{graphicx}
\usepackage{titlesec}
\usepackage[labelsep=period]{caption}
\usepackage{enumitem}
\usepackage[most]{tcolorbox}
\usepackage[backend=bibtex,style=authoryear,natbib=true,dashed=false,maxbibnames=99]{biblatex} % Use the bibtex backend with the authoryear citation style (which resembles APA)
\usepackage{float}    
\usepackage{array,adjustbox,amssymb} % array for p{..}, adjustbox to auto-fit
\usepackage[table]{xcolor}
\definecolor{rowgray}{gray}{0.94} 
\newcommand{\cmark}{\checkmark}
\newcommand{\xmark}{\(\times\)}
\newcolumntype{L}[1]{>{\raggedright\arraybackslash\hspace{0pt}}p{#1}}
\newcolumntype{C}[1]{>{\centering\arraybackslash}p{#1}}
\titleformat{\paragraph}
{\normalfont\normalsize\bfseries}{\theparagraph}{1em}{}
\titlespacing*{\paragraph}
{0pt}{3.25ex plus 1ex minus .2ex}{1.5ex plus .2ex}

\usepackage[autostyle=true]{csquotes} % Required to generate language-dependent quotes in the bibliography

\thesistitle{Practical Graph Optimisation and AI-Driven Models for Active Directory Security Hardening} % Your thesis title, this is used in the title and abstract, print it elsewhere with \ttitle
\supervisor{Dr.\textsc{Hung Nguyen}} % Your supervisor's name, this is used in the title page, print it elsewhere with \supname
\examiner{} % Your examiner's name, this is not currently used anywhere in the template, print it elsewhere with \examname
\degree{Doctor of Philosophy} % Your degree name, this is used in the title page and abstract, print it elsewhere with \degreename
\author{\textsc{Quang Huy Ngo}} % Your name, this is used in the title page and abstract, print it elsewhere with \authorname
\addresses{} % Your address, this is not currently used anywhere in the template, print it elsewhere with \addressname

\subject{Computer Science} % Your subject area, this is not currently used anywhere in the template, print it elsewhere with \subjectname, such as Computer Science
 % Keywords for your thesis, this is not currently used anywhere in the template, print it elsewhere with \keywordnames
\university{\href{http://www.adelaide.edu.au}{University of Adelaide}} % Your university's name and URL, this is used in the title page and abstract, print it elsewhere with \univname
\department{\href{https://set.adelaide.edu.au/computer-and-mathematical-sciences/}{School of Computer and Mathematical Sciences}} % Your department's name and URL, this is used in the title page and abstract, print it elsewhere with \deptname, CS is the department name
\group{\href{https://set.adelaide.edu.au/}{Computer Science Education Research}} % Your research group's name and URL, this is used in the title page, print it elsewhere with \groupname
\faculty{\href{http://faculty.university.com}{Faculty Name}} % Your faculty's name and URL, this is used in the title page and abstract, print it elsewhere with \facname

\AtBeginDocument{
\hypersetup{pdftitle=\ttitle} % Set the PDF's title to your title
\hypersetup{pdfauthor=\authorname} % Set the PDF's author to your name
\hypersetup{pdfsubject=\deptname}
\hypersetup{pdfkeywords=\keywordnames} % Set the PDF's keywords to your keywords
}

\begin{document}

\frontmatter % Use roman page numbering style (i, ii, iii, iv...) for the pre-content pages

\pagestyle{plain} % Default to the plain heading style until the thesis style is called for the body content

%----------------------------------------------------------------------------------------
%	TITLE PAGE
%----------------------------------------------------------------------------------------

\begin{titlepage}
\begin{center}

\includegraphics{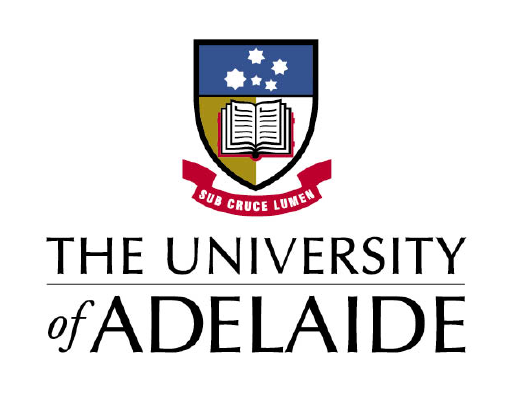}\\
\vspace{0.8cm}{\huge \bfseries \ttitle\par}\vspace{0.8cm}

{%
  \fontsize{14pt}{17pt}\selectfont
  \authorname\par
}
\vfill
{%
  \fontsize{15pt}{18pt}\selectfont
  \deptname\par
}
\vfill
\large A thesis submitted for the degree of\\ DOCTOR OF PHILOSOPHY\\The University of Adelaide\\[0.3cm] % University requirement text
 
\vfill

{\large 2nd Feb 2026}\\[1cm] % Date

{\large \itshape Awarded the Dean's Commendation for Doctoral Thesis Excellence}\\[3cm]

\vfill
\end{center}
\end{titlepage}

%----------------------------------------------------------------------------------------
%	LIST OF CONTENTS/FIGURES/TABLES PAGES
%----------------------------------------------------------------------------------------

\tableofcontents % Prints the main table of contents

\listoffigures % Prints the list of figures

\listoftables % Prints the list of tables

%----------------------------------------------------------------------------------------
%	ABSTRACT PAGE
%----------------------------------------------------------------------------------------

% \begin{abstract}

\chapter*{\abstractname}
\addchaptertocentry{\abstractname} 
Microsoft’s Active Directory (AD) is a directory service that enables the IT admin to manage security permissions and control access within a Windows domain network. As a core management system in up to 90$\%$ of companies, AD has become a primary target for adversaries. As a result, AD vulnerabilities are exploited in an estimated nine out of ten cyberattacks. Worse still, these networks are often alarmingly fragile when penetration testers show they can achieve full compromise in up to 82$\%$ of cases, often by leveraging known vulnerabilities. These statistics highlight an urgent need for more effective defence and hardening solutions.

Conceptually, an AD environment is represented as a directed graph where nodes represent entities like accounts, computers, and groups, and the directed edges represent accesses, permissions, or vulnerabilities. The primary challenge in securing this graph is managing the large number of vulnerabilities. While one could attempt to remediate every single vulnerability, the sheer scale and complexity of an enterprise AD graph make such a brute-force approach astronomically time-consuming and impractical. These constraints shift the focus from exhaustive remediation to prioritisation-driven attack-graph hardening. Rather than fixing every vulnerability, the defender targets a small set of high-risk issues that deliver the most defensive impact. The goal is to reduce the attack surface and limit traversable paths to high-privilege accounts while keeping remediation effort low. Traditional rank-based prioritisation methods often fall short, as they are not designed to mitigate attack paths tied to Active Directory structure. Recent optimisation-based hardening better disrupts these paths and has proved more effective. These optimisation-based approaches formulate the hardening task as a combinatorial problem on the attack graph, with the objective of improving graph-related security metrics and ultimately reducing the available attack surface.

While many solutions for hardening attack graphs exist, these efforts fall short in addressing several key practical challenges specific to the AD attack graph. First, existing models often assume the graph is static, whereas a real-world AD environment is highly dynamic. Second, most proposed solutions are limited to the defensive measure of revoking vulnerabilities (edge removal), while more active defence mechanisms, like the strategic placement of honeypots or decoys, are largely unstudied. Third, because not all remediations are implementable, as the recommended set may include essential permissions that cannot be removed due to operational requirements, a practical end-to-end model must incorporate system admin feedback into the prioritisation process, a constraint that previous works have often ignored.

This thesis aims to address these limitations by studying and proposing a number of game-theoretic and optimisation-based decision-making models specifically designed for defending AD-style attack graphs. This thesis introduces four key contributions. First, we propose a novel honeypot/decoy placement model based on the principle of minimising the number of shortest paths and the number of Domain Admin–reachable nodes. Second, building on this model, we introduce a defence strategy that considers the dynamic/temporal nature of the AD graph, where the objective is to find the location to deploy decoys that maximises the worst-case incident response time. Third, we introduce a novel adaptive prioritisation model that queries each high-risk attack path to the IT administrator for mediation. The objective of this model is to find a remediation policy that disconnects lower-privilege nodes from Domain Admin nodes with a minimum number of remediations. Finally, we introduce an end-to-end adaptive prioritisation model that minimises the approval effort of the system admin (IT admin) by finding a general adaptive edge-removal policy that generalises the system admin's decisions to edges with similar risk features. As a key result of this thesis, we show that the problems underlying all of the contributed models are computationally intractable. To address this complexity, we propose and evaluate a diverse range of algorithms. These solutions span from mathematical optimisation and evolutionary algorithms to AI-based methods, including clustering and reinforcement learning.

% \end{abstract}

%----------------------------------------------------------------------------------------
%	DECLARATION PAGE
%----------------------------------------------------------------------------------------

\begin{declaration}
\addchaptertocentry{\authorshipname} % Add the declaration to the table of contents
%source: https://www.adelaide.edu.au/graduatecentre/current-students/your-thesis-examination/preparation#thesis-declaration --> https://www.adelaide.edu.au/graduatecentre/system/files/media/documents/2020-01/examples-of-thesis-declarations-for-submissions.pdf --> the statement below is the one "For a thesis that contains publications"

I certify that this work contains no material which has been accepted for the award of any other degree or diploma in my name, in any university or other tertiary institution and, to the best of my knowledge and belief, contains no material previously published or written by another person, except where due reference has been made in the text. In addition, I certify that no part of this work will, in the future, be used in a submission in my name, for any other degree or diploma in any university or other tertiary institution without the prior approval of the University of Adelaide and where applicable, any partner institution responsible for the joint-award of this degree.

I also give permission for the digital version of my thesis to be made available on the web, via the University’s digital research repository, the Library Search and also through web search engines, unless permission has been granted by the University to restrict access for a period of time.

\vspace{40mm}
% \begin{figure}[h!]
% \begin{flushright}
%   \includegraphics[width=0.4\textwidth]{Figures/sampleSignature.png}
% \end{flushright}
% \end{figure}
\begin{flushright}
Quang Huy Ngo
\end{flushright}
\begin{flushright}September 2025\end{flushright}
 
\end{declaration}

\cleardoublepage

%----------------------------------------------------------------------------------------
%	ACKNOWLEDGEMENTS
%----------------------------------------------------------------------------------------

\begin{acknowledgements}
\addchaptertocentry{\acknowledgementname} PhD journey has been one of the best things to happen in my life—deeply rewarding and, at times, very challenging. This thesis would not have been possible without the tremendous support, guidance, and encouragement of many people. I hereby express my deepest gratitude to you all.

I am deeply grateful to my supervisors, Professor Hung Nguyen and Dr Mingyu Guo, for their unwavering support and guidance throughout my candidature. Their mentorship and companionship have been immensely important. They shaped how I think about research and what it means to be an independent scholar, and they made my candidature deeply rewarding.

I am especially thankful to Professor Hung Nguyen, whose mentorship began well before my PhD, during my Honours year of my bachelor’s degree. His constructive feedback and patience consistently pushed me to improve and grow into a more well-rounded researcher.

I am equally grateful to Dr Mingyu Guo, whose expertise and ability to make new concepts easy to understand, along with his enthusiasm for new research ideas, shaped my work and strengthened this thesis.

I am grateful to Dr Cheryl Pope for introducing and supporting me in teaching at the University of Adelaide. I also thank all my tutor peers for being supportive colleagues in the classroom and for hearing me out on professional challenges.

I also want to acknowledge Nhat, Long, Yumeng, and everyone in the HDR Lab for thoughtful feedback on ideas, for encouraging me forward, and for listening generously to my research and to my personal life. To all of my friends, thank you for your encouragement and for always being there.

Finally, I am immensely grateful to my family, especially my parents whose love and practical help were the bedrock of this work. Without their constant support and unconditional care, I could not have brought this thesis to completion. To my partner, thank you for your patience and steady support since day one of the PhD, and for cheering me on through the tough moments.

\end{acknowledgements}

\begin{publications}
    \addchaptertocentry{\publicationname}
    This thesis is \textbf{based on and extends} the following works that have been published or submitted for review:
    \begin{itemize}
        \item \textbf{Huy Q. Ngo}, Mingyu Guo, Hung Nguyen, Catch Me if You Can: Effective Honeypot Placement in Dynamic AD Attack Graphs. the IEEE International Conference on Computer Communications (INFOCOM’24) \textbf{[CORE A*]} (Chapter~\ref{chapter:paper1}).
        
        \item \textbf{Huy Q. Ngo}, Mingyu Guo, Hung Nguyen, Near-optimal Strategies for Honeypots Placement in Dynamic and Large Active Directory Networks (Extended Abstract). International Conference on Autonomous Agents and Multiagent Systems (AAMAS’23) \textbf{[CORE A*]} (Chapter~\ref{chapter:paper1}).
        
        \item \textbf{Huy Q. Ngo}, Mingyu Guo, Hung Nguyen, Optimizing Cyber Response Time on Temporal Active Directory Networks Using Decoys. Genetic and Evolutionary Computation Conference (GECCO’24) \textbf{[CORE A]} (Chapter~\ref{chapter:paper2}).
        
        \item \textbf{Huy Q. Ngo}, Mingyu Guo, Hung Nguyen, Adaptive Wizard for Removing Cross-Tier Misconfigurations in Active Directory. International Joint Conference on Artificial Intelligence (IJCAI’25)  \textbf{[CORE A*]} (Chapter~\ref{chapter:paper3})

        \item \textbf{Huy Q. Ngo}, Mingyu Guo, Hung Nguyen, Reinforcement Learning for Security Adaptive Connectivity Test. Submitted to AAAI Conference on Artificial Intelligence (AAAI’26) (Chapter~\ref{chapter:paper4})
    \end{itemize}
\end{publications}

%----------------------------------------------------------------------------------------
%	THESIS CONTENT - CHAPTERS
%----------------------------------------------------------------------------------------

\mainmatter % Begin numeric (1,2,3...) page numbering

\pagestyle{thesis} % Return the page headers back to the "thesis" style

% Include the chapters of the thesis as separate files from the Chapters folder
% Uncomment the lines as you write the chapters

\include{Chapter1/Chapter1}
\include{Chapter2/Chapter2}

\include{Chapter3/Chapter3}
\include{Chapter4/Chapter4}

\include{Chapter5/Chapter5}
\include{Chapter6/Chapter6}

\include{Chapter7/Chapter7}

%----------------------------------------------------------------------------------------
%	THESIS CONTENT - APPENDICES
%----------------------------------------------------------------------------------------

% \appendix % Cue to tell LaTeX that the following "chapters" are Appendices

% Include the appendices of the thesis as separate files from the Appendices folder
% Uncomment the lines as you write the Appendices

% \include{Chapter3/Appendix3}
% \include{Chapter4/Appendix4}
% \include{Chapter5/Appendix5}
% \include{Chapter6/Appendix6}

%----------------------------------------------------------------------------------------
%	BIBLIOGRAPHY
%----------------------------------------------------------------------------------------

 \printbibliography[heading=bibintoc]

%----------------------------------------------------------------------------------------

\end{document}

%% file: Chapter1/Chapter1.tex
\chapter{Introduction} % Main chapter title
\label{chapter:intro} % Change X to a consecutive number; for referencing this chapter elsewhere, use \ref{ChapterX}

Microsoft Active Directory (AD) is a directory service that enables administrators to manage identity and access within Windows domain networks. It serves as the core identity and access management platform for most large enterprises, including up to 90$\%$ of Fortune 500 companies \citep{adexploit}. Because of this ubiquity, AD has been a primary target for attackers for more than a decade. Microsoft estimates that organisations using AD face a sustained high volume of brute-force attempts, with roughly 95 million AD accounts targeted every day \citep{report3}. Worse still, adversaries continue to expand their attack techniques \citep{mandiant2025mtrend,schroeder2021cpo}, in part because methods that succeed in one AD often transfer to others. Consistent with this, Mandiant reports that Active Directory exploitation is involved in about 90$\%$ of the breaches it investigates \citep{report2}.

The importance of defending AD has been widely recognised, and organisations have invested heavily in tools and hardening programs to secure it \citep{report3}. Yet enterprise AD environments remain fragile in practice. A survey by Enterprise Management Associates found that penetration testers achieved full compromise in up to 82$\%$ of assessed networks \citep{adexploit}. Furthermore, Microsoft reports that 88$\%$ of enterprises affected by a cyberattack were not following AD best practices, leaving numerous critical misconfigurations \citep{report1}. Consistent with these findings, the 2024 CyberEdge report shows that 81.4 $\%$ of surveyed organisations experienced at least one successful attack \citep{cyberedge2024cdr}. These statistics underscore the urgent need for more efficient defences for Active Directory.

In Active Directory, defenders can extract entities and permission relationships from directory data to generate an AD attack graph. Using this graph, they can align the environment with best practices by separating nodes into tiers based on roles and privileges, with the most critical assets isolated in a Tier 0 zone. With the adoption of automated enumeration tools such as BloodHound \citep{BloodHound}, attack graphs have become more accessible, and attack-graph analysis has become standard practice in the industry. This graphical perspective reveals previously overlooked permissions and misconfigurations and makes AD environments easier to manage. To emphasise the importance of attack graphs in defenders’ workflow, this thesis cites the well-known quote by John Lambert, Corporate Vice President and Security Fellow at Microsoft:
\begin{quote}
“Defenders think in lists. Attackers think in graphs. As long as this is true, attackers win.”
\end{quote}

Although this approach is now standard in many enterprises, attack paths to Tier 0 still appear. As reported in \citep{atkinson2025iap}, about 70$\%$ of users in surveyed networks have at least one path to accounts and resources in Tier 0. Without sustained defence and regular maintenance, AD environments drift from best practices during day-to-day operations due to admins’ configuration errors and sign-in behaviour that does not follow best-practice policy.

To combat this, defenders adopt a hardening process that removes vulnerabilities or deploys active defences (honeypots and decoys). The primary objective is to eliminate open attack paths from least-privileged nodes to the most critical assets in the Tier 0 zone. Traditional rank-based models assess risk from the characteristics of individual vulnerabilities (e.g., CVSS \citep{cvss}) while overlooking graph structure, such as how risks interact and compound along multi-step paths. Recent work instead treats AD defence as a combinatorial optimisation problem on the attack graph, with objectives that minimise the attack surface \citep{guo2022practical,zhang2024practical,goel2022defending}. However, current optimisation-based approaches still face practical gaps: there is a lack of models for placing active defences, for operating on time-varying attack graphs, and with robust, usable human-in-the-loop formulations. This thesis addresses these shortcomings with a suite of optimisation, game-theoretic, and decision-making models tailored to AD-style attack graphs.

% To secure these AD environments, IT professionals often follow a four-stage vulnerability management lifecycle: Discovery,  Prioritization, Remediation, and Monitoring. The process begins with Discovery, where enumeration tools like BloodHound are used to map the entire network and identify all potential security weaknesses. This initial step, however, often uncovers thousands of vulnerabilities, making the goal of fixing every single one impractical and time-consuming. This challenge makes the next stage, Vulnerability Prioritization, arguably the most critical. Instead of attempting to fix everything, the goal is to analyze the attack graph and identify the handful of vulnerabilities that pose the greatest risk. During the subsequent Remediation phase, an IT administrator is responsible for acting on this prioritized list, deciding which vulnerabilities to remediate and which defensive controls to apply, often balancing security risks against operational needs. This allows them to focus their efforts where they will have the most impact, while continuous Monitoring ensures the network stays secure against new threats.

\section{Background and Scope}

In this section, we first provide a detailed description of Active Directory and its best practice design. Next, we explain how the attack graph degrades over time in the absence of maintenance and hardening. Finally, we describe the standard hardening process and review state of the art models for this process.

\subsection{Active Directory and Best Practice}

Active Directory (AD) is a Windows network directory service that stores network objects and security policy information and enables administrators to organise and manage these policies. AD stores security data as objects—for example users, computers, security groups, organizational units, and group policy objects—and arranges them in a hierarchical structure similar to a directory tree. This directory-tree-like structure is why it is called a “Directory”. A key capability of AD is that it lets the system admin implement and manage identity and access control for every object in the network. In the AD access-management model, an object can be granted permissions or relationships that give it control over other objects. For example, a user or group may be granted the right to reset passwords (ForceChangePassword), modify attributes or entire objects (WriteProperty, GenericWrite), add members to a group (AddMember), or change security settings (WriteDACL). Depending on the role of the object in the network, system admins will grant appropriate permissions for it. Consequently, AD can be represented as a graph in which the nodes are objects and the edges capture the permissions and relationships between them.

To understand what leads to an insecure AD, we first outline the best-practice design used by security engineers when building AD networks and then demonstrate why every AD eventually drifts from this safe state and becomes vulnerable. Specifically, Microsoft recommends applying the tier model \citep{bestpractice, bestpractice2} to Active Directory. In this model, each node is assigned to exactly one tier based on the permissions required by its role. The ultimate goal of this model is to isolate higher-privilege nodes from lower-privilege nodes. Following Microsoft’s guidelines, nodes are typically grouped into three or more tiers:

\begin{itemize}
    \item Tier 0 (highest privilege): The most critical assets and admins with full control over AD (e.g., Enterprise Admin accounts, Domain Admin accounts, Domain Controller, etc.).
    \item Tier 1: Admins and systems that manage enterprise servers, business-critical applications, and high-value data (e.g., application servers, database servers, server admin, etc.).
    \item Tier 2 and beyond: other non-administrative user, workstations and assets. 
\end{itemize}

To avoid confusion between tier numbers and their privilege, we use "privilege" to refer to a tier: “high-privilege” refers to lower-numbered tiers (Tier 0 being the highest) and “low-privilege” refers to higher-numbered tiers. This tiering model also enforces two rules to ensure the separation between tiers: (1) higher-privileged nodes can control resources at their tier and in the tiers below, but not the other way around (access-control restriction); (2) credentials from a higher-privileged tier (e.g. Tier 0 Admin) must not be exposed to lower-tier systems (log-in restriction). This architecture strictly prevents any path from a lower-privilege tier to a higher one, and only “downstream” edges are allowed \footnote{In practice, downstream edges are created sparingly, only when required by role.}. In other words, a node may be controlled or accessed only by nodes that are located in a tier with higher privilege than itself. The tiering model and privilege access rules ensure that all objects in the network follow the "Clean Source Principle" \citep{cleansource}. This principle requires all security dependencies to be as trustworthy as the object being secured. It is based on the assumption that if an adversary can control any dependency of a target, they can ultimately control the target object itself. Figure \ref{fig:tier}.a shows an exemplary AD that follows these best-practice principles.

\begin{figure}[h]
\centering
  \includegraphics[width=1\linewidth]{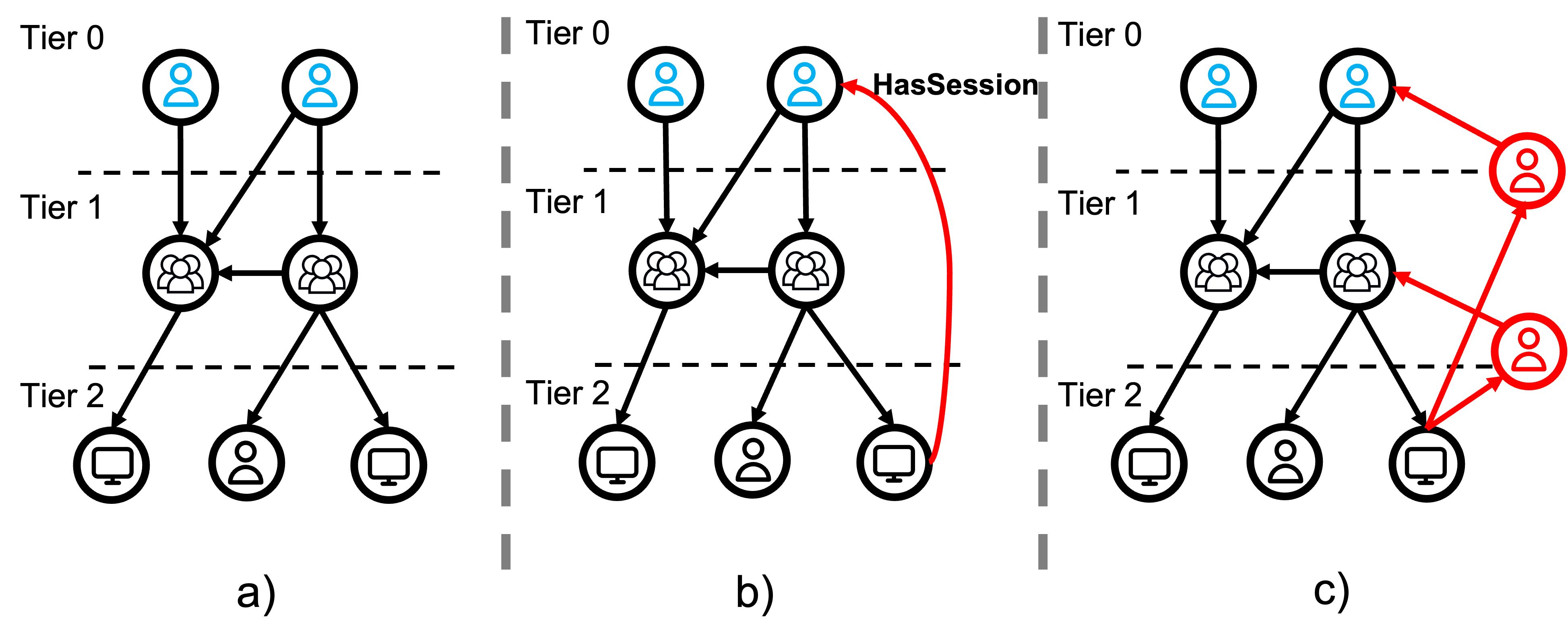}
  \caption{Misconfigurations in AD tiering model. (a) Secured AD follows best practice with only “downstream” edges. (b) Insecure AD with an “upstream” edge caused by a privileged admin signing in to a low-privileged computer. (c) Insecure AD caused by granting multi-tier permissions to an accounts (red nodes).}
  \label{fig:tier}
\end{figure}

% \begin{figure}
% \centering
% \begin{subfigure}{0.29\textwidth}
%     \includegraphics[width=\textwidth]{Chapter1/figure/bestpractice.png}
%     \caption{}
%     \label{fig:tier1}
% \end{subfigure}
% \hfill
% \begin{subfigure}{0.34\textwidth}
%     \includegraphics[width=\textwidth]{Chapter1/figure/insecuresign.png}
%     \caption{}
%     \label{fig:tier2}
% \end{subfigure}     
% \hfill
% \begin{subfigure}{0.33\textwidth}
%     \includegraphics[width=\textwidth]{Chapter1/figure/insecurerole.png}
%     \caption{}
%     \label{fig:tier3}
% \end{subfigure}     
% \caption{Misconfigurations in AD tiering model. (a) Secured AD follows best practice with only “downstream” edges. (b) Insecure AD with an “upstream” edge caused by a privileged admin signing in to a low-privileged computer. (c) Insecure AD caused by granting multi-tier permissions to an accounts (red nodes).}
% \label{fig:tier}
% \end{figure}

\subsection{Insecure Drift and Violations}

In practice, especially in large or complex enterprises, AD environments often drift from best practices during operations due to violations such as over-granting permissions or high-privilege admins signing in to lower-privilege machines. When these violations occur, they create \textit{cross-tier attack edges} \citep{nguyen2024adsynth} that short-circuit privilege boundaries and expose higher-tier systems. These violations and misconfigurations are hard to avoid in day-to-day operations and often result from administrator error driven by operational pressure \citep{verizon2024dbir}, overlooked legacy permissions and systems \citep{legacy}, users with overlapping roles, and users not following security policy \citep{violate}. Over time, these cross-tier edges will snowball into multiple attack paths and open more insecure pathways that allow an attacker to move laterally from low-privilege nodes to higher tiers.

Here are examples of how cross-tier edges can form and worsen the AD security posture. First, Figure \ref{fig:tier}.b shows a "upstream" cross-tier edge created when a high-privilege (Tier 0) admin signs in to a lower-privilege (Tier 2) machine. This happens when an admin user fails to follow restriction (2) (the log-in restriction). In Windows systems, user authentication leaves behind credential material, in the form of a hash or clear-text password \citep{mimikatz,Tilbury2017,Moore2015}, in the computer's memory. In this case, if the Tier 2 machine is compromised by the attacker, they can harvest those credentials and pivot to Tier 0. A second example, illustrated in Figure \ref{fig:tier}.c, involves a cross-tier edge caused by over-granting permissions where a user or group is given rights that span multiple tiers. This commonly occurs in large organisations when users hold multiple responsibilities requiring access across different tiers while administrators simultaneously fail to enforce the best practice that each node is assigned to a single tier \footnote{Microsoft's best practice recommends splitting these nodes into multiple accounts, each corresponding to a single tier\citep{bestpractice}}. An attacker starting from a lower privilege node can then use those accounts to move laterally to nodes in a higher privilege tier.

\subsection{AD Attack Graph}

As defenders, we aim to detect and prevent these cross-tier pathways from existing or going unmonitored in our AD environment. One way to defend against cross-tier misconfigurations is to visualise them with attack graphs. An attack graph is a graphical model that models chains of events or actions an attacker could exploit to compromise a system. The idea of representing Active Directory as an attack graph gained traction with the work of Dunagan et al. \citep{dunagan2009heat}, who first described these pathways as the "identity snowball attack" in the Active Directory. This type of attack leverages known vulnerabilities and cross-tier misconfigurations from an initially compromised user to gain access to other entities, uncover multiple attack paths in the network, and ultimately search for routes to the High Value Asset. Following this foundational work, research and industry focus on AD attack graphs grew significantly. In practice, system admins can construct an attack graph by extracting entities and permission relationships from the directory data. Notably, the development of commercialised attack graph enumeration tools like BloodHound \citep{BloodHound} has made attack graphs widely accessible, and attack-graph analysis has become standard practice in enterprise security.

\begin{figure}[h]
\centering
  \includegraphics[width=0.8\linewidth]{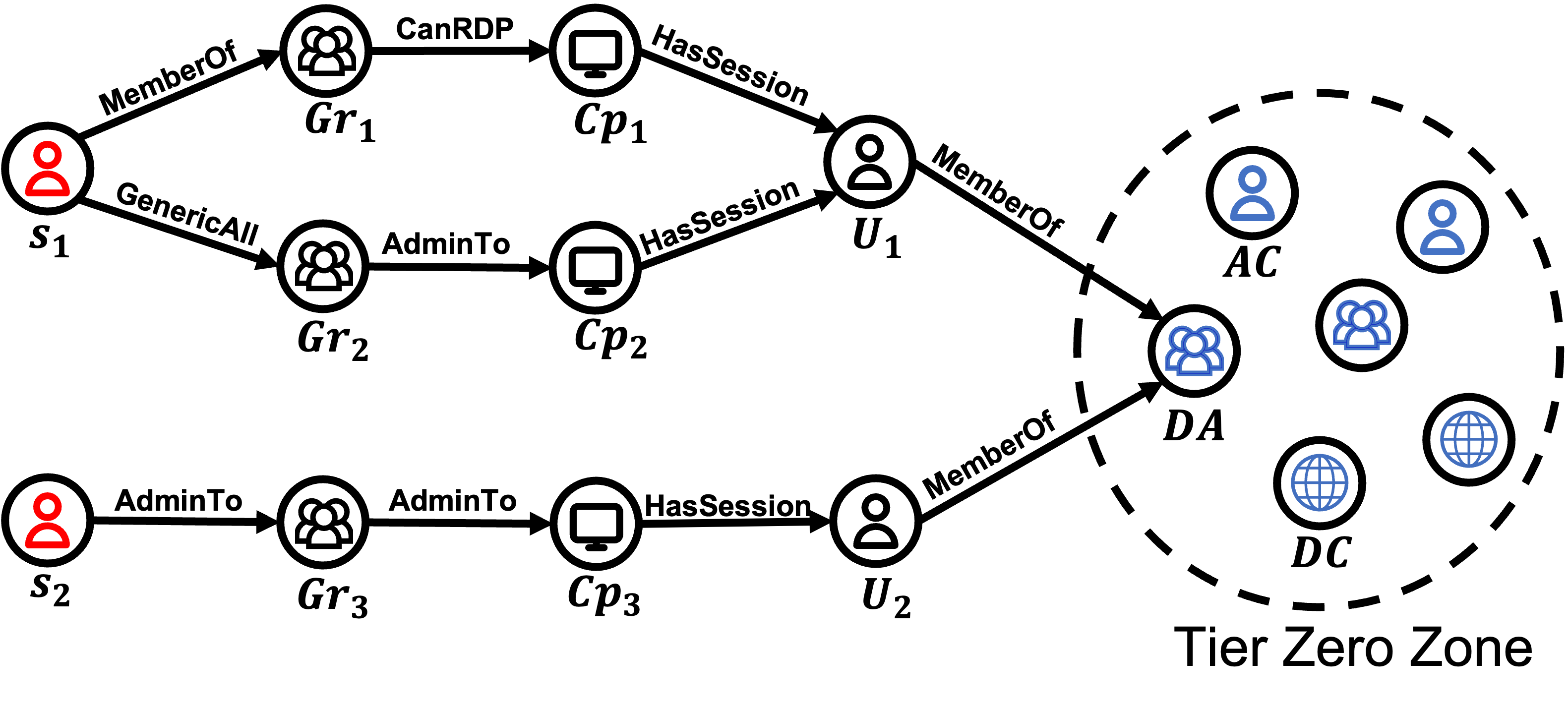}
  \caption{Simplified BloodHound attack graph with only Account ($U_x$), Computer ($Cp_x$) and Groups ($Gr_x$) nodes. The attack graph visualize attack paths from a low-privilege/compromise-prone source ($s_x$) to the Domain Admin (DA) node.}
  \label{fig:AD}
\end{figure}

Formally, an AD environment can be described as a directed graph $G(V, E)$ where the the node set $V$ represents entities such as users, computers, security groups, Organisational Units (OUs), Group Policy Objects (GPOs) and Domain. The directed edge set $E$ represents relationships that allow an attacker to move laterally from $u\in V$ to $v\in V$. These relationships include permissions or known-vulnerabilities for example local admin rights on an AD object (AdminTo), membership in a group (MemberOf), ability to establish Remote Desktop Access to a computer (CanRDP), full control over an AD object (GenericAll), and an active user session on a computer (HasSession), etc. Figure \ref{fig:AD} illustrates an attack graph for an insecure AD that shows a chain of potential privilege escalation from low-privilege nodes to Tier Zero, which is intended to be isolated from other network segments. Ultimately, as defenders, we seek to eliminate or minimise attack paths to Tier Zero and any nodes in this zone.

\subsection{Defending AD with Attack Graphs}

Even when an AD starts in a secure state, small misconfigurations can snowball into attack paths and eventually leave the whole environment vulnerable. To counter this, IT teams follow a vulnerability-management lifecycle to harden and maintain AD security. Figure \ref{fig:process} shows four main steps: Vulnerability Discovery, Prioritisation, Remediation, and Monitoring.
\begin{figure}[h]
\centering
  \includegraphics[width=1\linewidth]{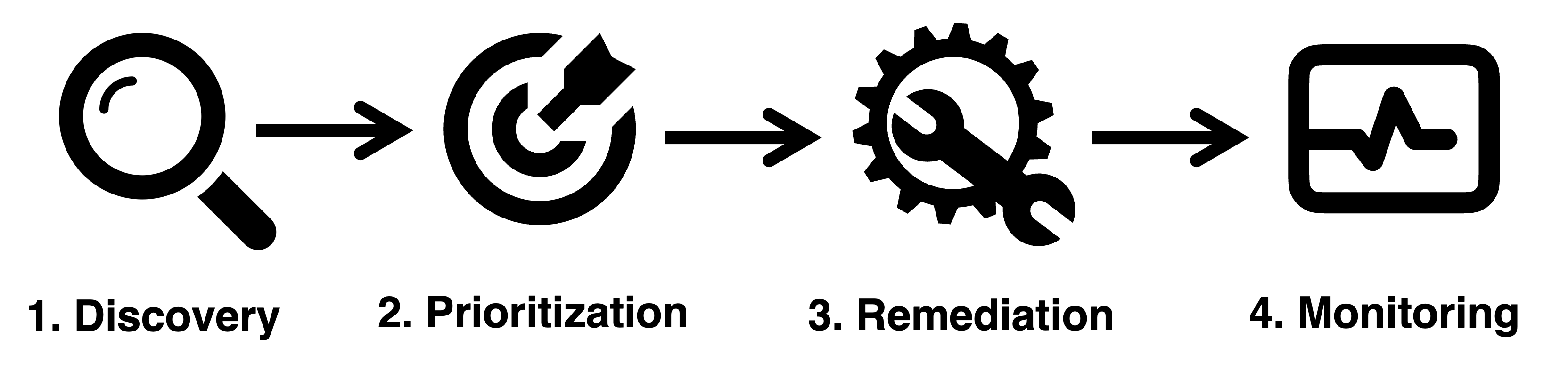}
  \caption{Vulnerability management lifecycle for Active Directory }
  \label{fig:process}
\end{figure}

The process begins with Discovery phase, where enumeration tools like SharpHound \citep{SharpHound} scan the whole network to list potential security weaknesses, and BloodHound maps the network into the attack graph. However, in enterprise networks, the scan typically reveals thousands of vulnerabilities, so fixing every single one is impractical and time-consuming. As a result, the next step is Vulnerability Prioritisation, which analyzes the attack graph to identify a list of few vulnerabilities that pose the greatest risk. During the subsequent Remediation phase, an IT administrator acts on this prioritized list, deciding which vulnerabilities to remediate and which defensive controls to apply (remove the vulnerability or deploy a honeypot), often balancing security risks against operational needs. Finally, the Monitoring phase ensures the environment does not drift back into an unsafe state and that new issues are caught early. If new risks are found, new findings are fed back into prioritization so the cycle continues.

In this thesis, we refer to "attack graph hardening" as the process involving two stages, prioritization and remediation. This thesis predominantly focuses on studying models for this subprocess. In hardening the attack graph, the main goal is to minimise the attack surface by either removing a limited number of edges or deploying a limited number of security devices (IDS or honeypots) to block open attack paths with as few resources as possible.

In the literature, there are two types of hardening, rank-based and optimisation-based. Rank-based methods \citep{Robbins2021patent,GoodHound,PlumHound,zhang2025rethinking} assign each vulnerability a score and a rank, and the ranked list is then given to the system admin, who remediates a limited number of risks. However, blocking attack paths to minimise risk under a limited budget is, at its core, a combinatorial optimisation problem. For example, disconnecting a high-privilege target from a low-privilege source in an attack graph resembles a minimum cut problem, which is a well-known combinatorial optimisation problem. This has given rise to more recent optimisation-based methods \citep{guo2022practical,zhang2024practical,goel2022defending}, which use optimisation algorithms to find a prioritized list of minimal vulnerabilities whose mitigation will reduce the attack surface the most.
In this setting, the hardening task typically involves optimising a security metric, such as minimising the shortest-path length to privileged assets \citep{guo2022practical} or minimising the number of nodes that can reach the Domain Admin \citep{zhang2024practical}, by modifying the attack graph (e.g., allocating honeypots or removing edges). 
Besides this, the hardening problem is subject to two constraints. First, most formulations impose a resource (budget) constraint on the number of modifications, since the implementation of each change requires human effort.
Second, because AD is a security critical system and changes to AD alter the permissions and access of entities, not every change is implementable, since some would disrupt operations or reduce usability.
In the current literature, there are three ongoing issues with optimisation based approaches that hinder their practical deployment.

First, most recent optimisation-based methods for AD attack graphs focus on removing edges (revoke vulnerabilities/permission) to disconnect attack paths to the Tier 0 zone. However, fully disconnecting all paths from low privilege nodes to Tier 0 is sometimes infeasible because certain edges are non-removable due to operational constraints or because the required cut set exceeds the available budget. In such cases, active defences such as honeypots or decoys can be deployed on nodes to monitor and intercept attackers along the remaining open paths. To diversify the defence inventory and reduce reliance on edge removal strategy alone, there is a need for models that guide system admin on where to place these active defences. At the time of writing, the literature lacks such allocation models for AD attack graphs.

Moreover, AD attack graphs are dynamic and evolve over time. Despite this, most existing work enumerates the graph once and optimises on a single static snapshot. As the graph continues to change, solutions derived from that snapshot deteriorate over time and become less effective or suboptimal. The graph changes for many reasons, from routine administrative updates, for example granting a user additional permissions to meet work requirements, to rapid shifts caused by user sign-in activity. On Windows systems, user authentication can leave credential material in memory (for example, password hashes or clear-text passwords), which adversaries can harvest for lateral movement. These motivate the need for a model that can deploy defensive measures in a (near-)optimal way on a dynamic, time-varying attack graph.

Another practical issue with optimisation-based models is that they often treat hardening as a standalone optimisation problem and ignore IT administrators’ decision-making. In practice, hardening follows an iterative two-step workflow: first, security proposes edges (vulnerabilities/permissions) in the attack graph to remove; second, the IT administration team manually reviews whether those changes are safe before implementation. Most optimisation-based models assume the IT administrator’s decision for every edge is known in advance, which is impractical because verifying each edge requires significant manual effort. Guo et al. \citep{guo2024limited} proposes an adaptive model that incorporates administrator decisions into hardening process, but the work is primarily theoretical and has two key drawbacks. First, it does not guarantee elimination of all attack paths. Second, its edge-by-edge querying does not scale to large AD networks, where IT admin typically define high-level security policies based on object properties. Consequently, there remains a lack of practical human-in-the-loop hardening models for AD.

\section{Research Objective and Challenges}
\label{sec:obj}
To tackle three issues identified above, this thesis develops a suite of novel, practical hardening models. Hence, the thesis pursues three research objectives aimed at resolving these objective:
\begin{itemize}
    \item \textbf{Objective 1:} Develop a model for allocating active defences on large-scale AD attack graph (e.g., honeypots and decoys). The model’s main contribution is a placement strategy that can effectively interdict attack paths and enable early detection.
    \item \textbf{Objective 2:} Design a hardening model for time-varying AD attack graphs that produces defences that remain effective as the graph evolves and prevents strategy deterioration over time.
    \item \textbf{Objective 3:} Develop a robust and practical human-in-a-loop hardening model for removing vulnerability in the large-scale AD attack graph. While prior works is largely theoretical, ineffective at removing attack paths, and does not scale at all, the proposed model should addresses these gaps.
\end{itemize}

In this thesis, we propose several defence models to directly achieve these objectives. Beyond the challenges above, which motivate our research objectives, there are additional issues we must address in algorithm design. First, most optimisation problems on attack graphs are computationally hard. Later, we will show that the core optimisation problems behind all of our models are intractable. Worse still, attack-graph-problem instances are often massive. Real-world AD graphs commonly contain hundreds of thousands of nodes and millions of edges. Inference time is another constraint we must consider when designing our algorithms. Slow inference is especially problematic during urgent incidents, when defenders need decisions within strict time limits.

\section{Thesis Outline and Contributions}

This section outlines the main content of the thesis and the corresponding contributions. To achieve the objectives in Section \ref{sec:obj}, this thesis develops a set of practical models for hardening the AD attack graph. While this chapter has introduced Active Directory, common misconfigurations, and the motivation for a more practical hardening approach. Chapter~\ref{chapter:review} will provide an in-depth overview of related research on attack graph hardening. The key original contributions of this thesis are as follows:

\begin{itemize}
    \item Chapter \ref{chapter:paper1} achieves Objectives 1 and 2. In this chapter, we present a novel honeypot placement model for Active Directory, framed as a Stackelberg game on the AD attack graph. We consider two attacker types: one that can observe honeypots and one that cannot. The defender’s goal is to jointly minimise the success probability of both attacker types. We prove the defender’s honeypot-placement problem is $\mathcal{NP}$-hard against an observable attacker and $\mathcal{W}[1]$-hard against a non-observable attacker. We formulate the problem as a mixed-integer program and use it to produce high-quality honeypot placements on very large AD graphs. We also extend the model to time-varying attack graphs and select placements that perform well across all snapshots. To handle the dynamic case efficiently, we develop two heuristics based on voting and clustering. The clustering method builds defences from representative snapshots, runs quickly on large graphs, and achieves results close to optimal.
    The preliminary version of this Chapter have been accepted for publication at \textit{the 22th International Conference on Autonomous Agents and Multiagent Systems (AAMAS'23)} as an extended abstract  under the title "Near optimal strategies for honeypots placement in dynamic and large active directory networks" \citep{ngo2023near}. The full content of the this Chapter accepted for publication at \textit{the IEEE International Conference on Computer Communications (INFOCOM'24)} under the title "Catch Me if You Can: Effective Honeypot Placement in Dynamic AD Attack Graphs" \citep{ngo2024catch}.

    \item Chapter \ref{chapter:paper2} achieves Objectives 1 and 2. In this chapter, we take a different approach to the honeypot allocation problem by modelling the time-varying AD attack graph as a temporal directed graph. In this model, edges are active only during specific time intervals, and a temporal attack path is a sequence of edges with non-decreasing timestamps. This captures a realistic time dimension of a persistent attacker behaviour, in which adversaries maintain a foothold and move when opportunities arise. We introduce an novel response time metric, defined as the duration from the first decoy trigger to compromise of the DA. The defender aims to choose honeypot placements that disconnect the source and the target when while maximising response time to ensure the early detection of the attacker. We prove that the defender’s optimisation problem is $\mathcal{NP}$-hard. To solve it, we design an Evolutionary Diversity Optimisation heuristic and address the fitness evaluation bottleneck with two advances: a Dijkstra-based earliest arrival algorithm that accelerates fitness computation on graphs with many static edges, and a lightweight surrogate that evaluates solutions on a set of important attack paths rather than the entire graph. We prove that the surrogate evaluation converges to feasible solutions and show experimentally that the approach scales well.
    The work in this Chapter has been accepted for publication at \textit{the Genetic and Evolutionary Computation Conference (GECCO'24)} under the title "Optimizing Cyber Response Time on Temporal Active Directory Networks Using Decoys" \citep{ngo2024optimizing}.

    \item Chapter \ref{chapter:paper3} achieves Objective 3. In this chapter, we introduce a human-in-the-loop edge-removal model for defending AD attack graphs, called Adaptive Path Removal (APR). In this model, a wizard proposes an attack path in each step and presents it as a set of multiple-choice options to the IT admin (system admin). In response, The IT admin then selects one edge from the proposed set to remove. This process continues until all attack paths from source $s$ to target $t$ are disconnected or a preset query budget is reached. Our goal is to design an adaptive wizard policy that minimises the expected number of queries. To our knowledge, this is the first path-based adaptive query framework that lets defenders choose the removal edge through multiple choice. Compared with prior adaptive defence work for Active Directory, our method can then guarantee elimination of all $s$ to $t$ attack paths when sufficient budget is available. We prove that finding an optimal policy for APR problem is $\#\mathcal{P}$-hard. To overcome this, we present an exact algorithm, an approximation algorithm, and a set of scalable heuristics. Among our heuristics, Dynamic Programming with Restrictions delivers the strongest results by building on the exact and approximate methods. A reinforcement learning heuristic is also promising. While it does not perform as well as DPR, it can be trained offline and then used with low inference time, which makes it suitable for larger graphs and time sensitive deployments.
    This work has been accepted for publication at the \textit{34th International Joint Conference on Artificial Intelligence (IJCAI’25)} as “Adaptive Wizard for Removing Cross-Tier Misconfigurations in Active Directory” \citep{ngo2025adaptive}.

    \item Chapter \ref{chapter:paper4} achieves Objective 3. In this chapter, we introduces another human-in-the-loop edge-removal framework for defending AD attack graphs, called Feature based Adaptive Connectivity Test (F-ACT). While prior work often oversimplifies administrator decisions by assigning each edge an independent fixed removal probability. F-ACT instead uses security context features to model the decision process. At query time, the wizard presents an edge together with its feature vector, and the administrator decides to cut or keep it. That decision is then applied to all edges that share the same feature pattern. The interaction continues until all attack paths are eliminated or we determine that elimination is impossible. The objective is to learn an adaptive query policy that minimises the expected number of queries and therefore reduces IT admin effort. We prove that finding an optimal policy for F-ACT problem is $\#\mathcal{P}$-hard. Since the adaptive nature of F-ACT problem class make it fit naturally to the RL framework, we proposed a RL solution for this problem. To handle the large search space and sparse rewards, we propose RL4FT algorithm with three novel components: a policy-agnostic self improvement mechanism, a policy invariant reward shaping scheme, and a custom prioritised experience replay. By design, RL4FT guarantees improvement upon any given reference policy, and we show that it significantly outperforms state of the art baselines.
    The work in this Chapter is submitted for review at \textit{the 40th AAAI Conference on Artificial Intelligence (AAAI'26)}.

    \item Finally, chapter \ref{chapter:end} provides a conclusive summary of the thesis and provide several potential direction for future work. 
\end{itemize}

% \begin{itemize}
%     \item Discovery: Scan the whole network and list every permission and known vulnerability. In a BloodHound workflow, run SharpHound (and AzureHound for Microsoft Entra ID) to collect AD entities, relationships, and properties, then use the data to build the attack graph.
%     \item Prioritization: Identify the few weaknesses that pose the highest risk, since fixing everything is not feasible. In an attack graph workflow, this often involves selecting a small set of nodes or edges and defining specific defensive actions for them (for example, removing edge or allocate honeypot on nodes) to break critical paths and reduce the attack surface.
%     \item Remediation: Implement the mitigation for the vulnerability identified by the prioritization stage. This stage, involved decision making, system admin and security engineer will colaborate access the security risk of not implement the measure and the usability potential operational disruption cause by implement it. 
%     \item Monitoring: ensures the environment does not drift back into an unsafe state and that new issues are caught early. If new risk araise, New findings are fed back into prioritization so the cycle continues.
% \end{itemize}

%% file: Chapter2/Chapter2.tex
\chapter{Literature Review} % Main chapter title
\label{chapter:review} % Change X to a consecutive number; for referencing this chapter elsewhere, use \ref{ChapterX}
% \section{Active Directory and Attack Graph}
Understanding the origins of insecure AD and the need for an efficient hardening model for AD attack graphs. This literature review explores existing defence models, moving from general graph application to AD-focused attack graphs. We begin by introducing attack graphs and their applications in network security (Section \ref{sec:lit1}). Next, we survey a broad range of defence models on general graphs, often formulated as Network Security Games. We then review defence models for attack-graph hardening (Section \ref{sec:lit2}). Finally, we examine hardening models for AD attack graphs, across both industry and research, which are most closely related to this thesis (Section \ref{sec:lit3}).

\section{Attack Graph and Application}
\label{sec:lit1}
Attack graphs have a long history in security which model chains of events or actions an attacker can exploit to compromise a system. It was first formalised by Phillips and Swiler as state-enumeration graphs \citep{phillips1998graph,jha2002two,swiler2001computer}. Since then, substantial work has focused on improving scalability and usability, and several types of attack graphs have been proposed, including dependency attack graphs \citep{ammann2002scalable,jajodia2005topological,ou2005mulval}, Bayesian attack graphs \citep{frigault2008measuring,munoz2017exact}, among many others. In a later survey, Lallie et al. \citep{lallie2020review} examined more than 180 studies on attack graphs and attack trees and identified over 90 distinct self-described definitions across the literature.

Attack graphs are primarily used to help system administrators and security engineers visualise and analyze attack paths, and to support decision-making for secure configuration and network hardening. For example, attack graphs can be used to identify and prioritise critical assets in the network. Sawilla et al. \citep{sawilla2008identifying} demonstrate this idea with AssetRank, which integrates public vulnerability data with dependency attack graph semantics to compute an importance score for each vertex and thereby identify critical assets. Attack graphs can also be correlated with real-time alerts from intrusion detection systems (IDS) to infer attacker intentions and predict their next possible exploits, Wang et al. \citep{wang2006using} demonstrate how combining attack-graph information with IDS alerts can recommend probable attacker next actions.

One of the most important applications of attack graphs is to prioritise vulnerabilities, which supports defenders’ decision-making in system hardening. Because resources and time are limited, defenders aim to reduce the attack surface with as few actions as possible. In practice, defenders have two approaches to hardening the network: security-device placement and vulnerability removal. For security-device placement, attack graphs help determine where controls are most effective. Examples include optimal allocation of IDS sensors \citep{noel2008optimal} and the deployment of honeypots \citep{durkota2015optimal,durkota2015approximate,durkota2019hardening,kiekintveld2015game}. For vulnerability removal, two well-studied approaches are rank-based hardening and optimization-based hardening. Rank-based hardening assigns each vulnerability a priority score, based on its exploitability and reachability to critical assets,which administrators then use to prioritise patching \citep{mehta2006ranking,ingols2006practical}. Optimisation-based hardening uses optimisation algorithms to select a set of vulnerabilities to remove so as to minimise the attack surface and the overall defence cost \citep{wang2006minimum,albanese2012time}. In this thesis, we particular interesting in using attack graph for hardening the network, which will be continued review in the next section.

\section{Defence Models for Graphs}
\label{sec:lit2}
Network hardening can be conceptualised as a strategic conflict on an attack graph, in which a defender seeks to prevent an attacker from reaching critical assets. Network Security Games (NSGs) are closely related and motivate ideas such as budgeted defence, path interdiction, and planning against an adaptive adversary. We next review security-focused models for attack-graph hardening from the literature, with particular emphasis on vulnerability prioritisation methods

\textbf{Network Security Games.} Network Security Games (NSGs) is a branch of security games \citep{pita2009security,tambe2011security,kar2016trends} that combine game theory and graph theory to model defender–attacker interplay on networked infrastructure. In these models, the defender allocates limited resources to protect the network, while the attacker targets nodes to maximise damage. The framework is versatile and has been applied across domains, with attacker goals and defender objectives tailored to each setting.
A common line of work studies urban attack interdiction \citep{jain2013security,jain2011double,tsai2010urban}, typically framed as a Stackelberg game in which the attacker chooses a path to a target $t$ from a set of potential targets $T$ to maximize the damage payoff $\mathcal{T}(t)$, and the defender allocates resources (e.g., checkpoints) to interdict paths so as to minimise $max_{t\in T} \mathcal{T}(t)$.
The interdiction idea extends to smuggling prevention \citep{guo2016optimal}, where the attacker seeks to maximise network flow and the defender blocks nodes or edges to reduce it.
Another significant line addresses contagion-style attacks that spread over a network \citep{bai2024resource,bai2021defending,gan2015security,li2020defending,adiga2016delay,tsai2012security}. Here, the attacker compromises an initial node $u$, causing damage that cascades to neighbours, and the defender allocates resources to minimise the resulting loss.
NSGs have also been applied in counter-terrorism to disrupt terrorist communication. The goal is not to block a single physical path but to degrade the network’s ability to coordinate by disrupting communication links. Defender success is measured by reducing a key network metric, such as the damage potential of a critical subgraph \citep{wang2016computing,mutzari2021coalition,guo2016coalitional} or overall connectivity (e.g., via the inverse geodesic length) \citep{aziz2017weakening,aziz2018defender,gaspers2019optimal}. Across many of these settings, the defender’s optimisation problem is computationally challenging—often intractable—due to the hardness of solving the underlying Stackelberg game \citep{conitzer2006computing}. Consequently, much of the literature focuses on scalable approximations and algorithmic heuristics that perform well on large, real-world graphs.

\textbf{Attack-Graph Hardening.} Focusing on attack graphs, hardening is commonly formulated as an optimisation problem in which the defender seeks to minimise the attack surface under a limited defence budget. These problems are often intractable because of the structure of common attack graph models. In state‑enumeration attack graphs, Sheyner et al. \citep{sheyner2002automated} prove that disconnecting all attack paths by removing exploit nodes is NP‑complete. Wang et al. \citep{wang2006minimum} and Albanese et al. \citep{albanese2012time} extend this line to interdependent exploits, where prerequisite constraints couple fixes, i.e., disabling a target exploit may require disabling its prerequisites, and removing those prerequisites can eliminate the target indirectly. Other papers cast hardening as multi-objective optimisation, jointly minimising attack surface and defence cost. For example, Wang et al. \citep{wang2013exploring} propose heuristics to jointly optimize security metrics (e.g., integrity loss, privilege escalation, service availability) and the defence cost. Khouzani et al. \citep{khouzani2019scalable} also jointly minimise defence cost and the probability of a successful attack on probabilistic attack graphs. 
A complementary line uses game-theoretic formulations to hardening on attack graph problem as a NSGs. For instance,  Durkota et al. \citep{durkota2015optimal,durkota2015approximate,durkota2019hardening} and Kiekintveld et al. \citep{kiekintveld2015game} formulated honeypot placement as NSG, where the defender introduces honeypots to raise attacker risk and cost and the attacker, uncertain about exact placements, plans under detection risk. Nguyen et al. \citep{nguyen2017multi} similarly frame hardening as a multi-stage security game on Bayesian attack graphs, using particle-filter belief updates with value-propagation/sampled-activation heuristics to allocate countermeasures under uncertainty. 
Finally, rather than selecting specific nodes to remove or harden, several works score vulnerabilities and return a ranked list for administrators, who then decide which edges to remove or which controls to apply. For example, Sawilla and Ou \citep{sawilla2008identifying} and Ingols et al. \citep{ingols2006practical} score vulnerability vertices by exploitability and by the number of hosts remaining reachable after a removal to generate ranked lists of high-priority fixes.

\textbf{Optimization-based Defense in Cyber Security.} Apart from defending networks via attack-graph optimization, a large body of work examines cybersecurity through the broader lens of optimization \citep{wang2018survey}. 
In the domain of optimal firewall configuration, \citet{katic2007optimization} address the firewall policy refinement problem by modeling rule sets as logical structures and applying static merging and redundancy-removal algorithms to eliminate policy conflicts. Similarly, \citet{hamed2006dynamic} optimize packet-filtering rule ordering to minimize matching latency, formulating the task as a dynamic traffic-aware sequencing problem solved via probabilistic matching heuristics. Tan et al. \citet{tan2016cyber} investigate the cyber password maintenance problem, utilizing a cost-benefit optimization framework to find an equilibrium between security risks and user dissatisfaction costs. 
In the context of system reliability, Shen et al. \citet{shen2012survivability} evaluate the survivability of Wireless Sensor Networks by formulating a zero-sum stochastic game to calculate transition probabilities and maximize the Mean Time To Failure. Furthermore, Law et al. \citet{law2014security} optimize the protection of Smart Grid Automated Generation Control systems against data injection attacks by combining risk management analysis with game-theoretic interdiction to identify critical system vulnerabilities.
There are several fundamental characteristics that distinguish the defense of Active Directory (AD) from other above-mentioned security optimization problems. Firstly, AD defense is fundamentally a combinatorial optimization problem on a graph, as the AD environment is naturally represented as an attack graph where defensive decisions are expressed as discrete actions. Secondly, the hardening of these graphs is centered on the Clean Source Principle \citep{cleansource}, where the objective is typically formulated as a path interdiction or reachability problem aimed at removing or minimizing attack paths to the Tier 0 zone. Thirdly, the AD attack graph and the defense workflow itself exhibit unique characteristics, such as a topology that changes dynamically over time and the requirement for human verification of each action. These factors required a new model and solution for the defense of the AD.

\section{Defence Model for AD Attack Graph}
\label{sec:lit3}
The study of AD attack graphs was pioneered by the seminal work of Dunagan et al. \citep{dunagan2009heat}. They introduced the identity snowball attack model, a concept that was later expanded and commercialised by Bloodhound \citep{BloodHound}. In the context of enterprise defence, the defender’s primary objective is to minimise accessibility from low privilege accounts to Tier Zero or Domain Admins (DA) nodes.

\textbf{Rank-based models.} BloodHound not only commercialises attack graph enumeration, it also integrates a vulnerability prioritization module that highlights risky edges as choke points~\citep{bloodhound-enterprise}. Although the full method is not disclosed, their patent briefly describes the algorithm at a high level \citep{Robbins2021patent}. The chokepoint score of an edge is the number of assets that can reach Tier Zero by traversing that edge. Concretely, for each directed edge $(u, v)$ where $v$ can reach a Tier Zero node, the score is the number of nodes in the attack graph that can reach $u$. In practice, practitioners first run BloodHound to enumerate the graph, then they receive a ranked list of edges by this choke point metric. Edges that serve as convergence points for many assets on the way to Domain Admins appear at the top and are labeled with higher severity. Using this ranked list, security engineers and system administrators proceed from the top, weighing security risk against operational disruption for each edge before deciding whether to remove it. Also in industry setting, GoodHound \citep{GoodHound} and PlumHound \citep{PlumHound} are popular open-source AD security tools that identify choke points by counting how many enumerated shortest paths traverse each edge. In a research setting, Zhang et al. \citep{zhang2025rethinking} extend this notion by redefining choke points in terms of attack paths, treating them as locations through which most paths to high-value targets pass.

\textbf{Optimization-based models.} Prior works often formulate AD defence as an combinatorial optimisation problem that improves a security posture measure by modifying the attack graph, for example by removing unsafe edges. The objective is to optimise a security-related metric, such as minimising the shortest path length to privileged assets \citep{guo2022practical}, minimising the number of nodes that can reach the Domain Admin \citep{zhang2024practical}, etc. 
Most formulations impose two constraints. First, a budget constraint on the number of actions, because each change is time-consuming to implement in practice. Second, an implementability constraint on proposed actions, because changes must be verified by the system admin to avoid unnecessary disruption of critical business operations.

One common way to enforce implementability is to assume it is known in advance. This leads to a non-adaptive model that assumes a set of removable edges and solves a single pass optimisation to produce a global hardening plan. Within this class, \citep{guo2022practical,guo2023scalable,zhang2023oracle} formulate defence as shortest path interdiction, which find a set of $B$ number of edges so as to maximise the length of the shortest paths from low privilege sources to DA. Because AD graphs can contain hundreds of thousands of nodes and millions of edges, they develop fixed parameter and linear programming methods that exploit structural features such as small treewidth and low budgets to scale to large graphs. To further improve scalability within this line, Zhang et al. \citep{zhang2023oracle} introduce an oracle-based algorithm that can scale on practical-sized AD graphs. In another defence formulation, Zhang et al. \citep{zhang2024practical} frame defence as minimising Domain Admin reachability by removing edges to minimize the number of low privilege nodes that can reach Domain Admin nodes. They further introduce an anytime algorithm that returns a usable solution at any time. Another line of work models defence as a Stackelberg game between attacker and defender on a stochastic graph \citep{goel2022defending,goel2023evolving}. In these work, each edge has a failure probability, the attacker adapts by trying alternate unexplored edges when one fails, and the defender removes edges to minimise the probability that the attacker reaches Domain Admins. In these papers, the attacker policy is learned with reinforcement learning, and the corresponding defence is constructed with an evolutionary diversity optimisation approach. While the works above optimise to select an actionable set of edges to remove, Zheng et al. \citep{zheng2011active} focus on learning implementability of actions. They use active learning to predict, for each edge, whether a proposed removal would be accepted, and the resulting labels can help the nonadaptive optimisation models.

Nonadaptive models are often impractical in real security pipelines. They require the implementability of every edge to be known in advance, yet in practice this would demand that a system administrator verify each edge by hand. To address this limitation, an adaptive approach is introduced that proceeds iteratively and takes system administrator decisions into account. Within this setting, the seminal study by Dunagan et al. \citep{dunagan2009heat} introduced Heatray. In their formulation each edge has a latent implementation cost that reflects the difficulty of removing it, and administrators are assumed to prefer lower cost changes. At each round it solves a directed sparsest cut that minimizes total edge cost, presents the suggested cut to the administrator for a keep or remove decision, and uses those decisions to train a support vector machine that updates the edge costs for the next round. 
Guo et al. \citep{guo2024limited} introduce an adaptive defence model called the Limited Query Graph Connectivity Test. A proxy algorithm proposes one edge at a time to the system administrator who will decide whether to remove or retain it. Given a low privilege source $s$ and a high privilege target $t$, the model aims to find a cut that separates $s-t$. If the system admin retains so many edges that no feasible cut remains, the procedure instead returns a s-t path. The design goal is a optimal query policy for the proxy which minimises the expected number of queries. Finally, the summary of the related work in AD hardening is provided in Table \ref{tab:hardening-comparison}

\clearpage
\begin{landscape}
\begin{table*}[p]
\centering
\caption{Summary of AD hardening tools and researches.}
\label{tab:hardening-comparison}
\smallskip
\setlength{\tabcolsep}{3pt}
\renewcommand{\arraystretch}{1.04}

\resizebox{0.79\linewidth}{!}{%
\begin{tabular}{L{4.7cm}|C{1.6cm}C{1.6cm}C{1.6cm}C{1.6cm}C{1.8cm}|L{13cm}}
\hline
\textbf{Study/Tool} & \textbf{Optimis-ation} & \textbf{Adaptive} & \textbf{Dynamic} & \textbf{Defence Type} & \textbf{Hardness} & \textbf{Model Summary} \\
\hline
BloodHound Enterprise \citep{bloodhound-enterprise} & \xmark & \xmark & \xmark & edge & -- &
Give each edge a "chokepoint" score by how many nodes can reach Tier~0 via that edge. Afterward provide the system admin with a ranked list of edges for removal. \\
\hline
GoodHound \citep{GoodHound}, PlumHound \citep{PlumHound} & \xmark & \xmark & \xmark & edge & -- &
Score each edge by how many shortest paths to Tier~0 traverse it. Afterward provide the system admin with a ranked list of edges for removal.  \\
\hline
\citep{zhang2025rethinking} & \xmark & \xmark & \xmark & edge & -- &
Score each edge by the number of nodes that can reach Tier Zero through it. Afterward provide the system admin with a ranked list of edges for removal. \\
\hline
\citep{dunagan2009heat} & \cmark\textsuperscript{1}  & \cmark & \xmark & edge & -- &
Help the system admin clean up misconfigured edges by repeatedly solving a cost weighted directed sparsest cut to propose edges for removal, while an SVM learns and updates per edge removal costs. \\
\hline
\citep{zheng2011active} & \xmark & \cmark & \xmark & edge & -- &
Use active learning to predict how likely the system admin is to remove each edge when proposed (predict edge cost). \\
\hline
\citep{guo2022practical},\citep{guo2023scalable},\citep{zhang2023oracle} & \cmark & \xmark & \xmark & edge & $W[1]$-hard &
Maximise shortest-path length to Tier Zero nodes by removing edges (shortest-path interdiction). \\
\hline
\citep{zhang2024practical} & \cmark & \xmark & \xmark & edge & NP-hard &
Minimise the number of low-privilege nodes that can reach Tier Zero nodes by removing edges. \\
\hline
\citep{goel2022defending}, \citep{goel2023evolving} & \cmark & \xmark & \xmark & edge & $\#\mathcal{P}$-hard &
Model defence as a Stackelberg game between attacker and defender on a stochastic graph. The attacker alters strategy when an fail to exploit an edge fails, the defender removes edges to minimize the chance of reaching DA. \\
\hline
\citep{goel2024optimizing} & \cmark & \xmark & \cmark\textsuperscript{2} & edge & $\#\mathcal{P}$-hard &
As above \citep{goel2022defending}, \citep{goel2023evolving}, but extended to time-varying graphs. \\
\hline
\citep{guo2024limited} & \cmark & \cmark & \xmark & edge & $\#\mathcal{P}$-hard &
Adaptive querying to determine whether $s$--$t$ can be separated in the attack graph. The policy minimizes the expected number of queries. \\
\hline
\rowcolor{rowgray} 
\textbf{Chapter~\ref{chapter:paper1}} & \cmark & \xmark & \cmark & node & $W[1]$-hard / NP-hard &
Stackelberg honeypot placement to interdict attack paths (static and time-varying graph). Jointly minimizing unmonitored shortest-path routes and the number of low-privilege nodes reaching Tier 0. \\
\hline
\rowcolor{rowgray} 
\textbf{Chapter~\ref{chapter:paper2}} & \cmark & \xmark & \cmark & node & NP-hard &
Stackelberg honeypot placement on temporal attack graph. Maximize cyber response time and disconnect $s$--$t$. \\
\hline
\rowcolor{rowgray} 
\textbf{Chapter~\ref{chapter:paper3}} & \cmark & \cmark & \xmark & edge & $\#\mathcal{P}$-hard &
Query the system admin with attack paths selected adaptively for removal. The policy aims to minimize expected queries and disconnect $s-t$. \\
\hline
\rowcolor{rowgray} 
\textbf{Chapter~\ref{chapter:paper4}} & \cmark & \cmark & \xmark & edge & $\#\mathcal{P}$-hard &
Adaptively query the attack graph to decide whether $s$--$t$ can be separated. Edge features guide the system admin’s decisions, and the policy minimizes the expected number of queries. \\
\hline
\end{tabular}}
\vspace{2pt}

\noindent\footnotesize
\begin{minipage}{0.78\linewidth}
\raggedright\fontsize{6}{7}\selectfont
\textsuperscript{1} This model is not formulated as an overall optimization problem. It uses a sparsest cut to guide the proposed edge set\\
\textsuperscript{2} Chapters \ref{chapter:paper1} and \ref{chapter:paper2} had already been published at the time of this paper \citep{goel2024optimizing}’s publication. So we are the first to consider time-varying AD attack graph
\end{minipage}
\end{table*}
\end{landscape}
\clearpage

%% file: Chapter3/Chapter3.tex
% Chapter Template

\chapter{Effective Honeypot Placement in Static and Dynamic Active Directory Networks} % Main chapter title
\label{chapter:paper1} % Change X to a consecutive number; for referencing this chapter elsewhere, use \ref{ChapterX}

%----------------------------------------------------------------------------------------
%	SECTION Background
%----------------------------------------------------------------------------------------

In this chapter, we model a Stackelberg game between an attacker and a defender on large, dynamic Active Directory (AD) attack graphs. The defender deploys honeypots to prevent the attacker from reaching high-value targets. We consider two attacker types: a simple attacker who cannot observe honeypots and a competent attacker who can. We formulate this as a bi-objective optimization problem where we propose a mixed-integer programming (MIP) formulation for defender problem. Unlike prior work only focus on small and static attack graphs, AD graphs typically contain hundreds of thousands of nodes and edges and constantly change over time. We observed that the optimal blocking plan for static graphs performs poorly in dynamic graphs. To address this, we redesign the formulation to combine $m$ MIP instances which we called $dyMIP(m)$ to produce a near-optimal blocking plan over time. To scale to many dynamic instances, we propose a cluster-based method that selects the m most representative graph snapshots, and run dyMIP(m) on them to compute defence strategies. This is motivated by the observation that a single defence strategy remains effective across neighbouring snapshots, so solving the MIP only on representative instances is sufficient. Furthermore, we establish a lower bound on the optimal dynamic blocking strategy and show that dyMIP(m) achieves results close to optimal across a range of AD graphs under realistic conditions.

\section{Introduction} 

In this chapter, we study a new active-defense model for securing Active Directory (AD) using honeypots. In our model, the defender places honeypots on nodes of the AD attack graph, and an attack campaign is considered unsuccessful if the attacker enters or interacts with any honeypot. In Active Directory, an attacker can enumerate vulnerabilities using SharpHound \citep{SharpHound} and use the accompanying tool BloodHound~\citep{BloodHound} to generate the attack graph. By default, BloodHound suggests the shortest attack path to the highest-privilege node, which in this chapter we assume is Domain Admin (DA).

Active defense with honeypots is not new. However, placing honeypots in an AD attack graph brings two challenges that have not been studied thus far. The first challenge is the scale of the attack graph. An AD attack graph can include thousands of nodes and hundreds of thousands of edges, and even a small or medium-sized organization can produce millions of attack paths. The second challenge is how to derive a decoy strategy that remains effective while the graph changes over time. One of the major sources of changes in the AD graphs is users' activities such as logging on and off a workstation. In an AD attack graph, these dynamics are represented by a special type of edges called HasSession edges \citep{BloodHound}. HasSession edges are added to the graph when a user signs on to a computer and has their credential stored in the computer memory. These edges will stay online until being removed from the graph when the user signs off from the computer after a period of time.

Our solution is designed to be effective against attackers with different tactics and techniques. In our proposed model, we consider two types of attackers: a \textbf{simple attacker}, who is unable to detect any honeypots, and a \textbf{competent attacker}, who employs sophisticated detection tools to obtain complete visibility of all honeypots. In practice, it is difficult to ascertain the level of competence of an attacker. However, we can assume the probability of our network being attacked by one of these attackers based on the observation from previous attacks and/or analysis of the attacker's tactics using databases such as MITRE ATT\&CK~\citep{MITRE}.

Our first proposal is to use a mixed-integer programming (MIP) method to solve the honeypot allocation problem for mixed-type attackers. For our problem, we observed that the Linear Programming approach can be very efficient in solving AD graphs due to the exploitability of MIP on the unique tree-like structure of AD graphs~\citep{guo2022practical}. 

The above solution, however, assumes a static graph and performs poorly when the graph dynamically changes. Our approach to addressing the graph dynamic problem is to simulate all possible variations of the HasSession edges and design a placement solution that works best for all of these graph instances. This solution would only work for very small AD systems with a small number of HasSession edges as the number of graph instances grows exponentially with the number of these edges. Instead, we expand the MIP formulation to include $m$ samples of the AD graph. We call it mixed-integer programming formulation for dynamic graphs (dyMIP(m)). As the run-time of dyMIP($m$) grows for $m \gg 1$, we propose two heuristics: the \textit{voting-based heuristic} and the \textit{clustering-based heuristic}. These allow dyMIP($m$) to approximate the optimal strategy very quickly in these cases. While \textit{voting-based heuristic} is a very natural heuristic based on the idea of running smaller batches of samples on dyMIP(m), hence reducing the $m$ value. \textit{Clustering-based heuristic} approach hypothesizes that there is a number of good-quality samples in our graph sample that could be used to improve the solution for our dyMIP(m) algorithm. The novel idea of \textit{clustering-based heuristic} is to solve combinatorial in each individual (static) graph to derive features and then apply a clustering algorithm to improve the quality of the sample for the general problem on a dynamic graph. The key contributions of this chapter are:

\begin{itemize}
\item We formulate a general honeypot placement problem in AD graphs considering 2 types of attackers. We show that the static graph version of our problem is already NP-hard when attackers are unable to detect any honeypots and $W[1]$-hard when attackers have complete visibility of all honeypots.
% \item We proposed a studies on the problem of jointly defense against these two types of attackers using honeypot;
\item We provide a mixed integer program formulation and solutions for placing honeypots on very large AD graphs. Our solution could solve the problem on graphs of sizes that none of the previous approaches could.

\item We study for the first time the honeypot placement problem with dynamic graphs. All of the existing approaches focused only on static graphs. We develop two heuristics based on voting and clustering to solve the dynamic graph problems.  Our solutions, especially the clustering-based method, which derives defense based on "representative" snapshots of the dynamic graph (cluster centroids), run very fast on large graphs and provide close to optimal results. Our experiment will consider the algorithm on a range of AD graphs under realistic conditions.
\end{itemize}

\section{Related work}

%There has been some existing work on the AD attack graph. 
Guo et al. \citep{guo2022practical} study the shortest path edge interdiction problem on AD attack graphs. The authors propose a dynamic programming-based tree decomposition method to derive an optimal solution for their edge-blocking problem in directed acyclic AD graphs. They then developed a Graph Neural Network as a heuristic to enhance the scalability of their solution. Guo et al. \citep{guo2022scalable} extend their previous edge-blocking problem by proposing several scalable algorithms including Reinforcement Learning for tree decomposition and mixed-integer programming (MIP) algorithms. The authors later extend the tree decomposition algorithm to a general AD graph with cycles and derive the mixed strategy solution for their problem via MIP. In addition, Zhang et al. \citep{zhang2023oracle} showed the scalability of a double oracle-based algorithm on the same problem. Goel et al. \citep{goel2022defending} assumed the detection and failure probability of attackers on every edge of the AD graph. They deploy a neural network to approximate the attacker's strategy and apply an evolutionary diversity optimization (EDO) strategy to solve the defender problem. Goel et al. \citep{goel2023evolving} further improved their defense algorithm (EDO) by improving attacker strategy (fitness function) via Reinforcement Learning. These works, however, only work on static graphs and are not applicable to real AD networks, which are naturally dynamic. One of our key observations in this chapter is that optimal strategies on a static graph perform poorly when applied to dynamic networks. Our main difference with these existing works on AD graphs is that our algorithms are designed for dynamic graphs and hence generate better defenses for real AD systems. 

Honeypots have been widely studied in the literature. Stackelberg games are typically used to formulate the honeypot allocation problem in previous works~\citep{durkota2015optimal,durkota2015approximate,durkota2019hardening,milani2020harnessing}. The models in these studies are very similar to ours: they use attack graphs to model the attacker's actions and the defender can stop the attacker by deploying honeypots in the network. The defender’s task is to minimize the damage caused by the attacker with the minimum cost of honeypot deployment. Also, Milani et al. \citep{milani2020harnessing} study games where the defender can either deploy honey-edge or remove an edge from a network to lure the attacker to honeypots. Although these studies can optimally place honeypots in the network, they only work on very small graphs and suffer scalability issues. Indeed, in all of these works, they only consider networks of a hundred nodes in their experiments. Real AD graphs are several orders of magnitude larger than these graphs. Our solutions in this chapter are the first solutions that can effectively handle honeypot placement at that scale. 

More recently, Shinde et al. \citep{shinde2021cyber} model the cyber deception game as an interactive partially observable Markov decision process.  The authors of~\citep{shinde2021cyber} focus on luring the attacker to fake data by allocating deception files on network hosts - that is when the attacker already gets in the honeypot. Our work focuses on finding the optimal honeypot allocation strategy at the network level. Lukas et al. \citep{lukas2021deep} study the honeypot placement in AD graphs. The author proposed to use deep learning variational autoencoder model to generate and place honey users in the network. Their focus is on the topological structure of the fake AD graphs, not the impact of their solution on the attacker's success rate - a real measure of effectiveness for these solutions. Again, these recent solutions only consider honeypot allocation in small and static networks. In this chapter, we aim to provide practical solutions for allocating honeypots on large and dynamic AD attack graphs against realistic attackers.

\begin{figure}[htbp]
\centering
  \includegraphics[width=0.7\linewidth]{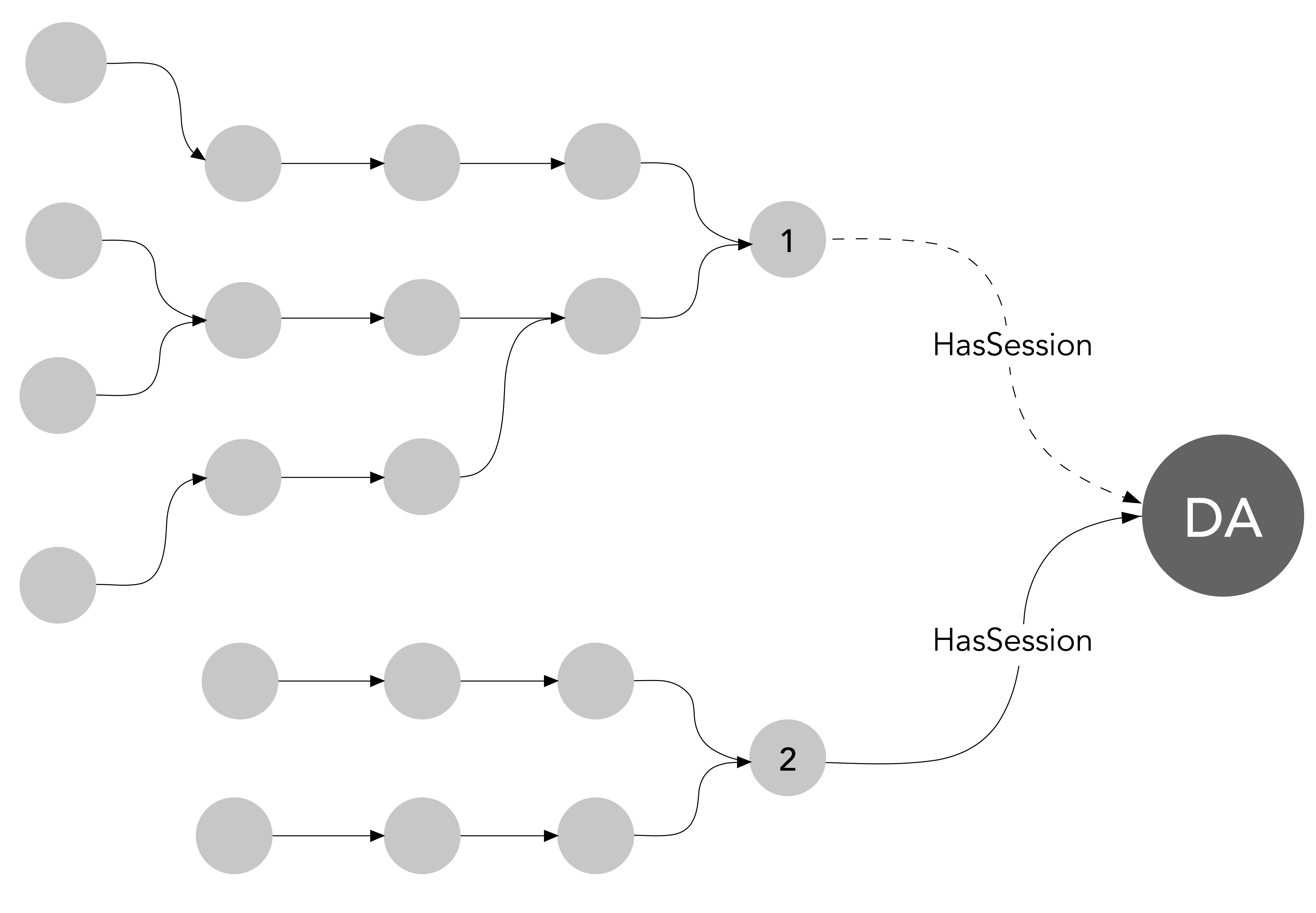}
  \caption{An attack graph with dynamic HasSession edges. The DA uses the edge 2 $\rightarrow$ DA (straight) regularly but only use 1 $\rightarrow$ DA (dash) when needed}
 \label{fig:dygraph}
\end{figure}

Fig. \ref{fig:dygraph} shows an attack graph of a typical dynamic AD environment. There are two HasSession edges to the DA. A static snapshot of the graph typically captures only the branch that has 2 $\rightarrow$ DA and completely misses the branch of the graph that contains 1 $\rightarrow$ DA. An attacker would wait in a system for a long time to wait for edge 1 $\rightarrow$ DA to appear and can bypass static defenses built on only the 2 $\rightarrow$ DA branch. None of the previous works on AD attack graphs ~\citep{dunagan2009heat, goel2023evolving, guo2022practical,guo2022scalable,goel2022defending, zhang2023oracle} considered dynamic AD graphs and would fail to defend against this type of attacks. Our solution can provide near-optimal honeypot allocation in dynamic networks with hundreds of thousands of nodes and millions of edges.

\section{Problem Formulation}
\label{sec:model}
\subsection{Honeypot placement model for static graphs}
We consider a Stackelberg game between an attacker and a defender on a directed AD attack graph $G = (V, E)$. There is a set $S\subseteq V$ of entry nodes, and the attacker can enter the graph via one of $s = |S|$ entry nodes. The attacker's ultimate goal is to reach the Domain Admin (DA). The attacker performs attacks by choosing paths with pure actions. (i.e.,  an attacker can run BloodHound \citep{BloodHound} to generate the paths to DA). From an entry node, there are many possible attack paths to DA, and the attacker can choose to attack the DA. Also, there is a fixed set of \textit{blockable} nodes, denoted by $N_b \subseteq V$. In this work, we use the term ``blocking'' to mean allocating a honeypot to a node (rather than removing the node from the graph); thus, $N_b$ represents the set of candidate nodes available for honeypot placement. The defender's task is to pick $b$ nodes in $N_b$ to allocate honeypots in order to intercept as many of the attack paths as possible. A honeypot will stop the attacker from going further into the network, that is, if the attacker stumbles into a honeypot, the attack campaign fails. We study two types of attacker models, the \textbf{simple attacker}  and \textbf{competent attacker}.  We introduce a defending problem denoted as $HATK(\varphi)$, where $\varphi$ ($0\leq \varphi \leq 1$) and represents the fraction of competent attacker agents ($1-\varphi$) is the probability of simple attacker agents).  When $\varphi = 1$, we focus on defending against pure competent attackers, whereas $\varphi = 0$ represents defending against pure simple attackers.

In the simple attacker model, the attacker is assumed to be unable to differentiate between a normal node and a honeypot. The simple attacker will perform the shortest attack path. Therefore, the attacker's optimal strategy is to randomly select one of the shortest paths to the DA. The defender's goal is to reduce the number of ``clean paths'', which are paths that can be used by the attacker to reach the target without any intervention. The defender achieves this by applying a blocking plan $B$ that reduces the number of clean paths. Let's denote by $y_{i}$ the total number of clean paths from entry node $i$ to DA and $y^B_{i}$ as the remaining clean paths after the defender applies a blocking plan $B$. The attacker's success probability is $\frac{y^B_{i}}{y_i}$, and the expected success probability can be obtained by averaging over all entry nodes. Our optimization problem for allocating honeypot against the \textbf{simple attacker} is formally defined as:

\begin{equation}
    \min_{B \subset N_b, |B| \leq b} \sum_{i=1}^{s} \frac{y^B_{i}}{y_i}
    \label{eq:obj_simple}
\end{equation}
% $\min\limits_{B \subset N_b, |B| \leq b} \sum\limits_{i=1}^{s} \frac{y^B_{i}}{y_i}$

In the competent attacker model, the attacker is capable of detecting every honeypot in the attack graph. Thus, when the attacker encounters a honeypot on its way to the target destination, it alternates to other paths. To prevent a competent attacker from reaching the DA, the defender must place honeypots to completely disconnect an entry node from the DA. The defender's goal is to maximize the number of entry nodes that are disconnected from the DA using a limited budget. Let's denote $R_{i}$ as the reachability of node $i$ to DA and given that allocating honeypot can change the reachability, we define $R^{B}_{i}$ as the reachability of node i after applying blocking plan $B$. Our optimization problem for allocating honeypot against the \textbf{competent attacker} is formally defined as: 

\begin{equation}
    \min_{B \subset N_b, |B| \leq b} \sum_{i=1}^{s} R^{B}_{i}
    \label{eq:obj_competent}
\end{equation}

By combining the objectives of both attacker models (Eqs. \ref{eq:obj_simple} and \ref{eq:obj_competent}), the problem $HATK(\varphi$) can be interpreted as: 
\begin{equation}
\min\limits_{B \subset N_b, |B| \leq b} \sum\limits_{i=1}^{s} (\varphi R^{B}_{i} + (1-\varphi) \frac{y^B_{i}}{y_i})
\end{equation}
To simplify notation, we introduce the following denotation: $z_{\varphi,i} = \varphi R^{B}{i} + (1-\varphi) \frac{y^B{i}}{y_i}$, where $z_{\varphi,i}$ represents the normalized attacker success rate of reaching the Domain Admin (DA) from node $i$ given a specific $\varphi$ value.

The two theorems below establish the computational hardness of the honeypot placement problems presented in this chapter. 

\begin{theorem} 
The static version of the optimal honeypot placement problem against \textbf{Simple Attacker} is NP-hard when $L\ge 7$ where $L$ be the maximum shortest path length from any entry node to DA.
\end{theorem}

\begin{proof}

    We present a reduction from the Vertex Cover problem. Let $G_{vc} = (V_{vc}, E_{vc})$ be an instance of the Vertex Cover problem, where $G_{vc}$ is an undirected graph with $n_{vc} = |V_{vc}|$. We define $L$ as the maximum shortest path length from the entry nodes to the DA node.
    We construct an attack graph as follows. We create $s$ entry nodes and a single destination node (DA). For every node $i \in V_{vc}$, we construct a gadget containing 6 non-blockable nodes: $i_1, i_2, i_4, i_5, i_6$ and 3 blockable nodes, which we refer to as $h_{i1}$, $h_{i3}$, and $h_{i5}$. Let $H$ denote the set of blocked nodes (honeypots).
    We connect these nodes using a set of multiple edges (each consisting of $M$ parallel edges, where $M$ is assumed to be a large value) and single edges. The multiple edges connect the pairs: $(s_i, h_{i1})$, $(h_{i1}, i_2)$, $(i_2, h_{i3})$, $(h_{i3}, i_4)$, $(i_4, h_{i5})$, and $(h_{i5}, i_6)$. The single edges connect the pairs: $(s_i, i_1)$, $(i_1, i_2)$, $(i_4, i_5)$, and $(i_5, i_6)$. Additionally, if node $i$ and node $j$ are connected in the original graph $G_{vc}$, we create two cross-paths in our gadget: from $i_2$ to $j_4$ (via a connector $c_{j3}$) and from $j_3$ to $i_4$ (via a connector $c_{i3}$).

    Observe that within a gadget, if we choose to block one node (e.g., $h_{i1}$), we guarantee that at least $M$ paths are removed between $s_1$ and $i_2$. We refer to any path that is not blocked by a honeypot as a ``clean path.'' Recall that in our model, to maximise the probability of the attacker triggering a honeypot, we aim to minimize the number of available clean paths to the DA by selecting a specific set of nodes for honeypot allocation.

    Let $V'_{vc} \subseteq V_{vc}$ be a vertex cover, and let $x = |V'_{vc}|$ be the size of the vertex cover. Given a defensive budget $b = n_{vc} + x$, the optimal strategy corresponds to the vertex cover set. An optimal blocking plan is constructed as follows: if $i \in V'_{vc}$, we add two blockable nodes, $h_{i1}$ (left) and $h_{i5}$ (right), to our blocking set $H$; if $i \in V_{vc} \backslash V'_{vc}$, we add one blockable node, $h_{i3}$ (middle), to the blocking set $H$.

    % Let $V'_{vc} \subseteq V_{vc}$ be a vertex cover and $x = |V'_{vc}|$ is the vertex cover size. Giving the defensive budget $b = n_{vc} + x$, the best way to spend defensive budget is to locate the vertex cover of size $x$. Optimal blocking plan could be obtained by the following \textbf{implication}: If $i \in V'_{vc}$ we add 2 blockable node h-type1 (left) and h-type5 (right) to our blocking set $H$ and for each node $u \in V_{vc} \backslash V'_{vc}$, we add 1 node h-type3 (middle) to blocking set $H$. 

    For every gadget corresponding to node $i \in V_{vc}$, at least 1 unit of the budget must be spent; otherwise, there will be at least $(M^2+1)^3$ paths to the DA via this gadget. Furthermore, if exactly one node of the gadget is in $H$, it must be the middle node ($h_{i3}$); otherwise, there will be at least $(M^2+1)^2$ paths via this gadget. Conversely, if two or more honeypots are placed in the gadget, they must be the left and right nodes ($h_{i1}$ and $h_{i5}$); this guarantees that there are at most $Z$ paths available via this gadget. Finally, following this construction, assume a gadget $i$ corresponds to a node in the vertex cover ($i \in V'_{vc}$) and let $d$ be the degree of node $i$ in $G_{vc}$.

    In this optimal scheme, the number of clean paths is at most $(M^2+1)(2d+1)$. We aim to avoid the sub-optimal case of $\mathcal{O}((M^2+1)^2)$ or $\mathcal{O}(M^4)$. Thus, we have:

    \begin{eqnarray}\label{eq:s}
    \mathcal{O}(M^4) > \mathcal{O}((M^2+1)(2d+1)) = \mathcal{O}(2dM^2) \nonumber
    \end{eqnarray}
    \begin{eqnarray}\label{eq:ss}
    \Rightarrow \mathcal{O}(M) > \mathcal{O}(\sqrt{2d}) \nonumber
    \end{eqnarray}

    In this proof, we assume that multiple edges exist in our graph. A similar proof for the case of a simple graph (without multiple edges) can be obtained by substituting every multiple edge (except edges connected to the DA) with a structure consisting of an intermediate node and two edges.

    \begin{figure}[h]
    \centering
      \includegraphics[width=0.7\linewidth]{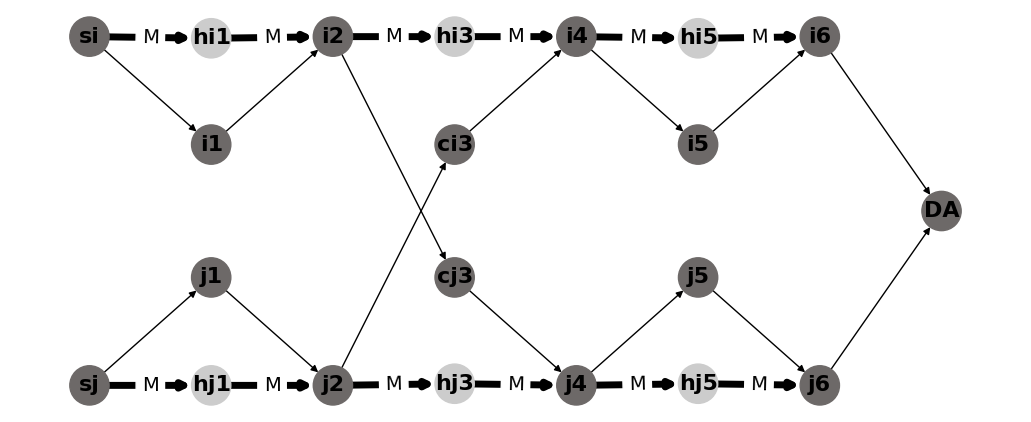}
      \caption{Proof gadget for Theorem 1. Light nodes are blockable nodes. Thick edge is multiple edges ($M$ of them) and thin edge are single edge.}
      \label{fig:NPgadget}
    \end{figure}
\end{proof}

\begin{theorem} The static version of the optimal honeypot placement problem against \textbf{Competent Attacker} is $W[1]$-hard with respect to budget b.
\end{theorem}

\begin{proof}
    To prove the hardness of our problem, we use the reduction from \textit{clique} problem to our. Clique problem is know to be $W[1]-hard$ with respect to clique size \citep{fellows2009parameterized}. Let define undirected graph $G_{c} = (V_{c}, E_{c})$ is the clique instance. We construct attack graph $G$ as follow. For every edge $(i,j) \in E_{c}$, we construct a corresponding entry node $s_{i,j}$. For every node $v \in G_{c}$ we construct a pair of node $v_{ub}$ and $v_{b}$ in the attack graph instance where $v_{ub}$ is unblockable node. Next, for every entry nodes, we create 2 path $s_{i,j}\rightarrow i_{ub} \rightarrow i_{b} \rightarrow DA$ and  $s_{i,j}\rightarrow j_{ub} \rightarrow j_{b} \rightarrow DA$. If we want to block all path from an entry node $s_{i,j}$ to DA, we need to block both $i_{b}$ and $j_{b}$. Let's say we can choose up to $b$ node to block, the optimal blocking plan could be obtained by locate the clique of size $b$ in $G_{c}$. By allocating a clique, we can disconnect maximum of $b(b-1)/2$ entry nodes.
\end{proof}

\subsection{Honeypot placement model for dynamic graphs}

The above formulation only optimizes the honeypot placement for a given graph instance $G$. In real AD networks, the graph $G$ changes constantly due to users' activities. 
There is a subset of edges $H \subseteq E$ (the HasSession edges) where each edge is turned on and off randomly. Formally, let $V_{comp}$ be the set of computers and $V_{user}$ be the set of users in the network. The authentication data is represented in the tuple format $<t_s, t_e, c_i, u_i>$, where $t_s$ is the sign-in (start time), $t_e$ is the sign-off time (end time), and $c_i$ and $u_i$ are a user and a computer in $V_{comp}$ and $V_{user}$, respectively. This tuple indicates that user $u_i$ signs in to computer $c_j$ at time $t_s$ and signs off that computer at time $t_e$. 

In practice, security teams often run SharpHound \citep{SharpHound} to enumerate a snapshot of the current network state, and they schedule it at regular intervals, for example every $t$ hours, to obtain multiple snapshots. Formally, let $G_s = \{g_1, g_2, ..., g_m\}$ be the set of possible graph snapshots. Let $E$ be the set of universal edges. In snapshot $x$, the network is represented as $g_x = (V,E_x)$ where $x\in \{1, \cdots, m\}$ indexes snapshots by order and $E_x \subseteq E$ is the set of edges that are active in that snapshot. Let $T(x)$ denote the time at which snapshot $x$ is taken. For example, given an authentication record $<t_s, t_e, c_i, u_i>$ and a snapshot $x$ taken at time $T(x)$, if $t_s \leq T(x) \leq t_e$ then edge $e_i = (c_i, u_i)$ will appears in the snapshot graph $g_x$. We assume only the edge set $E_x$ changes across snapshots while the node set $V$ remains the same. This is reasonable because most unexpected change is due to edges, while node changes are typically introduced by the defender with full awareness. The dynamic version of the HATK($\varphi$) problem can be interpreted:
 
\begin{equation}
\min\limits_{B \subset N_b, |B| \leq b} \sum\limits_{g \in {G_s}} \sum\limits_{i=1}^{s}  (\varphi R^{B}_{i,g} + (1-\varphi) \frac{y^B_{i,g}}{y_{i,g}})
\end{equation}

Again, to simplify the notation, we have: $z_{\varphi,i,g} = \varphi R^{B}{i,g} + (1-\varphi) \frac{y^B{i,g}}{y_{i,g}}$. Note that we do not make any assumption on the temporal pattern of the edge sets in our solutions. We will show in the evaluation section that our algorithms are robust to dynamic graphs with realistic patterns.

\newpage
\section{Honeypot Placement Solutions}

\subsection{Solution for static AD graphs}

\textbf{Mixed-integer programming for static AD graphs}

\noindent Firstly, we formulate a nonlinear program for solving the $HATK(\varphi$):

\begin{subequations}
\begin{align}
\text{min} \displaystyle\sum\limits_{i=1}^{s} z_{\varphi,i}  \nonumber\\ && \nonumber\\
% \end{flalign*}
% \begin{align}
  y^{B}_i =& \sum_{j \in n(i)}y^{B}_{j}, & \forall i \in V \backslash N_b \label{eq:1a} \\
  y^{B}_i =& \sum_{j \in n(i)} (1-B_i) y^{B}_{j}, & \forall i \in  N_b \label{eq:1b} \\
  y^B_{i} \leq& y_{i},         & \forall i \in V \label{eq:1c} \\
  y^B_{i} \geq& 0,              & \forall i \in V \label{eq:1d} \\
  R^B_{i} \geq& R^B_{j} - B_{i}, &(i,j) \in E, i \in N_{b} \label{eq:1e}\\
  R^B_{i} \geq& R^B_{j}        , & (i,j) \in E, \nonumber\\
  &&i \in N \backslash N_{b} \label{eq:1f}\\
  \sum_{i \in V}\ B_{i} \leq& b \label{eq:1g}\\
  B_{i},R_{i} \in & \{ 0,1 \} \label{eq:1h}\\
  R_{DA}, y_{DA} =& 1 \label{eq:1i}
\end{align}
\end{subequations}

\noindent where $n(i)$ is the set of the successor node of node $i$. The design of the formulation is based on the problem of defending against each individual attacker. We formulate it as a nonlinear program that is tailored to each type of agent. In our formulation, we use $B_i$ to denote the unit of budget spent on node i, $B_i$ is binary. Formally, $B_i = 0$ if $i \not \in N_b$ and $B_i \in \{0,1\}$ if $i \in N_b$.

Constraints (\ref{eq:1a}) to (\ref{eq:1d}) present the problem of defending against the simple attacker. The design of this set of constraints concentrates on computing the number of paths from an arbitrary node to the target (DA) on an \textbf{all-shortest path} graph by summing the number of paths to DA from each of its successors. We illustrate this idea using the example in Fig. \ref{fig:exadgraph}. To calculate the number of paths to DA from 5 or $y_5$, we need to recursively calculate $y_3$, $y_2$, $y_1$ and $y_{DA}$. First, we define that $y_{DA} = 1$ (there is 1 path from DA to itself). Then each node 1 and 2 has exactly 1 direct edge to DA so $y_1 = 1$ and $y_2 = 1$. Next, node 3 has two outgoing edges to nodes 2 and 1, meaning that attackers have two possible ways to reach DA via these nodes. Therefore, $y_3 = y_1 + y_2 = 1 + 1 = 2$. Finally, since node 5 has one outgoing edge to node 3, which could lead to two possible paths to DA, we set $y_5 = y_3 = 2$. Note that our path-calculating approach will fail in the presence of cycles in the graph, as the number of paths of a node in a cycle can go to infinity. Since the simple attacker is unable to observe any honeypot and they always go for the shortest paths, it is without loss of generality to consider only \textbf{all-shortest path} graphs, which is naturally acyclic.

\begin{figure}[h]
\centering
  \includegraphics[width=0.65\linewidth]{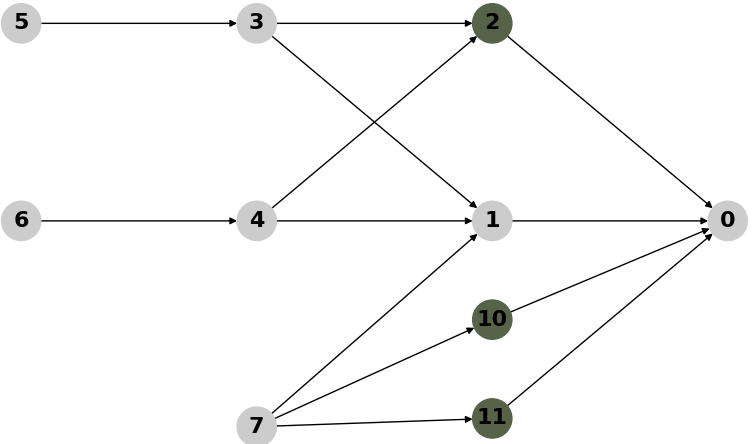}
  \caption{An example all-shortest path AD attack graph. Nodes 5, 6, 7 are entry nodes. Nodes 0 is the DA. \textbf{Bold} nodes 2, 10, 11 are not blockable nodes}
  \label{fig:exadgraph}
\end{figure}

In the case of defending against a competent attacker, we aim to reduce the attacker's success rate by disconnecting the maximum number of entry nodes from DA. Constraints (\ref{eq:1e}) and (\ref{eq:1f}) presented the problem of defending against a competent attacker. These constraints solve the problem based on the idea that if any of the successor nodes of node $i$ can reach DA, then DA is also reachable via node $i$. We assume that DA is reachable by all entry nodes by default.

Our formulation is nonlinear due to the nonlinear constraint (\ref{eq:1b}). To obtain a linear version, we follow several transformation steps. First, we will replace $y_i$ by an auxiliary variable $v_i$ by adding constraints $y^B_{i} \leq  v_{i}$, (\ref{eq:1a}) also become $v_i = \textstyle\sum_{j \in n(i)}y^{B}_{j}$. Next, we replace (\ref{eq:1b}) by 2 additional constraints, $y^B_{i} \leq y_{i}\cdot(1-B_i)$ and $y^B_{i} \geq v_{i} + y_{i}\cdot(-B_i)$. These constraints imply that if node $i$ is blocked, then $y^B_{i} = 0$; otherwise, we compute $y^B_i$ normally ($y^B_i = v_i = \textstyle\sum_{j \in n(i)}y^{B}_{j}$).

\subsection{Honeypot Placement for dynamic networks}
\label{sec:placedynamic}
% In this section, we will discuss our proposed graph sampling technique that generalise our algorithm for dynamic graph. We first refining our mixed-integer programming formulation for multiple graph problem. 

\textbf{Multi-graph Mixed-integer Programming}

% In the previous section, we formulated the mixed-integer programming for the honeypot allocation problem for static graphs. As shown later, blocking strategy for static graph will shortfall in dynamic graph scenario where HasSession edges are changing their states.
\noindent For dynamic graphs, we require a different approach as we need to design a blocking plan that performs well against an exponential number of possible graph instances. Assume that we have a collection of ``sample static instances'' (snapshots) based on the dynamic AD graph, denoted by $G_s$. As shown in section \ref{sec:model}, we model the node blocking problem as a global optimization problem where the defender tries to minimize the attacker's average success probability over the sample instances in $G_s$.

Let us denote by $m$ the number of graph samples in $G_s$. A naive approach is to merge all the graphs in $G_s$ (i.e, find a super graph that contains all instances, essentially assuming all HasSession edges are on) and run sMIP on it. Of course, the resulting blocking plan will not be optimal. Let $OBJ(g_i)$ be the objective function on graph $g_i \in G_s$ or the attacker's success probability facing graph $g_i \in G_s$ in the static scenario. 
We also define $OBJ_m(G_s)$ to be the sum of $m$ objective functions or the sum of the attacker's success probabilities for all sample graphs in $G_s$ (there are $m=|G_s|$ samples). 
Optimizing blocking strategy for dynamic graphs then boils down to minimizing the following objective function:

\begin{eqnarray}\label{eq:6}
OBJ_m(G_s) = OBJ(g_1) + OBJ(g_2) + ... + OBJ(g_m)
\end{eqnarray}

In practice, some snapshots of the graph will have a higher chance of appearing, in that case, we could adjust our objective function by multiplying by a weight  $k$  to that graph's objective function (i.e., $k  * OBJ(g_i)$). We have the following nonlinear program for the dynamic version (We provide the linear program version in Section \ref{sec:appxmmip}):

\begin{subequations}
\begin{align}
% \text{min} \displaystyle\sum\limits_{g \in {G_s}} \displaystyle\sum\limits_{i=1}^{s} \varphi &R^{B}_{i,g} + (1-\varphi) \frac{y^B_{i,g}}{y_{i,g}}&  \nonumber\\  \nonumber\\
\text{min} \displaystyle\sum\limits_{g \in {G_s}} \displaystyle\sum\limits_{i=1}^{s} z_{\varphi,i,g}&  \nonumber\\  \nonumber\\
% \end{flalign*}
% \begin{align}
  y^{B}_{i,g} =& \sum_{j \in n(i)}y^{B}_{j,g}, & \forall i \in V_g \backslash N_b, g \in G_s \label{eq:3a} \\
  y^{B}_{i,g} = \sum_{j \in n(i)} &(1-B_i) y^{B}_{j,g}, & \forall i \in  N_b, g \in G_s \label{eq:3b} \\
  y^B_{i,g} \leq& y_{i,g},         & \forall i \in V_g, g \in G_s \label{eq:3c} \\
  y^B_{i,g} \geq& 0,              & \forall i \in V_g, g \in G_s \label{eq:3d} \\
  R^B_{i,g} \geq& R^B_{j,g} - B_{i}, &\forall (i,j) \in E_g,  \nonumber \\ 
                &                    &i \in N_{b}, g \in G_s \label{eq:3e}\\
  R^B_{i,g} \geq& R^B_{j,g}        , & \forall (i,j) \in E_g,g \in G_s\nonumber \\
                &                    & i \in N \backslash N_{b} \label{eq:3f}\\
  \sum_{i \in V}\ B_{i} \leq& b \label{eq:3g}\\
  B_{i},R_{i,g} \in & \{ 0,1 \} \label{eq:3h}\\
  R_{DA}, y_{DA} =& 1 \label{eq:3i}
\end{align}
\end{subequations}

The above formulation can be converted to a linear program, similar to the approach involving static graphs. In the rest of this chapter, we will refer to the Mixed Integer Programming algorithm for dynamic graphs as $MIP_m(X)$, where $X$ represents the set of graph instances considered in MIP and $m$ denotes the cardinality of $X$. We use the term $dyMIP(m)$ when we do not need to specify the set of graphs under consideration.

\vspace{0.75cm}
\noindent \textbf{Optimal Lower Bound}

% Ideally, to obtained the optimal blocking policy of Sample 
\noindent In theory, we can find the optimal blocking policy by running dyMIP(m) algorithm on every possible graph instance. However, the linear program $MIP_m(X)$ requires $\mathcal{O}(m \cdot (n+e))$ variables, where $n$ is the number of nodes and $e$ is the number of edges in the graph $G \in X$. As a result, $MIP_m(X)$ becomes increasingly computationally challenging as $m$ grows larger. To overcome this issue, we split $G_s$ into $j$ equally-sized batches, each batch has $t$ graphs (i.e, $t \cdot j = |G_s|$). We denote each batch of graphs by $X_i \subset G_s$, and $X = \{X_1, X_2, ..., X_j\}$ is the (all disjoint) batch set, where $G_s = X_1 \cup X_2 \cup \cdots \cup X_j$. The following proposition shows the derivation of a Lower Bound on the overall objective using this batching approach.

\begin{proposition} The batching solution gives a lower bound on the overall objective:
\begin{eqnarray}\label{eq:7}
\displaystyle\sum\limits_{X_i \in X} \min_B OBJ_t(X_i) \leq \min_B OBJ_m(G_s) = OPT
\end{eqnarray}
\end{proposition}
\begin{proof}
    To prove the hardness of our problem, we use the reduction from \textit{clique} problem to our. Clique problem is know to be $W[1]-hard$ with respect to clique size \citep{fellows2009parameterized}. Let define undirected graph $G_{c} = (V_{c}, E_{c})$ is the clique instance. We construct attack graph $G$ as follow. For every edge $(i,j) \in E_{c}$, we construct a corresponding entry node $s_{i,j}$. For every node $v \in G_{c}$ we construct a pair of node $v_{ub}$ and $v_{b}$ in the attack graph instance where $v_{ub}$ is unblockable node. Next, for every entry nodes, we create 2 path $s_{i,j}\rightarrow i_{ub} \rightarrow i_{b} \rightarrow DA$ and  $s_{i,j}\rightarrow j_{ub} \rightarrow j_{b} \rightarrow DA$. If we want to block all path from an entry node $s_{i,j}$ to DA, we need to block both $i_{b}$ and $j_{b}$. Let's say we can choose up to $b$ node to block, the optimal blocking plan could be obtained by locate the clique of size $b$ in $G_{c}$. By allocating a clique, we can disconnect maximum of $b(b-1)/2$ entry nodes.
\end{proof}

For (\ref{eq:7}), the left-hand side consists of sum of the attacker's success probabilities for all batches, and the defender can
apply different optimal blocking plans for different batches. The right-hand side is the sum of the attacker's success probabilities for all samples, under the defender's optimal blocking plan that is applied to all samples (i.e., the same plan shared over all batches).
The inequality in (\ref{eq:7}) provides a lower bound on the optimal result for the set of samples $G_s$. This lower bound can be estimated by summing the optimal results of each batch in $X$.
Even for large values of $m$, this approach allows us to find the estimated lower bound with minimal computational effort. In simpler terms, solving a single instance of $MIP_{1000}(G_s)$ is more challenging than solving 100 instances of $MIP_{10}(X_i)$. By summing the objective values of all instances of $MIP_{10}(X_i)$, we are guaranteed to get a lower bound for the original problem. 
Furthermore, we could interpret the result of each batch as {\em an independent draw from the same (unknown) distribution} and our task is to calculate the mean based on sample averages,
so using Hoeffding’s Inequality Bound (more details presented below), we can achieve a high-confidence estimation without going over all batches.

%In our case, the objective value represented the chance of the attacker can get to the DA. It is noteworthy that the objective function is the sum of the attacker's success probability in every graph. For the representation, we will show the normalised result which takes the average over the number of possible graphs. 

% The above program has at most $O()$

% \textbf{Voting for best blocking strategy:} 

\vspace{0.75cm}
\noindent \textbf{Number of Sample for Monte Carlo Simulation}

\noindent To determine how many graph sample is enough to guarantee that our evaluation is good enough, we use the Hoeffding's Inequality Bound \citep{hoeffding1994probability}. Let us define $X_0, ..., X_n$ is the independence random variable such that $a_i \leq X_i \leq b_i$. Let us say $\epsilon > 0$ is a error of the estimated mean to the true mean. We have the following Hoeffding's Inequality:

\begin{eqnarray}\label{eq:hoefdori}
\Pr(|\hat{X} - \mu| \geq \epsilon) \leq 2e^{\frac{-2n^2\epsilon^2}{\sum_{i=1}^{n} (b_i-a_i)^2}} \leq \alpha
\end{eqnarray}

\noindent where $\hat{X}$ is the sample mean and $\hat{X_n} = \frac{\sum_i X_i}{n}$. $\mu$ is the true mean, n is the sample size and $\alpha$ is the level of significance (upper bound of the error probability). Now let us say we have a blocking strategy $B$ and applying it to sample $G_s$ with where number of sample $n = |G_s|$. Apply $B$ on graph $g_i \in G_s$ yield attacker success score $X^{B}_i$. $X^{B}_i$ is bounded by $0 \leq X^{B}_i \leq 1$ for any i in $G_s$. The sample mean and true mean  of the simulation is $\hat{X^B}$ and $\mu^B$, respectively. We could re-write the inequality as:

\begin{eqnarray}\label{eq:hoefeval}
\Pr(|\hat{X^B} - \mu^B| \geq \epsilon) \leq 2e^{-2n\epsilon ^2} \leq \alpha
\end{eqnarray}

Solving \ref{eq:hoefeval} give use the number of sample required for ensure $\hat{X^B} = \mu^B \pm \epsilon$ with confidence level of $(1-\alpha)$:
\begin{eqnarray}\label{eq:sample}
n \geq \frac{\ln(\frac{2}{\alpha})}{2\epsilon^2}
\end{eqnarray}

In our Monte Carlo Simulation, by testing our blocking polity on $n =$ 100,000 graph samples, our result can be at least 99\% confident ($\alpha = 0.01$) with error within $\epsilon = 5.147\times 10^{-3}$ of the true mean ($\hat{X^B} = \mu^B \pm 5.147\times 10^{-3}$)

\subsection{Generalization to Large Collections of Graph Samples}

As mentioned above, dyMIP(m) only works on a small number of graphs before running into scalability obstacles. We need to develop a better strategy for scaling dyMIP(m) to large $m$. We develop two heuristics for this purpose.

% \vspace{0.75cm}
%\subsubsection{Voting based heuristic}
\vspace{0.75cm}
\noindent \textbf{{Voting based heuristic}}

\noindent The overall idea behind our \textit{voting-based heuristic} is to leverage the lower-bound calculation in Section~\ref{sec:placedynamic}. The pseudocode appears in Algorithm~\ref{alg:voting}. In line 2 we partition $G_s$ into $\left\lceil \frac{|G_s|}{p} \right\rceil$ disjoint batches, each of size at most $p$. Lines 3–5 run $MIP(x,m)$ on every batch $x \in X$ and store each resulting blocking strategy $B_{\mathrm{tmp}}$ in the list $\mathcal{B}={B_1,B_2,\ldots,B_j}$. Lines 6–8 then count how often each node appears across the strategies in $\mathcal{B}$, writing $\text{freq}(v)$ for the number of strategies that contain $v$. Finally, line 9 selects the $b$ nodes with the largest values of $\text{freq}(v)$ (majority voting) to form the final blocking strategy.

\begin{algorithm}[H]
 \caption{Voting based heuristic for MIP(m)}
 \label{alg:voting}
 \SetAlgoLined
 \LinesNumbered

 % suppress printed "end" markers
 \SetKwFor{For}{for}{:}{}
 \SetKwFor{ForEach}{foreach}{:}{}
 \SetKwFor{While}{while}{:}{}
 \SetKwIF{If}{ElseIf}{Else}{if}{then}{else if}{else}{}

 \KwIn{Graph sample $G_s$, partition size $p$, number of graph in MIP(m) $m$, defensive budget $b$}
 \KwOut{Blocking strategy $B^*$}
 
 \BlankLine
Initialize the set of blocking strategies $\mathcal{B} \leftarrow \emptyset$\;

Partition $G_s$ into $\left\lceil \frac{|G_s|}{p} \right\rceil$ batches of size $p$ and store to $X$\;

 \For{$x \in X$}{
   $B_{tmp} \leftarrow MIP(x, m)$ \;
   
   Append $\mathcal{B} \leftarrow \mathcal{B} \cup B_{tmp}$ \;
 }

Initialize the node–frequency map $\text{freq} \leftarrow \emptyset$ (so $\text{freq}(v)$ stores how often node $v$ appears across $\mathcal{B}$)\;

 \For{every node $v$ appearing in any $B_{\mathrm{tmp}}\in\mathcal{B}$}{
    $\text{freq}(v) \leftarrow \bigl|\{\, v\in B_{\mathrm{tmp}} \mid B_{\mathrm{tmp}}\in\mathcal{B} \,\}\bigr|$\;
 } 

Set $B^{*}$ the top-$b$ nodes with the largest values of $\text{freq}$\;

 \Return{$B^*$}
\end{algorithm}

\vspace{0.75cm}
\noindent \textbf{Clustering based heuristic}

\noindent Our second heuristic is based on the idea that there is a set of most \textit{representative} graph instances (snapshots, realizations) for the whole collection of samples. This approach is particularly useful when there are duplicate graphs in a set of samples or samples with similar structures. Our idea is only to run dyMIP(m) once on this \textit{representative} set of graphs. To do that, we apply a clustering algorithm to explore different structures among the graph samples and group graphs with common structures into the same group. Formally, let's define $F(g)$ is a $s$-dimention features vector of graph $g$ which expressed as follow:

\begin{equation}
F(g) = \big( f(1,g), f(2,g), \cdots, f(s,g) \big )
\label{eq:feature}
\end{equation}

where $f(i,g) = \frac{y^{B_g}_{i,g}}{y_{i,g}}$ represents the attacker's success rate when they are at a starting node $i$ and facing a optimal blocking strategy $B_g$ on graph $g$ produced by running staticMIP on that instance. We use k-means as the clustering algorithm in our implementation. After getting $F(g)$ for every graph $g \in G_s$, we apply k-means to cluster $G_s$ into $k$ clusters. We consider graphs that are closest to centroids as the representative graphs for our dyMIP(m), where the distances to centroid are Euclidean distances~\citep{danielsson1980euclidean}. We apply dyMIP(m) on the set of representative graphs that we got from the clustering algorithm. 
% \begin{eqnarray}\label{eq:7}
% feature(g) = \bigg( \frac{y^{B_g}_{1,g}}{y_{1,g}}, \frac{y^{B_g}_{2,g}}{y_{2,g}}, \cdots, \frac{y^{B_g}_{i,g}}{y_{i,g}} \bigg )
% \end{eqnarray}

% \begin{eqnarray}\label{eq:7}

% \end{eqnarray}

There is a complication with the clustering algorithm where clusters are of varying sizes and densities. If we uniformly pick a graph from each cluster, then the algorithm is prone to be biased to the ``fringe'' clusters which degrade the quality of our samples. Our solution is to weigh the graph selection based on the cluster size. That is, when choosing a graph for dyMIP(m), we first choose $m$ clusters where each cluster $c$ with size $size_c$ has a probability $\frac{size_c}{|G_s|}$ for being picked. We then pick a graph instance that has the closest distance to centroid $c$ to be the \textit{representative graph}. The already chosen graph instances are removed from consideration to avoid duplication. \textit{Clustering heuristic} is more scalable than \textit{voting heuristic}. For example, suppose our batch size is $100$, while \textit{voting heuristic}  requires us to run dyMIP(100) multiple times (majority voting does not make sense if there is only one vote), the \textit{clustering heuristic} only requires us to run MIP(100) once.

Pseudocode for the clustering-based heuristic appears in Algorithm \ref{alg:kmean}. Lines 3–4 compute a feature vector for each graph using Eq. \ref{eq:feature}. Line 5 runs the clustering algorithm which returns the centroids and the size of each cluster. In our experiments, we use k-means as clustering algorithm. In lines 8–14, we selects the m most representative graphs from the sample. At each iteration we sample a cluster $c\in C$ with probability proportional to its size to avoid bias toward trivial clusters. We then choose the graph $g$ that is closest to its centroid $c$, add it to the representative set $G''_s$, and remove it from the pool to prevent duplicates. Finally, line 15 runs the MIP on $G''_s$.

\begin{algorithm}[H]
 \caption{Clustering based heuristic for MIP(m)}
 \label{alg:kmean}
 \SetAlgoLined
 \LinesNumbered

 % suppress printed "end" markers
 \SetKwFor{For}{for}{:}{}
 \SetKwFor{ForEach}{foreach}{:}{}
 \SetKwFor{While}{while}{:}{}
 \SetKwIF{If}{ElseIf}{Else}{if}{then}{else if}{else}{}

 \KwIn{Graph sample $G_s$, number of cluster $k$, number of graph in MIP(m) $m$, defensive budget $b$}
 \KwOut{Blocking strategy $B^*$}
 
 \BlankLine
 Initialize a map of features $F\gets \emptyset$\;
 Initialize list of centroid $C \gets \emptyset$\;
 \For{$g \in G_s$}{
   Calculate $F(g)$ based on Eq. (\ref{eq:feature})\;
 }

 Get list of centroid $C$ and size of each cluster $clusters\_size$ by running clustering algorithm\;
 
 Initialize 2-D list $G'_s$ where $G'_s(c)$ is the list of graph in cluster $c$\;
 
 Initialize most representative list $G''_s = \emptyset$\;
 
 \For{$i$= $[1, \cdots,m]$}{

    \Repeat{$G'_s(c)\notin \emptyset$}{
   Pick a cluster $c \in C$ randomly with probability of $\frac{size(c)}{|G_s|}$ where $size(c)$ is number of graph in cluster $c$\;
   }
   
   Get graph $g\in G'_s$ with minimum distance to centroid $c$\;
   
   Append $G''_s \leftarrow G''_s \cup g$\;
   
   Delete $g$ from $G'_s(c)$\;
 }

 $B^* {\leftarrow} \textit{MIP}(G''_s, m)$\;
 \Return{$B^*$}
\end{algorithm}

\section{Experiment}
We evaluate our proposed solutions on both synthetic AD graphs with both randomized HasSession edges and real HasSession data. All of the experiments are carried out on a high-performance computing cluster with 2 CPUs and 24GB of RAM allocated to each task for the experiments. Code and data will be made publicly available.

\subsection{Synthetic Graph Generation}
%Since the Active Directory attack graph data is confidential, we assess the performance of our algorithms using synthetic data. 
We start with synthetic graphs generated by DBCreator \footnote{https://github.com/BloodHoundAD/BloodHound-Tools/tree/master/DBCreator} and Adsimulator \footnote{https://github.com/nicolas-carolo/adsimulator} - two state of the art tools for creating AD graphs.
In particular, we present results for the 5 synthetic graphs in Table \ref{tab:compgreedy}). The graphs labeled as "Rxxxx" were generated using DBCreator, while the ones labeled as "ADxxx" were created using Adsimulator. DBCreator allows us to generate graphs given a number of computers in the targetted AD network (R2000 and R4000 have 2000 and 4000 computer nodes, respectively). Adsimulator provides more flexibility by allowing graphs with specific properties, including the number of users, computers, domain trusts, OUs, GPOs, and more. $ADX05$ and $ADX10$ are created by increasing all its default parameters by a factor of 5 and 10 respectively. $ADU10$, $ADU100$ are created by tuning the parameters to mimic the real AD network of {\em University of Anonymized} which are $10\%$, $100\%$ of the size of the mimicked network, respectively.

We perform the following pre-processing steps, First, when there are multiple DAs in the graph, we simply merge it into 1 unique DA. We only consider 3 types of edges in our graph: AdminTo, HasSession, and MemberOf. We only consider an entry node that can reach DA. $ADU100$ is the largest graph in our experiment with 137207 nodes and 1410679 edges. Table \ref{tab:graphstat} show the size of each graphs

\begin{table}
\begin{center}
\caption{This table provides the size of each graph in terms of the number of nodes, edges, user nodes, and computer nodes.}\label{tab:graphstat}
\smallskip\noindent
\resizebox{0.7\linewidth}{!}{%
\begin{tabular}{l|c|c|c|c}

\hline
       & \textbf{Nodes}   & \textbf{Edges}    & \textbf{Users } & \textbf{Computers} \\
       \hline
\textbf{R2000 } & 5997   & 18795   & 2000  & 2001 \\
\textbf{R4000}  & 12001  & 45780   & 4000  & 4001 \\
\textbf{ADX05 } & 1624   & 6955    & 482   & 526  \\
\textbf{ADX10}  & 3123   & 15767   & 972   & 1026 \\
\textbf{ADU10}  & 13834  & 94878   & 6372  & 366  \\
\textbf{ADU100} & 137315 & 1490766 & 63172 & 3378 \\
\textbf{AuthOrg}& 2624   & 18594   & 1804  & 326 \\
       \hline
\end{tabular}}
\end{center}
\end{table}

\subsection{Experiments on static graphs}
\textbf{Experiment Set Up.} We will refer mixed-integer programming approach as sMIP (or staticMIP). 
We also adopt 2 algorithms referred to as sMIP-S, sMIP-C where the former is sMIP formulation for solving the problem of (pure) simple attacker agent ($\varphi = 0$), and the latter is sMIP for solving the problem of (pure) competent attacker agent ($\varphi = 1)$. 
Throughout our experiment, we assume $\varphi = 0.5$ (equal chance of being attacked by either of the attackers). However, it is important to note that our technique remains effective for any value of $\varphi$. 
%In summary, our experiment will have three settings for $\varphi$ (0, 0.5, 1), which will be denoted as sMIP-S, sMIP-M, and sMIP-C, respectively.

% We also have another setting for dealing with mixed types of attackers and the MIP algorithm for solving that problem will be denoted as sMIP-M. We argue that the competent attacker, in the most extreme case, have a higher capability to infiltrate our network compared to the simple attacker. In our experiment, we assume a ratio of $s:1$ (or $\varphi =\frac{s}{s+1}$) between the competent attacker and the simple attacker. This signifies that the more the entry node, the greater the chance the network is vulnerable with competent attackers. 

 We use Gurobi 9.0.2 solver for solving staticMIP. For each graph, we run 10 trials and report the results for these runs. For each trial, we randomly draw 50 entry nodes and set the budget b to 10. We assumed only computer nodes are blockable (i.e., honeypots can be placed only on computers). In the experiment, we introduce 3 metrics to evaluate our algorithms: \textbf{S}imple-Agent \textbf{S}uccess \textbf{R}ate (\textbf{SSR}) and \textbf{C}ompetent-Agent \textbf{S}uccess \textbf{R}ate (\textbf{CSR}) and \textbf{M}ean \textbf{S}uccess \textbf{R}ate (MSR). SSR and CSR are the success rates of the Simple and Competent attacker agent when facing a blocking strategy from the defender. SSR is measured as the fraction of the clean paths over the number of paths in the all-shortest path graph, normalized by the number of starting nodes and CSR is the fraction of entry nodes that can reach DA over the total number of entry nodes; whereas MSR is the average of SSR and CSR. 
 %Basically, SSR is the fraction of the clean paths over the number of paths in the all-shortest path graph, normalised by the number of starting nodes, CSR is the fraction of entry nodes that can reach DA over the total number of entry nodes; and MSR is mean of SSR and CSR. 

We carry out 2 sets of experiments. The first experiment which shows in Table \ref{tab:compgreedy} compares our sMIP-S and sMIP-C with 2 greedy algorithms namely GREEDY-S, GREEDY-C. 
Additionally, we evaluated our approach against Zhang et al.'s Double Oracle (ZDO) algorithm \citep{zhang2023oracle}. ZDO is the state-of-the-art edge-blocking algorithm for AD graphs. We adopt this approach for honeypot placement by formulating it as a simple blocking problem.  The purpose of this comparison is to demonstrate that the existing techniques are not suitable for our honeypot placement problem without compromising its performance.

%which is a completely different problem to ours. We re-implemented their node-blocking version and compared it to ours. However, it is still noteworthy that our chapter is not simply the node blocking version of the previous works, we study a completely new defense technique that is honeypot placement problem which is completely difference from other works (difference objective function). The purpose of this comparison is to demonstrate that the existing techniques are not suitable for our defense approach without compromising its performance.

In the GREEDY-S algorithm, for each node, we can enumerate the shortest path that goes through it or calculate the betweenness centrality for each node. The idea is to iteratively block the node that has the highest betweenness centrality and repeat it $b$ times. GREEDY-C uses a minimum node cut algorithm to disconnect each entry node and then prioritizes the entry nodes requiring the least budget to be blocked until the budget is depleted. In the second experiment shown in Table \ref{tab:mixedstatic}, we simply show the performance of our sMIP against a scenario of 2 types of attacker agents in the network.

\textbf{Result Interpretation.} In Table \ref{tab:compgreedy}, it can be seen that sMIP-S and sMIP-C perform at least as well as, and often better than, GREEDYs on various graphs. While ZDO performed badly since it is designed for the shortest path interdiction problem. In Table \ref{tab:mixedstatic}, our goal is to produce a defence plan that can be effective against both types of attacker agents. 
As they do not make any relaxation in the MIP formulation, both sMIP-S and sMIP-C obtained the optimal defense against a Simple attacker and a Competent attacker, respectively.
As shown sMIP-M significantly improves the MSR for $\varphi = 0.5$. This is not surprising as it is specifically designed to defend against both types of attackers. In $ADU100$, 10 honeypots are not enough to stop the competent attacker from any entry nodes hence, the attacker success rate for the competent attacker is 1. More significantly, staticMIP runs fast enough for large graphs and could scale to graphs of sizes equivalent to our anonymous network. That is, staticMIP is able to produce optimal results for very large static graphs, making it scalable for realistic AD graphs. Note that even though the general problem is NP-hard, what we observe here is that AD graphs have certain tree-like structures \citep{guo2022practical} that allow staticMIP to solve them very fast. In fact, the defending problem against both types of attackers is polynomially solvable if the AD graph is a tree (via a top-down Dynamic Program from DA to allocate honeypots). This makes the Linear Programming approach very efficient against problems on tree-like AD graphs.

% \hn{tree-like, cite Mingyu chapter that talks about this?}\hn{Explain why? what in the property of the problem that allows staticMIP to solve an NP-hard problem very fast?}. 

% We note, however, that staticMIP is not designed to work with dynamic graphs --- a topic that we will address next.

\begin{table}
\begin{center}
\caption{Comparison of the sMIP algorithm against GREEDY and ZDO algorithms. SSR is the metric for simple attackers and CSR is used for competent attackers. We set the budget $b=10$. Values in parentheses indicate average computation time in seconds.}\label{tab:compgreedy}
\smallskip\noindent
\resizebox{\linewidth}{!}{%
\begin{tabular}{l|ccc|cc}
\hline
       & \multicolumn{3}{c|}{\textbf{Simple Attacker (SSR)}}       & \multicolumn{2}{c}{\textbf{Competent Attacker (CSR)}}                   \\
\cline{2-6}
                  & sMIP-S                 & GREEDY-S           & ZDO      & sMIP-C                  & GREEDY-C       \\
\hline
$\textbf{R2000 }$ & \textbf{0.233} (0.15s)          & \textbf{0.233}  (0.01s)     & 0.389 (0.67s)      & \textbf{0.354} (0.09s) & 0.542 (0.37s) \\
$\textbf{R4000 }$ & \textbf{0.279} (0.10s)          & \textbf{0.279 } (0.01s)     & 1.000 (0.13s)      & \textbf{0.530} (0.28s) & 0.560 (0.48s) \\
$\textbf{ADX05 }$ & \textbf{0.135} (0.12s) & 0.136  (0.02s)     & 0.272 (0.68s)      & \textbf{0.670} (0.12s) & 0.682 (1.04s) \\
$\textbf{ADX10 }$ & \textbf{0.144} (0.16s) & 0.155  (0.03s)     & 0.226 (0.48s)      & \textbf{0.606} (0.20s) & 0.774 (1.86s) \\
$\textbf{ADU10 }$ & \textbf{0.547} (1.27s)          & \textbf{0.547}  (0.27s)     & 0.717 (0.54s)      & 0.984 (1.29s)          & 0.984 (24.5s) \\
$\textbf{ADU100}$ & \textbf{0.602} (15.3s)          & \textbf{0.602}  (3.80s)     & 0.760  (0.15s)      & 1     (-)              & 1    (-) \\
\hline
\end{tabular}}
\end{center}
\end{table}

\begin{table}
\begin{center}
\caption{Comparison between variants of sMIP algorithm. Optimal results are in \textbf{bold}.}\label{tab:mixedstatic}
\smallskip\noindent
\resizebox{\linewidth}{!}{%
\begin{tabular}{l|ccc|ccc|ccc}

\hline
       &\multicolumn{3}{c|}{sMIP-S}&\multicolumn{3}{c|}{sMIP-M}&\multicolumn{3}{c}{sMIP-C}            \\
\cline{2-10}
       & \textbf{SSR} & \textbf{CSR} & \textbf{MSR} & \textbf{SSR} & \textbf{CSR} & \textbf{MSR} & \textbf{SSR} & \textbf{CSR} & \textbf{MSR}   \\
\hline
$\textbf{R2000}$  & \textbf{0.233}   & 0.682 &  0.46  & 0.303   & 0.356   &  \textbf{0.330}  & 0.313   & \textbf{0.354}  & 0.334     \\
$\textbf{R4000}$  & \textbf{0.279 }  & 0.856 &  0.568 & 0.421   & 0.536   &  \textbf{0.478}  & 0.452   & \textbf{0.53}   & 0.498\\
$\textbf{ADX05}$  & \textbf{0.135}   & 0.958 &  0.547 & 0.17    & 0.794   &  \textbf{0.482}  & 0.492   & \textbf{0.67}   & 0.581    \\
$\textbf{ADX10 }$ & \textbf{0.144}   & 0.992 &  0.568 & 0.409   & 0.63    &  \textbf{0.519}  & 0.468   & \textbf{0.606}  & 0.539   \\
$\textbf{ADU10 }$ & \textbf{0.547}   & 1     &  0.774 & 0.547   & 1       &  \textbf{0.774}  & 0.912   & \textbf{0.984}  & 0.953    \\
$\textbf{ADU100}$ & \textbf{0.602}   & 1     &  0.801 & 0.602   & 1       &  \textbf{0.801}  & 0.851   & \textbf{1}      & 0.925  \\
\hline

\end{tabular}}
\end{center}

\end{table}
% mingyu no explanation so far on table 2, I would include the "final" susccess rate, i.e., 0.5*SSR+0.5*CSR in table and make that the most prominent numbers for comparison
\subsection{Experiments on synthetic dynamic graphs}
\label{sec:dynamicex}

In Table \ref{tab:dynamic}, we present our experiments on the $ADX05$ graph with a budget of 20. Only computer nodes are blockable. We consider all user nodes in the graph as the entry nodes. After preprocessing, the numbers of user nodes and computer nodes are 72 and 500 respectively. Each edge is randomly ON in each run with probability $p = 0.5$. Ideally, we will run our dyMIP(m) on every possible instance to get the optimal blocking policy. However, this is computationally infeasible. A more practical way is to derive a blocking policy on a training set (sampled from the network) and evaluate it on a different testing set (future states of the network). In our experiments, we use 5,000 instances for training and run a Monte Carlo simulation on 100,000 graphs to evaluate the effectiveness of our blocking plan to ensure statistically significant results for the evaluation.

%In the next part, we will explain why 100,000 samples are enough to ensure statistically significant results for the evaluation.
%In our Monte Carlo Simulation, by testing our blocking polity on $n =$ 100,000 graph samples. We proved that the results test on 100,000 can be at least 99\% confident with error within $\epsilon = 5.147\times 10^{-3}$ of the true mean (proof provided in the technical appendix.

% \vspace{0.75cm}
\textbf{Estimating the Optimal Lower Bound.} We compute the Lower Bound in~\eqref{eq:7} for the AD graphs using a sample size of 100,000 graphs. 
%Based on the Central Limit Theorem, we can be confident in treating this result as an accurate estimate of the Lower Bound. Each batch in this experiment can be viewed as an independent draw, and the Lower Bound is calculated as the average over multiple batches.
We find a Lower Bound of CSR and SSR on a graph using dyMIP-$C(m)$  and dyMIP-$S(m)$  respectively. We set $m=50$ and run it 200 times to cover all 100,000 graph instances in Monte Carlo Simulation.
After conducting this experiment on ADX05, the estimated the Lower Bound for SSR is  $0.7240\pm0.0055$ and for CSR is $0.8889\pm 0.0019$.
These estimates are used to measure the optimality of our solutions.

\begin{table}
\begin{center}
\caption{Experiments on dynamic graphs ADS05. We show the success rate of each algorithm with respect to their intended attacker (e.x. We show the Simple attacker Success Rate (SSR) for dyMIP-S algorithm, etc.). Values in parentheses indicate average computation time in seconds.}\label{tab:dynamic}
\smallskip\noindent
\resizebox{\linewidth}{!}{%
\begin{tabular}{l|l|llll}

\hline
      % \multicolumn{9}{c||}\textbf{dyMIP-S} \\
    & &\multicolumn{1}{c}{m = 1$^{\mathrm{1}}$} & \multicolumn{1}{c}{m = 10} & \multicolumn{1}{c}{m = 50}   & \multicolumn{1}{c}{m = 100}        \\
% \cline{2-7}
  % \textbf{dyMIP-S}     & \textbf{SSR}     & \textbf{SSR}     & \textbf{SSR}  & \textbf{SSR}     \\
\hline
\multirow{3}{*}{dyMIP-S} &\textbf{Rand} & 0.7960 (2s)       & 0.7430 (16s)     & 0.7359 (89s)       & 0.7326 (174s)    \\
                         &\textbf{Kmean}  &                   & 0.7409 (919s)    & 0.7343 (987s)      & 0.7325 (1062s)         \\
                         &\textbf{Vote} & 0.7304 (10232s)   & 0.7293 (8403s)   & 0.7293 (9095s)     & \textbf{0.7293} (9643s)        \\
\hline

\hline

\multirow{3}{*}{dyMIP-M}   &\textbf{Rand} & 0.8851 (2s)        &0.8494 (17s)       & 0.8461 (98s)           & 0.8407 (216s)     \\
                          &\textbf{Kmean}  &                    & 0.8544 (741s)    & 0.8398 (1011s)         & 0.8373 (1310s)   \\
                          &\textbf{Vote} & 0.8417 (2416)      & 0.8391 (1716s)       & 0.8353 (2136s)            & \textbf{0.8353} (2666s)  \\
\hline

% \textbf{dyMIP-C}   & \textbf{SSR} & \textbf{CSR}      & \textbf{SSR} & \textbf{CSR}      & \textbf{SSR} & \textbf{CSR}  & \textbf{SSR} & \textbf{CSR}     \\
\hline
\multirow{3}{*}{dyMIP-C} &\textbf{Rand}    & 0.9722 (2s)    &  0.9028 (16s)       & \textbf{0.8889} (95s)      & \textbf{0.8889} (188s)     \\
                         &\textbf{Kmean}     &                &  0.9000 (949s)      & \textbf{0.8889} (1010s)    & \textbf{0.8889} (1069s)    \\
                         &\textbf{Vote}    & 0.9042 (9188s) &  0.9028 (8718s)     &  \textbf{0.8889} (10684s)  &  \textbf{0.8889} (12878s)    \\
       \hline

\multicolumn{6}{l}{$^{\mathrm{1}}$ Random with m = 1 is staticMIP (static graph's policy on the dynamic graph)}

\end{tabular}}
\end{center}

\end{table}

\textbf{Discussion of Dynamic Results.} Table~\ref{tab:dynamic} summarizes our experimental results. Each experiment is conducted with 10 trials. We run our algorithm with 3 settings: pure simple attacker (dyMIP-S and $\varphi = 0$), pure competent attacker (dyMIP-C and $\varphi=1$) and mixed-type attacker (dyMIP-M and $\varphi=0.5$).  We randomly draw a training set of 5,000 graphs for each trial. 
We also run the 2 proposed heuristics: voting-based heuristic (\textit{Voting}) and clustering-based heuristic (\textit{Kmean}). We also include a third algorithm (\textit{Random}), which randomly picks $m$ graphs from the samples and applies MIP(m) to that $m$ graphs only.  
For each setting, we set $m$ to 1, 10, 50, and 100 (note that Kmean does not make sense for $m=1$ since one cluster is not meaningful).
In Table~\ref{tab:dynamic}, note that \textit{Random} with $m = 1$ is equivalent to running the MIP algorithm on one static graph (\textit{staticMIP}), and \textit{Voting} with $m = 1$ is applying the voting heuristic on \textit{staticMIP}. For \textit{Voting}, we run it with 10 parallel tasks. For example, \textit{Voting} with $m = 50$ with 5,000 graphs will require 100 rounds of MIP(m) if we run it sequentially, which is infeasible without parallel computing. The reported time for Voting is the sum of the time of 10 parallel processes. 
As \textit{Random} and \textit{Kmean} only require running MIP(m) once,  we run these algorithms without parallelization. 
For evaluation, we test our algorithms on 100,000 randomly drawn graph samples. %We also split our testing into 10 parallel tasks to save time. 

\textbf{Result Interpretation.} We observed that static defense policy (staticMIP or Random with $m=1$) in all 3 $\varphi$ values performs poorly against all types of attackers.
We are able to obtain optimal CSR (= 0.8889) for various settings with dyMIP-C and dyMIP-M.
While for SSR (= 0.7293), we are able to obtain near-optimal results with dyMIP-S using Voting heuristic and m = 100.
Our experiments demonstrated that \textit{Kmean} outperforms \textit{Random} in defending against simple attackers (yielding better SSR) under various settings. 
This shows that the clustering method (\textit{Kmean}) can improve the ``quality'' of the samples and therefore improve the performance of MIP(m). 
Notably, when only considering settings that can obtain the optimal policy against a competent attacker (CSR = 0.8889), Kmean with m = 100 outperform others in term of SSR (= 0.7781).  
% mingyu I don't understand the above sentence
As the voting heuristic requires running dyMIP(m) on every possible graph in the training set to perform majority voting. Therefore, \textit{Voting} easily outperforms the other methods in terms of SSR as it scales.
However, this makes it computationally more expensive compared to Kmean and Random, which only require running dyMIP(m) once. Therefore, Kmean is more favourable when computational resources are limited. We remark that the success rate obtained from this experiment is quite high for the budget of 10. It is due to the high number of HasSession edges makes the graph very dense (each snapshot could contain $\approx 36,000$ HasSession edges) and hard to defend with 10 honeypots. However, in the next experiment, we will show that our approach could obtain a manageable success rate in a realistic attack graph.

\subsection{Experiment on real dynamic AD graphs}

\begin{figure}
\centering
\begin{subfigure}{0.48\textwidth}
    \includegraphics[width=\textwidth]{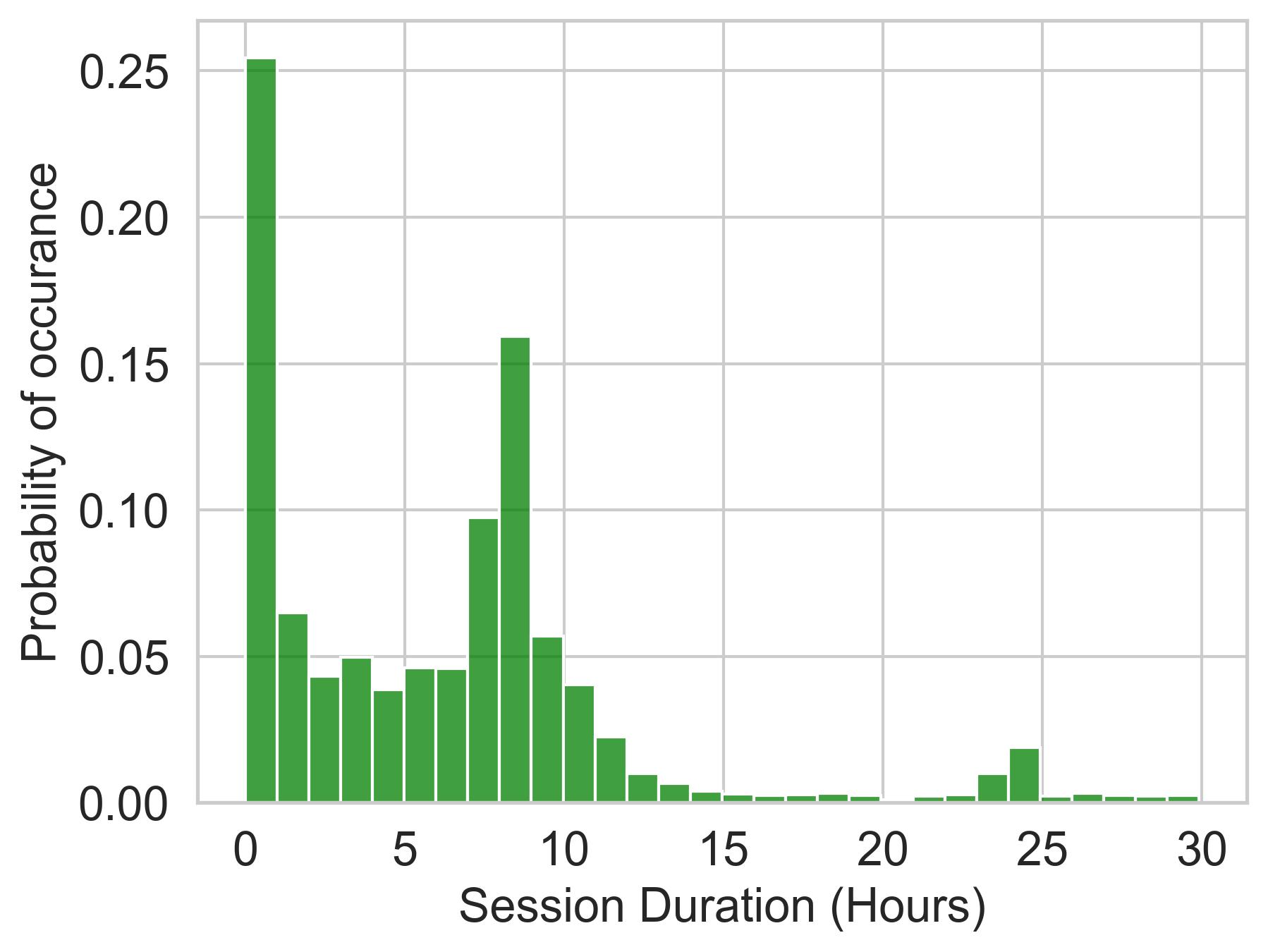}
    \caption{}
    \label{fig:session_length}
\end{subfigure}
\hfill
\begin{subfigure}{0.48\textwidth}
    \includegraphics[width=\textwidth]{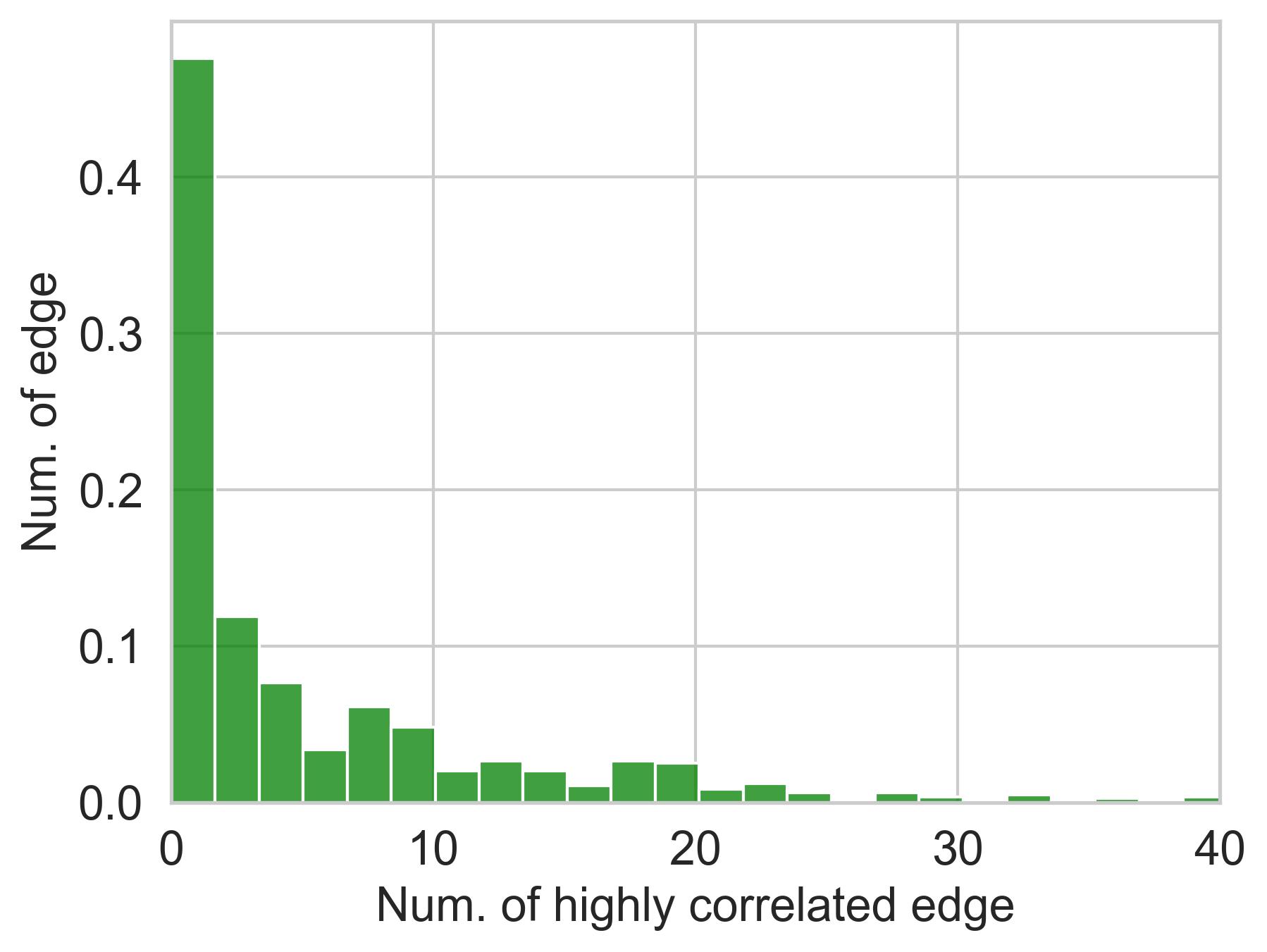}
    \caption{}
    \label{fig:correlation}
\end{subfigure}     

\caption{(a) Number of online sessions and newly established session sampling every 1 hour (b) Histogram of a number of edges which are highly correlated (in terms of logon/off session pattern) with others.}
\label{fig:figures}
\end{figure}

\begin{figure}
\centering
\begin{subfigure}{0.48\textwidth}
    \includegraphics[width=\textwidth]{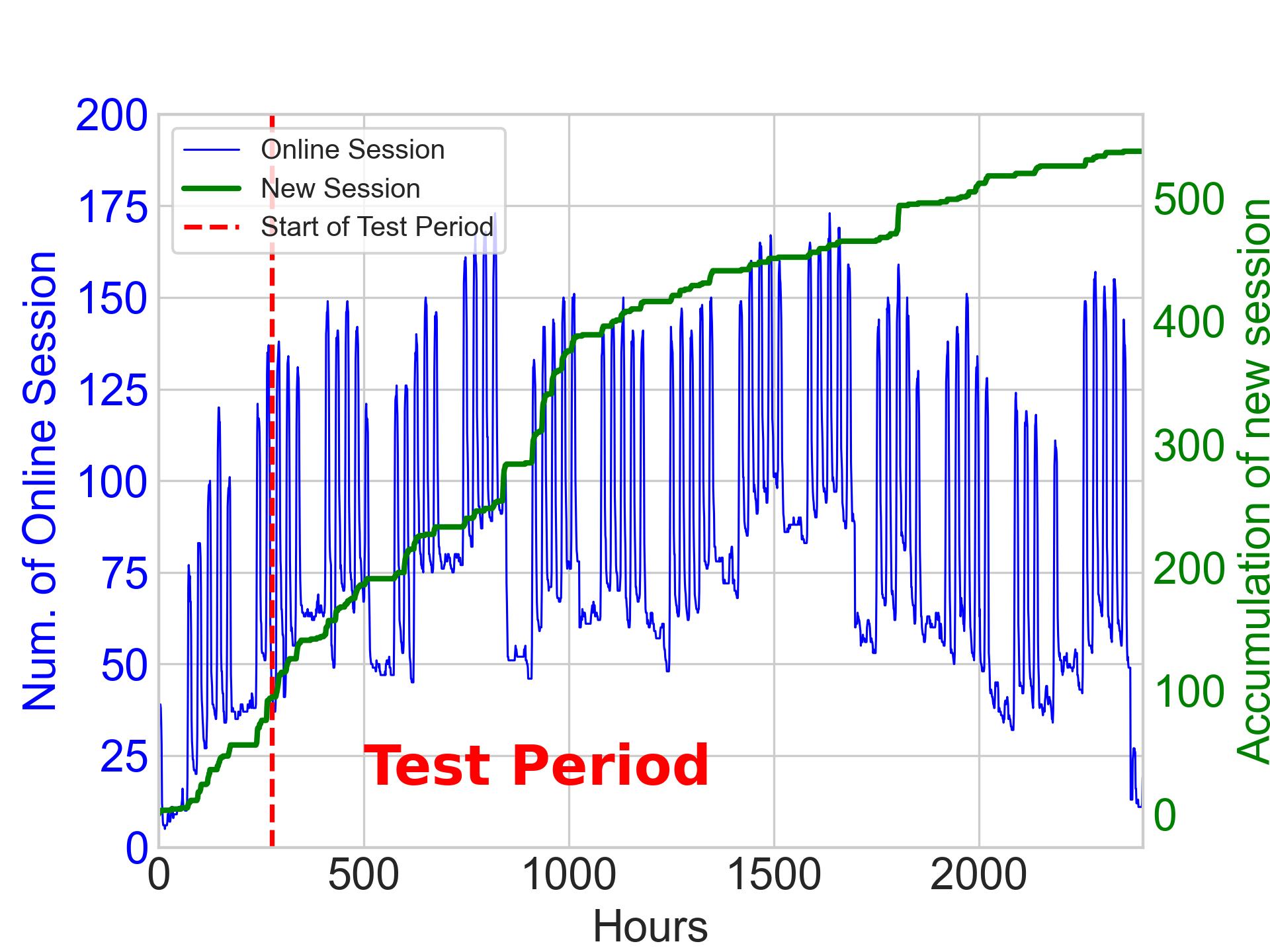}
    \caption{}
    \label{fig:online_session}
\end{subfigure}
\hfill
\begin{subfigure}{0.48\textwidth}
    \includegraphics[width=\textwidth]{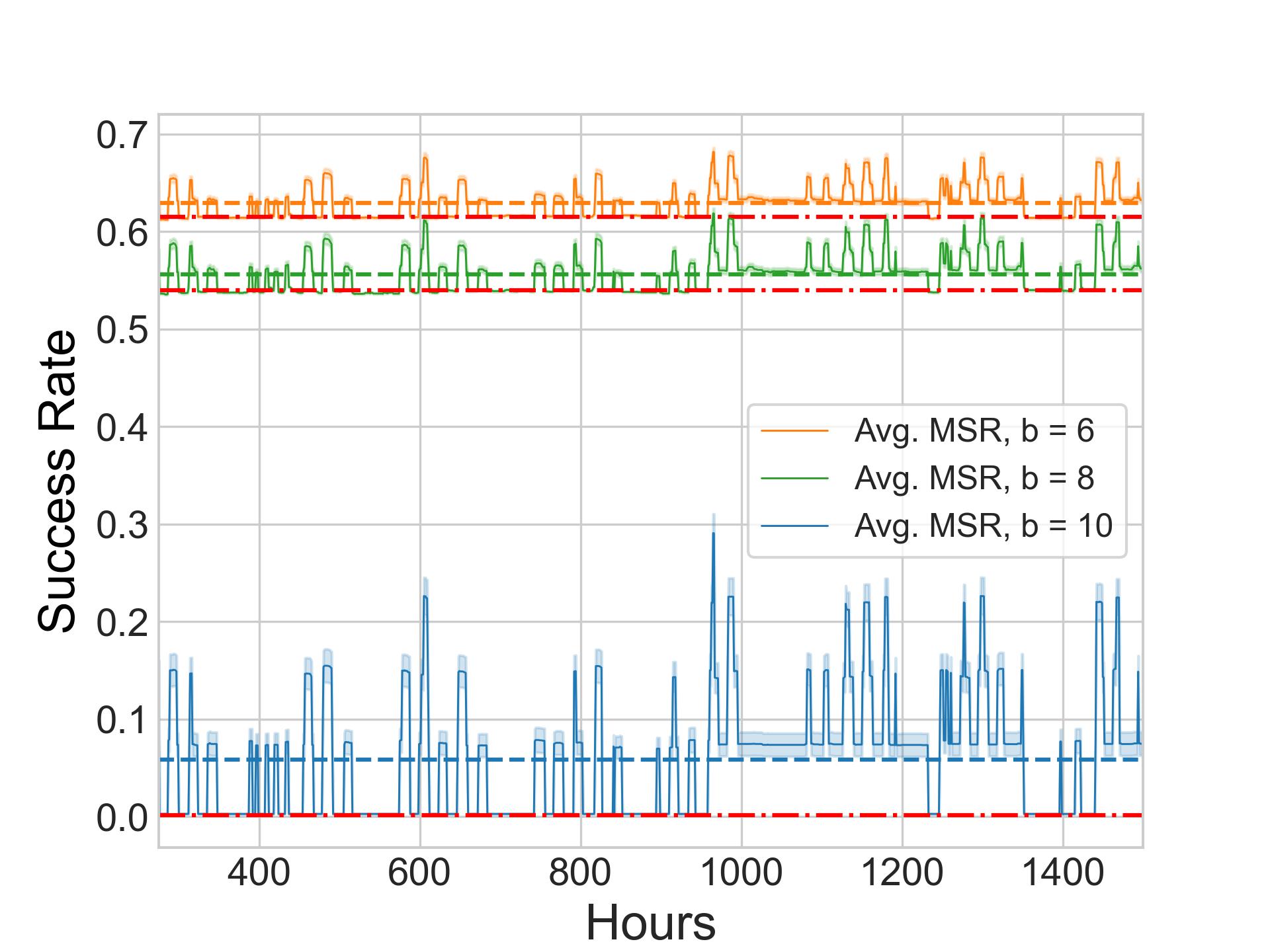}
    \caption{}
    \label{fig:resultcomp}
\end{subfigure}     

\caption{(a)  Number of online sessions and the accumulation of newly established session sampling every 1 hour (b)Voting with m = 10. The red dash-dotted lines are the optimal MSR of b = ${6, 8, 10}$, listed from top to bottom.}
\label{fig:figures}
\end{figure}

\textbf{Analyzing authentication data.} We analyze in this section the temporal properties of HasSession edges in a real network and apply these properties to simulate realistic dynamic AD graphs. The dataset used in this section is sourced from an anonymous organization and referred to as \textbf{AuthOrg} data. The dataset includes the authentication events observed over a 102-day period (from 24-06-2022 to 04-10-2022). There are 270 users, 338 computers, 14249 sessions and 1015 unique sessions in the dataset. Upon analyzing the \textbf{AuthOrg} data, we find that the random authentication scheme in Section \ref{sec:dynamicex} significantly diverges from realistic session patterns. Our analysis reveals that the activities of Hassession edges are not independent; rather, they exhibit correlations with each other. We highlight this with three observations: First, the graph shows a higher density of HasSession edges during working hours, indicating a greater likelihood of sessions appearing in that period (shown on Fig. \ref{fig:online_session}); HasSession edges can persist for extended periods, for example, a typical employee session lasts about eight hours before sign-off (Fig. \ref{fig:session_length}). Third, session profiles are often correlated across edges (Fig. \ref{fig:correlation}). In this Fig., we present a histogram of highly correlated session profiles for each edge. For every edge $e_i$, we count the number of session profiles in the set $E_t\backslash e_i$ that exhibit a high correlation with $profile(e_i)$ 
We define two vectors as \textit{highly correlated} if their Pearson correlation is greater than $0.5$. Fig. \ref{fig:correlation} provides evidence that a significant portion of the edges exhibits high correlations with at least one other edge where the percentage of edge profiles that have at least 1 or more highly correlated edges $P[X \geq 1]$ is more than $63.8 \%$.

In addition to correlated edges, another factor that could affect the performance of our honeypot placement algorithm is the number of new or rare link connections. New connections is connection that the establishment of unique sessions for the first time (edges that did not exist before).
The accumulative new connection session line in Fig. \ref{fig:online_session} illustrates a continuous establishment of new sessions throughout the entire dataset duration. These new links can introduce alternate attack paths that were not considered during the training period, thereby contaminating the graph with unexpected paths and degrading the performance of the defense policy over time.
We show in the next section that our algorithms by sampling representative behaviours across time horizons are robust to these degradation factors. 

\textbf{Simulated AD graphs with realistic HasSession edges.} 
We augment the Adsimulator AD graphs in the previous sessions with the log-on and log-off data to generate more realistic AD graphs. Here, we randomly map users and computers from the authentication data to the synthetic graph which we refer to as the "AuthOrg" graph. We train our algorithm using the first 250 hours of data, excluding the initial 25 hours for training stability. The subsequent 2,073 hours are used for testing, resulting in a total of 2,448 hours of data. For ease of comprehension, the figures in this chapter only display results up to 1,500 hours. We take snapshots on the graph every 0.25 hours, so we have 1000 snapshots for the training in total. In the experiments, 80\% of nodes are blockable (this number could be changed and didn't affect the overall results). We used the voting heuristic with $m = 10$ for all experiments. Additionally, each set is tested on 10 trials.

\textbf{Result Interpretation.} In Fig. \ref{fig:resultcomp}, we present the results for varying the budget, specifically $b = {6, 8, 10}$. Our algorithm demonstrates strong performance with low optimality gaps for each setting, with gaps of 0.0129, 0.0189, and 0.049 for $b = 6, 8,$ and $10$, respectively. We are able to archive the low attacker success rate of 5.83 $\%$ on average with a budget of 10. Notably, we observed that the attacker's success rate is influenced by two main factors: the number of online edges in the network and the presence of certain new link connections (rare links) established in the network.

The experimental results exhibit a high correlation with the number of online sessions plotted in Fig. \ref{fig:online_session}. For instance, with the $b = 10$ setting, the Pearson correlation between the number of online sessions and the success rate over time is 0.5870. Our algorithm shows overall resilience against the effects of new link connections when they are introduced linearly over time. The (Pearson) correlation between the number of new connections and the success rate over time, with $b = 10$ setting, is only 0.07. Still, we observed a burst of new connections occurring throughout week 6 (from hour 800th to 1000th, Fig. \ref{fig:online_session}), which had a noticeable impact on the overall algorithm performance, as visually observed in Fig. \ref{fig:resultcomp}. % \hn{Explain what causes this? rare links?} 
Our hypothesis is that during these periods there are \emph{rare new links} that the $m$ sample graphs do not capture. 
% \hn{Why don't you try to increase $m$}

Overall, the algorithm performs well with low optimality gaps. However, it tends to perform worse during periods with an increased number of connections in the network, such as working hours. However, it is explainable that as the graph becomes more populated with Hassession edges, the number of paths to the Domain Admin (DA) increases, subsequently elevating the attacker's success rate. Another reason is that we also assumed that we have a limited number of honeypots and having more honeypots will suppress this trend. Additionally, the presence of rare certain new link connections during the test period can also affect the algorithm's performance. This limitation can be addressed by periodically retraining the algorithm to adapt to changing network conditions or increasing the number of snapshots. 
% \hn{Show results of this!!!}

Finally, Table~\ref{tab:graphstat} shows the scalability of our algorithms.  As shown, the largest network that we applied our algorithms to (ADU100) has 137,315 nodes and 1,490,766 edges. This is several orders of magnitude larger than typical graphs that any existing algorithms for honeypot placement can handle.

\section{Chapter Summary}

This chapter investigated a Stackelberg game model between an attacker and a defender on Active Directory attack graphs. The model considers two types of attackers, one of which can observe honeypots while the other cannot. The defender aims to allocate honeypots on attack graph nodes to minimize the attacker's expected success chance to reach DA. We show that the problems for both types of attackers are computationally hard and propose a mixed-integer programming (MIP) formulation to solve the problem when there is a mixed-types of attacker agents exist in the network, for both static and dynamic graphs. Our experiments showed that MIP scales well for a graph with $\approx$ 100 thousand nodes and $\approx$ 1 million edges. We show that our algorithm that combines m MIP (dyMIP(m)) instances of the graphs performs well under realistic temporal patterns. Our experimental results show that the MIP(m) algorithms produce near-optimal blocking plans, with the help of two heuristic methods based on majority voting and k-means clustering.

\section{Appendix-Complete Formulation for Multi-graph MIP}
\label{sec:appxmmip}

In this section, we provide the supplementary the complete formuluation for Multi-graph MIP. We define $v_{i,g}$ as the temporary variable for $y_{i,g}$
\begin{subequations}
\begin{align}
\text{min} \displaystyle\sum\limits_{g \in {G_s}}\displaystyle\sum\limits_{i=1}^{s}& \varphi R^{B}_{i,g} + (1-\varphi) \frac{y^B_{i,g}}{y_{i,g}} & \nonumber\\  \nonumber\\
% \end{flalign*}
% \begin{align}
  v_{i,g} =& \sum_{j \in n(i)}y^{B}_{j,g}, & \forall i \in V_g \backslash N_b, g \in G_s  \\
  % y_{i,g} =& \sum_{j \in n(i)} (1-B_i) y^{B}_{j,g}, & \forall i \in  N_b, g \in G_s \label{eq:1b} \\
  % y^B_{i,g} \leq& y_{i,g},         & \forall i \in V_g, g \in G_s \label{eq:1c} \\
  % y^B_{i,g} \geq& 0,              & \forall i \in V_g, g \in G_s \label{eq:1d} \\
  y^B_{i,g} \leq& v_{i,g}, & \forall i \in N_b, g \in G_{s} & \\
  y^B_{i,g} \leq& y_{i,g} (1-B_i), & \forall i \in N_b, g \in G_{s}\\ 
  y^B_{i,g} \geq& v_{i,g} + y_{i,g} \cdot (-B_i) ,& \forall i \in N_b, g \in G_{s}\\
  y^B_{i,g} \geq& 0 ,& \forall i \in N_b, g \in G_{s}\\
  y^B_{i,g} =& v_{i,g},& \forall i \in V_g \backslash N_b, g \in G_{s}\\
  R^B_{i,g} \geq& R^B_{j,g} - B_{i}, &\forall (i,j) \in E_g,  \nonumber \\ 
                &                    &i \in N_{b}, g \in G_s \\
  R^B_{i,g} \geq& R^B_{j,g}        , & \forall (i,j) \in E_g,\nonumber \\
                &                    & i \in N \backslash N_{b}, g \in G_s \\
  \sum_{i \in V}\ B_{i} \leq& b \\
  B_{i},R_{i,g} \in & \{ 0,1 \} \\
  R_{DA}, y_{DA} =& 1 
\end{align}
\end{subequations}

%% file: Chapter4/Chapter4.tex
\chapter{Optimizing Cyber Response Time on Temporal Active Directory Network using Decoys} % Main chapter title
\label{chapter:paper2} % Change X to a consecutive number; for referencing this chapter elsewhere, use \ref{ChapterX}

In this chapter, we study the problem of placing decoys in AD network to detect potential attacks. We model the problem as a Stackelberg game between an attacker and a defender on AD attack graphs where the defender employs a set of decoys to detect the attacker on their way to Domain Admin (DA). Contrary to previous works, we consider time-varying (temporal) attack graphs. We proposed a novel metric called response time, to measure the effectiveness of our decoy placement in temporal attack graphs. Response time is defined as the duration from the moment attackers trigger the first decoy to when they compromise the DA. Our goal is to maximize the defender's response time to the worst-case attack paths. We establish the NP-hard nature of the defender's optimization problem, which leads us to develop Evolutionary Diversity Optimization (EDO) algorithms. EDO algorithms identify diverse sets of high-quality solutions for the optimization problem. Despite the polynomial nature of the fitness function, it proves experimentally slow for larger graphs. To enhance scalability, we proposed an algorithm that exploits the static nature of AD infrastructure in the temporal setting. Then, we introduce tailored repair operations, ensuring the convergence to better results while maintaining scalability for larger graphs.

\section{Introduction}

Active Directory is Microsoft's identity and access management system designed for Windows domain networks. It's widely adopted and plays a critical role in the networks of many enterprises and government bodies. However, its popularity has also made it a prime target for many cyber adversaries over the years. According to a report from Microsoft ~\citep{stat2}, there has been an alarming surge in attacks targeting AD users, with 30 billion attempted password attacks on AD accounts reported each month in 2023.

An Active Directory network naturally describes an attack graph, with \textbf{nodes} representing both physical and virtual entities such as users, computers, security groups, etc., and \textbf{directed edge} $(i, j)$ representing the vulnerability and accesses that an attacker can exploit to gain access from node $i$ to node $j$. BloodHound \footnote{https://github.com/BloodHoundAD/BloodHound} is one of the most influential tools for analysing/visualising the AD attack graph. BloodHound models the \textit{identity snowball attack}, a concept initially introduced by Dunagan et al.~\citep{dunagan2009heat}. The identity snowball attack models the sequence of attack in the network allowing them to gain access of higher privileges nodes from a low privilege node (ex. Account A $\xrightarrow[]{\text{AdminTo}}$ Computer B $\xrightarrow[]{\text{HasSession}}$ Account C).  However, Dunagan et al. \citep{dunagan2009heat} and several other works on defending Active Directory network \citep{guo2022practical, guo2023scalable, goel2022defending, goel2023evolving, zhang2023oracle} oversimply the problem with the assumption that AD network is static. 
% mingyu(c): edge is not just vulnerability, it could be access
% While this assumption might hold for highly scripted and fast-paced attacks, it is not always reflective of reality. 
In practice, the AD graph is very dynamic which will affect the security landscape overtime. One of the major sources of changes in the AD graphs is caused by users’ activities. In Windows systems, user authentication leaves behind credential material, typically in the form of a hash or clear-text password in the computer's memory. Adversaries can exploit this vulnerability, harvesting credentials for lateral movement. In BloodHound, this vulnerability is presented as 'HasSession'. HasSession edges will stay online until being removed from the graph when the user signs off from the computer after a period of time. In this work, we formally model the dynamics of the AD graph using the \textbf{temporal attack graph}, wherein attackers gain access to nodes in the AD graph through the \textit{identity snowball attack}, presented as \textit{temporal paths}. For example, in Figure \ref{fig:ADex}, the identity snowball attack in temporal graph for gaining access of account $U_3$ from compromised node $s_2$ can be the following temporal path: $s_2$ $\xrightarrow[]{\text{1}}$ $Cp_1$ $\xrightarrow[]{\text{2}}$ $Cp_3$ $\xrightarrow[]{\text{4}}$ $U_3$ where number on each arrow is time the attacker exploit the edge to gain the access to next node. The static attack graph can not capture this attack path if generated at time steps $t \in [1,4) \cup (6, 10]$

\begin{figure}[h]
  \includegraphics[width=0.9\linewidth]{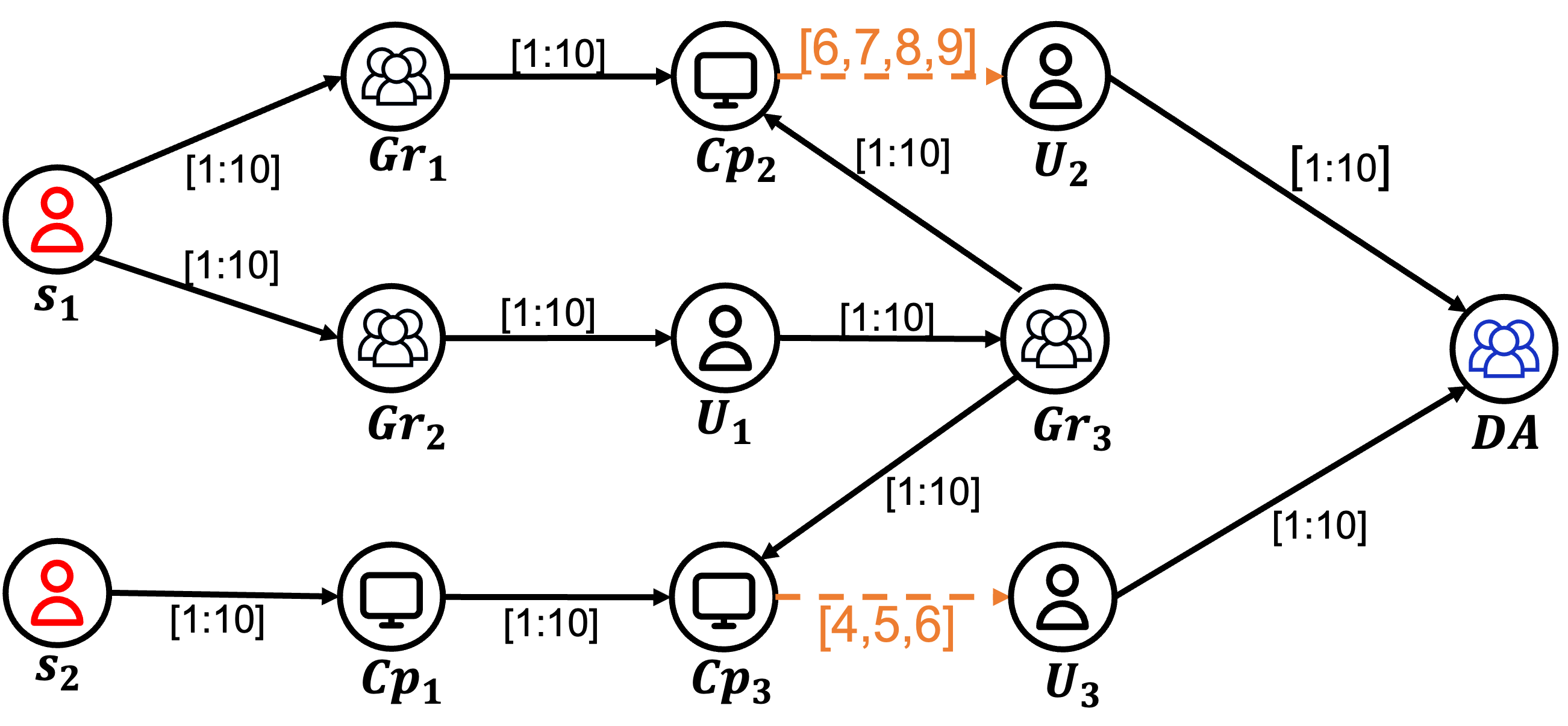}
  \caption{Example of an Active Directory graph sampled over a period of 10 time units. The timestamps on each edge indicates its appearance time. Black labels represent static edges, while orange labels denote dynamic edges (HasSession).}
  \label{fig:ADex}
\end{figure}

In this chapter, we study a method for defending temporal AD attack graph by using active defense with cyber decoys. Decoys or honeypots \citep{decoy1, decoy2} are fake assets such as fake users, and fake hosts that trigger an alert when attackers engage them. They are designed to lure attackers to them by mirroring authentic assets. By allocating decoys in strategic locations, they can serve as the sentinel for the early detection of the cyber threat. An early detection of a threat can increase the effectiveness of incident response process of IT admin and reduce the further damage caused by the attack~\citep{stat1}. Motivated by the early detection use case of the network decoy, we proposed a problem for decoy allocation in temporal network called $maxRT$. In this problem, defender aims to optimize the \textbf{response time} of the decoy to any attack path. The response time is defined as the duration from the moment attacker triggers the first decoy to when they can reach DA. Defender want to maximize the response time to ensure the early detection, providing IT admin with sufficient time to respond to the incident before the attacker reaches the DA. 

We model our problem as a two-player Stackelberg game model with a pure strategy. In our game, the defender (leader) wants to allocate at most $b$ decoys on a set of \textit{blockable} nodes. The defender's allocation intercepts the attacker's temporal attack path while maximizing the response time of the allocation to ensure early detection. The attacker (follower) has access to a set of compromised entry nodes. The attacker can also observe the entire temporal AD graph and the defensive strategy. This assumption is based on the practicality of attackers employing reconnaissance tools similar to Sharphound \citep{SharpHound} to collect data from domain controllers and build an AD attack graph. The attacker’s strategy specifies an entry node, and from it, a temporal attack path to DA. The attacker's best response is to choose a temporal attack path that has the minimum response time. We will later show that the attacker's optimal plan can be found in polynomial time.
% mingyu: maybe mention it is realistic to assume that the attacker can observe the whole graph and the defensive strategy (i.e., APT and sharphound etc.)

\textbf{Chapter Contribution.} This chapter aims to propose a solution for the decoys allocation in Active Directory problem. We first prove $\mathcal{NP}$-hard nature of the defender's combinatorial optimization problem. We then introduce the Evolutionary Diversity algorithm as a heuristic solver. However, the vanilla EDO algorithm does not scale well for our problem when it fails to converge to any feasible solution (response time > 0) in some graphs. When we mention "vanilla" EDO algorithm, we refer to directly applying the EDO implementation from Goel et al. \citep{goel2022defending, goel2023evolving} to our problem. In an attempt to run the vanilla EDO algorithm on the ADX10 graph in our experiment, the response time of the solution remains 0 even after 2 million iterations (equivalent to almost 3 days of computational effort).
To enhance our algorithm, we propose several improvements. Firstly, the computation of the attacker's optimal path relies on the earliest-arrival path, which is computationally slow in AD graphs. We observe an contradiction that the state-of-the-art algorithm for computing the earliest-arrival path becomes highly inefficient in temporal graphs with a large number of static edges. Despite the dynamic characteristics of the AD graph, a significant portion of the AD infrastructure remains static. To address this problem in AD-specific graphs, we present a novel Dijkstra-based algorithm for computing the earliest-arrival path which significantly improves the run-time of the fitness function.
Secondly, we introduce two constraint-handling techniques to tackle the difficulty of finding feasible solutions in the vanilla EDO. The first method introduces a repair mechanism using Integer Linear Program (ILP) to directly patch the infeasible solution every round. The second approach introduces the surrogate/penalty fitness function. The surrogate function is a lightweight fitness function that replaces the computationally expensive real fitness function, allowing the evaluation of individuals at a lower cost during each iteration. The surrogate function is designed to evaluate a solution on a set of "important" attack paths instead of the whole graph and penalize the infeasible individuals. We experimentally verify that our proposal effectively improves the scalability of the EDO algorithms on our problem.

\section{Model Description}

% In this chapter, we study a defence problem on the Active Directory attack graph assuming that the AD graph is a temporal graph. Here in this chapter, we consider 2 special elements of an advanced attacker such as APT, i.e., the persistence and the stealthiness of the attacker. To capture the time-depending element of the attacker model, we define a discrete sequence of time $\tau = \{1, 2, \dots, t_{max}\}$, where $t_{max}$ is the maximum lifetime. Next, we define the \textbf{temporal directed attack graph}:

\subsection{Background}
\textbf{Temporal directed graph} define as $G = (V, E_1,\cdots, E_{t_{max}}) = (V, E = (E_i)_{i\in [t_{max}]})$) where $V$ is a set of vertices in the graph and $E_i$ is the set of edges at time $i$. We denote the tuple $(u, v, t) \in E_t$ the edge from $u$ to $v$ appears at time $t$. For the sake of model simplicity, we assume that every edge has a duration of 1. In other words, if an attacker traverse edge $(u, v, t)$, they will start from node $u$ at time $t$ and arrive $v$ at time $t + 1$. Despite this simplification, all our algorithms remain effective in more general settings in which the duration of every edge is larger or equal 1. We call $t_{max}$ the lifetime of the graph. 
% We denote $G' = (V', (E'_i)_{i\in [t_{\alpha}, t_{\omega}]})$ as the temporal subgraph of $G$ where $V' \subseteq V$, $\bigcup_{i = t_{\alpha}}^{t_{\omega}}E_i = V'$, and $[t_{\alpha}, t_{\omega}]$ is called the interval of $G'$, where $1 \leq t_{\alpha} < t_{\omega} \leq t_{\text{max}}$. 
We also define the \textbf{underlying graph} of graph $G$ as $G_{\downarrow} = (V, E_{\downarrow})$ where $E_{\downarrow} = \bigcup_{t=1}^{t_{max}}E_t$. For the ease of demonstration in the chapter, we also denote $time\_label(u, v) = (t_i)_{i=1}^k$ as a (ascending) sorted list of time units that edge $(u, v)$ appears (or is on) if $(u, v, t_i) \in E_i$ where $t_{i} \in time\_label(u, v)$. Otherwise, we say edge $(u, v)$  disappears or is off at time step $t_i$ if $t_i \notin time\_label(u, v)$. We denote a set of \textbf{static edges} as $E_s$, we say an edge $(u, v) \in E_s$ is static if they appear in every time step throughout the graph's lifetime or $|time\_label(u, v)| = t_{max}$. Similarly, we denote set of \textbf{dynamic edges} as $E_d$, we say an edge $(u, v) \in E_d$ is dynamic if they disappear from the graph for some of the time units or $|time\_label(u, v)| < t_{max}$.

\textbf{Temporal $(s, d)$-path} is defined as a sequence of edges in graph $G$ exhibiting a monotonic increase in edge labels. For any two distinct nodes $s, d \in V$, a temporal path between two vertices $s$ and $d$ is represented by the sequence of edges: $\pi$ = $\pi(s, d) = \langle(s = v_0, v_1, t_1), (v_1, v_2, t_2), \dots, (v_{k-1}, v_k = d, t_k) \rangle = \langle(v_{i-1}, v_i, t_i )\rangle_{i=1}^{k}$ where $v_i \neq v_j$ and $t_i < t_j$ for all $i, j \in \{0, \dots, k\}$ with $i \neq j$.  
We denote $start(\pi) = t_1$ and $end(\pi) = t_k + 1$ as the \textbf{starting time} and \textbf{ending time} of a path $\pi(s, d)$. We further denote by $dur(\pi(s, d)) = end(\pi(s, d)) - start(\pi(s, d))$ the duration of travelling from the starting vertex to the ending vertex of the path $\pi(s, d)$. 
% Next, we may refer to a \textbf{sub-temporal path} of $\pi$ as $\pi(v_x, v_y) = \langle(v_{i-1}, v_i, t_i )\rangle_{i=x}^{y}$ where $\pi(v_x, v_y) \subseteq \pi = \langle(v_{i-1}, v_i, t_i )\rangle_{i=1}^{k}$ and $1 \leq x < y \leq k$. 
% We note that is a more general way to denote temporal path: $\pi(s, d) = \langle (v_{i-1}, v_i, t_i, \lambda_i) \rangle_{i=1}^{k}$, where $\lambda_i$ denotes the time taken to traverse from $v_{i-1}$ to $v_{i}$ through the edge $(v_{i-1}, v_i, t_i, \lambda_i)$. However, for the sake of model simplicity, we assume $\lambda_i = 1$ for all edges.
Next, we define a set of every possible temporal path from $s$ to $d$ between interval $[t_{\alpha}, t_{\omega}]$ as $\Pi(s, d, [t_{\alpha}, t_{\omega}]) = \{\pi:\pi \text{ is a (s, d)-temporal path such} $ $ \text{that} start(\pi) \geq t_{\alpha}, end(\pi) \leq t_{\omega} \}$. Then, a path $p \in \Pi(s, d, [t_{\alpha}, t_{\omega}])$ is an \textbf{earliest-arrival path} if $end(\pi) = min\{ end(\pi'): \pi' \in \Pi(s, d, [t_{\alpha}, t_{\omega}])\}$.
% mingyu: end is the min of the set?

% In the existing literature, two models for temporal paths have been identified: paths that can traverse multiple edges per time step (non-strict temporal path) and paths that are constrained to only one edge per time step (strict temporal path). We assume to use only strict temporal path in our model.

\textbf{Temporal $(s, d)$-cut}, also known as a temporal $(s, d)$-separator, refers to the set of nodes $C(s, d)$ in the graph $G$ that the removal of every node in set $C(s, d)$ will disconnects all temporal paths from $s$ to $d$. \textit{It is essential to note that in this chapter, the terms "cut" or "separator" specifically refers to the allocation of decoy on vertices}. When we employ a temporal $(s, d)$-cut the graph, we guarantee every (s, d) temporal path has a contact with the cut $C$.
% We denote $minC(s, d)$ as the minimum cut of the temporal graph if $|C(s, d)| \leq |minC(s, d)|$ for any temporal cut $C(s, d)$. 
Given a path $\pi(v_0, v_k) = \langle(v_{i-1}, v_i, t_i)\rangle_{i=1}^{k}$ in graph $G$ and a cut $C(v_0, v_k)$, we define a node $v$ as the \textbf{first point of contact} between the path $\pi(v_0, v_k)$ and the cut $C(v_0, v_k)$ if $v \in I : \forall u \in I, \text{dur}(v_0, v) \leq \text{dur}(v_0, u)$ where $I = V(\pi(v_0, v_k)) \cap C(v_0, v_k)$. In plain language, the first point of contact represents the first honeypot encountered when following the path. 
% mingyu: shall we emphasis there is always contact; we plan to place decoys on C; and only computer nodes can be decoys?

\textbf{Response time} denoted as $RT$ is a key parameter introduced in this chapter for our specific problem. 
The response time of a path $\pi$ is defined as the duration between the moment when the attacker encounters or triggers the first honeypot and the time when the attacker compromises the Domain Admin while following path $\pi$.  
Let's us consider the temporal path $\pi(s, DA) = \langle(s = v_0, v_1, t_1), \dots, (v_{k-1}, v_k = DA, t_k) \rangle = \langle(v_{i-1}, v_i, t_i )\rangle_{i=1}^{k}$ with $v_x$ where $1 < x < k$ is the first point of contact of $\pi(s, DA)$ and the defense solution $C(s, DA)$. The response time of path $\pi(s, DA)$ is defined as $RT(\pi, C) = dur(\pi(s, DA)) - dur(\pi(s, v_x)) = t_k - t_{x} $
As a defender, we want to maximize the response time of every temporal path in the attack graph to let IT admin have enough time to react to the incident. 
% Formally, we define the response time of path $\pi(s, DA) = \langle(s = v_0, v_1, t_1), \dots, (v_{k-1}, v_k = DA, t_k) \rangle = \langle(v_{i-1}, v_i, t_i )\rangle_{i=1}^{k}$ with defender honeypot allocation $C(s, DA)$ as $RT(\pi, C) = dur(s, DA) - dur(s, v)$ where $v$ is the first point of contact of $\pi$ and $C$. 
% As we want to measure the effectiveness of our honeypot deployment, we may interested in the response time of the whole graph. We may refer to the response time of an graph $G$ with the decoy deployment $C$ as $\Delta RT(G, C)$  .

\textit{Example 2.1.} Figure \ref{fig:ADex} illustrates a temporal Active Directory graph. The graph includes two compromised users, denoted as $s_1$ and $s_2$. The graph consists of two sets of edges: static edges, allowing the attacker to move between nodes at every time step, and dynamic edges, which appear for a limited time. In this example, we assume the defender allocates honeypots to a set of nodes $C = \{Cp_2, Cp_3\}$.
Consider the following temporal path $\pi = \langle(s_1, Gr_1, 2), (Gr_1, Cp_2, 4),$ $ (Cp_2, U_2, 6), (U_2, DA, 7) \rangle$ from $S$ to $DA$. Assuming the attacker from $s_1$ chooses this path, the honeypot on node $Cp_2$ is triggered at time 4 (as the attacker steps on it), alerting the IT admin to the attacker's presence. In this context, node $Cp_2$ is considered as the first point of contact for the attacker. The response time, defined as the time from honeypot alert to the attacker compromising the DA, is $RT = dur(\pi(s_1, DA)) - dur(\pi(s_1, Cp_2)) = (7+1-2) - (4+1-2) = 3$ units (plus 1 due to the assumption that traversing each edge takes 1 time unit). During this window, the IT admin has 3 time units to isolate compromised systems and terminate the attacker's unauthorized session. Note that the proposed response time is a realistic model of real hackers' behaviour where they would wait in the system for a long time before an opportunity arises for the next movement.

\subsection{Problem formulation}

\textbf{Temporal directed attack graph}
We define an AD attack graph in our model as a Temporal directed graph $G = (V, E_1,\cdots, E_{t_{max}})$. Set of vertices V represents all physical and virtual entities such as user, computer, security group, etc. The set of edge $E_i$ denotes the link modelling the security dependency and relationships between entities which represent vulnerabilities for attacker to make lateral movements.
There is a set $S\subseteq V$ of initial footholds called entry vertices, and the attacker has already compromised these vertices at the start of the attack. 
The attack goal is to compromise the Domain Admin (DA), the attacker can laterally move through the network using any of the \textbf{temporal (s, DA)-path}.

\textbf{Formulation with game theory} The problem of defending a temporal AD network with honeypots can be modelled as a Stackelberg game.
In our proposed model, the defender can deploy a set of honeypots on a set of vertices $C$ (a cut) such that form a temporal $(S, DA)$-cut. In our model, each honeypot will "monitor" any malicious activities on their allocated vertices. The honeypots will set an alert to IT admin once the attacker steps on one of these vertices. Defender can only allocate honeypot on a set of blockable vertices, denoted by $N_b \subseteq V$. 
% mingyu: i.e., blockable = computer nodes?
In consideration of a worst-case scenario, we assume the attacker has full visibility into the temporal graph and the honeypot placements. The attacker can bypass these honeypots if the honeypot's placement does not form a temporal $(s, DA)$-cut, in which case the response time is 0. 
Consequently, when the budget of the defender problem is exactly the size of the minimum temporal cut, our problem's solution is also the solution for the minimum temporal $(s, d)$-separator problem \citep{zschoche2020complexity} which is known to be a $\mathcal{NP}$-complete problem. However, our problem goes beyond this by also maximizing the response time of the temporal cut which tends to "push" the solution further away from the DA. Generally speaking, nodes further away from the DA tend to be lower privilege nodes instead of servers or admin. Therefore, our solution incurs lesser disruption to the network.
We say $C$ is a defender's \textbf{feasible solution} if $C$ is strictly a temporal $(S, DA)$-cut, otherwise, it is a \textbf{infeasible solution}. Strategically, when facing a defence solution $C$, the attacker selects a path that minimizes the response time. The \textbf{attacker optimal attack} path can be found via $\min_{\pi\in \Pi} RT(\pi, C)$ where $\Pi$ is the set of every possible temporal path between each vertex $s \in S$ to DA. In contrast, the defender aims to find a cut $C$ that maximize the response time. The \textbf{defender's objective} is formulated as

\begin{equation}
\max_{C\subseteq N_{b}, |C|<b}\{\min_{\pi\in \Pi} RT(\pi, C)\}.
\end{equation}

\begin{theorem} Defender's problem is $\mathcal{NP}$-hard.
\label{theorem:np}
\end{theorem}

\begin{proof} The proof is based on a reduction from the strict temporal $(s, d)$-seperator (strict-TS) problem which is $\mathcal{NP}$-complete \citep{zschoche2020complexity} for graph of lifetime $\geq 5$. 

\noindent \textbf{Problem: } Strict-TS
\begin{itemize}
    \item \textbf{Input:} A temporal graph $G = (V, E_1,\cdots, E_{t_{max}})$, source node $s\in V$, destination $d\in V$ and $k \in \mathbb{N}$
    \item \textbf{Output:} Does $G$ admit a temporal $(s, d)$-seperator of size at most $k$
\end{itemize}
The proof gadget for the strict-TS is illustrated in Figure \ref{fig:NPgadget}.a and the complete proof is provided in \citep{zschoche2020complexity}. 

Now for our max-$RT$ problem, the high level idea for the hardness proof is that the solution for the max-$RT$ problem can be found via solving the strict-TS problem. Let us define an instance of strictTS problem $G_{ts} = (V_{ts},  E_{ts,1},\cdots, E_{ts, t_{max}})$. We define a source node $s \in V_{ts}$ and destination node $d \in V_{ts}$. For the detailed construction of other nodes and edges in strict-TS, we refer the reader to Theorem 3.1 of \cite{zschoche2020complexity}. Let $minC_{ts}$ represent the solution to the strict-TS problem.

Subsequently, we construct the proof gadget for the max$RT$ problem (Figure \ref{theorem:np}.b) as follows. We introduce two entry nodes, $s_1$ and $s_2$. At time $t_\alpha$, node $s_1$ is connected to node $s$ of a sub-graph constructed following the strict-TS instance. At time $t_{\alpha}+6$, we also connect $d$ from the strict-TS subgraph to $y_2$. We delay every edges in the Strict-TS instance by $t_{\alpha}$. We assume that $s$ and $d$ is not blockable. Assuming $s$ and $d$ are not blockable, we finalize the instance by adding the following remaining edges: $(y_2, DA, t_{\alpha} + 7)$, $(s_2, y_1, t_{\alpha})$, and $(y_1, DA, t_{\omega})$ where $t_{\omega} \geq t_{\alpha}+7$. The full construction for max$RT$ can be seen in Figure \ref{theorem:np}.b.

\begin{figure}[h]
    \centering
  \includegraphics[width=0.75\linewidth]{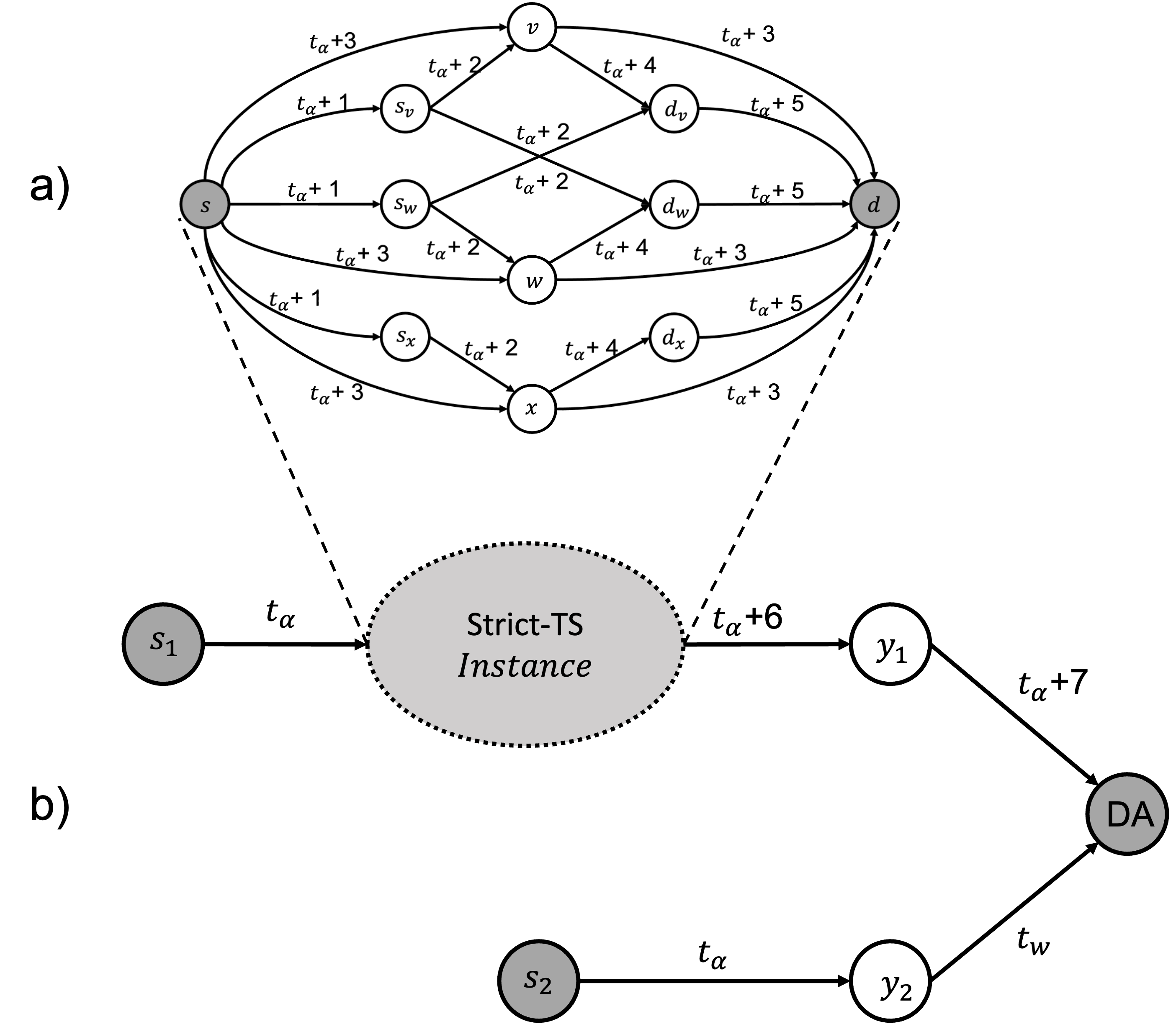}
  \caption{Proof gadget for Theorem \ref{theorem:np}. a) Proof gadget for Strict-TS problem. b) Proof gadget for max-$RT$ problem}
  \label{fig:NPgadget}
\end{figure}

With a defensive budget of $b = |minC_{ts}| + 1$, the optimal allocation involves locating the solution for the strict-TS instance and blocking vertices $y_2$. As the optimal solution of max$RT$ yield the optimal solution for Strict-TS, this implies that max$RT$ is $\mathcal{NP}$-hard.
% mingyu: and also block y_2?
\end{proof}

\section{Related Work}

\textbf{Identity Snowball Attack in dynamic environment}. In the literature, there are several efforts to model the identity snowball attacks with consideration of the dynamic nature of the attack graph. Chapter \ref{chapter:paper1} also study the honeypot/decoys allocation on Active Directory network with the consideration of the dynamic setting. However, their approach to modelling dynamism is somewhat simplistic. They capture the dynamic nature by taking independent static snapshots of the attack graph at each time step, treating each snapshot as an attacker's scenario in a static graph. Their allocation strategy jointly optimizes the number of attack paths in each snapshot. Chapter \ref{chapter:paper1} fail to model the identity snowball attack in the temporal graph. In practical scenarios, attackers can patiently "lurk" in a node until a more opportune path emerges. This characteristic makes our model more sophisticated and practical than theirs. Albanese et al. \citep{albanese2022formal} attempted to model the credential hopping attacks/identity snowball attacks on the time-varying user-computer graph. They assume that attacker does no observation on the network topology and employ a heuristic algorithm to find the upper-bound of the attacker attack effort. In contrast, our work considers the worst scenario where attacker have the observation on the attack graph and we can derive the optimal attack response. Pope et al. \citep{pope2018automated} also consider a similar model to Albanse et al. except they employ genetic programming to predict the attacker success rate/effort. We highlight that none of these works considers the temporal graph for modelling the dynamic of AD graph. 
% mingyu: incomplete

\textbf{Active Directory.} In the literature, two primary defender strategies have been explored for defending Active Directory: edge-blocking and decoy allocation (node-blocking). The seminal work by Dunagan et al. \citep{dunagan2009heat} was the first to study the defense problem in Active Directory through edge-blocking by introducing the heuristic edge-blocking algorithm. Follow-up researches on the edge-blocking optimization problem includes Guo et al. \citep{guo2022practical} proposed an optimal edge-blocking strategy using Fixed-Parameter Tractable algorithms; \citep{guo2023scalable, zhang2023oracle} improved scalability through Mixed-Integer Programming and the Double Oracle algorithm; Goel et al. \citep{goel2022defending, goel2023evolving} proposed the Evolutionary Diversity Optimization (EDO) algorithm to defend against attackers in a configurable environment; and Guo et al.\citep{guo2024limited} studied optimal edge-blocking problem with minimal human input. Another approach for defending Active Directory found in the literature involves node-blocking, which abstracts the concept to decoy allocation. Chapter \ref{chapter:paper1} are the first to study the honeypot allocation problem for defending Active directory where they proposed MIP algorithm to solve the problem.

\textbf{Evolutionary Diversity Optimization} \citep{ulrich2010integrating} is a recent branch of Evolutionary Computation. EDO is designed to identify a set of solutions that is both high-quality and structurally diverse. In the literature, there have been considerable efforts exploring the EDO algorithm for various combinatorial problems, including the travelling salesperson problem \citep{nikfarjam2021entropy, do2022analysis, bossek2019evolving}, minimum spanning tree problem \citep{bossek2021evolutionary}, knapsack problems \citep{bossek2021breeding}, and more. Among these studies, the work of Goel et al. \citep{goel2022defending, goel2023evolving} is particularly relevant to our research. Goel et al. consider the edge-blocking problem against attacker in AD graph where edges are associated with a failure rate and detection rate. They deploy a neural network/reinforcement learning to approximate the attacker's strategy and apply EDO algorithm to solve the defender problem. In our study, our EDO algorithm draws inspiration from Goel et al. \citep{goel2022defending}, including the design of the mutation/crossover operator and diversity measure strategy. However, our experimental findings reveal that the vanilla EDO algorithm performs poorly when directly applied to our specific problem.

\section{Proposed Algorithm}

\subsection{Game-theoretical rational attacker}
\label{sec:optatk}
% mingyu: optimal attacker => game-theoretical rational attacker?
In our model, the game-theoretical rational/optimal attacker will choose the attack path that has the minimal response time. We illustrate such paths using the following example from Figure \ref{fig:ADex}. We assume the defender allocates honeypots to a set of nodes $C = {Cp_2, Cp_3}$. 
Starting from the entry node $s_1$, let's examine two potential attack paths: $\pi_1 = \langle(s_1, Gr_1, 1), (Gr_1, Cp_2, 2),(Cp_2, U_2, 6), (U_2, DA, 7) \rangle$ and $\pi_2 = \langle(s_1, Gr_1, 1), (Gr_1, Cp_2, 5),(Cp_2, U_2, 6), (U_2, DA, 7) \rangle$. The difference between these 2 paths lies in the departure time of exploiting the second edge $(Gr_1, Cp_2)$. After exploiting the first edge $(s_1, Gr_1)$ at time 1, the attacker has 2 options: either immediately exploit the next edge at time 2 ($\pi_1$) or wait until time 5 to continue ($\pi_2$). Despite both paths leading to the attacker reaching DA at time 7, the attacker is more "troublesome" if they opt for $\pi_1$. This is because the decoy only identifies them at time 5 ($RT = 2$) for path $\pi_1$, whereas for path $\pi_2$, the attacker is detected at time 2 ($RT = 5$), providing the defender with significantly more time to respond to the incident. $\pi_1$ in this example is actually the worst-case/optimal attack path.

% From entry node $s_1$, the following path is one of the optimal attack path: $\pi = \langle(s_1, Gr_1, 1), (Gr_1, Cp_2, 5),(Cp_2, U_2, 6), (U_2, DA, 7) \rangle$.  . However, the attacker still wait in node $Gr_1$ until edge $(Cp_2, U_2)$ appear at time 6 to minimize the response time which is only $RT = 7 - 6 = 1$. 

Algorithm \ref{alg:optatk} for finding such paths can be described as follows. Let's consider an attack graph $G$ and a defender's honeypot allocation $C \in V$. We define a tuple $(\pi_1, c, t_c)$, where $\pi_1$ represents a temporal path, $c$ is a node in $C$, and $t_c$ is a time. Firstly, for each node $c \in C$, we verify if it is reachable from any of the entry nodes $s\in S$ at time $t_c$ in a graph $G' = (V\setminus (C\setminus c), E)$ (line 4) —here, we remove all nodes in $C$ except node $c$ (line 2). The condition in line 4 ensures the \textit{worst-case} condition of the optimal attack path. If we can reach node $c$ from $S$ at time $t_c$ using path $\pi_1$, we then find the earliest-arrival path $\pi_2$ from $c$ to DA within the interval $[t_c, t_{\omega}]$ (line 5-6). We add the tuple $(\pi_1, \pi_2)$ to $\Psi$ (line 7). Next, for every tuple $(\pi_1, \pi_2) \in \Psi$, we merge 2 path to form a temporal $(s,DA)$-path $\pi = \pi_1 + \pi_2$. We identify the tuple with the smallest duration $dur(\pi_2)$, \textit{the duration of the earliest-arrival path $\pi_2$ is actually the response time for the attack path $\pi = \pi_1 + \pi_2$} (line 8). Therefore, the optimal attack path $\pi_{OPT}$ is the one where the $\pi_2$ sub-path has the smallest duration. The fitness function giving a defender solution $C$ can be defined as: 

\setcounter{AlgoLine}{0}
\begin{equation}
\label{eq:fullfit}
  f(C)=\begin{cases}
    \min_{\pi \in \Pi} RT(\pi, C), & \text{if $C$ is feasible (temporal cut)}.\\
    % mingyu: reviewers already forget what is feasible, so good to mention it again
    0, & \text{otherwise}.
  \end{cases}
\end{equation}

\begin{algorithm}[H]
 \caption{Algorithm for Computing Optimal Attack Path}
 \label{alg:optatk}
 \SetAlgoLined
 \LinesNumbered

 % suppress printed "end" markers
 \SetKwFor{For}{for}{:}{}
 \SetKwFor{ForEach}{foreach}{:}{}
 \SetKwFor{While}{while}{:}{}
 \SetKwIF{If}{ElseIf}{Else}{if}{then}{else if}{else}{}

 \KwIn{Temporal graph $G$, set of source nodes $S$, set of honeypot $C$, destination node DA, time interval $[t_{\alpha}, t_{\omega}]$}
 \KwOut{Optimal attack path $\pi$}
\BlankLine
 \For{$c \in C$}{
   Remove nodeset $C \setminus c$ from graph $G$;

   \For{$t \in [t_{\alpha}, t_{\omega}]$}{
     \If{$c$ can be reached from any $s\in S$ at time $t$}{
       Store the path used to reach c by time $t$ to $\pi_1$;

       $\pi_2 \leftarrow compute\_earliest\_arrvl\_path(G, c, DA, [t, t_{\omega}])$;

       Add $(\pi_1, \pi_2)$ to $\Psi$;
     }
   }
 }

 $\pi = \arg\min_{(\pi_1, \pi_2) \in \Psi} \text{dur}(\pi_2)$;

 \Return{$\pi$}
\end{algorithm}

For computing earliest-arrival path subroutine (line 6) we can use the state-of-the-art algorithm proposed by Wu et al. \citep{wu2014path} which has been proven to be time-polynomial. This makes computing attacker optimal attack plan time polynomial. Despite this, Wu's algorithm is inefficient when running on AD-specific graph, slowing down the computation of the optimal attack plan. We will discuss this issue in the next section and propose a more efficient approach for calculating the earliest-arrival path.

\subsection{Faster computation for earliest-arrival path}

As outlined in Section \ref{sec:optatk}, the response time of an attack path is determined by the duration of the earliest-arrival path from an initial point of contact to the DA. Therefore, the computation of optimal attack for fitness function required the call of computing the earliest arrival path subroutine. The first candidate algorithm that we use for computing the earliest arrival path in our implementation is Wu's algorithm \citep{wu2014path}. In \citep{wu2014path}, the author explored the computation of minimal paths in temporal graphs, including the earliest-arrival path. Wu et al. introduced a one-pass algorithm for computing the earliest-arrival path, which stands as one of the state-of-the-art algorithms for this task. Wu's algorithm generates a set of edge streams, a chronological sequence of all edges $E$ ordered by the time at which the edge is collected. The algorithm scans through the edge stream, greedily updating the earliest arrival time at each node that satisfies the arrival condition. This process required the duplication of every static edge to correctly update the earliest arrival time which explain the contradictory of the inefficiencies of Wu's algorithm in AD-specific temporal graph. In general, Wu's algorithm poses inefficiencies when applied to graphs with a substantial number of static edges as Wu's algorithm requires the scan of every edge in $E$. In practice, while the AD graph exhibits dynamic characteristics, a significant portion of the AD infrastructure remains static. For instance, in a snapshot taken from the University of Anonymous on 13/10/2021 at 02:00 pm, a total of 1,151,962 relationships (edges) were identified as online at that time while only 4,039 of these edges were the HasSession edges, which are deemed as the primary source contributing to the dynamism of the AD graph.

Our proposed approach utilises the Dijkstra's edge scanning strategy which allows us to perform the scan only on the underlying edges $E_{\downarrow}$. The intuition behind this algorithm lies in using the Dijkstra greedy scanning strategy, which scans through each underlying edge only once to expand the earliest-arrival paths. The pseudocode is given in Algorithm \ref{alg:dijkstraea}. The idea of using Dijkstra for finding earliest-arrival path has been proposed in \citep{xuan2003computing}. However, we further enhance the runtime on graphs with numerous static edges by introducing a conditional statement between lines 11-14 in Algorithm \ref{alg:dijkstraea}

\setcounter{AlgoLine}{0}
\begin{algorithm}[H]
 \caption{Dijkstra-based algorithm for Computing Earliest-Arrival Time}
 \label{alg:dijkstraea}
 \SetAlgoLined
 \LinesNumbered

 % suppress printed "end" markers
 \SetKwFor{For}{for}{:}{}
 \SetKwFor{ForEach}{foreach}{:}{}
 \SetKwFor{While}{while}{:}{}
 \SetKwIF{If}{ElseIf}{Else}{if}{then}{else if}{else}{}

 \KwIn{Temporal Graph $G$, source nodes $S$, time interval $[t_{\alpha}, t_{\omega}]$}
 \KwOut{The earliest-arrival time from source node $s \in S$ to every vertex}
 \BlankLine
 $PQ = Priority\_Queue$;
 
 $INSERT_{PQ}(t_{i}, s), \forall s \in S$;

 $seen[s] = t_{i}, \forall s \in S$;

 $arrvl\_time \leftarrow empty$ $dictionary$;

 \While{$PQ \neq \emptyset$}{
   $(t_u, u) \leftarrow POP\_MIN_{PQ}()$;

   \If{$u$ in $arrvl\_time$}{
     continue;
   }

   $arrvl\_time[u] = t_u$;

   \ForEach{successor $v$ of $u$}{
     \If{$(u, v)$ is a $static$ $edge$}{
       $v\_arrvl \leftarrow t_u + 1$;
     }
     \Else{
       $v\_arrvl \leftarrow min\{t : t \in time\_labels(u, v),$ $and$ $t > t_u \}$;
     }

     \If{$v$ in $arrvl\_time$}{
       \textbf{continue};
     }
     \ElseIf{$v$ not in $seen$ or $v\_arrvl < seen[v]$}{
       $seen[v] \leftarrow v\_arrival$;

       $INSERT_{PQ}(v\_arrvl, v)$;
     }
   }
 }

 \Return{$arrvl\_time$}
\end{algorithm}

The correctness of the Dijkstra Greedy Strategy for computing the earliest-arrival path is provided in Theorem \ref{theo:dijkstra}. In the general case, the time complexity of Wu et al.'s algorithm can be expressed in our notation as $\mathcal{O}((\varepsilon_{s} + \varepsilon_{d}) \cdot t_{max})$, whereas the time complexity of our proposed algorithm is $\mathcal{O}((\varepsilon_{s} + \varepsilon_{d} \cdot t_{max})\log{}(|V|))$. In scenarios where the number of static edges $\varepsilon_{s} = |E_s|$ outweighs the number of dynamic edges $\varepsilon_{d} = |E_d|$, our algorithm demonstrates more efficient runtime, as theoretically presented in Theorem \ref{theo:staticea}.
% mingyu: complexity unmatched parenthesis, there is another unmatched parenthesis below in theroem 3

Experimentally, when we use these algorithms to find the earliest path from every source to every node in graph $ADX10$ (section ), while Wu's algorithm takes $18.370$ seconds to complete the task, Dijkstra Greedy's runtime is only about $3.389$ seconds (5x faster).

\begin{theorem}
\label{theo:dijkstra}
Algorithm \ref{alg:dijkstraea} correctly compute the earliest-arrival path from a source vertex $x$ to every vertex $v \in V$ within a given interval $[t_{\alpha}, t_{\omega}]$ with the complexity of $\mathcal{O}((\varepsilon_{s} + \varepsilon_{d} \cdot t_{max}) \cdot \log{}(|V|))$
\end{theorem} 

\begin{proof} To proof the correcness, we first provide the following Lemma: 
\begin{lemma}
\label{lemma:wu}
Let a node sequence $V(\pi) = \langle x, v_{1}, v_{2}, \cdots, v_{k}  \rangle$ be the earliest-arrival path from vertex $x$ to vertex $v_{k}$ within some interval $[t_{\alpha}, t_{\omega}]$. Every prefix-subpath $V(\hat{\pi}) = \langle x, v_{1}, v_{2}, \cdots, v_{i}  \rangle \subset \pi$ where $0 < i < k$, is also an earliest-arrival path from $x$ to $v_{i}$ within $[t_{\alpha}, t_{\omega}]$.
\end{lemma} 
\begin{proof} Admit proof from Lemma 6 of \citep{wu2014path}
\end{proof}

The classic Dijkstra's algorithm computing single-source shortest paths based on the observation that the prefix-subpath of the shortest path is also a shortest path. Lemma \ref{lemma:wu} implied that the prefix-subpath of an earliest-arrival path is also an earliest-arrival paths. This proof the correctness of the use of Dijkstra greedy strategy for computing earliest-arrival path.

Next, we continue with the complexity analysis. We assume the use of a Priority Queue to identify the minimum arrival time of unvisited nodes in the Dijkstra-based algorithm. The algorithm grow the earliest arrival path by scan through each out-bound underlying edges in underlying edge the from the current node. Eventually, vertices $v \in V$ will be added to the heap once, hence, the worst-case heap size is $|V|$. Consequently, the complexity of the extract-min operation of the priority queue is $\mathcal{O}(\log{}(|V|))$. Iteratively popping the minimum value from the priority queue takes $\mathcal{O}(|V| \cdot \log{}(|V|))$. Since each node is only extracted once and not revisited, the for loop at line 10 will visit each underlying edge $E_{\downarrow} \in G_{\downarrow}$ only once. The updated earliest arrival time for each successor requires $\mathcal{O}(1)$ for static edges $e_{s} \in E_{s}$ and $\mathcal{O}(t_{max})$ for dynamic edges $e_{s} \in E_{d}$ where $t_{max} = t_{\omega}-t_{\alpha}$. Consequently, the overall complexity of the algorithm is $\mathcal{O}(|V| \cdot \log{}(|V|) + (\varepsilon_{s} + \varepsilon_{d} \cdot t_{max}) \cdot \log{}(|V|))$. As $\varepsilon_{\downarrow} = V^2$ and $\varepsilon_{\downarrow} = \varepsilon_{s} + \varepsilon_{d}$, we can simply rewrite as $\mathcal{O}((\varepsilon_{s} + \varepsilon_{d} \cdot t_{max}) \cdot \log{}(|V|))$.
\end{proof}

\begin{theorem}
\label{theo:staticea}
When $\varepsilon_{s} \gg \varepsilon_{d}$, the complexity of Dijkstra-based algorithm become $\mathcal{O}(\varepsilon_{s} \cdot \log{}(|V|)$ while complexity of Wu's algorithm become $\mathcal{O}(\varepsilon_{s} \cdot t_{max})$
\end{theorem}

\begin{proof} When $\varepsilon_{s} \gg \varepsilon_{d}$, we can safely assume that $\varepsilon_{d} \to 0$ to present the complexity in term of $\varepsilon_{s}$. The complexity of our Dijkstra-based algorithm can be reformulated as $\mathcal{O}(\lim_{\varepsilon_{d} \to 0}(\varepsilon_{s} + \varepsilon_{d} \cdot t_{max})\cdot \log{}(|V|))$, which simplifies to $\mathcal{O}(\varepsilon_{s} \cdot \log{}(|V|))$. Similarly,  the complexity of Wu's algorithm in the same limit is $\mathcal{O}(\lim_{\varepsilon_{d} \to 0}(\varepsilon_{s} \cdot t_{max} + \varepsilon_{d} \cdot t_{max}))$, which simplifies to $\mathcal{O}(\varepsilon_{s} \cdot t_{max})$.
\end{proof}

\subsection{EDO Algorithm for max-$RT$}

In this section, we discuss the application of the Evolutionary Diversity Optimization (EDO) algorithm within our problem context. The pioneering work of Goel et al. \citep{goel2022defending, goel2023evolving} introduced the EDO technique for addressing the edge-blocking problem in AD attack graphs. We initially applied Goel's EDO algorithm to our scenario. In our problem, the defender employs EDO to acquire a diverse set of defensive plans denoted as $C$, where the fitness function $f(C)$ can be obtained by computing the optimal attack plan. Let's define $P$ as the population of defensive solutions. An individual $p\in P$ is defined as the binarization of solution $C$ where each individual has a length of $|N_b|$, with 1 signifying the decision to block the corresponding node and 0 implying no blocking.
% mingyu: maybe mention goel is doing edge blocking while you are doing vertex blocking

We initiate the process by generating a random population \(P\) of defensive solutions. An individual $p$ is randomly selected from $P$ to undergo either mutation or crossover, each with a probability of 0.5. The number $x$ of mutated bits in the offspring is chosen randomly based on a Poisson distribution. For \textbf{mutation}, we randomly select an individual $p'$ from $P$ and flip $x$ random bits, changing 0s to 1s and 1s to 0s. For example, if we choose $p' = \langle 1, 0, 1, 1, 0, 1 \rangle$ from $P$ and $x = 2$, the resulting offspring could be $p = \langle 0, 1, 0, 1, 1, 1 \rangle$. For \textbf{crossover}, we again randomly select two parents $p'$ and $p''$ from $P$. We identify $x$ coordinates where $p'$ has 0s and $p''$ has 1s, and flip the bits at those coordinates on both $p'$ and $p''$. Similarly, we identify $x$ coordinates where $p'$ has 1s and $p''$ has 0s, and flip the bits at those coordinates. After having the offspring using mutation and crossover operation, we add the new offspring to the population only if their fitness score is close to the best fitness score of the population and reject the individuals that contribute the least to the diversity of the population. We follow the \textbf{diversity measure} of population implementation of \citep{goel2022defending,goel2023evolving}. Here we summarise our diversity objective, which aims to maximise the presence of unique nodes in the population. Let $Cnt_P(v_i)$ denote the number of individuals in population $P$ that contain node $v_i$, and define the set of nodes that appear at least once as $V_P = \{\, v_i \mid \mathrm{Cnt}_P(v_i) > 0\}$. The diversity is formally defined as: 

\begin{equation}
\label{eq:diversity}
D(P) = \sum_{v_i \in V_P} \frac{1}{\mathrm{Cnt}_P(v_i)}
\end{equation}

We say that $v_i$ is more "unique" to the population if they have a lower $Cnt_P(v_i)$ score. We noted that this chapter is not intended to redesign mutation, crossover operations, or diversity measures. Instead, our focus lies in the design of an algorithm aimed at enhancing the overall runtime and the convergence time to feasible solutions. We provide the pseudocode for EDO framework in Algorithm \ref{alg:edo-classic}

\begin{algorithm}[H]
 \caption{Evolutionary Diversity Optimization (EDO) framework}
 \label{alg:edo-classic}
 \SetAlgoLined
 \LinesNumbered

 % suppress printed "end" markers
 \SetKwFor{For}{for}{:}{}
 \SetKwFor{ForEach}{foreach}{:}{}
 \SetKwFor{While}{while}{:}{}
 \SetKwIF{If}{ElseIf}{Else}{if}{then}{else if}{else}{}

 \KwIn{Problem instance, diversity acceptance threshold $\alpha$}
 \KwOut{A diverse population of solutions $P$}
\BlankLine
 Initialize the population $P$\;

 \While{a termination criterion is not met}{
   Let $p_{OPT}$ be the best individual currently in $P$\;
   Generate an offspring $p'$ by either mutation or crossover\;
   \If{$f(p') > f(p_{OPT}) - \alpha$}{
     Add $p'$ to population $P$\;
     Remove the individual from $P$ that causes the smallest loss to the diversity indicator $D(P)$ in \ref{eq:diversity}\;
   }
 }

 \Return{$P$}
\end{algorithm}

Our preliminary investigation of the EDO algorithm revealed that most of generated offspring solution are infeasible. This challenge arises due to the expansive nature of the defender solution space (${|N_b| \choose b}$ combinations), which makes it difficult to generate feasible solutions using conventional evolution operators alone. 
Another challenge with the vanilla EDO algorithm is its requirement to execute the full fitness function. Although we have demonstrated that the fitness function can be computed in polynomial time and pushed the runtime frontier by proposing a modification of a Dijkstra-based algorithm for computing earliest-arrival paths, the execution time remains slow for larger graphs. In the following section, we will explore two constraint-handling techniques that we propose to enhance convergence to feasible solutions and improve the algorithm's runtime efficiency.

\subsection{Constraint-Handling Evolutionary Algorithm}
In this section, we shall introduce two constraint-handling approaches for our EDO algorithm.
\subsubsection{Integer Linear Programming repair operator}

In this proposal, we aim to address the issue of infeasible offspring directly by employing a problem-specific \textit{repair} operator. The repair mechanism involves solving an Integer-Linear Programming (ILP) to "patch" the cutting solution. The complete algorithm can be describe as following. Suppose we encounter an infeasible offspring, denoted as $p$ after the mutation or crossover. For each blocked node, excluding those that have undergone a state change during the mutation or crossover, we probabilistically unblock them (i.e., change 1s to 0s) with a probability of 1/2. The purpose of this unblocking operation is to reserve additional space for the subsequent repair process and fulfill the cardinality condition. Finally, we solve our problem-specific ILP repair operator. 

We begin the description for ILP repair operator by introducing its key \textbf{variables}. Let $R_{i, t}$ be a binary variable representing the DA-reachability of node $i$. A value of 1 indicates that we can reach the DA from node $i$ when starting the journey at time $t$, while 0 indicates otherwise. Additionally, we define $B_i$ as a binary decision variable; a value of 1 mean we decide to block node $i$, and 0 otherwise. The \textbf{objective function} of the repair process is formulated as follows: 

\begin{equation}
\label{eq:repairopt}
\min \sum\limits_{s\in S}\sum\limits_{t = t_{\alpha}}^{t_{\omega}} R_{s,t}
\end{equation}

The ILP minimise number of starting nodes that can reach DA. A resulting objective function score of 0 mean the ILP have successfully patch of the solution. Conversely, if the objective function score is greater than 0, it indicates the infeasibility of patching the cutting solution.

The ILP is subject to various \textbf{constraints}. Firstly, for all $(u, v, t) \in E$ where $v\in N_b\setminus V$, we impose the constraint $R_{u, t} \geq R_{v, t+1}$. This constraint implies that if node $v$ can reach the DA when starting to traverse at time $t+1$, then we can reach the DA from $u$ when starting to traverse at time $t$. The "$\geq$" sign, rather than "$=$," accommodates cases where there is an alternate edge from $u$ to reach the DA. Similarly, for all $v \in V$ and $t \in [t_{\alpha}, t_{\omega}]$, we have the constraint $R_{u, t} \geq R_{u, t+1}$. This indicates that if node $u$ can reach the DA when departing from this node at time $t+1$, then we can also reach the DA when departing from this node at time $t$. Next, the blocking constraint is expressed as follows: for all $(u, v, t) \in E$ where $v\in N_b$, the constraint is $R_{u, t} \geq R_{v, t+1} - B_v$. This states that if node $v$ is decided to be blocked, then $u$ cannot reach the DA via the edge $(u, v)$. Finally, we incorporate budget constraints: $\sum_{i \in V}\ B_{i} \leq b$ to conclude the formulation. The complete formulation is presented as following:
\begin{subequations}
\begin{align}
\text{min} \displaystyle \sum\limits_{s\in S}\sum\limits_{t = t_{\alpha}}^{t_{\omega}} R_{s,t}  \nonumber\\ & \nonumber\\
% \end{flalign*}
% \begin{align}
    R_{u, t} \geq R_{v, t+1}& - B_v, & \forall (u, v, t) \in E, v \in N_b \\
    R_{u, t} \geq R_{v, t+1}&, & \forall (u, v, t) \in E, v \in V \setminus N_b \\
    R_{u, t} \geq R_{u, t+1}&, & \forall v \in V, t \in [t_{\alpha}, t_{\omega}] \\
    \sum_{i \in V}\ B_{i} \leq b,& \\
    R_{u, t}, B_{i} \in \{0, 1\}&
\end{align}
\end{subequations}

The pseudocode for this approach is shown in Algorithm \ref{alg:ilprepair}

\begin{algorithm}[H]
 \caption{EDO with ILP repair operator}
 \label{alg:ilprepair}
 \SetAlgoLined
 \LinesNumbered

 % suppress printed "end" markers
 \SetKwFor{For}{for}{:}{}
 \SetKwFor{ForEach}{foreach}{:}{}
 \SetKwFor{While}{while}{:}{}
 \SetKwIF{If}{ElseIf}{Else}{if}{then}{else if}{else}{}
 
 \KwIn{Problem instance with objective $f$, diversity indicator $D$, acceptance threshold $\alpha$}
 \KwOut{A diverse population of solutions $P$}
 \BlankLine
 Initialize the population P\;

 $p_{OPT} \subseteq P$ is the best individual, $f(p)$ is the fitness function\;

 \While{A termination criterion is not met}{
   Generate an offspring $p'$ by either mutation or crossover\;

   \If{$p'$ is infeasible}{
     Generate $p''$ by randomly flipping blocked nodes to unblocked with probability of $1/2$ except mutated/crossovered nodes\;
     $p' \leftarrow ILP(p'')$\;
   }

   \If{$f(p') > f(p_{OPT}) - \alpha$}{
     Add $p'$ to population P and remove individual with smallest loss to the diversity indicator $D(P)$ in \ref{eq:diversity}\;
   }
 }
\end{algorithm}

While the repair operator ensures convergence to a feasible solution in each iteration, it is worth noting that this approach is very memory costly. The ILP requires $\mathcal{O}(|V|\cdot t_{max})$ variable and upto $\mathcal{O}(\varepsilon \cdot t_{max} + |V|)$ constraints, which can become exponentially large for certain graphs. Additionally, solving the ILP itself is known to be a $\mathcal{NP}$-hard problem. 
% Experimentally, the bottleneck of this approach is the IP-solving process, and it is difficult to run on large graphs due to out-of-memory issues.

\subsubsection{Surrogate-assisted and penalty-based repair operator}

During our experiments, we observed that evaluating the full fitness function (Eq. \ref{eq:fullfit}) is computationally costly, as it requires executing Algorithm \ref{alg:optatk} on the entire graph structure in every iteration. This cost becomes particularly prohibitive as graph size increases. Furthermore, when utilizing the vanilla EDO algorithm, we observed convergence issues where the solution frequently failed to reach feasibility.

To tackle the challenges mentioned earlier, we propose Algorithm \ref{alg:edo}. The central premise of this approach is to evaluate the population using a lightweight surrogate fitness function rather than the high-overhead complete fitness function. Our idea for design is that we only need to focus on a set of "important" paths that are likely to have the most impact on the evaluation, instead of spending time on the entire graph. We will have two separate sets of populations in our algorithm namely global population $P_{global}$ and local population $P_{local}$. The local population is evaluated every iteration by the surrogate fitness, while the global population is only evaluated by the complete fitness function when a specific condition is met. Let $\Phi$ be the set of "important" temporal $(s, DA)$-path for the surrogate function. We initialize the set $\Phi$ by adding a random set of temporal paths in graph. Then we iteratively improve the function by adding to $\Phi$ the most up-to-date optimal attack path by the attacker when facing the current population. Our experimental results demonstrate that the surrogate function eventually becomes as effective as the complete fitness function. The proposed algorithm is designed to guide the solution towards convergence of the feasible solution. The pseudocode of the algorithm is presented in Algorithm \ref{alg:edo}. It involved the call of 3 other subroutines: 

\textit{Local Search (line 14):} In the local search, the algorithm performs the standard mutation or crossover, diversity measure and rejection. The key difference is that instead of using a resource-intensive fitness function, we employ a lightweight surrogate fitness function for evaluation. We say an individual $p$ is a \textbf{locally feasible} solution if $p$ can intercept all paths in $\Phi$. Individuals failed to block all paths in $\Phi$ will be penalized. The penalty score is determined by the number of paths in $\Phi$ that an individual $p$ cannot block. The \textbf{surrogate fitness function} can be presented as follows:

\begin{equation}
  f_s^{\Phi}(C)=\begin{cases}
    \min_{\pi \in \Phi} RT(\pi, C), & \text{if $C$ is locally feasible}.\\
    -|\{ \pi \in \Phi : \pi \cap C = \emptyset\}|, & \text{otherwise}.
  \end{cases}
\label{eq:localfit}
\end{equation}

\textit{Global Search (line 22):}  We define that global search starts only when there are no locally infeasible individuals in the local population, and a specified number of local iterations have been completed. In the Global Search, the algorithm adds each "candidate" individual from the local population to the global population and employs diversity measures and rejection on the global population. We use the complete fitness function to evaluate each individual. \textit{It's important to note that a solution $C$ is locally feasible may not necessarily be globally feasible}. This concern arises because the local search evaluates only a fraction of the graph ($\Phi$), which may not provide enough samples to form a cut in the graph. However, as stated in Theorem \ref{theo:converge}, we establish that eventually, the locally feasible solution yields the globally feasible solution after a certain number of iterations.

\textit{Update the Surrogate Fitness Function (line 8 - 12):} Following every global search, we improve the surrogate function by updating the important path set $\Phi$. The update is based on the performance of each individual in the local population. For every $p\in P_{local}$ that is globally infeasible, we add some random temporal (s, DA)-path to $\Psi$ in graph $G'=(V\setminus p, E)$ after removing nodes in cut set $p$. Those are the paths that make the individual $p$ globally infeasible. We use the modification of the Depth First Search algorithm for temporal graphs to find the random paths. For every $p\in P_{local}$ that is globally feasible, we improve the surrogate function by adding the optimal attack path when facing the defense solution $p$ to $\Phi$.
% denoted as $\pi_{opt}^{p}$,
% To keep the surrogate function light-weight, every time when we add a path $\pi$ to $\Psi$, we remove a path $\hat{\pi}$ that have $V(\pi) = V(\hat{\pi})$.
% mingyu: don't understand the last sentence

\begin{theorem}
\label{theo:converge}
In Algorithm~\ref{alg:edo}, after at most $\mathcal{O}(|V|)$ updates of $\Phi$, any cut $C$ that is feasible for $f_s^{\Phi}(C)$ is also feasible for $f(C)$. Hence the algorithm converges to global feasibility.
\end{theorem}
\begin{proof}
Let's us denote $\Pi(\pi)  = \{\hat{\pi} : V(\hat{\pi}) = V(\pi)\}$ is the set of path where each of the element $\hat{\pi}$ have the identical path sequence with $\pi$. We make a following observations regarding the first point of contact of $\pi$: Let's say $i\in V(\pi) \cap C$ is the first point of contact of temporal path $\pi$, then, $i$ is also the first point of contact of every temporal path $\hat{\pi} \in \Pi(\pi)$. Based on the above mentioned observation, for each time the algorithm execute line 12 to add a random path to $\phi$, the algorithm will add a temporal path that will not overlap with any node sequence of any path in $\phi$.

Let's consider an instance of temporal graph denoted as $G = (V, E)$. In this graph, we have source vertices $s\in V$ and destination vertices $d \in V$, forming the underlying graph $G_{\downarrow} = (V, E_{\downarrow})$. It is specified that $G_{\downarrow}$ contains $\mathcal{O}(|V|-2)$ (excluding the source and destination vertices) disjoint paths from $s$ to $d$. Additionally, it is assumed that there is a budget available for deploying at least $|V|-2$ honeypots. The number of budget is $|V|-2$ since it is the size of the minimal temporal cut of our instance. If $b < |V|-2$, the response time is 0, defining the best defense. To simplify our proof, we make the assumption that the algorithm adds only one path to the set $\phi$ in each global iteration. If more than one path is added to the set, the algorithm may achieve faster convergence. The algorithm continues to append new temporal paths to the set $\phi$ until no further paths remain. In the worst-case scenario, each path $\pi$ added to $\phi$ corresponds to vertices disjoint paths in the underlying graph $G_{\downarrow}$ (every paths in $\phi$ are vertices disjoint with each other). Consequently, all $\mathcal{O}(|V|-2)$ paths must be incorporated into the surrogate path set $\phi$ until the Local Search's feasible solution produces a (s, d)-cut on the graph $G$, meeting the feasibility condition for Global Search. This leads to the conclusion that, at worst, we need $\mathcal{O}(|V|)$ Global Search iterations until the feasible solution of Local Search can yield a feasible solution for Global Search. It's worth noting that in the event of tie-breaking, where paths added to $\phi$ aren't disjoint, the algorithm converges faster (i.e. $\mathcal{O}(|V|-2)$ is the lower bound). Blocking common vertices demands less budget, resulting in $\phi$ containing only disjoint paths as the worst-case scenario.

\end{proof}

% $min\{fitness_{local}(p)$ : $ p \in P_{local})\} > 0$. 

\begin{algorithm}[H]
 \caption{EDO with surrogate-assisted/penalty-based fitness function}
 \label{alg:edo}
 \SetAlgoLined
 \LinesNumbered

 % suppress printed "end" markers
 \SetKwFor{For}{for}{:}{}
 \SetKwFor{ForEach}{foreach}{:}{}
 \SetKwFor{While}{while}{:}{}
 \SetKwIF{If}{ElseIf}{Else}{if}{then}{else if}{else}{}

 % inline function blocks (no "end")
 \SetKwProg{Fn}{Function}{}{}

 \KwIn{Temporal Graph $G$, honeypot budget $b$}
 \KwOut{Blocking population $P$}
\BlankLine
 Initialize local population $P_{local}$\;
 Initialize global population $P_{global}$\;
 Initialise set of paths $\Phi$\;

 \While{A termination criterion is met}{
   $P_{local} \leftarrow LocalSearch(P_{local}, \Phi)$\;

   \If{$\prod_{c\in P_{local}, \pi \in \Phi} |c \cap \pi| \neq 0$ and global criterion is met}{
     $P_{global} \leftarrow GlobalSearch(P_{global}, P_{local})$\;

     \ForEach{$c \in P_{local}$}{
       \If{$c$ is a $(S, DA)$-cut in $G$}{
         Compute $\pi_{opt}^{c}(G)$ and add to $\Phi$\;
       }
       \Else{
         Add a random paths from $s\in S$ to $DA$ in graph $G'=(V \textbackslash c, E)$ to $\Phi$\;
       }
     }
   }
 }

 \Return{$P_{global}$}

\BlankLine
 \Fn{LocalSearch($P_{local}, \Phi$)}{
   Initialize the population $P$\;

   $f_s^\Phi(p)$ is the surrogate fitness function from Eq. \ref{eq:localfit}\;    
   
   $p_{OPT} \subseteq P$ is the best individual\;
   \While{A termination criterion is not met}{
     Generate an offspring $p'$ by either mutation or crossover\;

     \If{$f_s^\Phi(p') > f_s^\Phi(p_{OPT}) - \alpha$}{
       Add $p'$ to population $P$ and remove individual with smallest loss to the diversity indicator $D(P)$ in \ref{eq:diversity}\;
     }
   }

   \Return{$P_{local}$}\;
 }
\BlankLine
 \Fn{GlobalSearch($P_{global}, P_{local}$)}{
   
   $f(p)$ is the full fitness function from Eq. \ref{eq:fullfit}\;  
   $p_{OPT} \subseteq P_{global}$ is the current best individual\;
   \ForEach{$p \in P_{local}$}{
     \If{$f(p) > f(p_{OPT}) - \alpha$}{
       Add $p$ to population $P_{global}$ and remove individual with smallest loss to the diversity indicator $D(P_{global})$ in \ref{eq:diversity}\;
     }
   }
   \Return{$P_{global}$}\;
 }
 
\end{algorithm}

\section{Experiment Result}
\label{sec:exp}

\subsection{Experiment Set Up}
All of the experiments are carried out on a high-performance computing cluster with 1 CPU and 24GB of RAM allocated to each trial. As the real-world AD graph is sensitive, we will conduct experiments on synthetic graph generated by DBCreator \footnote{https://github.com/BloodHoundAD/BloodHound-Tools/tree/master/DBCreator} and Adsimulator \footnote{https://github.com/nicolas-carolo/adsimulator} - two state of the art tools for creating AD graphs. Every graph starting with R ("Rxxx") is generated by DBCreator while the one starting with label AD ("ADxxx") is generated by the ADsimulator. 
DBCreator only allows us to fine-tune the number of computers and users. In contrast, Adsimulator provides greater flexibility by enabling adjustments to various entities in the AD graph, including Security Groups, Organizational Units (OUs), Group Policy Objects (GPOs), and more. Consequently, we have two types of graphs generated by Adsimulator: 'ADX\textbf{x}', where default parameters are increased by a factor of '\textbf{x}' (e.g., ADX10 is 10 times the default setting), and 'ADU\textbf{y}', mimicking '\textbf{y}' fractional proportions of the structure of the real AD network at the University of Anonymized (e.g., ADU05 represents 5 $\%$ of the mimicked network). Due to space constraints, detailed information about the size of each graph will be provided in the technical report. However, for a quick estimate, here are the sizes of the largest graph for each type: R4000 (12001 nodes and 45780 edges), ADX20 (6013 nodes and 26671 edges), and ADU (6875 nodes and 37292 edges).

\begin{table}[h]
\begin{center}
\caption{Comparison of all algorithms with DBCreator's graph. The results show the average response time (higher is better) and the average last improvement time (lower is better) of each setting. The numbers in the parenthesis are the average last improvement time.}\label{tab:dbresult}
\smallskip\noindent
\resizebox{\linewidth}{!}{%
\begin{tabular}{l|llll}
\hline
               & \textbf{R2000+C}   & \textbf{R4000+C} & \textbf{2000+L} & \textbf{R4000+L} \\
\hline
\textbf{VAN-V}     & 2.10 (53177s) & 3.60 (42445s) & 0              & 0                \\
\textbf{VAN-D}     & 2.03 (68268s) & 4.09 (28514s) & 0              & 0           \\
% \hline
\textbf{ILP-V}     & 4.07 (16715s) & 4.49 (17037s) & 2.62 (40999s) & 3.18 (41826s)              \\
\textbf{ILP-D}     & 4.09 (25177s) & 4.58 (19067s) & 2.90 (43272s) & 2.90 (33305s)               \\
% \hline
\textbf{EST-V}     & 4.17 (302s)   & 4.70 (706s)   & 4.50 (11862s)  & 3.70 (20536s)              \\
\textbf{EST-D}     & 4.17 (473s)   & 4.70 (302s)   & 4.60 (10257s)  & 3.70 (20536s)              \\
\hline
\end{tabular}}
\end{center}
\end{table}

\begin{table*}[ht]
\begin{center}
\caption{Comparison all algorithms with ADsimulator's graph. No feasible result found by VAN so we did not include it here. OOM is stand for Out-of-Memory. All notion in Table \ref{tab:dbresult} will be also applied here. }\label{tab:adresult}
\smallskip\noindent
\resizebox{\linewidth}{!}{%
\begin{tabular}{l|llllllll}
\hline
                                    & \textbf{ADX5+C}   & \textbf{ADX10+C} &\textbf{ ADX20+C} & \textbf{ADU5+C} & \textbf{ADX5+L}   & \textbf{ADX10+L} &\textbf{ ADX20+L} & \textbf{ADU5+L} \\
\hline
\textbf{ILP-V}                      & 3.21 (31362s)        & OOM               & OOM                   &  OOM       & 3.50 (26796s)        & 0.83 (78836s)              & OOM                   &  OOM                                 \\
\textbf{ILP-D}                      & 3.30 (25250s)         & OOM                & OOM                   & OOM    & 3.40 (25228s)         & 0.60 (73151s)                & OOM                   & OOM                                    \\

\textbf{EST-V}                      & 3.30 (10517s)        & 3.50 (20498s)        & 2.50 (34762s)            & 1.60 (67432s)     & 5.27 (1965s)        & 1.05 (38057s)        & 1.4 (68776s)            & 1.8 (70165s)                                      \\
\textbf{EST-D}                      & 3.40 (2415s)         & 3.30 (14839s)        & 2.50 (34365s)            & 1.70 (65532s)      & 4.77 (3175s)         & 1.75 (23709s)        & 1.90 (74284s)         & 1.40 (73482s)                                     \\
\hline

\end{tabular}}
\end{center}
\end{table*}

However, these tools only generate static snapshots of the graph. To generate a temporal AD attack graph, we will merge a "mould" of static AD graph with authentication data which simulates the characteristic of Hassession edge. The first source is the authentication data from The Comprehensive, Multi-Source Cyber-Security Events dataset \citep{LANL}, referred to as \textbf{LANL}. The second source is from an anonymous organization, labelled as \textbf{COMP}. We will provide the details of each dataset in the appendix. Combining these datasets involved the following process. First, in each static mould AD graph, we removed all HasSession edges. Next, we randomly mapped users, computers and authentication events from the authentication data to the mould graph to create an instance of the temporal graph. For clarity in denoting the generated instances, we referred to a temporal graph in the format $\{graph\}+\{auth\_source\}$. For instance, \textbf{R2000+C} indicates a temporal graph derived from the mould static graph \textbf{R2000}, with Hassession edge data sourced from the \textbf{C}OMP authentication dataset. In this notation, L refers to the LANL dataset, and C refers to the COMP dataset.

We use Gurobi 9.0.2 solver for solving the ILP module. For each experiment instance, we ran 10 trials. In each trial, we randomly choose a set of 10 starting nodes and randomly re-map the authentication data to the mould graph. 
% We only choose nodes with at least 6 time unit distance to DA as starting nodes.
To define the defensive budget for our problem, we have to determine the size of the minimum temporal cut $|minC|$. We will discuss how we determine $minC$ in our Section \ref{sec:minc}. Given that the condition $b \geq |minC|$ has to be met to ensure our problem is feasible, we define the budget for our problem as $b = b_f*|minC|$ where $b_f > 1$ is the budget factor. We set the budget factor to $b_f = 1.5$ for every experiment. We define that only 90 percent of nodes in the graph is blockable. To construct the HasSession edges, we captured snapshots of the authentication dataset every 1 hour. In the experiment, we considered a total of 1000 snapshots in each setting (about 40 days). To avoid confusion in metrics, we will use "time unit" as the metric for the response time. We generate a population of 10 defensive blocking plans. The termination condition for the evolution algorithms was set at 2,000,000 iterations or 24 hours, whichever came first.

In our experiment, we adopt specific denotation for clarity: the Integer Linear Programming approach is denoted as ILP, the surrogate-assisted approach as EST, and the vanilla EDO algorithm as VIN. Additionally, we introduce a \textbf{Value-based Evolutionary Computation (VEC)} which greedily rejects the worst individual from the population instead of rejecting individuals based on diversity measure. In total, we will have 6 sets of algorithm includes: Vanilla EDO algorithm (\textbf{VAN-D}), Vanilla VEC algorithm (\textbf{VAN-V}), ILP-repair approach with EDO framework (\textbf{ILP-D}), ILP-repair approach with VEC framework (\textbf{ILP-V}), Surrogate-assisted approach with EDO framework (\textbf{EST-D}), Surrogate-assisted approach with VEC framework (\textbf{EST-V}). 
Note that our vanilla EDO algorithm's fitness function is implemented with our Dijkstra-based algorithm for computing the earliest-arrival path. Results would significantly degrade if Wu's algorithm were employed.
 
% \begin{itemize}
%     \item \textbf{VAN-D}: Vanilla EDO algorithm.
%     \item \textbf{VAN-V}: Vanilla VEC algorithm.
%     \item \textbf{ILP-D}: ILP-repair approach with EDO framework.
%     \item \textbf{ILP-V}: ILP-repair approach with VEC framework.
%     \item \textbf{EST-D}: Surrogate-assisted approach with EDO framework.
%     \item \textbf{EST-V}: Surrogate-assisted approach with VEC framework.
% \end{itemize}

% mingyu (c): shouldn't it be $b\ge |minC|$

% The minimum cut number is derived by solving the ILP formulation for the minimum cut for temporal graph (appendix for more details). Given the $\mathcal{NP}$-hard nature of the problem of finding minimum cut on temporal graph, solving the ILP for this proves to be computationally expensive. In cases where finding the minimum cut is infeasible with our computational resources, we resort to a heuristic approach to identify the minimum cut in the static underlying graph.
\subsection{Result Interpretation}

\begin{figure}[h]
\centering
  \includegraphics[width=0.85\linewidth]{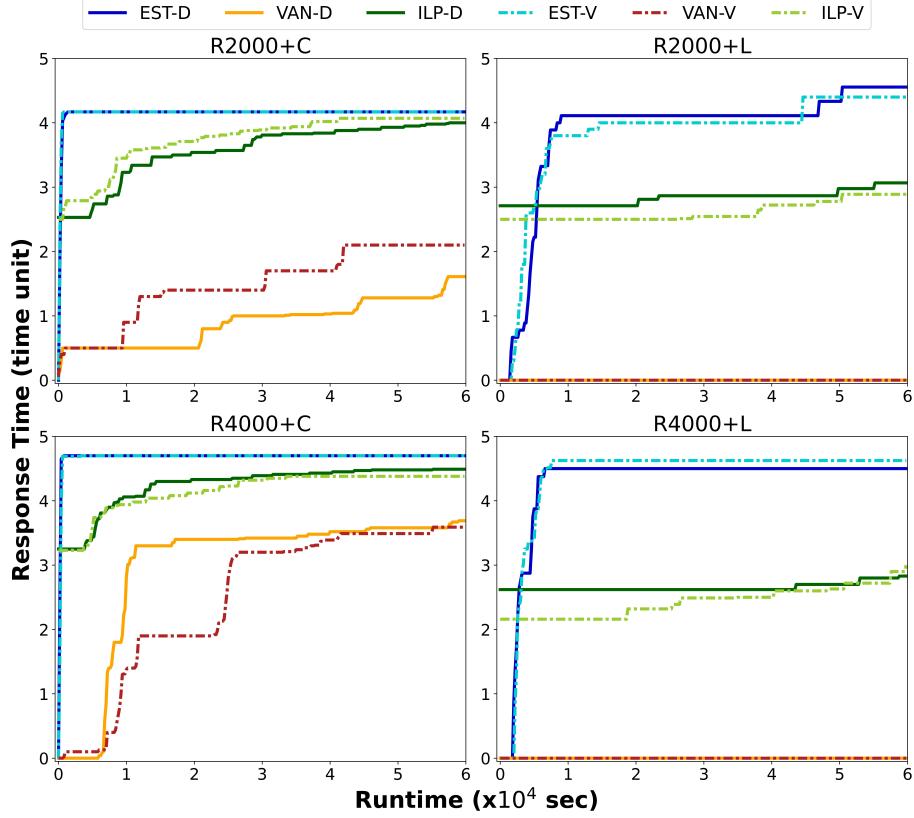}
  \caption{Performance comparison of all 6 algorithms. The \textbf{EST} approaches exhibit significantly faster convergence to the best result compared to the other methods.}
  \label{fig:converge}
\end{figure}

From Figure \ref{fig:converge}, the EST approach significantly improves the convergence speed of the Evolutionary Algorithm, allowing it to reach the best result much faster than ILP and VAN. Notably, ILP approach can find feasible solution from the early iteration since the repair operator will guarantee the mutation/crossover yields a feasible defensive solution. However, solving Integer Linear Programming itself is highly resource-intensive and is the bottleneck for this technique. Unfortunately, ILP failed to run in 5 out of 12 graphs due to Out-of-Memory errors.

Among the graphs, R2000+C and R4000+C are the only two where VAN can find any feasible solution. When we record the time to find the feasible solution (R2000+C and R4000+C), while vanilla takes about 21,728 seconds to find the feasible solution, EST takes on average 201 seconds which is about 108 times faster. For our setting, the EST method performs, on average, about $23\%$ better than the ILP. The convergence speed of EST is also superior to ILP.

To compare the performance of EDO-based algorithms (ended with D) with VEC-based algorithms (ended with V), we conducted a head-to-head comparison between these two approaches. Out of 21 comparable settings (excluding those with OOM errors and infeasible solutions), EDO outperformed VEC in 12 settings, while VEC performed better in only 9 cases (in instances where the response times were equal, we compared the average last improvement time). Overall, EDO outperformed VEC when applied to our problem. 

\section{Conclusion}

This chapter investigated a Stackelberg game model between an attacker and a defender in temporal Active Directory attack graphs. We propose the use of Evolutionary Diversity Optimization algorithms to address this problem. However, the vanilla EDO encounters challenges when scaling to larger graphs and struggling to find feasible solutions.
To improve our solution, we first improve the computation of the attacker's optimal path (fitness function) by refining the calculation of the earliest-arrival path. Our novel Dijkstra-based algorithm for computing the earliest-arrival path, based on the observation that a significant portion of the AD infrastructure remains static. Experimentally, our algorithm is approximately 5 times faster than the SOTA algorithm when running on AD-specific graphs. Next, we introduce two constraint-handling techniques: a repair mechanism using Integer Linear Program (ILP) and a surrogate-assisted model with a penalty fitness function (EST). 
While ILP guarantees to find a feasible solution in early iterations, the EST method achieves this approximately 108 times faster than the vanilla approach. Moreover, EST outperforms ILP, demonstrating approximately a 23\% improvement in our specific setting.

\section{Appendix-ILP for finding minimum temporal cut}
\label{sec:minc}

Zschoche et al. \cite{zschoche2020complexity} provide the complexity analysis on the minimum temporal cut problem (Strict-TS). Despite our effort in finding algorithm for Strict-TS in the literature, we have not come across any algorithm for this algorithm yet. Here, we proposed an ILP formulation to optimally solve the problem. The ILP formulation based on the idea that if node $u$ can reach the DA when departing from this node at time $t+1$, then we can also reach the DA when departing from this node at time $t$. The formulation is inspired by an ILP repair operator, with slight modifications to accommodate our problem requirements. We remove the budget constraints, and the objective function is tailored to minimize the number of budget allocations for the cut. The formulation is presented as follow:

\begin{subequations}
\begin{align}
\text{min} \displaystyle \sum\limits_{v \in N_b} B_{v}  \nonumber\\ & \nonumber\\
% \end{flalign*}
% \begin{align}
    R_{u, t} \geq R_{v, t+1}& - B_v, & \forall (u, v, t) \in E, v \in N_b \\
    R_{u, t} \geq R_{v, t+1}&, & \forall (u, v, t) \in E, v \in V \setminus N_b \\
    R_{u, t} \geq R_{u, t+1}&, & \forall v \in V, t \in [t_{\alpha}, t_{\omega}] \\
    R_{u, t}, B_{i} \in \{0, 1\}&
\end{align}
\end{subequations}

%% file: Chapter5/Chapter5.tex
\chapter{Adaptive Wizard for Removing Cross-Tier Attack Path in Active Directory} % Main chapter title
\label{chapter:paper3} % Change X to a consecutive number; for referencing this chapter elsewhere, use \ref{ChapterX}

In this chapter, we study and propose a novel human-in-the-loop model for edge-removal in AD. Security vulnerabilities in Windows Active Directory (AD) systems are typically modeled using an attack graph and hardening AD systems involves an iterative workflow: security teams propose an edge to remove, and IT operations teams manually review these fixes before implementing the removal.  As verification requires significant manual effort, we formulate an Adaptive Path Removal Problem to minimize the number of steps in this iterative removal process.  
In our model, a wizard proposes an attack path in each step and presents it as a set of multiple-choice options to the IT admin. The IT admin then selects one edge from the proposed set to remove.
This process continues until the target $t$ is disconnected from source $s$ or the number of proposed paths reaches $B$. The model aims to optimize the human effort by minimizing the expected number of interactions between the IT admin and the security wizard. We first prove that the problem is $\mathcal{\#P}$-hard. We then propose a set of solutions including an exact algorithm, an approximate algorithm, and several scalable heuristics. 
Our best heuristic, called DPR, can operate effectively on larger-scale graphs compared to the exact algorithm and consistently outperforms the approximate algorithm across all graphs.
We verify the effectiveness of our algorithms on several synthetic AD graphs and an AD attack graph collected from a real organization.

\section{Introduction}

We propose the Adaptive Path Removal Problem, a model motivated by the challenge of eliminating attack paths in cybersecurity. We begin by describing the cybersecurity use case that motivates our approach and by explaining the design rationale behind our model. The main contributions of this paper are the introduction of a novel theoretical model and the exploration of scalable algorithms for solving this problem. Our model’s design rationale is heavily influenced by practical cybersecurity scenarios and by the urgent demand for workable solutions from security teams.

% edge-blocking => adaptive => LQGCT => drawback => model => why our model is more elegant 

Windows Active Directory (AD) is Microsoft's directory service that enables IT administrators to manage security permissions and control accesses across Windows domain networks. 
An AD environment is naturally described as a graph where nodes are accounts/computers/groups, and the directed edges represent accesses/permissions/vulnerability.
One of the main focus in this line of work is minimizing “attack paths”—routes an attacker might use to escalate privileges and move laterally within the network.

Existing security models \citep{guo2023scalable,zhang2024practical,goel2023evolving} and commercial tools such as BloodHound \citep{BloodHound} reduce these paths by suggesting actionable fixes, typically presented as sets of edges to remove from the graph.
Unfortunately, not every proposed fix (edge) is implementable. Some edges may appear redundant, but removing them could cause significant disruptions.
Since removing edges equates to revoking permissions or accesses within the network, each fix must be approved and implemented by IT operations teams.
This has been referred to in the literature as the “implementable fixes” problem \citep{dunagan2009heat,guo2024limited}.
In industry practice, network hardening workflow typically unfolds in two stages: the security team first proposes necessary fixes, and then IT operations team review those fixes before implementation. This practical constraint has motivated the development of adaptive security models, models that incorporate human feedback and are thus better suited to real-world usage.

In the same way a “proxy” in auction theory places bids on user’s behalf, our proposed wizard model acts as a “proxy” security operator that guide the IT administrator through the attack path removal process.
At every step, the wizard model proposes an attack path to remove.
The IT admin view this as a multiple-choice list of edges and will require to choose one edge to remove.
This process continues until all attack paths are eliminated, or until the number of proposals reaches a preset limit.
The wizard’s goal is to minimize the expected number of proposals.
The wizard is adaptive, meaning it proposes subsequent edges to remove based on the IT admin’s choices in previous steps. 
Unlike previous work \citep{guo2024limited,zheng2011active,dunagan2009heat}, which modeled IT admin’s decision as simply removing or retaining an edge (binary decision), without guaranteeing that all attack paths would be eliminated; our path-based proposal mechanism provides a cut-guarantee solution.

\begin{figure}[h]
 \centering
  \includegraphics[width=0.75\linewidth]{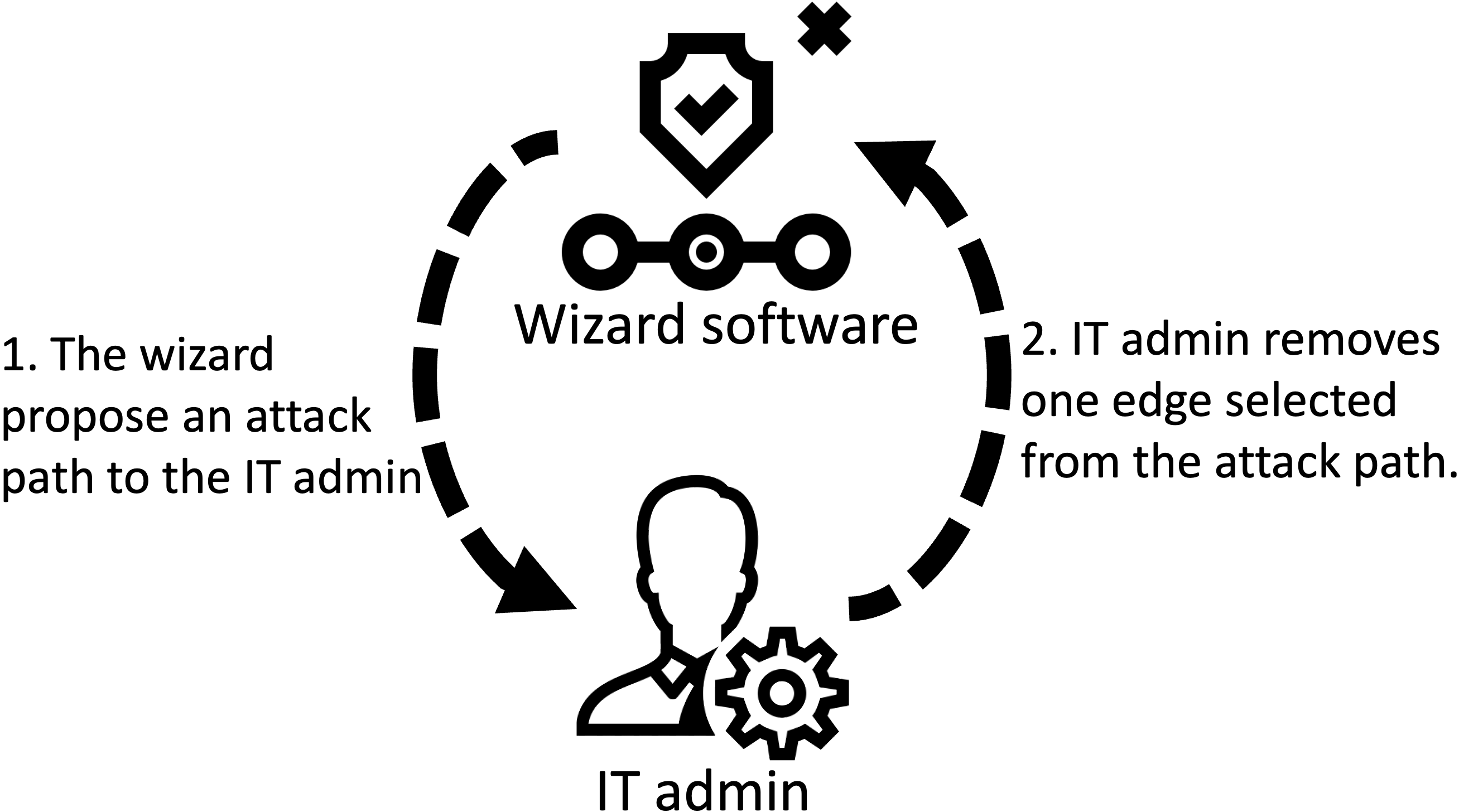}
  \caption{The wizard is a software step-by-step guide to assist the user in performing correction actions without requiring extensive technical knowledge.}
  \label{fig:wizard}
\end{figure}

This chapter contributions can be summarized as follows:

\begin{itemize}
    \item We introduce a new theoretical combinatorial optimization model called Adaptive Path Removal, motivated by the network security use case in AD systems. This is the first adaptive graph-focused model to incorporate path proposals and provide a cut-guarantee solution.
    \item We prove that the problem is $\#\mathcal{P}$-hard and introduce both an exact and an approximate algorithm.
    \item We develop a scalable heuristic called Dynamic Programming with Restriction (DPR), which builds on our exact and approximate algorithms. DPR achieves better scalability than the exact algorithm and outperforms the approximate algorithm.
    \item We also introduce several baseline methods, including two RL-based heuristics, and evaluate them on multiple synthetic graphs and a real AD network. Our experimental results show that DPR consistently achieves superior performance over all other methods.
\end{itemize}

% First, we present a formal formulation of the Adaptive Path Removal (APR) problem, followed by a discussion of the practical motivations and design rationale underlying our model.

\section{Problem Formulation}

The Adaptive Path Removal (APR) problem can be formally defined as follows: Given a directed attack graph $G = (V, E)$ with a source $s$ and a destination node $t$. Each edge $e \in E$ is associated with a confidence score, defined by a function $conf : E \mapsto [0, 1]$. Every round, the system will propose a simple path from the current attack graph. A simple path $p$ is defined as a sequence of edges $p = \langle (v_0 = s, v_1), (v_1, v_2), \ldots, (v_{k-1}, v_k = t) \rangle$ such that no edge is repeated, i.e., $((v_i, v_{i+1}) \neq (v_j, v_{j+1})), \forall ( i \neq j )$. When a path $p$ is queried, the IT admin selects exactly one edge $e \in p$ to remove. We model the IT admin’s choice using the Bradley–Terry model \citep{bradley1952rank}, which assigns a probability to each edge $e \in p$ proportional to its confidence score relative to the others in $p$: 

\begin{equation}
\label{eq:conf}
Pr(e|p) = \frac{conf(e)}{\sum_{e' \in p} conf(e')}
\end{equation}

In other words, an edge with a higher confidence score is more likely to be chosen for removal. This models the administrator’s relative preference or belief about which edge’s removal is most effective.
Let $C$ be the set of edges removed by the IT administrator after $|C|$ round. At round $|C| + 1$, the system will proposed a path $p \in P_{G'}$ in the temporary graph $G' = (V, E \setminus C)$ where $G'$ is called the temporary graph which evolved from the original graph $G = (V, E)$ by removing set of $C$ edges and $P_{G'}$ is the set of every possible path from $s$ to $t$ in $G'$. The query process terminates when either $C$ forms an $(s, t)$-cut (i.e., $s$ is disconnected from $t$) or the query budget $B$ is reached (i.e. $|C| = B$). In this paper, all cuts refer to $s-t$ cuts. To minimize human effort during the cutting process, our optimization goal is to design a policy that minimizes the expected number of queries (or iterations) required to complete the cutting process.

% \HN{Make a comment on the fact that in AD defense you care about paths to high-value targets but here we consider all paths.}

\begin{theorem}
    The APR Problem is $\#\mathcal{P}$-hard 
\end{theorem}
% We defer the proof to the Appendix.
\begin{proof}

The proof is based on a reduction from the $(s, t)$- network reliability problem \citep{reliability} which is $\#\mathcal{P}$-hard.

\textbf{PROBLEM: } $(s, t)$-Network Reliability Problem 
\begin{itemize}
    \item \textbf{Input:} A graph $G = (V, E)$, source node $s$ and destination node $t$, probability $p_e \in \left[0, 1\right]$ associated with the present of each edge.
    \item \textbf{Question:} What is the probability that there is a path between two distinguished vertices s and t
\end{itemize}

First, let $Rel(G)$ denote the $(s-t)$-reliability of graph $G$, i.e., the probability that there exists a path between nodes $s$ and $t$. We assume that each edge in $G$ is operational (i.e., "On") with a probability of $p_e = 0.5$. We known that $Rel(G)$ is $\#\mathcal{P}$-hard to compute. 
Now, we proceed with the construction of the reduction instance. Suppose there exists a directed graph $G' = (V', E')$ and a query limit of $B$. Let $m = |E'|$ represent the number of edges in $G'$ and $B > m$. 
For this construction, we define two types of edges: high-confidence edges (which are more likely to be misconfigurations) and low-confidence edges (which are less likely to be misconfigurations). In this way, we assume that when a path involves $x$ high-confidence edges ($x \geq 1$), then the IT admin will never pick one of the low-confidence edge and they will only pick one of the high-confidence edges with equal probability ($1/x$). Conversely, if a path involves only low-confidence edges, the IT administrator will have to choose one of these edges with equal probability. This essentially assumes that the confidence score of high-confidence edge is infinitely. As defined in the APR problem, the IT administrator behaviour is random, i.e., they choose an edge from the proposed path to remove according to the probability distribution defined by the Equation 1 
%It is just to make our proof cleaner. We can change the cost of difficult edges to a really large value – as long as the precision of the final calculation is good enough for Rel(G), which requires a precision within 1/2^m

Given a APR instance graph $G'$, we construct the following graph $G$: For each edge $(u, v)$ in $G'$, we introduce an auxiliary node called $uv$. We then add one low-confidence edge between  $u$ and $uv$ and one low-confidence edge between $uv$ and $v$. Additionally, we introduce $B$ parallel high-confidence edges between $uv$ and $v$. We refer to this constructed graph as $G$.  In this graph, for the segment $u->uv->v$, we classify the edges as follows: an edge is called an $m_l$-type edge if it connects $u$ to $uv$; an edge is called an $n_l$-type if it is a low-confidence edge connecting $uv$ to $v$; and an edge is called $n_h$-type edge if it is a high-confidence edge connecting $uv$ to $v$

The core idea behind this construction is to ensure that the optimal policy avoids presenting any path with $n_h$-type edges to the IT administrator. Querying path contain $n_h$-type edges is suboptimal because resolving any scenario will require at least $B + 1$.

The overall idea of the construction above is to ensure that the optimal policy avoids presenting any high-confidence edges to the IT administrator, since query high-confidence edges will be never useful as any situation will required to proposed at least $B+1$ query to cut the graph. Therefore, the optimal policy prioritizes querying paths that only contain $m_l$-type and $n_l$-type. 

Now, consider a single segment $u\xrightarrow[]{} uv \xrightarrow[]{} v$, which involves two low-confidence edges and $B$ high-confidence edges, being presented to the IT administrator. Following the optimal policy, $m_l$-type edge and $n_l$-edge of the segment will always be presented first. There is 50\% that the IT admin will choose an $m_l$-type edge and a 50\% chance they will choose an $n_l$-type edge. If the IT admin selects the $m_l$-type edge, the segment will be successfully disconnected (i.e., $u$ will be disconnected from  $v$ within this segment). However, if the IT admin chooses the $n_l$-type edge, the segment becomes impossible to disconnect with $B$ query. This is because once the $n_l$-type edge is removed, all $B$ parallel $n_h$-type edges must be queried to disconnect the segment (we will always hit the query limit $B$). 

Now, we will think about optimal query policy on $G$. We denote tuple $\mathcal{I} = \langle G, B \rangle$ as the APR problem instance on graph $G$ with budget of $B$
Next, let us introduce the concept of a\textit{ realization}. In general terms, a \textit{realization} is a specific outcome or instance of a random variable or process—essentially, the actual occurrence of a particular event within a probabilistic framework. In our problem context, a realization represents the decision made by the IT administrator when a path is proposed. We denote \textit{realization} with a function $\psi : P \mapsto E$. The function $\psi (p)$ acts as an oracle that returns the edge $e \in p$ that will be removed when the path $p$ is proposed to the IT administrator under realization $\psi$

Next, let $q(\mathcal{I}, \pi, \psi)$ represent the number of queries made when following the query policy $\pi$ for the problem instance $\mathcal{I}$ under the realization $\psi$. We denote $\psi_{\mathcal{I}, \pi, b>x}$ as the set of realizations where the query cost is at most $x$ when applying policy $\pi$ to instance $\mathcal{I}$. The expected number of queries across all realizations for instance $\mathcal{I}$ under policy $\pi$ is denoted by $\mathbb{E} \left[ q(\mathcal{I}, \pi) \right]$.
Additionally, the conditional expected number of queries, given that the query cost is at most $\mathbb{E} \left[ q(\mathcal{I}, \pi, \Psi) | \Psi \in \psi_{\mathcal{I}, \pi, b>x} \right]$. We define $\pi^*$ as the optimal policy.
According to the law of total expectation, the following equation holds:

\begin{equation}
\begin{split}
    \mathbb{E} &\left[ q(\mathcal{I}, \pi^*) \right] \\
    &= \mathbb{E} \left[ q(\mathcal{I}, \pi^*, \Psi) | \Psi \in \psi_{\mathcal{I}, \pi^*, b<m}\right] \cdot Pr(\Psi \in \psi_{\mathcal{I}, \pi^*, b<x}) \\
    &+ \mathbb{E} \left[ q^*(\mathcal{I}, \pi, \Psi) | \Psi \in \psi_{\mathcal{I}, \pi, b>m}\right] \cdot Pr(\Psi \in \psi_{\mathcal{I}, \pi, b>x})
\end{split}
\end{equation}

We observe that for realizations where the query cost is at most $m$ (i.e. $\psi_{\mathcal{I}, \pi^*, b<m}$), we will also successfully disconnect $(s, t)$ in graph $G$, Similarly, for realizations where the query cost exceeds $m$ (i.e. $\psi_{\mathcal{I}, \pi^*, b>m}$), we fail to disconnect $(s, t)$ in $G$. This is because there are only $m$ edges of $m_{l}$-type, and exceeding $m$ queries will always deplete all of the budget (reminding that successful disconnection of an segment requires the IT administrator to select an $m_{l}$-type). Therefore, if we define $\psi_{\mathcal{I}, \pi, s\leftrightarrow t}$ as the set of realizations where $(s,t)$ remains connected after applying the optimal query process to instance $\mathcal{I}$, and $\psi_{\mathcal{I}, \pi, s \not \leftrightarrow t}$ as the set of realizations where $(s,t)$ becomes disconnected, we have $\psi_{\mathcal{I}, \pi^*, b\leq m} = \psi_{\mathcal{I}, \pi^*, s \leftrightarrow t}$ and $\psi_{\mathcal{I}, \pi^*, b>m} = \psi_{\mathcal{I}, \pi^*, s \not \leftrightarrow t}$. This allows us to express the expected number of queries using the following equation:

\begin{equation}
\begin{split}
    &\mathbb{E} \left[ q(\mathcal{I}, \pi^*) \right] = \mathbb{E} \left[ q(\mathcal{I}, \pi^*, \Psi) | \Psi \in \psi_{\mathcal{I}, \pi^*, s \leftrightarrow t}\right] \cdot Pr(\Psi \in \psi_{\mathcal{I}, \pi^*, s \leftrightarrow t}) \\ 
    & + \mathbb{E} \left[ q^*(\mathcal{I}, \pi, \Psi) | \Psi \in \psi_{\mathcal{I}, \pi, s \not \leftrightarrow t}\right] \cdot  Pr(\Psi \in \psi_{\mathcal{I}, \pi, s \not \leftrightarrow t}) \\
    &= \mathbb{E} \left[ q(\mathcal{I}, \pi^*, \Psi) | \Psi \in \psi_{\mathcal{I}, \pi^*, s \leftrightarrow t}\right] \cdot  Pr(\Psi \in \psi_{\mathcal{I}, \pi^*, s \leftrightarrow t}) + B \cdot Pr(\Psi \in \psi_{\mathcal{I}, \pi, s \not \leftrightarrow t}) \\
    &= \mathbb{E} \left[ q(\mathcal{I}, \pi^*, \Psi) | \Psi \in \psi_{\mathcal{I}, \pi^*, s \leftrightarrow t}\right] (1 - Rel(G)) + B \cdot Rel(G) 
\end{split}
\end{equation}

From this equation, we observe that the optimal query strategy is independent of $B$ (since the second term is not controlled by the policy). Calculate the optimal query strategy is all about minimizing the expected number of queries, conditional on $G$ being $s-t$ disconnected (i.e.$\mathbb{E} \left[ q(\mathcal{I}, \pi^*, \Psi) | \Psi \in \psi_{\mathcal{I}, \pi^*, s \leftrightarrow t}\right]$. When $G$ is $s-t$ connected, then we will always hit the budget limit. We can calculate $Rel(G)$ by taking the difference between $\mathbb{E} \left[ q(\mathcal{I'}=\langle G, B+1 \rangle, \pi^*) \right]$ and $\mathbb{E} \left[ q(\mathcal{I}=\langle G, B \rangle, \pi^*) \right]$. This shows that calculating the expected number of queries in the APR problem is as difficult as the reliability problem. 

\end{proof}

\textbf{Reason for edge's confidence score and how to assign it} Integrating confidence scores helps us effectively embed domain-specific security information into our model, making it easier to identify edges that are more likely to be removable. This matters because not all edges are equally prone to be removed by IT admin; some edges, such as outdated privileges or overly permissive group assignments, are clear candidates for removal. Using insights from the security context in our edge preference model could substantially reduce the number of required queries.
To automate the assignment of confidence scores, we can train a binary classifier that predicts the likelihood of each edge being safely removable. For example, Zheng et al.~\citep{zheng2011active}  propose an active learning approach that learns an IT admin’s decisions about which edges to remove. We can automatically assign confidence scores to edges by using a binary classifier, defined as a function $f : E \mapsto [0, 1]$ where the output represents the classifier’s confidence that a given edge can be safely removed.

\section{Related Works}

\textbf{Active Directory and non-adaptive defense models.}
% In the literature, two primary defender strategies have been explored for defending Active Directory: edge-blocking and decoy allocation (node-blocking). 
The seminal work by Dunagan et al. \citep{dunagan2009heat} proposed the Active Directory (AD) attack graph which modelled the identity snowball attack that developed further and commercialized by Bloodhound \citep{BloodHound}. Follow-up works by Guo et al. \citep{guo2022practical,guo2023scalable} and Zhang et al. \citep{zhang2023oracle} formulate the problem of hardening the AD system as the shortest path interdiction via edge-removing problem. \citep{goel2022defending,goel2023evolving} proposed the Evolutionary Diversity Optimization (EDO) algorithm to defend against attackers in a configurable environment. Another work by Zhang et al. \citep{zhang2024practical} studied the problem of minimizing the number of users with paths to the domain admin via edge removal. Another approach for defending Active Directory found in the literature involves node-removal, which abstracts the concept of decoy allocation as introduced in Chapter \ref{chapter:paper1} and \ref{chapter:paper2}. The main drawback of non-adaptive models in real-world deployments is that they are not amenable to include human feedback.

\textbf{Adaptive models for Active Directory defense}. Several studies have integrated manual feedback from IT admin into network defenses process, emphasizing the importance of human involvement in configuration changes. 
Dunagan et al. \citep{dunagan2009heat} proposed Heat-ray, a system aimed at minimizing snowball identity attacks in Active Directory (AD) by iteratively proposing edge removals to IT administrators based on the sparest cut. 
% The system refines its proposals through iteration by updating edge costs using a Support Vector Machine algorithm, based on IT administrators' feedback on configuration changes.
Zheng et al. \citep{zheng2011active} enhanced Heat-ray with active learning to improve edge cost learning process. 
Guo et al. \citep{guo2024limited} introduced an adaptive defense model called the Limited Query Graph Connectivity Test (LQGCT), which is closely related to our approach. In their model, a proxy algorithm proposes one edge at a time, and the IT admin's decision is binary (i.e. whether to remove or retain it). By contrast, the proxy of our model proposes an entire attack path instead of a single edge which offers a multiple-choice selection rather than a binary decision.
Proposing a path provides several practical advantages over proposing an edge. Firstly, an edge proposal can fail to form a graph cut if the IT admin is overly conservative and retains too many edges.
This leaves the possibility of an attack even after the clean-up. In our experiments, path-based proposals guarantee that no attack path remains, provided the proxy algorithm has a sufficiently large budget.
Secondly, by presenting a list of edges to compare, our model encourages more deliberate choices, whereas a binary question as in LQGCT may incentivize conservative behaviour.
From a theoretical view, path-based proposals fundamentally differ and are harder to solve than the previous edge-based model. In LQGCT, the policy tree is binary, while our policy tree can branch into up to $l$ outcomes at each step, where $l$ is the length of the longest proposed path. As a result, existing algorithms cannot be directly applied to our setting, requiring us to develop an entirely new class of solutions.

\textbf{Related models from other research communities}. The sequential testing problem in operations research \citep{unluyurt2004sequential}, is often described through medical testing use cases. For instance, \citep{short2013iron,yu2023deep} employs adaptive strategies to reduce testing costs to diagnose diseases. Another related area is the problem of learning with attribute costs problem in machine learning \citep{sun1996hill,kaplan2005learning,golovin2011adaptive}. In this problem, each feature incurs a cost, and the task is to construct a classification tree that minimizes the total feature costs. Stochastic Boolean Function Evaluation (SBFE) problem \citep{allen2017evaluation,deshpande2014approximation} is also closely related. An SBFE instance involves a Boolean function $f$  with multiple hidden binary inputs and one binary output. Each input bit can be queried at a cost, and the objective is to find a query strategy that minimizes the expected cost to determine $f$'s output. 
While these models are relevant, they are not designed for our graph-based problems and lack scalability for large graphs. Consequently, similar to LQGCT, solutions for these models cannot be directly applied to our work.

\section{Algorithms}
In this part, we will present our solution for the APR problem. 
To help with the solution formulation, we will convert our problem into an equivalent Markov Decision Process (MDP). 
% Our proposed algorithm will be presented in MDP language. 
\subsection{Markov Decision Process formulation and preliminary}

Let us define the MDP as a tuple $\langle \mathcal{S}, \mathcal{A}, \Phi, R \rangle$, where $\mathcal{S}$ is the set of state, $\mathcal{A}$ is the set of action, $\Phi : \mathcal{S} \times \mathcal{A} \times \mathcal{S} \mapsto [0, 1 ]$ is the state transition and the reward function $R : \mathcal{S} \times \mathcal{A} \mapsto \mathbb{R}$. 

\textbf{State:} In an APR problem, the IT admin will remove an edge from a proposed path in every round. This process will evolve the graph into a series of temporary graphs by removing edges. We present these temporary graphs using a temporary state variable $s$ with $|B|$-dimension: $s = \{(x_1, x_2, \cdots, x_{B}) : x_i \in E \cup \{*\}, \forall i \in \{1, 2, \cdots, B\}\}$ where $x_i$ will be the edge that is removed by the IT admin at round $i$ and  $x_i = $ '*' means we have not query any path in this round.  We define the state of the original attack graph $G$ as the root state $s_r$, which will have the form: $(*, *, \cdots, *)$. A state $s'$ evolves from a state $s$ by removing an edge $e$ expressed as $s' = s \setminus e$. We denote $\mathcal{S}_i$ the set of possible states in round $i$ of the process. Hence, the state space can be represented as $\mathcal{S} = \mathcal{S}_0 \cup \mathcal{S}_{1} \cup \mathcal{S}_{2} \cup \cdots \cup \mathcal{S}_{B}$. We also define two sets of terminal states: $\perp_b$ is the terminal state reached when the budget is exhausted without identifying a cut, and $\perp_d$ is the terminal state reached when a cut is successfully found as a result of the query sequence. 
% The worst-case size of the state space can goes up to $\mathcal{O}\binom{|E|}{B}$. 

\textbf{Action:} Each path proposal in the APR problem is associated with an action in MDP. Let's say we are at state $s'$ associated with a temporary graph $G'$ and $A_{s'}$ is the action space at state $s'$. The action $a \in A_{s'}$ associates with a simple path $p \in P_{G'}$. We have the following Lemma for the action space in our problem:

\begin{lemma}
\label{lemma:subpath}
Given an MDP construction $\langle \mathcal{S}, \mathcal{A}, \Phi, R \rangle$ for the APR problem. We have $A_s \subseteq A_{s_r}$ for every $s \in \mathcal{S}$ where $s_r$ is the root state.
\end{lemma}

\begin{proof}
    The action available at each temporary state $s$ is the enumeration of every possible path in the corresponding graph $G$. A temporary graph $G'$ is actually the subgraph of the root graph $G$ (as $G'$ is evolved from $G$ by removing edge) which implies $P_{G'} \subseteq P_{G}$. Therefore, we have $A_{s} \subseteq A_{s'}$, $\forall s, s'$ where $s \in successor(s')$ which imply $A_s \subseteq A_{s_r}$. 
\end{proof}

Lemma \ref{lemma:subpath} states that action set $A_s$ of every state $s \in \mathcal{S}$ is a subset of the action set $A_{s_r}$ of the root state $s_r$. As a result, the overall action space for the APR problem can be expressed as $A = \bigcup_{s\in \mathcal{S}} A_s = A_{s_r}$. Lemma \ref{lemma:subpath} is particularly useful in the design of our algorithms, as discussed in the following sections. 

\textbf{Transition Probabilities:} In each state, an action can lead to different outcomes, which are defined by the transition probabilities in the MDP. In our problem, a transition probability shows the probability of an edge being removed by the IT admin when a path is proposed. The removal probability is defined by the Bradley-Terry preference model, as defined in Equation~\eqref{eq:conf}. Let's say we have a state $s' = s \setminus e$ where $s$ evolved to $s'$ by removing edge e. The transition probability from $s$ to $s'$ when taking action $a$ can be expressed as  $\Phi(s'|s,a) = \Phi(e|s, a) = Pr(e|p)$.

\textbf{Reward:} In our problem, each query will be penalized by a cost of exactly one unit of budget with no discount factor. The reward is difference at two terminal states: when $\perp_b$ is reached, meaning that we ran-out of budget before identifying any cut, we will penalize it with a constant of $\alpha > 0$; 
\begin{equation}
\label{eq:reward}
  R(s, a)=\begin{cases}
    -\alpha, & \text{if } s \in \perp_{b}.\\
    0, & \text{if } s \in \perp_{d}.\\
    -1, & \text{otherwise}.
  \end{cases}
\end{equation}

\textbf{Realization:} While "realization" is not a standard notation in MDP literature, it is commonly used in the context of adaptive submodularity optimization \citep{golovin2011adaptive}. We introduce this concept here as it will help us in describing the algorithms. We define a function $\phi : P_{s_r} \mapsto E$  as the full realization. We can view function $\phi(p)$ as an oracle that returns the IT administrator's edge removal decision for a given path $p$. Additionally, we define the partial realization $\psi_{s}: P_{s_r} \mapsto E$ as the observations made so far at state $s$. Specifically, $\psi_{s}(p)$ returns the edge removed by the IT administrator when $p$ is proposed, and $\psi_s(p) = *$ if $p$ has not yet been proposed or contains any edge have been removed by the IT admin. The domain of the partial realization is defined as  $dom(\psi) = \{p \in P_{s_r}: \psi(p) \neq *\}$, representing the set of actions for which outcomes have been observed. The range of the partial realization is defined as $range(\psi) = \{\psi(p) : p \in P_{s_r}, \psi(p) \neq *\}$, representing the edges that have been removed by IT admin. A partial realization $\psi$ is consistent with realization $\phi$ if they are equal everywhere in the $dom(\psi)$, denoted $\phi \sim \psi$. We say a partial realization $\psi$ is a subrealization of $\psi'$, denoted by $\psi \subseteq \psi'$, if they are both consistent with some $\phi$ and $dom(\psi) \subseteq dom(\psi')$.

\subsection{Dynamic Programming Exact Algorithm (OPT)}
\label{sec:opt}

As we establish the equivalence between the APR problem and the MDP, we also find that our problem satisfies the ``Principle of Optimality" in MDP. In our problem, given a state $s$ and $s' = s \setminus e$, $\forall e \in E_s$ and $E_s$ is the set of edges in the graph corresponding to state $s$, the optimal path query in $s$ is independent of previous queries and solving for the optimal strategy at $s'$ can be viewed as a subproblem of $s$. This allows us to introduce the optimal utility function following the Bellman equation:
\begin{equation}
\label{eq:dp}
    u_{\pi}(s) = \min_{p \in A_{\textbf{s}}}\{r(s, a) + \sum_{e \in p} \Phi (e|s,p) u_{\pi}(s\setminus e) \}, \text{  } \forall s \in \mathcal{S}
\end{equation}

Based on the the optimal utility function in \eqref{eq:dp}, we can design a Bellman's style dynamic programming, called OPT. Now, let $G_s$ be the graph that is associated with state $s$. The Dynamic Programming follows a top-down approach. In each subproblem, we require the utility $u_{\pi}(s)$. However, as shown in Equation \eqref{eq:dp}, obtaining the optimal utility requires invoking the action set $A_s$ for every subproblem, which in turn requires enumerating every simple path $P_{G_{s}}$ in $G_s$. This approach is impractical to run on any graph of reasonable size because the number of subproblems can grow as large as $|E|^B$, and path enumeration, known to be $\#\mathcal{P}$-hard, takes $\mathcal{O}(|V|^k)$ time complexity under a DFS-based approach \citep{peng2019hop}, where $k$ is the longest path length. However, thanks to Lemma \ref{lemma:subpath}, we can run the enumeration one time only for the original problem. The action space for the subproblem is simply $A_s = A_{s_r} \setminus \{p : e \in p, p \in A_{s_r}, e \in range(\psi_s)\}$, i.e., we can obtain $A_s$ by removing paths in $A_{s_r}$ that contains edges that have been chosen to be removed by the IT admin. Since running Depth First Search procedure to check if $s, t$ connectedness in every subproblem will take $\mathcal{O}(|E|+|V|)$. The overall Exact algorithm will take us $\mathcal{O}(|V|^{k}+(|E|+|V|)*|E|^B)$. The pseudoscope for the exact algorithm is illustrated in Algorithm \ref{alg:dp}

\begin{algorithm}[H]
 \caption{Dynamic Programming (OPT)}
 \label{alg:dp}
 \SetAlgoLined
 \LinesNumbered

 % suppress printed "end" markers
 \SetKwFor{For}{for}{:}{}
 \SetKwIF{If}{ElseIf}{Else}{if}{then}{else if}{else}{}

 \KwIn{Directed Attack Graph $G(V, E)$}
 \KwOut{Optimal query policy $\pi$}
\BlankLine
 \For{$i \in \left[ 0, B\right]$}{
   \For{$\textbf{s} \in \mathcal{S}_i$}{
     \If{$\textbf{s}$ is in $T_d$}{
       $u_{\pi}(\textbf{s}\setminus e) = 0$, $\pi(\textbf{s}\setminus e) = \emptyset$\;
     }
     \ElseIf{$\textbf{s}$ is in $T_b$}{
       $u_{\pi}(\textbf{s}\setminus e) = \alpha$, $\pi(\textbf{s}\setminus e) = \emptyset$\;
     }
     \Else{
       $a^* = \arg\min_{p \in A_{\textbf{s}}} \{ \sum_{e \in p}\left[  \Phi(e|s,a) * u_{\pi}(\textbf{s}\setminus e)\right]\}$\;
       $u_{\pi}(\textbf{s}) = 1 + \sum_{e \in p}\left[   \Phi (e|s,a) * u_{\pi}(\textbf{s}\setminus e)\right]$\;
       $\pi(\textbf{s}) = p^*$\;
     }
   }
 }

 \Return{$\pi$}
\end{algorithm}

\textbf{Complexity Analysis} According to Lemma \ref{lemma:subpath}, we can enumerate the action space $A_{s_r}$ in advance and the action space $A_s$ of any state $s$ can be obtained by removing any path in $A_{s_r}$ that contain edges that have been removed by IT admin during the process. Enumerating action is equal to enumerating $s-t$ path in the graph which takes $|V|^{k}$ under the DFS-based approach \cite{peng2019hop}.
The number of subproblem can grow up to $|E|^B$, running Depth First Search procedure to check if $s, t$ connectedness in every subproblem will take $\mathcal{O}(|E|+|V|)$.
Hence, the overall Exact algorithm will take us $\mathcal{O}(|V|^{k}+(|E|+|V|)*|E|^B)$. And thank to lemma \ref{lemma:subpath}, we avoid the complexity $\mathcal{O}(|V|^{k}*(|E|+|V|)*|E|^B)$.

\subsection{Adaptive Submodular Approximation Algorithms (APP)}
\label{sec:appx}

In this section, we present an approximation algorithm, called APP, by utilizing the adaptive submodularity framework \citep{golovin2011adaptive}. The proposed algorithm provides square-logarithmic approximation to the number of enumerations of possible simple paths in the original graph $G$. The Adaptive Submodular Algorithm is proposed for the Stochastic Submodular Set Coverage (SSSC) problem in \citep{golovin2011adaptive} which has a close connection with APR problem.

\begin{problem}The SSSC problem involves a ground set of elements $U = \{u_1, u_2, \cdots u_n\}$ and a collection of items $E = \{e_1, e_2, \cdots, e_m\}$, where each item $e$ is associated with a distribution over subsets of U. When an item is selected, a set is sampled from its distribution, i.e., it will reveal which subset of $U$ will be covered. The objective of this problem is to find an adaptive policy $\pi$ that selects items to cover all elements in $U$ while minimizing the expected number of items. To define the coverage, we define a utility function $f: 2^E \mapsto \mathbb{R}$ that quantifies the coverage achieved by the current state. The complete coverage is represented as states with utility meeting a predefined quota $Q$, i.e., $f(s) = Q$. 
\end{problem}

We can see that our problem can be viewed as a special case of the SSSC problem. In the APR problem, the ground set consists of all simple paths $P_{s_r}$. In every round, when a path $p$ is proposed, the IT admin will choose an edge $e \in p$ to remove, every path in the set $P' = \{p' | e \in p', \forall p' \in P_{s_r} \}$ (the set of paths containing $e$) will also be removed. Each path in the action space can be viewed as an item in the SSSC problem, with each path associated with a distribution over the potential removal of other paths. This distribution is presented as in Equation \ref{eq:conf}. The goal of APR problem is to cover all paths in $P_{s_r}$ (a cut eliminates all paths) while minimizing the number of queries. Next, we have the following definitions.

\begin{definition}
    \textbf{(Conditional expected marginal benefit)} Given a state $s$, an action $a$ and a utility function $g$, the expected marginal benefit of $a$ is defined as:
    \begin{equation}
    \label{eq:marginal}
    \Delta (a\mid s) = \sum_{e \in a}  \{ \Phi (e\mid s,a) * \left[g(s\setminus e) - g(s) \right] \}
    \end{equation}
\end{definition}

\begin{definition}
\label{def:monotone}
    \textbf{(Adaptive Monotonicity)} A utility function $g : S \mapsto \mathbb{R}_{\geq 0}$ is adaptive monotone if the \textit{benefit of selecting an action is always nonnegative}. Formally, function $g$ is adaptive monotonic if $\forall s\in S$ and $\forall e \in \{e \mid e \in a, a \in A_s \}$, we have: 
    \begin{equation}
    \label{eq:monotone}
    g(s\setminus e) - g(s) \geq 0
    \end{equation}
\end{definition}

\begin{definition}
\label{def:submodular}
    \textbf{(Adaptive Submodular)} A utility function $g : S \mapsto \mathbb{R}_{\geq 0}$ is adaptive submodular if the marginal benefit of selecting an action does not increase as more actions are selected. Formally, function $g$ is adaptive submodular if for all temporary state $s$, $s'$ such that $\psi_{s} \subseteq \psi_{s'}$, $\forall a \in A_{s_r}$, we have:
    \begin{equation}
    \label{eq:submodular}
    \Delta(a\mid s) \geq \Delta(a \mid s')
    \end{equation}
\end{definition}

\textbf{The reason for introducing these concepts} is to port algorithms from the Adaptive Submodular framework to the APR problem while ensuring the theoretical approximation bound. To do so, we need to design a utility function $g$ associated with the APR problem that satisfies two key conditions: (1) adaptive monotonicity and (2) adaptive submodularity.

\textbf{Utility Function:} The utility function $g : S \mapsto \mathbb{N}$ is defined as:
    \begin{equation}
    \label{eq:utility}
        g(s) = |\bigcup_{e \in range(\psi_{s})} h(e, P_{s_r})|
    \end{equation}

where the function  $h(e, P) = \{p | e \in p, \forall p \in P\}$ returns the set of paths $p\in P$ that contains $e$. To remind, $range(\psi_{s})$ is the set of edges that have been removed upon state $s$. We have the following lemma for the utility function $g$:

\begin{lemma}
\label{lemma:g}
    Function $g$ is both adaptive monotonic and adaptive submodular 
\end{lemma}
   \begin{proof}

        First, consider the function $g$, which counts the number of paths removed up to the current state. Suppose we are at any state $s$, and path $p$  is proposed to the IT admin, who then chooses to remove an edge $e$. Removing $e$ eliminates at least the proposed path $p$ since $e \in p$. This means $g(s \setminus e) \geq  g(s) + 1$ and satisfying Equation \eqref{eq:monotone}. Therefore, $g$ is adaptive monotone. 

        Moreover, since the utility function $g$, as defined in Eq. \eqref{eq:utility}, returns the number of paths removed by IT admin decision from the root state. Furthermore, removing an edge $e$ in a state $s$ with $\psi_s \subseteq \psi_{s'}$ will result in more paths being eliminated than removing $e$ in the successor state $s'$ which is formally expressed as $g(s\setminus e) - g(s) \geq g(s'\setminus e) - g(s')$. This implies: 

        \begin{equation}
        \nonumber
        \begin{split}
            \Delta(a \mid s) &=  \sum_{e \in a}  \{ \Phi (e\mid s,a) * \left[g(s\setminus e) - g(s) \right] \} \\
                            &\geq \sum_{e \in a}  \{ \Phi (e\mid s',a) * \left[g(s'\setminus e) - g(s') \right]\} = \Delta(a \mid s')
                             % &\geq  \sum_{e \in a}  \{ \Phi (e\mid s,a) * \left[g(s'\setminus e) - g(s') \right]\} \\
                             % % &= \sum_{e \in a}  \{ \Phi (e\mid s',a) * \left[g(s'\setminus e) - g(s') \right]\} = \Delta(a \mid s')
        \end{split}
        \end{equation}
        The transition from line 1 to line 2 is valid because $\Phi (e\mid s,a) = \Phi (e\mid s',a)$, $\forall s, s', a$ as the transition probabilities of an action do not depend on the state, as defined in the equation \ref{eq:conf}. This proves the adaptive submodularity of $g$.
\end{proof}

Note, to ensure Theorem \ref{theorem:approx1} holds, $g$ must be both \textit{strongly adaptive monotonic and submodular}. While Definition \ref{def:monotone} aligns with the concept of strongly adaptive monotonicity \citep{golovin2011adaptive}, Definition \ref{def:submodular} is only adaptive submodularity. A function is strongly adaptive submodular if it is (1) adaptive submodular and (2) pointwise submodular. Although we have proven the former, we admit the second property for $g$ and hence $g$ is also strongly adaptive submodular.
We defer the definition and proof of (2) to the section \ref{sec:asub}

\textbf{Greedy with marginal gain strategy (Algorithm \ref{alg:greedy}):} By applying a greedy strategy with our problem-tailored utility function $g$, we have our Adaptive Submodular Algorithm as shown in Algorithm \ref{alg:greedy}.
In this algorithm, during each query round, we greedily propose the action $a \in A_{s}$ that yields the highest expected marginal benefit with respect to the utility function $g$. 
Here again, due to Lemma \ref{lemma:subpath}, we only need to enumerate the action space (for the $G_{s_r}$) once beforehand.
Note, while the path enumeration problem is known to be $\#\mathcal{P}$-hard, our experiments with attack graphs demonstrate that this enumeration can be performed within a reasonable runtime. The efficiency of this process is largely because attack graphs typically involve only subgraphs of the larger AD structure. 

\begin{algorithm}[H]
 \caption{Adaptive Submodular Strategy (APP)}
 \label{alg:greedy}
 \SetAlgoLined
 \LinesNumbered

 % suppress printed "end" markers
 \SetKwFor{For}{for}{:}{}
 \SetKwFor{ForEach}{foreach}{:}{}
 \SetKwFor{While}{while}{:}{}
 \SetKwIF{If}{ElseIf}{Else}{if}{then}{else if}{else}{}

 \KwIn{Directed graph $G(V, E)$, source $s$ and destination $t$}
 \KwOut{approximate propose strategy}
\BlankLine
 Initialise set of simple path $A_{s_r}$ of original state $s_r$, set of proposed path $A_{\pi}$\;

 \While{$s \notin \{\perp_b \cup \perp_d\}$}{
   \ForEach{$a \in A_{s_r}\setminus A_{\pi}$}{
     $\Delta(a|s) = \sum_{e \in a}  \{ \Phi (e|s,a) * \left[g(s\setminus e) - g(s) \right] \}$\;
   }
   $a^{*} = \arg\max_a \Delta(a|s)$\;
   $A_{\pi} = A_{\pi} \cup \{a^*\}$, proposed $a^*$, observe outcome $e*$\;
   $s = s \setminus e^*$, progressing to new state\;
 }
\end{algorithm}

Once we have proven that $g$ is both adaptive monotonic and adaptive submodular in Lemma \ref{lemma:g}, the following theoretical approximation ratio for Algorithm \ref{alg:greedy} follows.

\begin{theorem}
\label{theorem:approx1}
 Algorithm \ref{alg:greedy} achieves a $(\ln{|P_{s_r}| + 1)^2}$-approximation for the APR problem with $B = |P_{s_r}|$.    
\end{theorem}
\begin{proof}
    Applying Theorem 17 from \citep{golovin2011adaptive}, the approximation ratio for the greedy algorithm, which maximizes marginal gain, is bounded by:
    \begin{equation}
        c(\pi) \leq \alpha c(\pi^*)\left(\ln{\frac{Q}{\eta}} + 1\right)^2
    \end{equation}
    where $c(\pi)$ is the average cost of the greedy policy which is $\alpha$-approximate w.r.t items cost, $c(\pi^*)$ is the cost of the optimal policy, $Q$ is the utility target for covering the ground set $U$, $\eta$ is the threshold parameter such that $f(S) > Q - \eta$ will implies $f(S) = Q$.

    For algorithm 1 and utility function $g$, we have $range(g) \subseteq \mathbb{N}$ which imply that $\eta = 1$. In the context of the APR problem, achieving full coverage requires to disconnect every path in $P_{s_r}$. This condition is satisfied when $Q = |P_{s_r}|$. Also, $\alpha = 1$ for algorithm 1. Therefore, the following approximation ratio holds for Algorithm 1:
    \begin{equation}
        q(\pi) \leq q(\pi^*)\left(\ln{|P_{s_r}|} + 1\right)^2
    \end{equation}
\end{proof}

\subsection{Scalable Heuristics Algorithm (DPR)}

In this section, we present a scalable heuristic designed based on the exact algorithm and the approximate algorithm as shown in Algorithm \ref{alg:dpheu}

% Additional to approximate algorithm, in this section, we introduce several heuristics to improve the performance. Although we do not have any approximation guarantee for these heuristics. These heuristic is based on the exact algorithm and approximate algorithm, hence, experimentally outperform the approximate algorithms and allow us to scale to larger graphs. 

\textbf{Heuristics based on Exact Algorithm (DPR)} The exact dynamic programming algorithm struggles to scale in realistic scenarios due to two main issues. First, APR problem's MDP often has large action space, particularly in the initial states, where the action space size corresponds to the enumeration of simple paths in the original graph. Second, for each action, the number of child subproblems to solve can grow up to $\mathcal{O}(|E|^k)$. We proposed a scalable heuristic in Algorithm \ref{alg:dpheu}, designed to restrict the subproblem space in dynamic programming to a manageable size and enable efficient execution on graphs of practical scale. We call this algorithm dynamic programming with restriction (DPR). Function $DPR(s', r, B')$ in Algorithm \ref{alg:dpheu} shows a modification of the Exact Dynamic Programming algorithm with restriction. The first restriction is that instead of considering subproblems from $B$ steps ahead, \textit{DPR} reduces the lookahead to $B'$ steps, where $B' < B$. The second restriction is to avoid enumerating every possible state (which can be problematic in the early stages). Instead, we only consider a set of $\tau$ candidate paths, implemented as the $path\_sampling(A_{s'}, \tau)$ function in Algorithm \ref{alg:dpheu}. In general, we modify each heuristic to return the top $k <\tau$ candidate paths, rather than a single best path, based on each heuristic’s ranking criterion. For example, the approximate heuristic ranks paths by their marginal gain according to $g$, then selects the top $k$ among them. This function will draw from multiple heuristic methods—such as those derived from an approximation strategy (Algorithm \ref{alg:greedy}), shortest paths, approximate strategies on sets of shortest paths, or paths likely to remove an edge in a minimum cut. As our experiments show that these approaches provide strong performance. This modification decreases the number of subproblems to $\mathcal{O}((\tau)^{B'})$ for each DP, making DP more feasible in larger settings.
While this adjustment may affect the optimality of the solution (compared to the exact DP algorithm), it significantly improves the scalability. Empirically, we will experimentally demonstrate that the \textit{DPR} algorithm scales better than the exact algorithm and outperforms the approximate algorithm on every graph. The idea of DPR is to run the dynamic program with the restricted number of subproblem. Here we restricted the number of subproblem to $\mathcal{O}(\tau)^{B'}$. This make the complexity of DPR become $\mathcal{O}(|V|^{k}+(|E|+|V|)*(\tau)^{B'})$

\begin{algorithm}[H]
 \caption{Heuristics based on Exact Algorithm (DPR)}
 \label{alg:dpheu}
 \SetAlgoLined
 \LinesNumbered

 % suppress printed "end" markers
 \SetKwFor{For}{for}{:}{}
 \SetKwFor{ForEach}{foreach}{:}{}
 \SetKwFor{While}{while}{:}{}
 \SetKwIF{If}{ElseIf}{Else}{if}{then}{else if}{else}{}
 \SetKwProg{Fn}{Function}{}{}

 \KwIn{Directed Attack Graph $G(V, E)$, budget $B$, look-ahead budget $B'$}
 \KwOut{Heuristic query strategy}
\BlankLine
 \While{$s \notin \{\perp_b \cup \perp_d\}$}{
   $\pi \gets \text{DPR}(s, |A_{\pi}|, B')$\;
   $a^* \gets \pi(s)$\;
   $A_{\pi} \gets A_{\pi} \cup \{a^*\}$, propose $a^*$, observe outcome $e^*$\;
   $s \gets s \setminus \{e^*\}$, progressing to new state\;
 }

\BlankLine
 \Fn{DPR($s', r, B', \tau$)}{
   \For{$i \in [0, B']$}{
     \For{$s' \in \mathcal{S}_{|A_{\pi}| + i}$}{
       \If{$s'$ is in $\perp_{b}$ or $\perp_{d}$}{
         $u_{\pi}(s' \setminus \{e\}) \gets
          \begin{cases}
           \alpha, & \text{if } s' \in \perp_{b}\\
           0, & \text{if } s' \in \perp_{d}
          \end{cases}$\;
         $\pi(s' \setminus \{e\}) \gets \emptyset$\;
       }
       \Else{
         $A \gets \text{path\_sampling}(A_{s'}, \tau)$\;
         $a^* \gets \operatorname*{argmin}_{p \in A} \left\{ \sum_{e \in p} \left[ \Phi(e|s',p) \cdot u_{\pi}(s' \setminus \{e\}) \right] \right\}$\;
         $u_{\pi}(s') \gets 1 + \sum_{e \in a^*} \left[ \Phi (e|s',a^*) \cdot u_{\pi}(s' \setminus \{e\}) \right]$\;
         $\pi(s') \gets a^*$\;
       }
     }
   }
   \Return{$\pi$}\;
 }
\end{algorithm}

\section{Experiment}

In this section, we present the evaluation of our algorithm on 13 synthetic graphs of different sizes and an AD attack graph from a real organization. 
\subsection{Experiment Set Up}

In our experiment, we evaluate our algorithm using synthetic AD attack graphs generated by ADSynth \citep{nguyen2024adsynth}, a state-of-the-art AD graph generator. ADSynth models AD graphs based on Microsoft’s best practices tiering model \citep{bestpractice,bestpractice2}, where Tier 0 contains the highest privilege nodes with administrative control, Tier 1 includes high-privileged servers, and Tier 2 and beyond contain non-administrative nodes. ADSynth simulates the AD attack graph in two steps: (1) generating a best-practice AD infrastructure and (2) creating cross-tier edges. 
If every node had a predefined tier, the defense problem would become trivial, as attack paths could be easily identified and removed by removing all edges connecting lower-privilege nodes to higher-privilege nodes \citep{improhound}. 
Open-source tools like ImproHound \citep{improhound} are designed to automate this process. 
However, assigning roles to nodes is inherently challenging due to the dynamic nature of roles, overlapping responsibilities, and exceptions such as temporary access \citep{tiertalk}. We called these \textit{undefined tier nodes}. 
In our simulated attack graph, we assume the presence of a set of nodes with undefined tier connections which create attack paths from lower-privilege nodes to higher-privilege nodes. We assume that IT admin use our adaptive model with the goal of removing all attack paths from the lowest tier to Tier 0. 
Our model treats the attack graph as a single-source, single-target graph so we merge all Tier 0 nodes into a single supernode $t$ and all lowest-tier nodes into a single supernode $s$. 

For our synthetic attack graphs, we labelled them from $G1$ to $G9$. In these graphs, the number of tiers is fixed at 3, and $95\%$ of nodes in graph have well-defined tier assignments. Additionally, we also have 4 smaller versions of the graph denoted from $GS1$ to $GS4$, used in the experiment in Table \ref{tab:smallgraph}. In this small graph, defined-role ratio is about $99\%$. We also included one real AD graph that we collected from an anonymous organization, we denoted this graph as ORG. 

\begin{table}[h]
\begin{center}
\caption{Expected number of query under different algorithm \textbf{($\downarrow$ is better)}. Here, we only consider graphs where OPT can run on. }\label{tab:smallgraph}
\smallskip\noindent
\resizebox{0.45\linewidth}{!}{%
\begin{tabular}{ccccccc}
\toprule
        ~      & \textbf{$GS1$} & \textbf{$GS2$} & \textbf{$GS3$} & \textbf{$GS4$}  \\ \hline
        % \#node/\#edge   & 7               & 9               & 10              & 15              \\ \hline
        OPT    & \textbf{2.513}  & \textbf{2.592}  & \textbf{2.545}  & \textbf{2.383}  \\ \hline
        APP& \textbf{2.513}  & \textbf{2.592}  & 2.546           & 2.385  \\ \hline
        OTH1 & \textbf{2.513}  & \textbf{2.592}  & 2.546          & \textbf{2.383}  \\ \hline
        OTH2 & \textbf{2.513}  & 2.596           & \textbf{2.545} & 2.388  \\ \hline
        PPO    & \textbf{2.513}  & \textbf{2.592}  & 2.546           & \textbf{2.383}  \\ \hline
        SAC    & 2.514           & \textbf{2.592}  & 2.546           & 2.384  \\ \hline
        DPR   & \textbf{2.513}  & \textbf{2.592}  & 2.546           & \textbf{2.383} \\ 
\toprule
\end{tabular}}
\end{center}
\end{table} 

All of the experiments are carried out on a high-performance computing cluster with 1 CPU and 24GB of RAM allocated to each trial. In Tables \ref{tab:smallgraph} and \ref{tab:biggraph}, we report the average number of queries over 16,000 trials. The budget constraint $B$ is set at 10 for all experiments on synthetic graphs. For the real AD graph ORG, due to computing resource limitations, we report the average number of queries over 200 trials. Also for ORG, we reserve a higher budget of 20 and 30 queries due to the size of this graph, denoted ORG(20) and ORG(30) respectively. For the DPR algorithm, we set $\tau = 16$ actions and a lookahead budget of $B' = 4$ step. We reserve a higher budget of 20 and 30 queries due to the size of this graph. 
For the DPR algorithm, we set $\tau = 16$ actions and a lookahead budget of $B' = 4$ step. 

\textbf{Data collection and preprocessing for the ORG graph} The data was collected from an anonymous organization using SharpHound ~\cite{SharpHound}, which gathers detailed information on user sessions, group memberships, ACLs (Access Control Lists), and permissions within the AD environment. SharpHound is commonly used tool to collect domain data, which is then imported into BloodHound for visualizing potential attack paths. The data collection took place on Thursday, 28th October 2022, at approximately 10:00 AM during regular working hours in the local time zone. The collected Active Directory dataset consists of 125,444 nodes and 1,195,432 edges. While we do not have specific insight into the number of tiers in this AD instance, nor the exact rules the organization uses to define them (due to the confidentiality), we assume the AD follows the common three-tier model. We heuristically classify nodes into three tiers: nodes flagged as "adminaccount" (which SharpHound identifies as highly privileged groups) are categorized as Tier 0, servers and services are placed in Tier 1 (determined by their name), and all remaining nodes are assigned to Tier 2. In this experiment, we assumed the  undefined-tier ratio of 0.95. Given the shear size of this AD graph, we increased the budget from 10 to $B = 20$ and $B = 30$ queries. Additionally, due to runtime constraints, we were unable to include results for the DPR algorithm, as generating queries for 16,000 episodes for this algorithm would take approximately 800 seconds per episode.

\begin{landscape} 
\vspace*{\fill}
\begin{table*}[h]
\begin{center}
\caption{Expected number of query under different algorithm \textbf{($\downarrow$ is better)}. AVG.RANK represents the average head-to-head performance ranking of each algorithm across all evaluated graphs. $\#n/\#e$ show the number of nodes and edge in the graph. MC is the min-cut. }\label{tab:biggraph}
\smallskip\noindent
\resizebox{1\linewidth}{!}{%
\begin{tabular}{c|ccccccccccc|c}
\toprule
        ~      & \textbf{G1}& \textbf{G2} & \textbf{G3} & \textbf{G4} & \textbf{G5} & \textbf{G6} & \textbf{G7} & \textbf{G8} & \textbf{G9} & \textbf{ORG}(20)& \textbf{ORG}(30) & AVG.RANK  \\ 
        \#n/\#e   & 1047/5078              & 1047/5091               & 1047/5116               & 5147/25376             & 5139/25153              & 5139/25161              & 10070/48161               & 10070/48170               & 10070/48192              & 125444/1195432              & \_              &  \_     \\ 
        % \#n/\#e   & 36/70              & 44/86               & 40/78               & 154/314              & 166/343              & 153/323              & 261/547               & 267/578               & 304/652              & 583/1393              & 583/1393              &  $\_$     \\ \hline 
        \textbf{MC}   & 3               & 3                & 3                & 3                & 3                & 4
               & 3               & 3                & 3                & 8                & \_  & \_ \\ \hline

        \textbf{APP}& 3.821           & 3.762            & 4.534            & 4.334            & 3.879            & 4.594            & 3.807             & 3.869             & 3.590            & 17.605           & 18.840           & 3.889   \\ \hline
        \textbf{OTH1} & 3.816           & \textbf{3.755}   & \textbf{4.409}   & \textbf{3.885}   & 3.880            & 4.593            & 3.810             & 3.893             & 3.584            & 17.485           & 18.600           & 3.185  \\ \hline
        \textbf{OTH2} & \textbf{3.813}  & 3.756            & 4.437            & 3.904            & 3.883            & 4.592            & 3.799             & 3.871             & 3.570            & 17.535           & \textbf{18.535 }          & 3.333  \\ \hline
        \textbf{PPO}    & 3.816           & \textbf{3.755}   & 4.425            & 3.905            & 3.876            & 4.587            & 3.797             & 3.876             & 3.573            & 17.665           & 18.835           & 2.667  \\ \hline  
        \textbf{SAC}    & 3.854           & 3.799            & 4.490            & 3.901            & \textbf{3.874}   & 4.606            & \textbf{3.792}           & 3.876            & 3.575            &  18.005          & 18.560            & 2.667  \\ \hline
        \textbf{DPR}   & 3.816           & \textbf{3.755}   & \textbf{4.409}   & 3.901            & 3.876            & \textbf{4.589}   & 3.797             & \textbf{3.869}    & \textbf{3.568}   &  \textbf{17.480}          & 18.555          & \textbf{1.444}  \\ 
\toprule
\end{tabular}}
\end{center}
\end{table*}
\vspace*{\fill}
\end{landscape}

\subsection{Baseline Algorithms}

\textbf{Reinforcement Learning.} This approach shares a similar concept with DPR but replaces the use of Dynamic Programming with restricted lookahead by a model-free reinforcement learning to learn the query strategy. We utilize two model-free reinforcement learning models: Proximal Policy Optimization (PPO) \citep{schulman2017proximal} and Soft Actor-Critic for Discrete Action (SAC) \citep{christodoulou2019soft}.
These Actor-Critic methods allow us to train a reinforcement learning agent across multiple environments simultaneously, where each environment represents a possible realization $\psi$ of the APR problem. The objective in the APR problem is to derive a policy that minimizes the expected number of queries across all possible realizations. By interacting with several realizations, the RL agent learns a policy that minimizes the overall reward (corresponding to the number of queries) for these scenarios.
For the implementation of RL, we encode the observation space as a vector of $(E + \tau B)$ binary bits. The first $E$ bits represent a one-hot encoding of the edges that have been removed through queries, while the remaining $\tau B$ bits encode the taken actions. We allocated a100 GPUs for the training of RL agents.

\noindent\textbf{Others Heuristics} We also introduce two other heuristics called OTH1 and OTH2 which are designed based on the approximate algorithm:
\begin{itemize}
    \item \textbf{Minimum-cut-based heuristic (OTH1)} The idea of this heuristic is to proposes paths with the highest likelihood of removing an edge in the minimum cut set. Formally, it selects the path $p = \arg\max_{a \in P_{G'}:a \cap mc(G')} \sum_{e \in a}  \{ \Phi (e|s,a) * \left[g(s\setminus e) - g(s) \right] \}$ where $mc(G')$ return the $s-t$ minimum cut of the temporary graph $G'$. Suppringingly, MC is actually an approximate algorithm  with the approximation ratio of $|P_{s_r}|$ when queries have a cost (i.e., $c > 1$). However, since our problem assumes unit cost (i.e., $c = 1$), this approximation bound becomes trivial, and we treat MC as a heuristic. Despite this, the MC heuristic outperforms the APP algorithm in 9 out of 12 graphs in our experiments, suggesting the potential for a tighter approximation bound and leave this for future research.
    \item \textbf{Approximate Algorithm on Shortest Paths (OTH2)} In this heuristic, instead of running the APP algorithm on all simple paths, we restrict it to only the set of shortest paths. To remind, the APP algorithm is used to find the path that yields the highest marginal gain within the given set of paths. Surprisingly, this simple modification performs exceptionally well, outperforming the original APP algorithm in 10 out of 12 graphs.
\end{itemize}

\subsection{Sampling Scheme for DPR and RL} 

Now we will discuss about the sampling scheme for DPR, and both of the RL technique.  In total, we have 3 sampling scheme for each algorithm. We implement our sampling algorithm by modifying the APP, OTH1 and OTH2 algorithm to sample the action. The idea is instead of return the top action proposed by these algorithm, we return the top-k path by run these algorithm. The idea is to use DPR or RLs techniques to find the best action among $k$ proposed paths in each step. 
Specifically, the APP will propose the path with the highest marginal gain with respect to utility $g$, here, we use APP to return top $k$ path with highest marginal gain. Similarly, for OTH2, we return top $k$ shortest path. For OTH1, we return top $k$ path that have the highest likelihood of removing a edge in minimum cut. In our experiment, we fix $k=4$ for DPR and $k=16$ for the reinforcement learning methods. This choice is due to DPR's limited scalability at $k=4$, while PPO and SAC can scale effectively to $k=16$.  In the experiment, for each of algorithm, we will report the result of the sampling scheme that yield the highest performance for each algorithm.

\subsection{Performance Interpretation}

In Table \ref{tab:smallgraph}, we report the performance of our proposed algorithm in the graph where OPT can optimally come up with the query policy without out-of-memory error. As we mentioned, OPT is very costly computationally, we are only able to scale it to a graph with 17 nodes, 32 edges and 16 attack paths ($GS4$ graph). Overall, all of our heuristics (OTHs, PPO, SAC and DPR) perform very well with the small optimality gap.

In Table \ref{tab:biggraph}, we present the performance of our algorithm on 13 large synthetic graphs and one real-world AD graph (\textbf{ORG}) from an anonymous organization. We observed that DPR consistently achieved the best performance. Noticeably, DPR outperformed APP in every graph. All proposed algorithms outperformed APP. Nevertheless, APP's performance is theoretically guaranteed which may be useful for some worst-case scenarios. This algorithm is also useful as a path sampling scheme for the DPR algorithm as the action space of DPR algorithm contain the approximate strategy which somewhat helps DPR to have a guaranteed performance. While the RL algorithms (PPO and SAC) performed well on synthetic graphs, their performance was worse on the real AD graph. We suspect this is due to the real graph requiring a significantly larger number of queries. The RL policies are myopic, meaning they excel in scenarios with fewer queries by prioritizing short-term gains but struggle when a higher number of queries is needed, as they fail to account for long-term gains.   
% In this experiment, we also include the 

The adaptive hardening of AD security has been studied as the LQGCT problem by Guo et al. \citep{guo2024limited}. Their model simulates IT admin' behaviour as a binary decision-making process: at each step, an edge is queried, and the administrator labels it as either "cut" or "retain." However, this approach often fails when IT admins are overly conservative, retaining too many edges and leading to unsuccessful cuts. In contrast, our model queries a path in each step, presenting a multiple-choice decision for the IT admin to select one edge to remove. In table \ref{tab:guocomp}, we compare the graph-cutting performance of the RL algorithm from LQGCT with our DPR algorithm. The results from 512 trials demonstrate that our model achieved more successful cuts compared to Guo's model. We observe that a larger budget leads to a higher cutting success rate in our model, this means we can guarantee a successful cut in every trial by allocating a sufficiently larger budget.

\begin{table}[h]
\begin{center}
\caption{The number of trials with successful cuts between RL from LQGCT model and DPR algorithm over 512 trials. \textbf{($\uparrow$ is better)}}\label{tab:guocomp}
\smallskip\noindent
\small
\resizebox{0.45\linewidth}{!}{%
\begin{tabular}{ccccc}
\toprule
       % & \textbf{Guo's RL} & \textbf{DPR} \\ 
       &\multicolumn{2}{c}{LQGCT's RL }&\multicolumn{2}{c}{DPR}\\
                        & B = 5 & B = 10 & B = 5 & B = 10 \\
        \hline
       \textbf{G7}&  128  & 144    &   278 &  459   \\ \hline
       \textbf{G8}&  96   &  96    &   419 &  504   \\ \hline
       \textbf{G9}&  96   &  96    &   339 &  487   \\ 
\toprule
\end{tabular}}
\end{center}
\end{table}

The length of the proposed path is an important consideration in the proposing process, as longer paths could place additional burden on the IT administrator. Table~\ref{tab:len} reports the expected average lengths of the proposed paths in our experiments. The results show that the proposed paths are relatively short, even for real AD environments (DPR average path length is only 2.776 for ORG). This indicates that our approach does not create significant burdens for IT administrators.

\begin{table*}
\begin{center}
\caption{Expected length of the queried path by different algorithms. (\textbf{$\downarrow$ is better)} }\label{tab:len}
\smallskip\noindent
\resizebox{1\linewidth}{!}{%
\begin{tabular}{c|ccccccccccc}
\toprule
        ~      & \textbf{G1}& \textbf{G2} & \textbf{G3} & \textbf{G4} & \textbf{G5} & \textbf{G6} & \textbf{G7} & \textbf{G8} & \textbf{G9} & \textbf{ORG}(20)& \textbf{ORG}(30)  \\\hline
        \textbf{APP}& 2.228           & 2.219            & 2.524            & 2.178            & 2.252            & 2.140            & 2.864             & 1.862             & 2.162            & 3.488          & 3.400   \\ \hline
        \textbf{OTH1} & 2.228           & 2.218         & 2.473           & 1.923         & 2.252            & 2.139            & 2.864            & 1.863             & 2.162            & 2.769           & 2.809        \\ \hline
        \textbf{OTH2} &  2.258            & 2.238          & 2.476           & 1.894           & 2.279            & 2.159             & 2.902             & 1.861            & 2.184           & 2.724          & 2.778  \\ \hline
        \textbf{PPO}    & 2.259           & 2.236        & 2.474            & 1.893            & 2.279            & 2.158            & 2.901             & 1.862             & 2.185            & 2.726           & 2.788    \\ \hline  
        \textbf{SAC}    & 2.259           & 2.238            & 2.477            & 1.894            & 2.279   & 2.158            & 2.907           & 1.861            & 2.186            &  2.728          & 2.783      \\ \hline
        \textbf{DPR}   & 2.259          & 2.237   & 2.473   & 1.894            & 2.277            & 2.159   & 2.901             & 1.86    & 2.184   &  2.72          & 2.776         \\ 
\toprule
\end{tabular}}
\end{center}

\end{table*}
\section{Conclusion}

In this paper, we proposed a practical human-in-the-loop combinatorial problem for network security called Adaptive Path Removal problem. This problem was motivated by the technical requirements and limitations of current industrial models. The goal of our model is to reduce the workload for security teams in an adaptive manner. 
% We first provided that our problem is $\#\mathcal{P}$-hard.
We proposed a comprehensive set of solutions, including an exact algorithm, an approximate algorithm, and several scalable heuristics.
Among these, our DPR heuristic, designed based on both the exact and approximate algorithms, exhibited superior performance.
Specifically, DPR demonstrated the ability to run effectively on larger-scale graphs compared to the exact algorithm and consistently outperformed the approximate algorithm across all tested graph scenarios.
We verify the effectiveness of our algorithm on several synthetic AD graphs and an AD attack graph collected from a real organization.

\section{Appendix-note about Strong Adaptive Monotonicity and Strong Adaptive Submodular}
\label{sec:asub}

To define strong adaptive submodularity, we first need the following extension of $\Delta(a\mid s)$
\begin{definition}
    \textbf{(Conditional expected marginal benefit (extended version)} Given a state $s$ and $s'$ where $\psi_s \subseteq \psi_{s'}$, an action $a$ and a utility function $g$, the expected marginal benefit of $a$ is defined as $\Delta (a | s;s') = \sum_{e \in a}  \{ \Phi (e\mid s',a) * \left[g(s\setminus e) - g(s) \right] \}$
\end{definition}

\begin{definition}
    \textbf{(Strong Adaptive Monotonicity)} A utility function $g : S \mapsto \mathbb{R}_{\geq 0}$ is adaptive monotone if the \textit{benefit of selecting an action is always nonnegative}. Formally, function $g$ is adaptive monotonic if for all $s\in S$ and $e \in \{e \mid e \in a, a \in A_s \}$, we have $g(s) - g(s\setminus e) \geq 0$
\end{definition}

\begin{definition}
    \textbf{(Strong Adaptive Submodular\footnotemark[1])} A utility function $g : S \mapsto \mathbb{R}_{\geq 0}$ is adaptive submodular if marginal benefit of selecting an action not increase as more action are selected. Formally, function $g$ is adaptive submodular if for all temporary state $s$ and $s'$ such that $\phi_{s} \subseteq \phi_{s'}$ and action $a \in A_{s_r}$, we have $\Delta(a | s;s') \geq \Delta(a \mid s')$
\end{definition}

We also provide a the definition of pointwise submodularity as follows:

\begin{definition}
    \textbf{(Pointwise Submodular)} A utility function $g : S \mapsto \mathbb{R}_{\geq 0}$ A function $g$ say to be pointwise submodular if $g$ is submodular in every state for any realization $\psi$. Formally, function $g$ is pointwise submodular if for all temporary state $s$ and $s'$ such that $\phi_{s} \subseteq \phi_{s'}$ and for all $e \in \{e \mid e \in a, a \in A_s \}$, we have $g(s\setminus e) - g(s) \geq g(s'\setminus e) - g(s')$
\end{definition}

As our definition of adaptive monotonicity is already satisfy strong adaptive monotonicity condition so we admit the proof.

A sufficient condition for strong adaptive submodularity is that the function $g$ is both \textit{adaptive submodular} and \textit{pointwise submodular}. 
The utility function $g$ is pointwise submodular as for every state $s$ and $s'$ such that $s\subseteq s'$, and every action $a \in A_{s_r}$, we have $g(s\setminus a) - g(s) \geq g(s'\setminus a) - g(s')$. By definition, if a function $g$ is both adaptive submodular and pointwise submodular, then $g$ is strongly adaptive submodular.

%% file: Chapter6/Chapter6.tex
\chapter{A Reinforcement Learning Approach to Security Adaptive Connectivity Testing in Active Directory} % Main chapter title
\label{chapter:paper4} % Change X to a consecutive number; for referencing this chapter elsewhere, use \ref{ChapterX}

In this chapter we present another robust human-in-the-loop model for edge-removal in AD attack graph. Existing solutions often simplify the complex, feature-driven decisions that IT administrators make during review. We propose a feature-based edge query model that bridges practical gaps in existing work.
Motivated by the need for a practical tool for this security problem, we propose the Feature-based Adaptive Connectivity Test (F-ACT) model. The model is defined on a graph with a source $s$ and a target $t$, where each edge has an initially hidden state: ON/traversable or OFF/blocked, governed by a latent decision function over an associated feature vector. The objective is to find an adaptive policy that tests the $s-t$ connectivity by finding either a path of ON edges or a cut of OFF edges separating them, while minimizing the expected number of queries in the process. 
The sequential nature of the problem makes it a natural fit for a reinforcement learning (RL) formulation. However, applying standard RL algorithms is challenging due to two key obstacles: a significant exploration problem and the sparse reward issue. To overcome these, we design a new RL algorithm with three novel components: a policy-agnostic self-improvement mechanism, a policy-invariant reward shaping and a custom prioritized Experience Replay. Our self-improvement mechanism is capable of refining any given policy, allowing it to push performance boundaries beyond existing state-of-the-art approaches. We demonstrate the effectiveness of our algorithm on six large-scale AD graphs and four other real-world miscellaneous networks.

\section{Introduction}
Microsoft's Active Directory (AD) is a directory service that enables the IT admin to manage security permissions and control accesses within Windows domain networks. This management sevices have been deployed and widely used by large enterprises \citep{adper}. An AD environment is naturally described as a graph where nodes are accounts/computers/groups, and the directed edges describe accesses/permissions/vulnerability. 

One way to harden the network is to disconnect attack paths, routes an attacker might use to gain the access to high-privilege nodes. While existing models \citep{guo2023scalable,zhang2024practical,goel2023evolving} and commercial tools such as BloodHound \citep{BloodHound} can identify edges for removal to reduce these paths, these removals are often not directly implementable as they can disrupt critical business operations. This creates what the literature terms the “implementable fixes” problem \citep{dunagan2009heat,zheng2011active,guo2024limited}, where each proposed fix must be manually approved by an IT admin. This practical constraint has motivated the development of adaptive security models that incorporate human intervention into the hardening process. Our proposed model functions as a 'security wizard' that guides an IT admin through the hardening process by proposing a edge for removal at each step. In response, the IT admin can either approve the removal, which labels the edge as 'OFF' (removed), or disapprove the removal, which labels it as 'ON' (traversable). The process continues until all attack paths are eliminated, or we establish that elimination is impossible.
% (When it is impossible, IT admin can instead disable the account or split its into 2 with difference privileges)
% \footnote{When direct elimination is impossible, IT admin can instead disable the account or split its into 2 with difference privileges}. 
Since each query corresponds to a costly manual intervention, our model's objective is to minimize this human effort by designing a model that can derive and generalize adaptive edge-removal policy.

The effectiveness of this adaptive tool depends on its ability to model the IT admin's decision process. Prior works \citep{dunagan2009heat,guo2024limited,ngo2025adaptive} often oversimplifies this decision-making process by assigning a fixed, independent probability of removal to each edge (i.e., how likely the edge will be approved for removal by an IT admin). 
However, modeling an IT admin's decision as a simple probabilistic event is insufficient. An admin's choices are not independent but are highly dependent on the security context of a given network edge, which is defined by its features. They consider a variety of factors, such as the permission type, account inactivity, departmental roles, and existing access rights. For instance, if an admin approves removing an edge with \textit{GenericAll} permission (granting full control over an object) from an inactive account, the model should learn that edges with these properties are more prone to removal and adaptively update this knowledge before making its next query.
An edge's security context is represented by a feature vector that acts as a list of its associated security risks such as the permission type, the status of the account, its role, etc. This allows us to formalize the IT admin’s decision-making ‘rulebook’ as a hidden decision function, which maps an edge’s feature vector to its state. The proposed algorithm aims to learn this function, thereby generalizing from specific decisions to a network-wide policy. 
% While we use a linear model for f in our exposition, our framework is general and accommodates more complex, non-linear functions.
% Our model formalizes the IT admin's decision making 'rulebook' as a hidden decision function that takes an edge's feature vector as input to determine its true state (removable or not). The proposed algorithm aiming to learn this function to generalize from an administrator's decisions. 

Furthermore, motivated by the goal of saving human effort in the querying process, we introduce the following monotonic risk assumption: ground-truth labeling for edges are monotone with respect to their risk level. This enables deductive reasoning: If an IT admin decides an edge with a specific set of risks is removable ('OFF'), our agent can infer that any other edge with all the same risks plus additional ones (a formal superset) is also 'OFF', and the reverse holds for 'ON' decisions.
% from an 'OFF' decision on , the agent can infer that any edge with a superset of its risk factors is also 'OFF'; conversely, from an 'ON' decision, it can infer that any permission with a subset of those risk factors is also 'ON'. 
The soundness of this rule can be illustrated with an example from Figure \ref{fig:instance}, where we define four binary risk features for each edge:
(1) if the permission is GenericAll, (2) if it grants access to a HighPrivilege object, (3) if the source account has been \textit{Inactive} for an extended period, and (4) if the permission is an AdminTo relationship. Therefore, if an IT admin decided a permission with feature vector $a = [1, 1, 0, 0]$ is removable ('OFF'), our wizard agent can also conclude that a permission with vector $b = [1, 1, 1, 0]$ is also 'OFF' without a new query. This deduction is sound because the additional risk feature (Inactive) only strengthens the reasoning for removal. This allows our method to scale to massive graphs where reviewing individual edges is impractical. By generalizing a single decision to all related and higher-risk edges, it allows IT admin to effectively generalize a security policy from a few queries, significantly pruning the search space.

\begin{figure}[h]
    \centering
    % center the two subfigures as one block; tweak the gap in \hspace{...}
    \makebox[\textwidth][c]{%
        \begin{subfigure}[t]{0.45\linewidth}
            \centering
            \includegraphics[width=\linewidth]{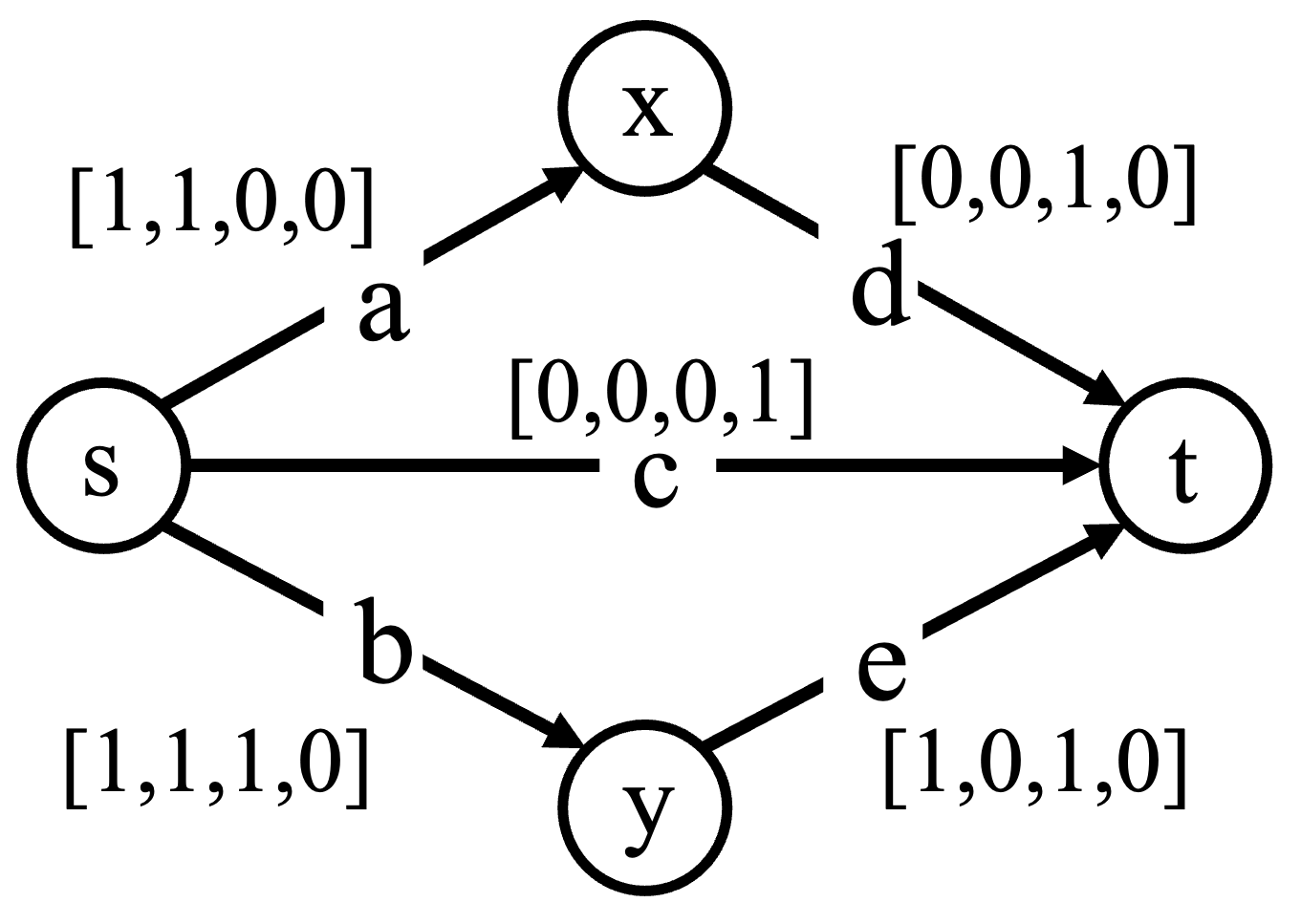}
            \caption{An instance of a AD attack graph}
            \label{fig:instance}
        \end{subfigure}
        \hspace{1.75em}%
        \begin{subfigure}[t]{0.25\linewidth}
            \centering
            \includegraphics[width=\linewidth]{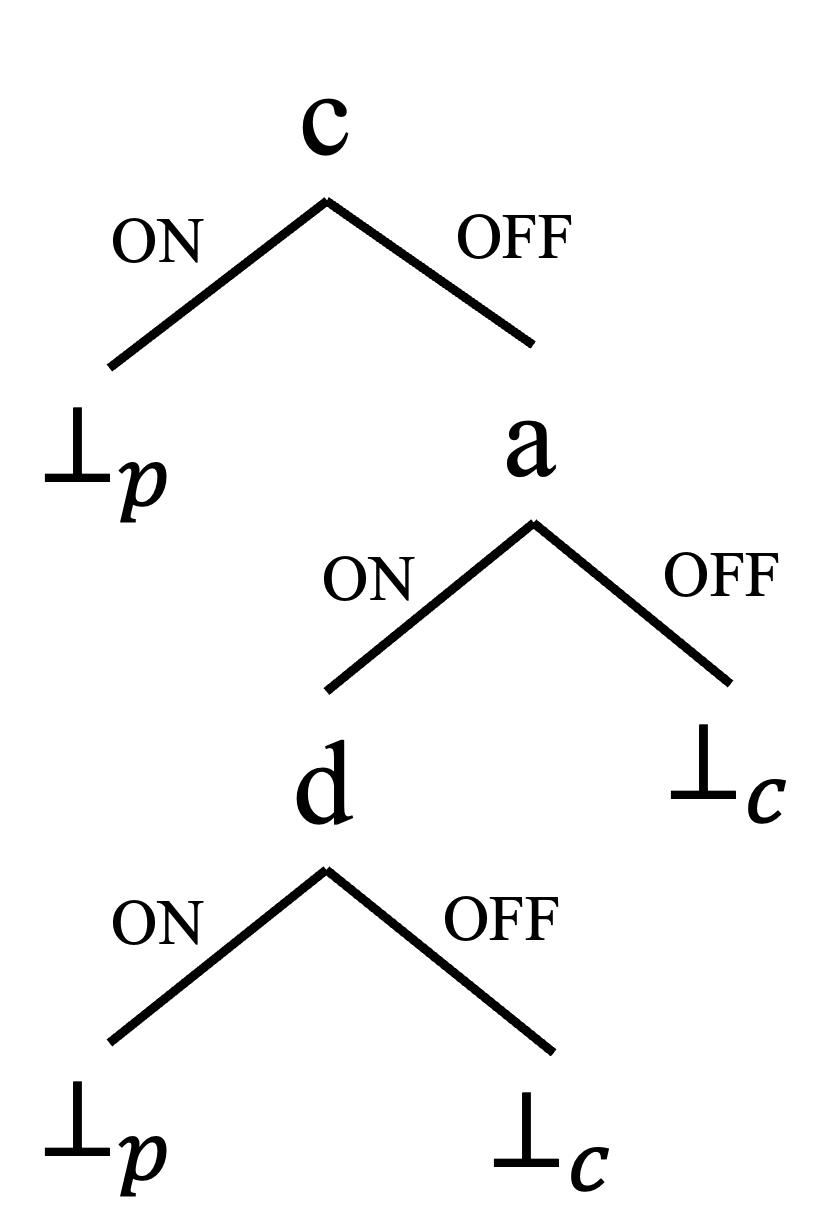}
            \caption{Optimal policy.}
            \label{fig:second_in_col}
        \end{subfigure}%
    }
    \caption{Example of F-ACT instance. $\perp_p$ means found a path of ON edges, $\perp_c$ means found a cut of OFF edges.}
    \label{fig:policy}
\end{figure}

% The Adaptive Network Connectivity Test (ACT) is a combinatorial problem that seeks to determine $s-t$ connectivity in a graph with a minimum number of edge queries. First proposed by Fu et al. \citep{fu2017determining,fu2017complexity}, this class of network-based framework has recently been successfully applied to practical problem such as hardening cybersecurity networks \citep{guo2024limited,ngo2025adaptive}. The ACT problem is defined on a graph where each edge has 2 true state: either online (ON) and thus traversable, or offline (OFF) and blocked. Starting with state of all edge are hidden, an agent can query one edge per round to reveal its true state. The hidden state is governed by a known prior probability $p(e)$ that the edge is 'ON'. The objective is to adaptively select a sequence of queries to find either a path of ON edges establishing $s-t$ connectivity or a separating cut of OFF edges, while minimizing the total number of queries made. 

This chapter contributions can be summarized as follows:

% \begin{itemize}
%     \item We introduce a new theoretical combinatorial optimization model called Feature-based Adaptive Connectivity Test. This is the first adaptive AD-focused model that incorporates security-context features to model an IT admin decision-making process.
%     \item We prove that the F-ACT problem is $\#\mathcal{P}$-hard and develop RL4FT, a novel RL algorithm that addresses the key challenges of exploration and sparse rewards through a heuristic-guided self-improvement mechanism, a provably policy-invariant reward shaping scheme, and a custom Prioritized Experience Replay
%     \item We conduct extensive experiments on six large-scale Active Directory graphs and four other real-world networks, demonstrating that our algorithm significantly outperforms baselines.
%     \item Through the ablation study, we show our self-improvement mechanism capable of refining any given policy and pushing performance boundaries beyond existing state-of-the-art approaches.
% \end{itemize}

\begin{itemize}
    \item We introduce a new theoretical combinatorial optimization model called Feature-based Adaptive Connectivity Test. This is the first adaptive AD-focused model that incorporates security-context features to model an IT admin decision-making process.
    \item We prove that the F-ACT problem is $\#\mathcal{P}$-hard and develop RL4FT, a novel RL algorithm that addresses the key challenges of exploration and sparse rewards through a heuristic-guided self-improvement mechanism, a provably policy-invariant reward shaping scheme, and a custom Prioritized Experience Replay.
    \item We conduct extensive experiments on ten large-scale real-world networks, demonstrating that our algorithm significantly outperforms baselines across both linear and non-linear decision functions.
    \item Through an ablation study, we show our self-improvement mechanism is capable of refining any given policy and pushing performance boundaries beyond existing state-of-the-art approaches.
\end{itemize}

\section{Related Work}

\textbf{Defense models for Active Directory.} The study of AD attack graphs was pioneered by the seminal work of Dunagan et al. \citep{dunagan2009heat}. They introduced the identity snowball attack model, a concept that was later expanded and commercialized by Bloodhound \citep{BloodHound}. Defenses models for AD can be broadly categorized into two lines of work: non-adaptive and adaptive. Non-adaptive models (one-shot) often involves running an optimization algorithm a single time to find a global hardening strategy. A popular defense strategy is edge removal, where the defense model are formulated as a shortest path interdiction problem \citep{guo2022practical,guo2023scalable,zhang2023oracle} or as a problem of minimizing the Domain Admin-reachable nodes \citep{zhang2024practical}. A different model was proposed by Goel et al. \citep{goel2022defending,goel2023evolving}, who used an Evolutionary Diversity Optimization algorithm to find robust defense for AD with uncertain configuration. Another method of defense involves node removal, which models the strategic allocation of network decoys as introduced in Chapter \ref{chapter:paper1} and \ref{chapter:paper2}.

The primary drawback of non-adaptive models in real-world deployments is their inability to incorporate human feedback which lead to the development of adaptive models.
Adaptive model was first explored in Heat-ray, a system proposed by Dunagan et al. \citep{dunagan2009heat} that iteratively suggests edge removals to an administrator based on the sparest cut of the attack graph. Zheng et al. \citep{zheng2011active} used active learning predict the likelihood of an edge being removable, rather than performing an optimization to prescribe an actionable defense. More recently, Chapter \ref{chapter:paper3} proposed the Adaptive Path Removal model where a proxy adaptively proposed attack paths for IT admin, who requiring to select a edge on a path for removal. The Limited Query Graph Connectivity Test (LQGCT) by Guo et al. \citep{guo2024limited} is a closely related problem. Their work involves adaptively querying edges to determine network connectivity where they assume the known prior probabilities for each edge and a fixed query budget. However, these prior edge-by-edge query methods are impractical for large-scale networks like AD, where IT admin usually define high-level security policies based on object properties. Our work is the first to model this directly by using features to learn an underlying decision function, which represents the security policy itself. This approach allows feedback on a single edge to generalize across the entire network, bridging the gap between individual queries and scalable policies.

% In contrast, our approach associates each edge with a feature vector governed by an unknown linear function and seeks to optimize the total number of queries in an unrestricted horizon, which makes the problem prohibitively more challenging. To the best of our knowledge, ours is the first work on AD security to incorporate edge features into the sequential decision-making process.

\textbf{Partially Observable Stochastic Optimization Problems (POSO).} POSO is an emerging problem in operational research and artificial intelligence. This problem involves making an adaptive sequence of decisions under uncertainty, where choices are informed by observing the outcomes of past actions. This class of optimization has been applied to diverse domains. 
For instance, in medical testing, \citep{short2013iron,yu2023deep} developed adaptive strategies to diagnose diseases while minimizing testing costs. In social influence, related work studied the adaptive version of the influence maximization problem, which aims to sequentially select seed nodes to maximize influence spread based on observed propagation \citep{tong2016adaptive, fujii2019beyond, hatano2016adaptive}. 
Furthermore, the principles of POSO are central to some active learning and active fine-tuning techniques \citep{hubotter2024transductive,golovin2011adaptive,wei2015submodularity}, where the goal is to adaptively select the most informative data points to query. 
This paradigm also extends to recommendation systems for adaptively learning user preferences \citep{gabillon2014large,gabillon2013adaptive}.
A more theoretical work including the Stochastic Boolean Function Evaluation problem \citep{allen2017evaluation,deshpande2014approximation} seeks to find a query policy that determines the output of a Boolean function $f$ by probing its hidden binary input while minimize the total expected query cost. Golovin and Krause \citep{golovin2011adaptive} studied a special case of POSO governed by an adaptive submodular utility function, 
% they provide a theoretical analysis of greedy algorithms on this problem class. 
While these models are relevant, their solution is not designed for our graph-based problems and lack scalability for large graphs. Consequently, solutions for these models cannot be directly applied to our work

% \textbf{RL for POSO}

% Large body of apply RL in one-shot combinatorial optimization \citep{bengio2021machine,yang2023survey} where most of the use RL to search for a one unique optimal solution or find generalize policy for difference problem instance \citep{manchanda2020gcomb,chen2021contingency,kamarthi2019influence}. In POSO, the solution itself is a policy so RL is a natural formulation for this family of optimization problem. \citep{gabillon2013adaptive,gabillon2014large} develop optimism bandit algorithm for the problem of maximizing adaptive submodular functions for the application of recommendation system. \citep{fadda2024math,huang2022provably} RL for IM
% In AD research, \citep{ngo2025adaptive} and \citep{guo2024limited} proposed RL for their adaptive netork connectivity problem. 

\section{Problem Formulation}

The Feature-based Adaptive Connectivity Test problem involved a graph $G = (V, E)$, either undirected or directed. There is a single source and a single destination $s, t \in V$. Each edge $e\in E$ has two possible labels (ON/OFF) which are hidden initially. There will be two termination conditions: (1) Among the revealed edges, the ON edges form an s-t path; (2) Among the revealed edges, the OFF edges form an s-t cut. We aim to design an adaptive query strategy $\pi$ that selects the next edge to query based on the history of outcomes. The objective is to find an optimal policy $\pi^*$ that minimizes the expected number of queries.
% Extending from the original ACT problem \citep{fu2017determining,guo2024limited}, we introduce the feature for each edge. 

In our model, we associated each edge $e\in E$ with a $d$-dimension binary feature vector $\mathbf{x}_e = (x_{e,1}, \dots, x_{e,d}) \in \{0, 1\}^d$ where $x_{e,i}$ denoted $i$-th feature for edge $e$. In AD context, feature vector is a list of security risk associated to an edge. The true label $\sigma_e$ of edge $e$ is determined by a parameterized decision function $f(\mathbf{x}_e)$. For the sake of simplicity, we define this function as a linear model: $f(x_e) = \mathbf{w}^T \mathbf{x}_e + b$ where $w \in \mathbb{R}^d$ is the parameterized weight vector and $b \in \mathbb{R}$ is the bias. The true label $\sigma_e$ of an edge is determined as following:
\begin{equation}
    \sigma_e = 
\begin{cases} 
\text{OFF} & \text{if } f(x_e) \ge 0 \\
\text{ON} & \text{if } f(x_e) < 0 
\end{cases}
\label{eq:admin}
\end{equation}

Next, we formally define the deduction assumption that enables efficient querying. First, let us define the set of features with value 1 in the feature vector (support of vector $\mathbf{x}_e$) as $\text{supp}(\mathbf{x}_e) = \{i \in \{1, \dots, d\} \mid x_{e,i} = 1\}$. Conversely, we have $\text{supp}(\bar{\mathbf{x}}_e)$ denote the set of features with value 0 in the feature vector where $\bar{\mathbf{x}}_e$ is the element-wise inverse of $\mathbf{x}_e$ (i.e., $\bar{x}_{e,i} = 1 - x_{e,i}$). Let $D_{\text{OFF}}(\mathbf{x}) = \{\mathbf{x}' \in \mathcal{X} \mid \text{supp}(\mathbf{x}) \subseteq \text{supp}(\mathbf{x}')\}$ denotes the set of all feature vector $\mathbf{x}' \in \mathcal{X}$ where every feature with a value of 1 in $x$ also has a value of 1 in $\mathbf{x}'$. Similarly, let $D_{\text{ON}}(\mathbf{x}) = \{\mathbf{x}' \in \mathcal{X} \mid \text{supp}(\bar{\mathbf{x}}) \subseteq \text{supp}(\bar{\mathbf{x}}')\}$  denotes the set of all feature vector $\mathbf{x}' \in \mathcal{X}$ where every feature with a value of 0 in $\mathbf{x}$ also has a value of 0 in $\mathbf{x}'$. The assumption is as follows:
\begin{assumption}
    % \textbf{(Monotone Deduction Rule)} We assume these features are monotone with respect to risk. This means the presence of a feature  $x_{e_i}=1$ ($x_{e_i}=0$) can only increase, never decrease, the likelihood that an edge's true state is 'OFF' ('ON'). Formally, this gives rise to two deductive rules:
    \textbf{(Monotone Deduction Rule)} We assume ground-truth labeling are monotone with respect to risk; this means a feature combination resulting in an 'OFF' (or 'ON') state will maintain that state if more (or fewer) risk features are present. Formally, this resulted in deductive rules:
    \begin{itemize}
        \item If an edge $e$ is revealed to be "OFF", any other edges $e'$ where $\mathbf{x}_{e'} \in D_{OFF}(\mathbf{x}_e)$ must also be "OFF"
        \item Conversely, if an edge $e$ is revealed to be "ON", any other edges $e'$ where $\mathbf{x}_{e'} \in D_{ON}(\mathbf{x}_e)$ must be "ON".
    \end{itemize}
    \label{assu:mono}
\end{assumption}
Without loss of generality, we can assume that the components of the unknown weight vector $\mathbf{w}$ are non-negative, i.e., $\mathbf{w_i} \geq 0, \forall i \in \{1, \dots, d\}$. Another implication is that all edges sharing the same feature vector must have the same true status. This allows us to simplify the problem significantly, as formalized in the following lemma.
\begin{lemma}
    \textbf{(Query Equivalence)} Let $e_1, e_2 \in E$ be two edges that have identical feature vectors, i.e., $\mathbf{x}_{e_1} = \mathbf{x}_{e_2}$. Then their true label must also be identical, i.e., $\sigma_{e_1} = \sigma_{e_2}$.
    
\label{lemma:equi}
\end{lemma}
We also have the following hardness result:

\begin{theorem}
    It is $\#\mathcal{P}$-hard to find the optimal policy for F-ACT problem
\end{theorem}

\begin{proof}
    The proof is based on a reduction to the LQGCT problem \citep{guo2024limited,fu2017determining} which is known to be $\#\mathcal{P}$-hard. The whole idea is to reduce the LQGCT to the F-ACT instance with edge with one-hot coding feature. 
    
     The LQGCT problem is defined on a undirected graph $G' = (V', E')$ with a source $s'$ and target $t'$. The LQGCT also test for $s'-t'$ connectivity by adaptively query each edge at a time. However, the difference is that each edge $e \in E'$ in LQGCT model is associated with a fixed prior probability $p$ without the risk deduction rule and have the query limit $B$. But the hardness proof of LQGCT is flexible with value of $B$ so we assume $B \geq |E'|$ (no budget at all).

     We will start the reduction as following. Suppose there exist an instance of LQGCT $G' = (V', E')$. We construct a undirected graph $G = (V, E)$ with edge, nodes, source and destination is a clone of $G'$. For each edge $e_i \in E'$ we associated it with a index $i\in \{1, 2,\dots, |E|\}$. Next, we create a associated feature vector $x_i$ where $x_i$ encode the one-hot coding of the edge index $i$ (only the i-th feature is 1). Without the loss of generality, we assume the decision function $f(x_e) = \mathbf{w}^T \mathbf{x}_e + b$ with fixed bias value $b = -p$ and the value of the weight vector $\mathbf{w}$ is uniformly drawn from range $[0, 1]$. We complete the construction of $G$.
     
    The construction $G$ is a clone of $G'$ except the associated feature. But since we fixed bias as $p$ and the feature is the one-hot coding of the edge index, the true label for each edge $e_i \in$ become: $ \sigma_e = 
            \begin{cases} 
            \text{OFF} & \text{if } w_i \ge p \\
            \text{ON} & \text{if } w_i < p 
            \end{cases}$. 
    Since the $w_i$ is unknown and randomly draw from range $[0, 1]$, the optimal query policy $\pi$ for $G$ instance of F-ACT is also optimal policy for $G'$ instance of LQGCT. Therefore, our F-ACT problem must be as difficult as the LQGCT. We conclude the proof.
\end{proof}

Although the presented formulation assumes a linear and monotonic function for analytical clarity, our work is more general. The hardness proof and our designed algorithm can be extended for the broader case involving arbitrary non-linear decision functions and non-monotonic features.

\section{MDP formulation and RL Preliminary}

In this section, we aim to design a RL agent to aid the IT administrator in the query process. The high-level interaction between agent and IT admin is as follows: in each round, the agent proposes an edge $e$ and the associated feature vector $\mathbf{x}_e$ to the IT admin. The IT admin then provides the true label of that feature vector/edge ('ON' or 'OFF') and their decision is modeled as the function $f(\mathbf{x}_e)$. The agent uses this new information to adaptively decide its next query.

\subsection{MDP formulation for F-ACT}
\label{sec:mdp}
Let us define the MDP as a tuple $\langle \mathcal{S}, \mathcal{A}, \Phi, R \rangle$, where $\mathcal{S}$ is the set of states, $\mathcal{A}$ is the set of actions, $\Phi : \mathcal{S} \times \mathcal{A} \times \mathcal{S} \mapsto [0, 1 ]$ is the state transition and the reward function $R : \mathcal{S} \times \mathcal{A} \mapsto \mathbb{R}$. A deterministic policy is a mapping from states to actions: $\pi : \mathcal{S} \mapsto \mathcal{A}$. In our setting, an environment instance is defined by the tuple $\langle G, f \rangle$, where $G$ is the graph and $f$ is the parameterized decision function that determines the hidden true labels of the edges.

\textbf{State:} In principle, the state is defined by the status of every edge $e \in E$. However, following from Lemma \ref{lemma:equi}, it is sufficient to define the state over the set of unique feature vectors $\mathcal{X} = \{\mathbf{x}_e \mid e \in E\}$. The state space of our problem is thus $\mathcal{S} = \{0, 1, *\}^{|\mathcal{X}|}$. Here, each state $s \in \mathcal{S}$ is a vector where each component corresponds to agent knowledge of a unique feature vector $x\in \mathcal{X}$. We define a function $s(\mathbf{x})$  to denote the status of feature $\mathbf{x}$ in state $s$. $s(\mathbf{x})$ return 1 if the status is known to be 'ON', 0 if it is 'OFF', and * if it's status is unknown. Abusing notation slightly, we will also write $s(e)$ to refer to the status of an edge $e$, shorthand for $s(\mathbf{x}_e)$. The initial state $s_0$ will be the vector of all-*. We denote $\mathcal{S}_i$ the set of possible state in round $i$ of the query process. 
% Hence, the state space can be represented as $\mathcal{S} = \mathcal{S}_0 \cup \mathcal{S}_1 \cup \dots \cup \mathcal{S}_{|\mathcal{X}|}$. 
We also define two set of terminal states: $\perp_{p}$ reached when an $s-t$ path of 'ON' edges emerges, and $\perp_{c}$ reached when an $s-t$ cut of 'OFF' edges emerges.

\textbf{Action and Transition Function:} Given a non-terminal state at query round $t$, $s_t \in \mathcal{S}_t \setminus \{\perp_p \cup \perp_c\}$, an action $a_t$  is the selection of a unique feature vector $\mathbf{x}_{t}$ whose status is currently hidden. This action corresponds to querying the IT admin to reveal the true status of the chosen feature vector. Since the true parameter of the decision equation \ref{eq:admin} is unknown to the agent, we can assume the state transition is stochastic and the new state $s_{t+1}$ is sampled from the distribution $p(. |s_{t},a_{t})$. When the agent queries $\mathbf{x}_{t}$, there are two possible outcomes from the oracle: 'ON' or 'OFF'. This leads to a transition to the new state $s_{t+1}$ which updates the label for $\mathbf{x}_{t}$ (and its associated edges) and all other feature vectors deduced via Assumption \ref{assu:mono}.

\textbf{Reward and Objective: } The primary objective is to find a policy that minimizes the expected number of queries (steps) required to reach a terminal state. Hence, the reward function is defined as $R = r(s_t, a_t) = -1$ for any action $a_t$ taken in a non-terminal state $s_t$. The sum of future rewards or \textit{return} is defined: $R_t = \sum_{i=t}^{\infty} r(s_i, a_i)$. Here we consider the discount factor $\gamma = 1$. The goal is to find the optimal policy $\pi^*$ such that maximize: $ \mathbb{E}_{\tau \sim \pi} \left[ R_0|s_0 \right]$. 

\subsection{RL Preliminaries}

RL is a learning algorithm where an agent learn an optimal policy $\pi$ by interacting with an environment (MDP). This interaction process involves the agent selecting an action $a_t$ in state $s_t$ , receiving a reward $r_t$ which refects how good the action/policy is. The search for an optimal policy is often guided by an action-value function $Q(s, a)$ which estimates the expected long-term return for each state-action pair. In this work, we mainly use Soft Actor-Critic (SAC) \citep{haarnoja2018soft} as demonstrative RL algorithm. Instead of solely maximizing the cumulative discounted reward, SAC seeks a policy that also maximizes its own entropy. 
% The objective of SAC will be $J(\pi) = \mathbb{E}_{\tau \sim \pi} [\sum_{t=0}^{\infty} \gamma^t ( r(s_t, a_t) + \alpha \mathcal{H}(\pi(\cdot|s_t)) )]$, where $\mathcal{H}$ is the policy entropy and $\alpha$ is the temperature parameter which balance the trade-off between reward (exploitation) and entropy (exploration). 
To optimize this objective, SAC learns a parameterized action-value function $Q_\theta(s, a)$ which act as the critic Q-network. The critic's parameters are learned by minimizing the soft Bellman residual via the following loss function: $J_Q(\theta) = \mathbb{E}_{(s_t,a_t) \sim \mathcal{B}}[\frac{1}{2} (Q_\theta(s_t, a_t) - \hat{Q}(s_t, a_t))^2]$ where $\mathcal{B}$ is the replay buffer and the target value $ \hat{Q}(s_t, a_t) = r_t + \gamma \mathbb{E}_{a_{t+1} \sim \pi} \left[ Q_{\bar{\theta}}(s_{t+1}, a_{t+1}) - \alpha \log \pi(a_{t+1}|s_{t+1}) \right]$. Here, $Q_{\bar{\theta}}$ is the target network, which is updated using a moving average of the value network parameters $\bar{\theta}$. The learned critic guides the policy network (the actor) to produce actions that maximize the soft Q-values, and this interplay between them drives the learning process. As our task involves a discrete action space, we will use an adaptation of SAC for this setting \citep{christodoulou2019soft}. We note that, although our proposal is mainly demonstrated within the framework of off-policy RL (SAC), our proposed technique can be generalised to any model-free RL algorithm.

\section{RL4FT Algorithm}
\subsection{Motivation for RL4FT}
% Research in one-shot combinatorial optimization has largely applied Reinforcement Learning (RL) in two main ways: either as a powerful search heuristic for a single problem instance \citep{bengio2021machine,yang2023survey}, or to learn a generalizable policy capable of constructing solutions for various unseen instances \citep{manchanda2020gcomb,chen2021contingency,kamarthi2019influence}.
% Our problem, however, have the distinct challenge of adaptive sequential optimization. 
% The nature of F-ACT problem class requires learning a policy that adapt to new information which make it fit naturally to the RL framework. While RL has been applied to other instances of POSO, such as adaptive submodular maximization \citep{gabillon2013adaptive,gabillon2014large}, adaptive influence maximization \citep{fadda2024math,huang2022provably}, and most relevantly to our work, by Guo et al. \citep{guo2024limited} for the LQGCT problem.

The nature of F-ACT problem class requires learning a policy that adapt to new information which make it fit naturally to the RL framework.
While RL has been applied to other instances of POSO \citep{gabillon2013adaptive,gabillon2014large,fadda2024math,huang2022provably,guo2024limited}, our problem setting presents unique challenges that cause standard RL to struggle. First, the F-ACT's environment poses a significant exploration challenge, as the combination of a long-tail pattern in query lengths and a large action space (which can be as large as $|\mathcal{X}|$ at each step) causes the agent to get trapped in suboptimal regions, making it struggle to converge to the optimal policy. The second problem is the sparse reward signal, caused by the default reward structure of the problem. 
The most relevant prior work \citep{guo2024limited} attempts to mitigate these issues by restricting the action space to a set of "good" actions and limiting the learning horizon to a fixed budget of $B$ steps. Although this can improve performance over vanilla RL, this approach sacrifices theoretical optimality (as restricted action and horizon) and fails to experimentally outperform even simple heuristics. This motivates our design of the RL4FT architecture, which introduces several novel components to improve RL performance without making these same compromises to optimality.

\subsection{Overall Architecture}
RL4FT architecture introduces four key novelty: a heuristic-guided mechanism to warm-start the agent and significantly improve convergence speed; a self-improvement mechanism bootstrap the policy and iteratively improve policy quality; a customized experience replay component adapted for this mechanism; and a reward shaping component to address the sparse reward issue. The complete training procedure is detailed in Algorithm \ref{alg:RL4FT}. The self-improvement and heuristic-guided mechanism forms the core learning loop that orchestrates the overall process, which is augmented by our other contributions: reward shaping (Line 10) and a customized Prioritized Experience Replay (Line 12).

\subsection{Self-improvement with Heuristic-guided}

% Our heuristic-guided (HG) mechanism is motivated by the observation that heuristic algorithms for the ACT problem \citep{guo2024limited,fu2017determining}, while not guaranteed to be globally optimal, are effective at finding competitive solutions in a fraction of the time.
In this component, we introduce two key mechanism: a heuristic-guided mechanism and an iterative self-improvement mechanism. The core idea of the \textbf{heuristic-guided mechanism} is to use a reference (curriculum) policy $\pi^{ref}$ which initially a sub-optimal RL policy or a heuristic policy to warm-start and steer the learning of the RL policy $\pi^\phi$. During the initial phase of learning, $\pi^{ref}$ is significantly better than the untrained $\pi^\phi$. Therefore, we use $\pi^{ref}$ for data collection, while $\pi^\phi$ controls the exploration of only the first $h$ steps. As $\pi^\phi$ improves under the guidance of $\pi^{ref}$, we gradually increase the number of steps $h$ that $\pi^\phi$ can explore. Progressively giving $\pi^\phi$ more control allows the agent to explore different states, which helps avoid the out-of-distribution shift that can occur when learning is guided too heavily by a fixed reference policy \citep{uchendu2023jump}. 
Furthermore, we introduce the \textbf{self-improvement mechanism}, which acts as an outer loop over the heuristic-guided training procedure and allows the agent to iteratively bootstrap its policy. The policy learned in each iteration becomes the reference for the subsequent iteration, a process that begins with a suboptimal policy as the initial reference.

\begin{algorithm}[H]
\caption{RL4FT Training Algorithm}
\label{alg:RL4FT}

\SetAlgoLined
\LinesNumbered

% suppress printed "end" markers
\SetKwFor{For}{for}{:}{}
\SetKwFor{ForEach}{foreach}{:}{}
\SetKwFor{While}{while}{:}{}
\SetKwIF{If}{ElseIf}{Else}{if}{then}{else if}{else}{}
\SetKwProg{Fn}{Function}{}{}

\KwIn{Graph $G(V, E)$ with source $s$ and target $t$, initial sub-optimal policy $\pi^{sub}$, number of guide steps $GS$, number of improvement iterations $I_{max}$}
\KwOut{RL policy $\pi$}
\BlankLine

% Initialize the reference policy with the initial heuristic
Initialize reference policy $\pi^{ref} \leftarrow \pi^{sub}$\;
% Main improvement loop

\For{$i = 1 \dots I_{max}$}{
    Initialize actor network $\pi^{\phi}$, critic network $Q^{\theta}$, and replay buffer $\mathcal{B} \leftarrow \emptyset$\;
    
    % Inner loop to incrementally increase RL agent's control
    \For{$h = 1 \dots GS$}{
        \Repeat{$EvalPolicy(\pi^{comb}) \geq EvalPolicy(\pi^{ref})$}{
            % Define the combined policy for the current guide step
            $\pi^{comb} \leftarrow \pi^{\phi}_{1:h} \cup \pi^{ref}_{h+1:H}$\;
            
            Reset env. and randomize parameters of $f$
            % Collect experience with the combined policy
            
            Collect trajectories $\mathcal{T}$ using $\pi^{comb}$\;
            
            % Process and store transitions
            \For{each transition $(s_k, a_k, r_k, s_{k+1}) \in \mathcal{T}$}{
                Calculate shaped reward $r'_k$ using Eq. (\ref{eq:rw})\;
                
                Append $\mathcal{B} \leftarrow \mathcal{B} \cup (s_k, a_k, r'_k, s_{k+1})$ \;
            }
            % Train the RL policy
            Sample mini-batch $\mathcal{MB} \subseteq \mathcal{B}$ using Eq. (\ref{eq:per})\;
            
            $\pi^\phi, Q^\theta \leftarrow TrainPolicy(\pi^\phi, Q^\theta, \mathcal{MB})$\;
        }
    }
    % Update the reference policy for the next iteration
    $\pi^{ref} \leftarrow \pi^{comb}$\;
}

\Return{$\pi^{comb}$}
\end{algorithm}

The training process, detailed in Algorithm \ref{alg:RL4FT}, consists of two nested loops. The inner loop (lines 4-14) implements the heuristic-guided mechanism, which is responsible for gradually transferring control to the RL agent. The outer loop (lines 2-15) drives the self-improvement mechanism by using the policy learned from the completed inner loop as the new reference for the next iteration.
We first describe the heuristic-guided mechanism within the inner loop. First, we initialized reference policy $\pi^{ref}$ with a given suboptimal or heuristic policy, $\pi^{sub}$. In each training episode, a trajectory $\mathcal{T}$ is generated by letting the RL agent interact with the environment using the learned policy $\pi^{\phi}$ for the first $h$ step. After that, we switch to rolling out the guide policy $\pi^{ref}$ for the remaining $H-h$ steps until we meet the horizon $H$ (line 6-8). This means the RL agent has the liberty to explore for the first $h$ steps, and the remaining $H-h$ steps will be completed by the heuristic policy. For brevity, we use $\pi^{\phi}_{i:j}$ to denote the policy restriction $\pi\bigl|_{\bigcup_{k=i}^j S_k}$, i.e., the policy applied only to states encountered from round $i$-th to $j$-th. The collected transitions are stored in a replay buffer $\mathcal{B}$ and used to update $\pi^\phi$ via $TrainPolicy$ function. $TrainPolicy$ can be any standard model-free algorithm (e.g., SAC or PPO). After each episode, we evaluate the combined policy $\pi^{comb} \leftarrow \pi^{\phi}_{1:h} \cup \pi^{ref}_{h+1:H}$ with the standard evaluation function (line 14). If the performance of $\pi^{comb}$ is better than the current $\pi^{ref}$, we increment $h$ and give RL agent more control of exploration. This process continues until $h$ reaches a predefined maximum of $GS$ step. Once the inner loop finishes, the self-improvement mechanism performs its update (line 15) where the newly trained policy $\pi^{comb}$ becomes the new reference policy for the next learning iteration.

Our heuristic-guided self-improvement procedure is conceptually related to Curriculum Learning (CL) in Reinforcement Learning. CL is a well-established method for solving the challenges of hard exploration in complex tasks \citep{zhou2021curriculum,zhou2020robust,ao2021co,graves2017automated,bengio2009curriculum}. In the context of combinatorial optimization, CL is often employed to navigate vast search spaces by decomposing a difficult problem into a curriculum of simpler instances \citep{feng2020solving,nabli2020curriculum}. However, these approaches are fundamentally designed to find a static solution for a one-shot combinatorial problem, while the goal of POSO is to learn an adaptive policy. This required us to design a difference learning procedure. 

Our heuristic-guided procedure is also closely related to Jump-start Reinforcement Learning (JSRL) \citep{uchendu2023jump}, as both methods combine a sub-optimal guide with a learned policy during the roll-out phase. However, our approach has several fundamental differences and serves different purposes. JSRL is designed for robotic optimal-control problems (e.g., robotic grasping, escaping a maze), which often have a specific goal and a unique terminal state. In contrast, our is designed for partially observable stochastic optimization problems \citep{golovin2011adaptive} which typically have a large number of terminal states and sparse rewards. This leads to a key difference in the algorithm design. Our method of agent exploration expansion follows a "top-down" approach, where the agent explores for the first $h$ steps before $\pi^{ref}$ takes over. In contrast, JSRL employs a guide policy to restrict the exploration space and lets the agent explore only in the final few steps of an episode. Second, our technique does not rely on a single, fixed heuristic. Instead, it features the self-improvement mechanism that improves over time. After each full learning iteration, the newly improved policy becomes the reference guide for the next iteration which allow the progressive solution improvement. While , JSRL uses a single guide policy and terminates when the agent's performance surpasses a predefined threshold $\beta$. This makes its potential for improvement finite and introduces the additional time-consuming task of tuning $\beta$.

\subsection{Reward Shaping}

The reward signal in our problem is sparse, as the agent receives a reward of -1 for every action it takes. To address this, we employ reward shaping. This technique modifies the original reward function with an additional shaping function $g$ to provide more informative feedback on an action's quality. In the general additive form, the modified reward is $r' = r + g$. However, an improper design of reward shaping $g$ can alter the optimal policy \citep{hu2020learning,devlin2012dynamic,randlov1998learning}. Our objective is to design a shaping function that effectively incorporates domain knowledge while guaranteeing policy invariance.

We now introduce our proposed reward shaping function.  Given a graph $G(V, E)$, let $\mathcal{P}$ be the set of all paths and $\mathcal{C}$ be the set of all minimal cuts. For any edge $e \in E$, we define $\mathcal{P}_e = \{p \in \mathcal{P} \mid e \in p\}$ as the set of paths containing $e$, and similarly, $\mathcal{C}_e = \{c \in \mathcal{C} \mid e \in c\}$ as the set of minimal cuts containing $e$. We define two auxiliary functions based on the state $s$: $g_p(s) = \left| \cup_{e \in E : s(e)=0} \mathcal{P}_e \right|$ and $g_c(s) = \left| \cup_{e \in E : s(e)=1} \mathcal{C}_e \right|$. We define the potential function $g : \mathcal{S} \to \mathbb{N}$ as follows:
\begin{equation}
g(s) = |\mathcal{P}||\mathcal{C}| - (|\mathcal{P}| - g_p(s))(|\mathcal{C}| - g_c(s))
\label{eq:utility}
\end{equation}
The function $g$ was introduced in \citep{fu2017determining} as a utility function for an approximation algorithm for the LQGCT problem. We utilize $g$ as the potential function for our reward shaping. The shaped reward is represented as:
\begin{equation}
    r_{shaped}(s_k, a_k, s_{k+1}) = g(s_k) - g(s_{k+1}) 
\label{eq:rwshape}
\end{equation}
Finally, the modified reward is represented as: 
\begin{equation}
    r'(s_k, a_k, s_{k+1}) = r(s_k, a_k) + \omega \cdot r_{shaped}(s_k, a_k, s_{k+1})
\label{eq:rw}
\end{equation}
where $\omega$ is the reward shaping weight. Next, we show that our reward shaping term is policy-invariant, meaning the optimal policy remains the same after reward modification. The concept of policy-invariant reward shaping was first introduced with Potential-Based Reward Shaping (PBRS) \citep{ng1999policy}. The policy invariance guarantee of PBRS works by taking advantage of the diminishing effect of a discount factor ($\gamma < 1$) over a long horizon (i.e., $\sum_{k=0}^{\infty} [\gamma \Phi(s_k)-\Phi(s_{k+1})] = \Phi(s_0)$).
Although our reward shaping is a potential-based, however, policy invariance for such functions in the undiscounted setting ($\gamma = 1$) is only guaranteed for MDPs with a single absorbing state. Since our problem has multiple absorbing states by definition, we must show that the value of the potential function is equivalent at any absorbing state. This is proven in the following lemma:
\begin{lemma}
    \textbf{(Policy Invariance)} Let $\tau = (s_0, s_1, \dots, s_H)$ be any trajectory that terminates in $s_H \in \{\perp_p \cup \perp_c\}$. The cumulative shaped reward is constant, i.e, $R_{\text{shaped}}(\tau) = -|\mathcal{P}||\mathcal{C}|$. Hence, the shaped reward in (\ref{eq:rw}) is policy invariant.
\end{lemma} 
\begin{proof}
To prove this, we analyze the behavior of the potential function $g$ at the start and terminal states of any $\tau$. First, we have two types of terminal states by definition. If the trajectory ends by finding a cut ($s_H \in \perp_c$), the auxiliary function will be $g_p(s_H) = |\mathcal{P}|$. This is because the union of all possible paths that contain an edge from the discovered cut must constitute the set of all paths $\mathcal{P}$ in $G$. Similarly, if the trajectory terminates by finding a path ($s_H \in \perp_p$), the auxiliary function will be $g_c(s_H) = |\mathcal{C}|$, as the union of all cuts containing an edge from the discovered path must comprise the set of all possible cuts $\mathcal{C}$. Therefore, for any terminal state, the potential function is $g(s_H) = |P||C|$. Next, at the starting state $s_0$, the potential value is $g(s_0) = 0$ as both component $g_p(s_0) = g_p(s_0) = 0$. Consequently, the cumulative shaped reward for any trajectory $\tau$ simplifies:  $R_{\text{shaped}}(\tau) = \sum_{k=0}^{H-1} r_{\text{shaped}}(s_k, s_{k+1}) = \sum_{k=0}^{H-1} [g(s_k)-g(s_{k+1})] = g(s_0)-g(s_H) = -|P||C|$. This conclude the proof.
\end{proof}
% About why we decided to use an undiscounted $\gamma = 1$ setting. . However, our problem fundamentally seeks to minimize the length of the horizon, making the discounted PBRS unsuitable for our problem.

\textbf{Reproducibility of potential function.} Regarding the implementation of the potential function $g$, a direct approach requires enumerating all paths and cuts. This is computationally infeasible for some graphs, as both are $\#\mathcal{P}$-hard problems. However, since $g$ can serve as heuristic, a precise enumeration is not strictly necessary. So, for graphs where full enumeration is intractable, we approximate the path and cut sets using a sampling procedure: we repeatedly generate a random weight vector $\mathbf{w}$ and bias $b$ and then determine the resulting label of each edge based on Eq. 1, check if a valid path or cut is formed, and add any found to their respective pools. This process continues until a predefined time constraint is met.

\subsection{Custom Prioritized Experience Replay}

The heuristic-guided mechanism is used in roll-out to sample trajectory and transition. These transition will be stored into the experience replay buffer and later used to sample batches of training data. 
In classical experience replay, data is sampled uniformly from the buffer for training  \citep{lin1992self}. To further encourage learning from effective solutions, we instead adopt the idea of Prioritized Experience Replay (PER) \citep{schaul2015prioritized}. The main hypothesis of PER is that an agent learns more effectively from 'surprising' transitions as they are most informative for learning. So, transitions that are more surprising are given a higher probability of being sampled for training. The "surprise" of the transition is measure via the temporal difference (TD) error which is calculated in SAC as $\delta = Q_\theta(s_t, a_t) - \hat{Q}(s_t, a_t)$. Let $\delta_\tau$ be the TD error for transition $\tau$. The probability of sampling transition $\tau$ is then $P_\tau = \frac{|\delta_\tau|+\epsilon}{\sum_j (|\delta_j| + \epsilon)^\alpha}$ where $\epsilon > 0$ is a very small to prevent zero denominator and $\alpha$ is a temperature parameter. 

Initially, when we combined the original PER with our heuristic-guided mechanism, we observed that the convergence speed slowed down and performance deteriorated. This is because the original PER can \textit{oversample} transitions from the reference policy $\pi^{ref}$, causing an out-of-distribution shift where the learning is influenced too heavily by a fixed reference policy. To address this, we introduce a new variant of PER that accounts for the out-of-distribution shift. We define the sampling probability $P_\tau$ for a transition $\tau = (s_k, a_k, r_k)$ occurring at round $k$:

\begin{equation}
    P_\tau = \frac{(|\delta_\tau|\beta_\tau+ \epsilon)^\alpha}{\sum_j (|\delta_j|\beta_\tau + \epsilon)^\alpha}
\label{eq:per}
\end{equation}
where $\beta_\tau = \begin{cases} 
      1, & k \le h \\
      \gamma^{k-h+1}, & k > h 
   \end{cases}$.
For transitions that is generated by heuristic policy ($k>h$), we penalize an exponential decay factor. While the learning of the heuristic-guided transitions is beneficial, standard PER may assign them excessively high priority, leading to oversampling. We avoids this by assigning transitions from the agent's own exploration ($k \leq h$) a higher priority. \textit{This avoids the risk of a distributional shift toward the reference policy and prevents the learned policy from simply only mimic the sub-optimal heuristic.}

\section{Experiment}
\label{sec:paper4:experiment}
\subsection{Experiment Set up}
We first evaluate the performance of RL4FT with baselines algorithms on various of real networks. All experiments were conducted on a high-performance computing cluster. Each experiment use a single node with Intel(R) Xeon(R) Platinum 8360Y CPU, 24GB RAM, NVIDIA A100-40GB.

We evaluate our methods on ten networks, which are categorized into two groups: AD Networks and Miscellaneous Networks. There are a few difference between these graph. The AD network are directed graphs. AD graph in experiment including:
\begin{itemize}
    \item \textbf{AS1-AS5}: These graph are generated using ADSynth \citep{nguyen2024adsynth}.
    \item \textbf{UNI}: This data is the real data, collected from an University of Anonymous using SharpHound \citep{SharpHound}. The data collection took place on Thursday, 28th October 2022, at approximately 10:00 AM during regular working hours in the local time zone. The collected Active Directory dataset consists of 125,444 nodes and 1,195,432 edges
\end{itemize}

The Miscellaneous Network are undirected graph. These graph including
\begin{itemize}
    \item \textbf{EUROROAD} \citep{euroroad} and \textbf{MINNESOTA} \citep{minnesota}: Europe road network and Minnesota’s road network.
    \item \textbf{BCSPWRs} \citep{bcspwr}: Power network patterns collected by Boeing Computer Services
\end{itemize}

Choosing source and destination: The selection of source and target for our experiments is performed as follows. For the Active Directory (AD) graphs, the destination t is set to a high-privilege node (identified by the admincount property). A corresponding source s is then selected from the ten nodes with the greatest shortest-path distance to t. For the miscellaneous graphs, both s and t are chosen uniformly at random.
To avoid trivial instances (such as when the source and target are connected by a single edge, requiring only one query), and to ensure a consistent level of difficulty across our evaluation, we use a "median-seed" selection strategy inspired by \citep{guo2024limited}. We generate eleven problem instances for each graph using random seeds 0-10, measure the difficulty of each by running the H1 heuristic, and then select the single instance that yields the median query cost for all of our main experiments.

The features for Active Directory (AD) and miscellaneous graphs are detailed in Table \ref{tab:ad_features} and Table \ref{tab:misc_features}, respectively. Given the rich security context of AD graphs, we extracted 20 distinct features that fall into three categories: permission-based, topological, and source vulnerabilities. The permission-based features are derived from the type of edge, while topological features are determined by comparing the degree of an edge's endpoints to the graph's median node degree. The source vulnerability is the feature capture the potential vulnerability of the source node based on their property. 

\begin{table*}[htbp]
\centering
\caption{Statistics of the graph datasets used in the experiments, showing the number of nodes and edges.}
\label{tab:graph_statistics}
\resizebox{\textwidth}{!}{%
\begin{tabular}{@{\hspace{1em}}c|cccccc|cccc@{\hspace{1em}}}
\hline
% This new row adds the multi-column category headers
& \multicolumn{6}{c|}{\textbf{AD Network}} & \multicolumn{4}{c}{\textbf{Miscellaneous Network}} \\
\cline{2-11} % This adds a line underneath the new category headers
\textbf{Statistic \textbackslash  Graph} & \textbf{AS1} & \textbf{AS2} & \textbf{AS3} & \textbf{AS4} & \textbf{AS5} & \textbf{ORG} & \textbf{EUROAD} & \textbf{MINNESOTA} & \textbf{BCSPWR05} & \textbf{BCSPWR10} \\
\hline
\textbf{Node} & 1049 & 5176 & 10377 & 16657 & 48275 & 125444 & 1175 & 2643 & 444 & 5301 \\
\hline
\textbf{Edge} & 4836 & 25264 & 48176 & 126178 & 236139 & 1195432 & 1417 & 3303 & 1033 & 13571 \\
\hline
\end{tabular}%
}
\end{table*}

For these graphs, the rich security context allows us to extract 20 security-specific features. The Miscellaneous graphs comprises four real-world, undirected graphs representing road and power grid infrastructures. As these networks provide only topological information, we derive six graph-based features for each. Detailed descriptions of corresponding features of all network are provided in the Table \ref{tab:ad_features} and \ref{tab:misc_features}.

\subsection{Baseline Algorithms}
\textbf{Heuristic Baselines.} We evaluate our algorithm against several strong adaptive heuristics for comparison. The first is \textbf{H1}, a fast and elegant heuristic for the LQGCT problem proposed by \citep{guo2022practical}, which at each step identifies a path and a cut consisting solely of unqueried edges and then queries an edge guaranteed to be in their intersection. Pseudocode for this baseline is shown in Algorithm \ref{alg:H1}. Our second baseline is a \textbf{Greedy} heuristic, inspired by the approximation algorithm in \citep{fu2017determining}. This heuristic takes a myopic, one-step lookahead approach: for each unqueried edge, it calculates the expected marginal benefit with respect to the utility function $g$ of revealing its state (assuming a uniform prior) and queries the edge with the highest expected benefit. Pseudocode for this baseline is shown in Algorithm \ref{alg:Greedy}. Finally, we propose \textbf{HMC} which is a Sequential Monte Carlo (particle filter) approach. HMC maintains n particles, each representing a different hypothesis of the true latent decision function. At each step, it queries the edge that appears most frequently across the candidate paths or cuts generated by these hypotheses. Pseudocode for this baseline is shown in Algorithm \ref{alg:HMC}.

\textbf{Reinforcement Learning Baselines.} We compare against two vanilla model-free RL algorithms: Discrete-\textbf{SAC} \citep{christodoulou2019soft} and \textbf{PPO} \citep{schulman2017proximal}. 
We also introduce \textbf{LIMIT}, a RL-based baseline inspired by \citep{guo2024limited} that uses a BFS lookahead with the H1 heuristic to restrict the action space (unlike the original work, our version didn't have restricted horizon).
\textbf{RL4FT} is our full algorithm, which includes the iterative self-improvement mechanism. 
For comparison, we also introduce a simpler version, \textbf{RL4FT-S}, which does not include the self-improvement mechanism and only performs a one-time improvement over a fixed heuristic policy. 
Our configurations follow the notation RL4FT(\textit{TrainPolicy})+\textit{RefPolicy}, which specifies the training policy (PPO or SAC) in parentheses and the initial reference policy (+H1 or +LIMIT). For example, \textbf{RL4FT(SAC)+LIMIT} use SAC as training policy and LIMIT is the initial reference policy.

\subsection{Hyperparameter}
\label{subsec:hyper}
We have two distinct hyperparameter configurations, one for agents trained with SAC and another for those trained with PPO (Table \ref{tab:RLpara}). We also have hyperparameter for each of our proposed component in Table \ref{tab:proposedpara}. For some hyperparameter we have range of value, for those, we run a combination of those range of parameter and record the setting with the best result. For the implementation, our RL framework is implemented on top of the Tianshou Reinforcement Learning library \citep{weng2022tianshou}.

\begin{table*}[htbp]
    \centering
    \caption{Hyperparameter Settings for All Experiments. Values in curly braces denote the search space.}
    \label{tab:RLpara}

    % Center both blocks as a unit; adjust the gap with \hspace{...}
    \makebox[\textwidth][c]{%
      % Left block: General
      \parbox[t]{0.38\linewidth}{
        \centering
        \textbf{General RL Settings}\\
        \begin{tabularx}{\linewidth}{l|X}
            \hline
            \textbf{Parameter} & \textbf{Value} \\
            \hline
            step-per-epoch   & 4096            \\
            step-per-collect & 512             \\
            batch-size       & 128             \\
            hidden-sizes     & $128 \times 128$  \\
            training-num     & 16              \\
            test-num         & 16              \\
            \hline
        \end{tabularx}
      }
      \hspace{1.5em}%
      % Right block: SAC on top, PPO below (stacked)
      \parbox[t]{0.47\linewidth}{
        \centering
        \textbf{SAC Settings}\\
        \begin{tabularx}{\linewidth}{l|X}
            \hline
            \textbf{Parameter} & \textbf{Value} \\
            \hline
            NN Optimizer & Adam \\
            actor-lr     & $5 \times 10^{-4}$ \\
            critic-lr    & $3 \times 10^{-3}$ \\
            alpha        & \{0.05, 0.15, 0.3\} \\
            tau          & \{0.0025, 0.005, 0.01\} \\
            repeat-collect & 3 \\
            \hline
        \end{tabularx}

        \vspace{0.9\baselineskip}

        \textbf{PPO Settings}\\
        \begin{tabularx}{\linewidth}{l|X}
            \hline
            \textbf{Parameter} & \textbf{Value} \\
            \hline
            NN Optimizer & Adam \\
            lr        & $1 \times 10^{-3}$ \\
            ent-coef & \{0.01, 0.025, 0.05\} \\
            vf-coef  & \{0.3, 0.5, 0.7\} \\
            \hline
        \end{tabularx}
      }%
    }

\end{table*}

\begin{table*}[htbp]
\centering
\caption{Hyperparameters for proposed components.}
\label{tab:proposedpara}
\resizebox{\textwidth}{!}{%
\begin{tabular}{l|l|l|l}
\hline
\textbf{Component} & \textbf{Parameter} & \textbf{Description} & \textbf{Value} \\
\hline
\multirow{2}{*}{Heuristic Guide} & GS & Total number of guide steps & 6 \\
\cline{2-4} 
 & $I_{\max}$ & Total number of policy improvement steps & 5 \\
\hline
Reward Shaping & $\omega$ & Reward Shaping Weight & \{0.1, 0.2, 0.3\} \\
\hline
\multirow{3}{*}{PER} & $\gamma$ & Diminishing weight & 0.95 \\
\cline{2-4} 
 & $\alpha$ & Prioritization exponent & 0.6 \\
\cline{2-4} 
 & $\beta$ & Initial importance-sampling exponent & 0.4 \\
\hline
\end{tabular}
}
\end{table*}

\subsection{Comparing with Baseline}

As discussed in above Section \ref{subsec:hyper}, we run the experiment on a range of hyperparameter and report the configuration with the best performance. 
As summarized in Table \ref{tab:baseline}, our proposed RL4FT(SAC)+LIMIT, achieves the best performance, consistently ranking as the top method on 9 out of 10 graphs. All RL4FTs variants constantly outperformed their reference policy on every single instance. Notably, RL4FT(SAC)+H1 surpasses its H1 reference policy by $30.87\%$ and the vanilla SAC baseline by $55.79\%$. Similarly, RL4FT(SAC)+LIMIT improves upon its LIMIT reference by $4.21 \%$ and the SAC baseline by a significant $57.87\%$. 

% To further validate the framework's versatility, we tested it with non-linear decision functions, namely a quadratic function and a Radial Basis Function (RBF) \citep{buhmann2000radial}. We again found that RL4FT improved the H1 reference policy by  $18.94 \%$ and $47.29 \%$, respectively. 

\begin{landscape}

\vspace*{\fill}
\begin{table}[H]
\centering
\caption{Performance comparison of algorithms based on the expected number of queries ($\downarrow$ is better). The \textbf{best} and \underline{second-best} results for each graph are shown in \textbf{bold} and \underline{underlined}, respectively.}
\label{tab:baseline}
\resizebox{\linewidth}{!}{% <-- use \linewidth here
\begin{tabular}{|@{\hspace{1em}}l|cccccc|cccc|c@{\hspace{1em}}|}
\hline
& \multicolumn{6}{c|}{\textbf{AD Network}} & \multicolumn{4}{c|}{\textbf{Miscellaneous Network}} & \\
\cline{2-11}
\textbf{Method \textbackslash{} Graph} & \textbf{AS1} & \textbf{AS2} & \textbf{AS3} & \textbf{AS4} & \textbf{AS5} & \textbf{UNI} & \textbf{EUROAD} & \textbf{MINN.} & \textbf{BCSPW05} & \textbf{BCSPW10} & \textbf{AVG} \\
\hline
\textbf{H1} & 5.062 & 6.434 & 4.980 & 5.952 & 4.944 & 5.966 & \underline{1.205} & 3.006 & 1.879 & 3.407 & 4.283 \\
\hline
\textbf{GREEDY} & 9.341 & 15.481 & 13.531 & 22.400 & 5.471 & 12.364 & 3.079 & 3.227 & 2.120 & 2.181 & 8.920 \\
\hline
\textbf{HMC} & 5.333 & 6.040 & 4.686 & \underline{4.981} & 4.244 & 5.049 & \underline{1.205} & \underline{1.320} & \underline{1.595} & 2.257 & 3.671 \\
\hline
\textbf{SAC} & 5.493 & 10.064 & 11.995 & 17.402 & 4.930 & 12.024 & \textbf{1.000} & \textbf{1.303} & \textbf{1.099} & \textbf{1.673} & 6.698 \\
\hline
\textbf{PPO} & 5.942 & 10.601 & 13.583 & 20.213 & 9.552 & 17.783 & \textbf{1.000} & 1.482 & \textbf{1.099} & \underline{1.684} & 8.294 \\
\hline
\textbf{RL4FT-S(SAC)+H1} & 4.711 & 5.160 & 4.535 & 5.563 & 4.605 & 5.263 & \textbf{1.000} & 1.482 & \textbf{1.099} & \textbf{1.673} & 3.509 \\
\hline
\textbf{RL4FT(SAC)+H1} & \underline{3.338} & \underline{4.248} & 3.925 & 5.056 & \textbf{3.089} & 4.879 & \textbf{1.000} & \textbf{1.303} & \textbf{1.099} & \textbf{1.673} & 2.961 \\
\hline
\textbf{RL4FT(PPO)+H1} & 3.550 & 4.802 & \underline{3.885} & 5.364 & 3.281 & 4.894 & \textbf{1.000} & 1.482 & \textbf{1.099} & 1.725 & 3.108 \\
\hline
\textbf{LIMIT} & 3.647 & 4.435 & 4.171 & \textbf{4.445} & \underline{3.162} & \underline{4.527} & \textbf{1.000} & \textbf{1.303} & \textbf{1.099} & \textbf{1.673} & \underline{2.946} \\
\hline
\textbf{RL4FT(SAC)+LIMIT} & \textbf{3.333} & \textbf{3.878} & \textbf{3.870} & \textbf{4.445} & \underline{3.162} & \textbf{4.452} & \textbf{1.000} & \textbf{1.303} & \textbf{1.099} & \textbf{1.673} & \textbf{2.822} \\
\hline
\end{tabular}%
}
\end{table}
\vspace*{\fill}

\end{landscape}
% \clearpage
% Our non-RL heuristic, HMC, also performs surprisingly well, outperforming other standard heuristics. However, since the algorithm is based on simulating across n instances, the resulting computational overhead makes it impractical as a reference policy. 

\subsection{Ablation Study}

To evaluate the effectiveness of our proposed components, we also conduct a ablation study on every graph. The result is report in Table \ref{tab:ablation}. In each comparison, we remove one of each components at a time and report the percentage of decrease in performance. We have 5 version of the RL4FT: (1) without Reward Shaping (\textbf{woRWS}), (2) without iterative policy-guide improvement, basically run RL4FT-S (\textbf{woITER}), (3) without Heuristic-guided mechanism (\textbf{woHG}), (4) without custom prioritized experience replay (\textbf{woPER}), (5) without all proposed component (\textbf{woALL}). The result show the effectiveness of each proposed component of RL4FT.

\begin{table}[htbp]
    % Explicitly top-align the parbox
    \caption{Ablation study results for RL4FT+H1 ($\downarrow$ indicates a larger performance degradation). The PER is applicable only to off-policy RL (SAC).  } 
    \label{tab:ablation} % Added an overall label for the combined table
    \parbox[t]{.48\linewidth}{ % Adjust width to leave space for separation
        \centering
        % Removed individual caption from here, as there's an overall one
        \label{tab:sac_results}
        \begin{tabular}{l|r}
            \hline
            \multicolumn{2}{c}{\textbf{RL4FT(SAC)}} \\
            \hline
            woRWS  & -10.08\% \\ % Added %
            woITER & -12.19\% \\ % Added %
            woHG   & -106.45\% \\ % Added %
            woPER  & -8.24\% \\ % Added %
            woALL  & -126.22\% \\ % Added %
            % vs $\pi^{ref}$ & -44.66\%\\
            \hline
        \end{tabular}
    } % No % needed here because \hfill follows directly
    \hfill
    % Explicitly top-align the parbox
    \parbox[t]{.48\linewidth}{ % Same width as the first
        \centering
        % Removed individual caption from here, as there's an overall one
        \label{tab:ppo_results}
        \begin{tabular}{l|r}
            \hline
            \multicolumn{2}{c}{\textbf{RL4FT(PPO)}} \\
            \hline
            woRWS  & -3.15\% \\ % Added %
            woITER & -8.50\% \\ % Added %
            woHG   & -151.04\% \\ % Added %
            woALL  & -166.85\% \\ % Added %
            % vs $\pi^{ref}$ & -37.81\% \\
            \hline
        \end{tabular}
    }
\end{table}

\subsection{Robustness of Self-improvement and Heuristic-Guided mechanism}

To illustrate how our Heuristic-Guided mechanism addresses the long-tail distribution of sampled trajectories during roll-out, we analyze the distribution of trajectory lengths from the initial learning phase. Figure \ref{fig:traj} shows the density of these lengths collected at the end of the first training epoch. We can see that standard SAC exploration (black, dashed) tends to sample long trajectories, resulting in a long-tail distribution. Since the objective is to minimize the number of queries (i.e., trajectory length), these long trajectories are poor samples for the RL agent to learn from. With guidance from the H1 heuristic (red), our method produces significantly shorter trajectories, while standard SAC exhibits a long-tail pattern. This demonstrates that our Heuristic-Guided mechanism effectively warm-starts the RL agent with higher-quality samples.

\begin{figure*}[htbp]
\centering
\caption{Distribution of trajectory lengths comparing our Heuristic-Guided mechanism against Standard SAC across AD network graphs. }
\label{fig:traj}

\begin{subfigure}[b]{0.32\textwidth}
    \centering
    \includegraphics[width=\textwidth]{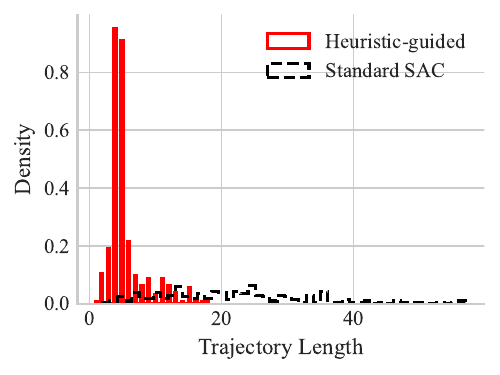}
    \caption{AS1}
    \label{fig:traj_1k}
\end{subfigure}
\hfill
\begin{subfigure}[b]{0.32\textwidth}
    \centering
    \includegraphics[width=\textwidth]{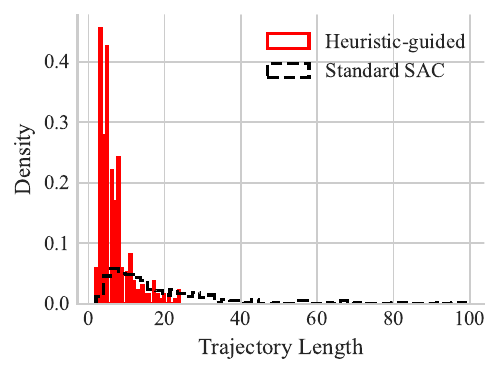}
    \caption{AS2}
    \label{fig:traj_5k}
\end{subfigure}
\hfill
\begin{subfigure}[b]{0.32\textwidth}
    \centering
    \includegraphics[width=\textwidth]{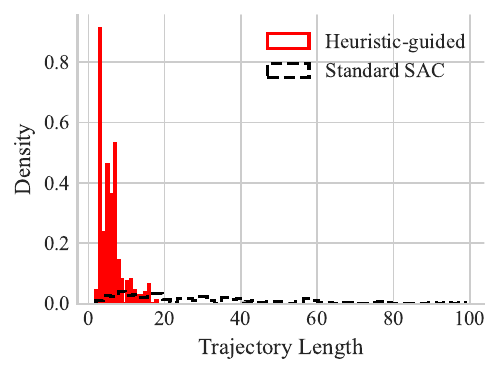}
    \caption{AS3}
    \label{fig:traj_10k}
\end{subfigure}

\vspace{0.5cm} % Adds a small vertical space between rows

\begin{subfigure}[b]{0.32\textwidth}
    \centering
    \includegraphics[width=\textwidth]{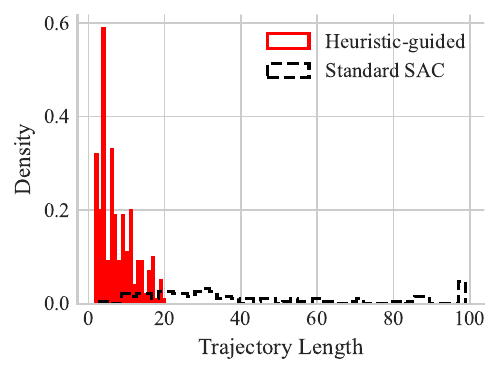}
    \caption{AS4}
    \label{fig:traj_20k}
\end{subfigure}
\hfill
\begin{subfigure}[b]{0.32\textwidth}
    \centering
    \includegraphics[width=\textwidth]{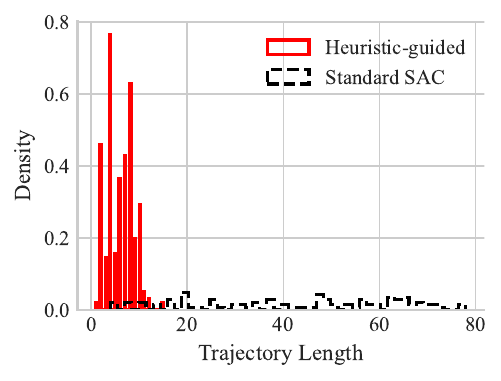}
    \caption{AS5}
    \label{fig:traj_50k}
\end{subfigure}
\hfill
\begin{subfigure}[b]{0.32\textwidth}
    \centering
    \includegraphics[width=\textwidth]{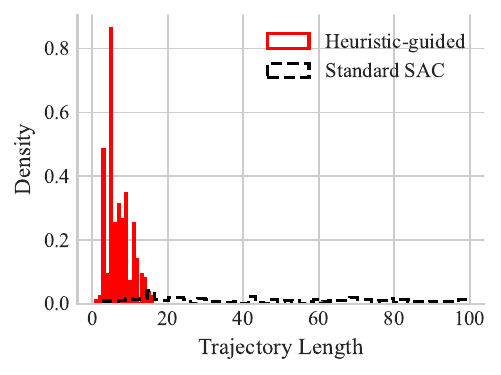}
    \caption{UNI}
    \label{fig:traj_uni}
\end{subfigure}

\end{figure*}

Figure \ref{fig:improve} demonstrates the performance improved by of the self-improvement mechanism through every iteration. In this experiment, we let the the algorithm iteratively improves until policy convergence or a limit of 10 iterations is met. The results shows the effectiveness of heuristic-guided mechanism, as it yields a policy that outperforms the reference heuristic after only a single iteration. Furthermore, this results show that the RL4FT-I framework is more effective when integrated with the SAC training algorithm compared to PPO.

\begin{figure*}[htbp]
\centering
\caption{Expected query cost versus self-improvement iterations for the RL4FT-I agent. The RL4FT-I agent outperforms the H1 reference policy from the very first iteration and keep improving their policy over iteration.}
\label{fig:improve}

\begin{subfigure}[b]{0.32\textwidth}
    \centering
    \includegraphics[width=\textwidth]{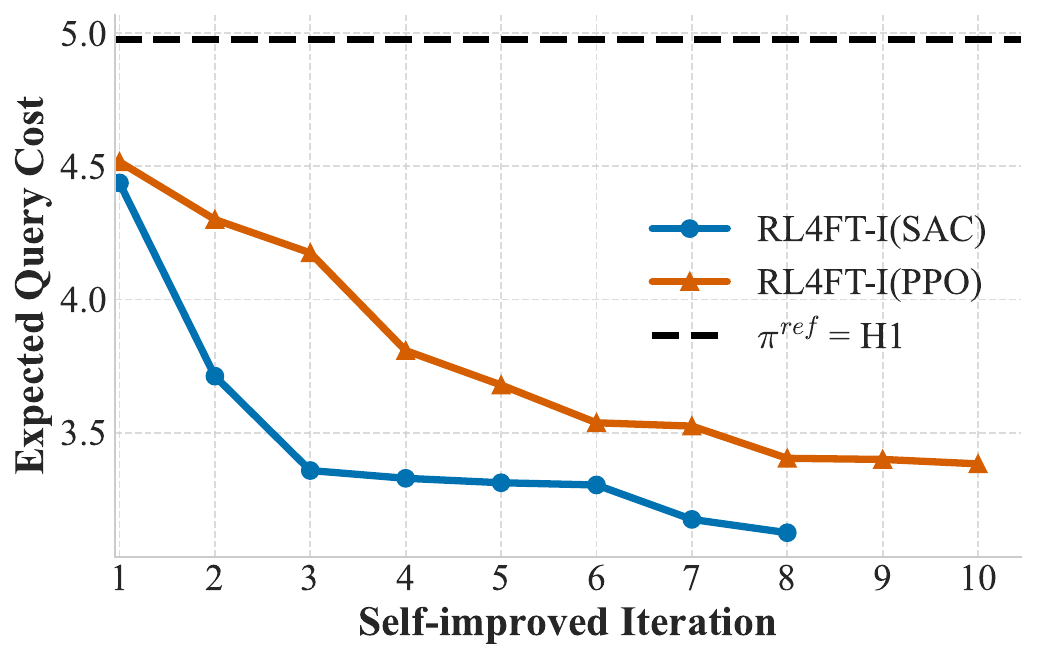}
    \caption{AS1}
    \label{fig:traj_1k}
\end{subfigure}
\hfill
\begin{subfigure}[b]{0.32\textwidth}
    \centering
    \includegraphics[width=\textwidth]{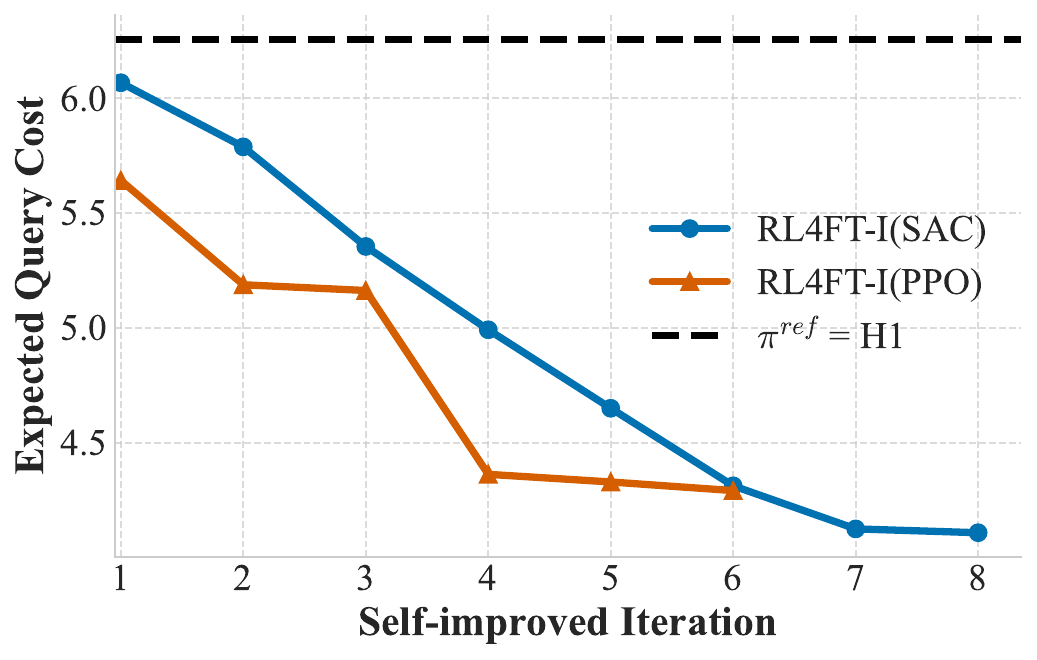}
    \caption{AS2}
    \label{fig:traj_5k}
\end{subfigure}
\hfill
\begin{subfigure}[b]{0.32\textwidth}
    \centering
    \includegraphics[width=\textwidth]{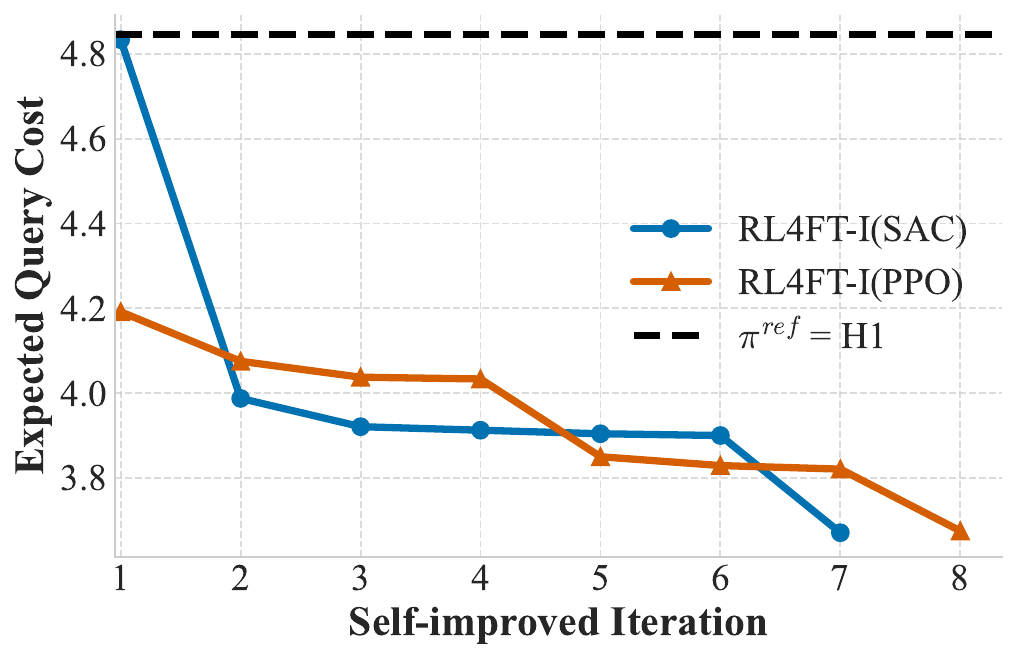}
    \caption{AS3}
    \label{fig:traj_10k}
\end{subfigure}

\end{figure*}

\subsection{Robustness to Non-linear Decision Function}
\label{subsec:nonlinear}

To demonstrate the flexibility of our approach with non-linear decision function, we extend our evaluation to two distinct function classes: a \textbf{Quadratic Model} and a \textbf{Radial Basis Function (RBF) Model} \citep{buhmann2000radial}.
The Quadratic Model captures non-linearities from pairwise feature interactions (e.g., rules based on combinations of two features). It's score function, $f_Q(\mathbf{x}_e)$, is formally defined as:
\[
    f_Q(\mathbf{x}_e) = \mathbf{w}^T \mathbf{x}_e + \mathbf{x}_e^T \mathbf{V} \mathbf{x}_e + b
\]
where $\mathbf{w} \in \mathbb{R}^d$ and $b \in \mathbb{R}$ are the linear weight and bias terms, and $\mathbf{V} \in \mathbb{R}^{d \times d}$ is a matrix of weights capturing the interaction terms.
The Radial Basis Function (RBF) Model captures non-linearities based on an edge's feature similarity to a set of learned prototypes. Its score function, $f_{RBF}(\mathbf{x}_e)$, is defined as:
\[
    f_{RBF}(\mathbf{x}_e) = \sum_{i=1}^{k} w_i \cdot \exp(-\gamma ||\mathbf{x}_e - \mathbf{c}_i||^2) + b
\]
Here, $\{ \mathbf{c}_i \}_{i=1}^k$ is a set of $k$ prototype center vectors, $\gamma$ controls the kernel width, and $\{ w_i \}_{i=1}^k$ and $b$ are the learned weights and bias. Figure \ref{fig:decision} reports results with a non-linear decision function. In this setting, we run RL4FT with H1 as the reference policy to assess robustness under non-linear decision function. We found that RL4FT improved the H1 reference policy by  $18.94 \%$ and $47.29 \%$, respectively. This results show that RL4FT remains robust and improves upon the reference heuristic.

\begin{figure*}[h] % Use the REGULAR figure environment
        \centering
        \caption{Expected query cost with different decision function {linear, quad, rbf}}
        \includegraphics[width=\linewidth]{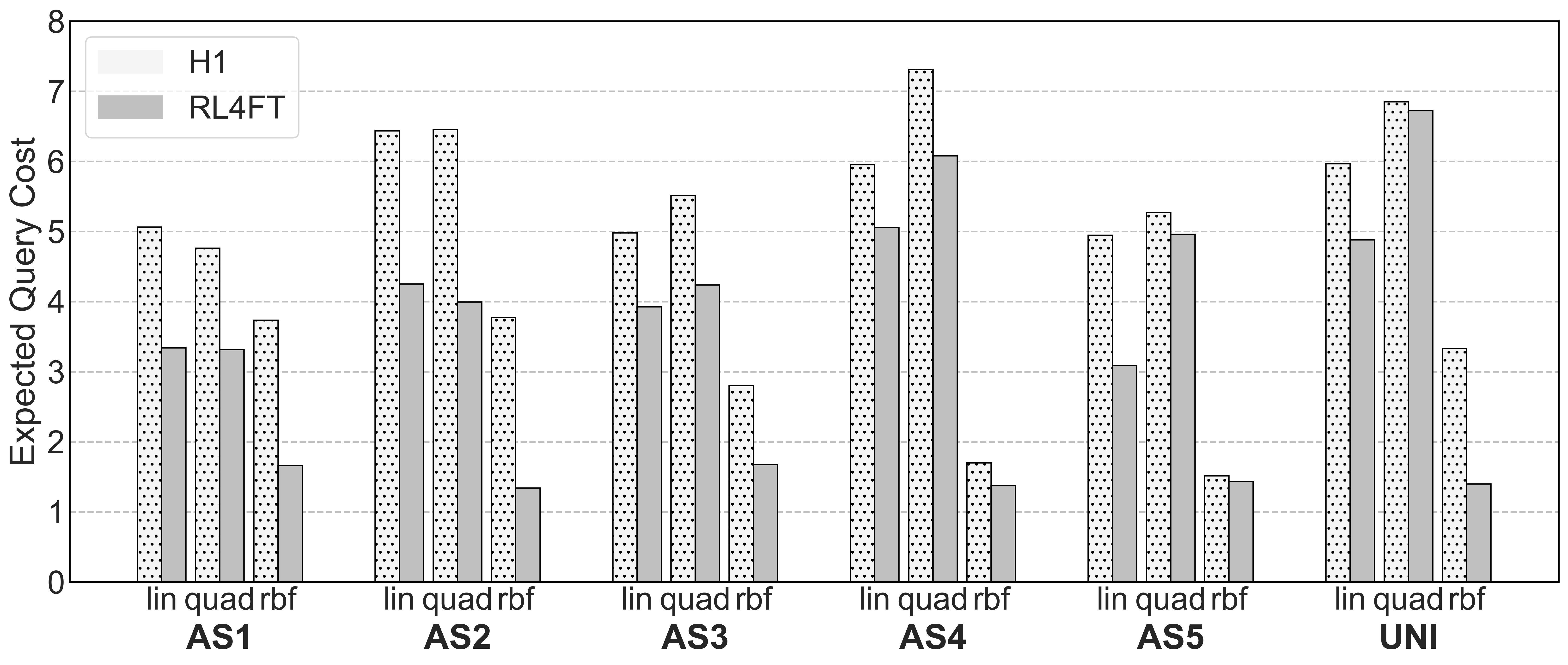}
        
        \label{fig:decision}
    \label{fig:policy}
\end{figure*}

\section{Conclusion}
In this work, we introduced the F-ACT model, a novel security model motivated by practical challenges faced by security teams. The key novelty of our model is its ability to generalize a adaptive security policy from a few queries and modeling IT admin decisions as a function of an edge's security features. We first proved that our problem is $\#\mathcal{P}$-hard and proposed a RL formulation for this problem. However, standard RL agent performed poorly due to challenges like sparse rewards and hard exploration. To address this, we developed RL4FT framework which introduces three key component to overcome these limitations including: (1) a heuristic-guided self-improvement mechanism; (2) a provably policy-invariant reward shaping function; and (3) a custom prioritized experience replay mechanism. 
The framework's flexibility is a key advantage, as it can be applied on top of any policy, enabling it to take an existing state-of-the-art method and further push the boundaries of performance.
We validated the effectiveness of our algorithm through extensive experiments on 10 large-scale, real-world graphs. The results show that our RL4FT algorithm guarantees improvement upon any given reference policy and significantly outperforms standard SOTA baselines.

\section{Appendix—Pseudocode for Heuristic Baselines}

\begin{algorithm}[H]
\caption{H1 Heuristic}
\label{alg:H1}
\KwIn{Graph $G=(V, E)$, source $s$, target $t$}
\KwOut{Query strategy}
\BlankLine

\Repeat{IsTerminated($G$)}{
    $p \leftarrow \text{FindONPath}(G, s, t)$\;
    
    $c \leftarrow \text{FindOFFCut}(G, s, t)$\;
    
    Select edge $e \in p \cap c$\;
    
    Query $e$ and observe its label $\sigma(e)$\;
    
    % Get the set of all newly deduced labels
    $D \leftarrow \text{ApplyDeductionRule}(e, \sigma(e))$, apply Assum. 1\;
    
    % Update the graph with all new information
    \For{each $(e', \sigma(e')) \in D$}{
        $G \leftarrow \text{updateLabel}(G, e', \sigma(e'))$\;
    }
}
\end{algorithm}

\begin{algorithm}[H]
\caption{Greedy Heuristic}
\label{alg:Greedy}
\KwIn{Graph $G=(V, E)$, source $s$, target $t$, utility $g$}
\KwOut{Sequence of queried edges}
\BlankLine

\Repeat{IsTerminated($G$)}{
% For each unqueried edge, calculate the utility of its potential outcomes
    \For{each $e \in E_u$}{
        % Utility if e is revealed to be OFF (after deductions)
        $U_{OFF}(e) \leftarrow g(\text{state after } e==\text{OFF})$\;
        % Utility if e is revealed to be ON (after deductions)
        
        $U_{ON}(e) \leftarrow g(\text{state after } e==\text{ON})$\;
    }
        
    $e \leftarrow \arg\min_{e \in E_G} \left\{ U_{OFF}(e) + U_{ON}(e)\right\}$\;
    
    Query $e$ and observe its label $\sigma(e)$\;
    
    % Get the set of all newly deduced labels
    $D \leftarrow \text{ApplyDeductionRule}(e, \sigma(e))$, apply Assum. 1\;
    
    % Update the graph with all new information
    \For{each $(e', \sigma(e')) \in D$}{
        $G \leftarrow \text{updateLabel}(G, e', \sigma(e'))$\;
    }
}
\end{algorithm}

\begin{algorithm}[H]
\caption{HMC Heuristic (for linear decision function)}
\label{alg:HMC}
\KwIn{Graph $G=(V, E)$, source $s$, target $t$, number of particles $n$}
\KwOut{Sequence of queried edges}
\BlankLine

Initialize particle set $\mathcal{H} = \{h_1, \dots, h_n\}$ where $h_i = (W_i, b_i)$\;

\Repeat{IsTerminated($G$)}{
    Initialize edge list $E' \leftarrow \{\}$\;

    \For{each hypothesis $h \in \mathcal{H}$}{
        % Find the critical path or cut under this hypothesis
        $E_{candidate} \leftarrow \text{FindPathOrCut}(G, h)$\;
        
        % Concatenate the candidate edges to the aggregate list
        $E_{agg} \leftarrow E_{agg} + E_{candidate}$\;
    }
    $e \leftarrow \arg\max_{e \in E_{agg}} \sum_{e' \in E_{agg}} \mathbb{I}(e = e')$, get edge with highest appearance frequency in $E_{agg}$\;
    
    Query $e$ and observe its label $\sigma(e)$\;
    
    % Get the set of all newly deduced labels
    $D \leftarrow \text{ApplyDeductionRule}(e, \sigma(e))$, apply Assum. 1\;
    
    % Update the graph with all new information
    \For{each $(e', \sigma(e')) \in D$}{
        $G \leftarrow \text{updateLabel}(G, e', \sigma(e'))$\;
    }
}
\end{algorithm}

\begin{algorithm}[H]
\caption{LIMIT algorithm - RL Training with Restricted Action Space Sampling}
\label{alg:rl_training_limit}

\SetAlgoLined
\LinesNumbered

% suppress printed "end" markers
\SetKwFor{For}{for}{:}{}
\SetKwFor{ForEach}{foreach}{:}{}
\SetKwFor{While}{while}{:}{}
\SetKwIF{If}{ElseIf}{Else}{if}{then}{else if}{else}{}
\SetKwProg{Fn}{Function}{ is}{end}

\KwIn{Initial Graph $G$, Branching factor $n$, Number of episodes $M$}
\KwOut{Trained policy $\pi_\phi$}
\BlankLine

% Initialize networks and buffer
Initialize policy network $\pi_\phi$, value network $Q_\theta$, and replay buffer $\mathcal{B} \leftarrow \emptyset$\;

\For{episode = $1 \dots M$}{
    Reset env. and randomize parameters of $f$\;
    
    % Collect trajectories T using the policy and restricted action space
    Initialize trajectory $\mathcal{T} \leftarrow \emptyset$ and state $s_t$\;
    
    \While{$\neg$IsTerminated($s_t$)}{
        % Generate restricted action space for current step
        $A_t \leftarrow \text{GetRestrictedActionSpace}(s_t, n)$\;
        
        % Policy samples edge e* from the restricted set At
        $e^* \sim \pi_\phi(s_t, A_t)$\;
        
        Execute query $e^*$, observe reward $r_t$ and next state $s_{t+1}$\;
        
        % Append to current trajectory
        $\mathcal{T} \leftarrow \mathcal{T} \cup (s_t, e^*, r_t, s_{t+1}, A_t)$\;
        $s_t \leftarrow s_{t+1}$\;
    }
    
    % Process and store transitions from the trajectory
    $\mathcal{B} \leftarrow \mathcal{B} \cup \mathcal{T}$
    
    \If{Update Criteria Met}{
        Sample mini-batch $\mathcal{MB} \subseteq \mathcal{B}$\;
        $\pi_\phi, Q_\theta \leftarrow \text{TrainPolicy}(\pi_\phi, Q_\theta, \mathcal{MB})$\;
    }
}

\BlankLine

\Fn{GetRestrictedActionSpace($s$, $n$)}{
    $A \leftarrow \emptyset$\;
    $Q \leftarrow \text{Queue}([s])$\;
    
    \While{$|A| < n$ \textbf{and} $Q$ is not empty}{
        $s'_{curr} \leftarrow Q.\text{pop}()$\;
        
        \If{$\neg$IsTerminated($s'_{curr}$)}{
            % H1 proposes exactly one intermediate edge for the current state
            $e_j \leftarrow \text{H1\_Heuristic}(s'_{curr})$\;
            
            \If{$e_j \notin A$}{
                $A \leftarrow A \cup \{e_j\}$\;
                
                % Explore the two possible scenarios (ON/OFF) in a BFS manner
                $s'_{ON} \leftarrow \text{updateLabel}(s'_{curr}, e_j, \text{ON})$\;
                $s'_{OFF} \leftarrow \text{updateLabel}(s'_{curr}, e_j, \text{OFF})$\;
                
                $Q.\text{push}(s'_{ON})$\;
                $Q.\text{push}(s'_{OFF})$\;
            }
        }
    }
    \Return{$A$}\;
}
\end{algorithm}

\section{Appendix-About LIMIT algorithm}

The pseudocode for the LIMIT algorithm is detailed in Algorithm \ref{alg:rl_training_limit}. While LIMIT follows a standard reinforcement learning training procedure, it diverges in that the agent's action space is constrained to a restricted set $A_t$ at each timestep. This restricted action space is constructed via the \textit{GetRestrictedActionSpace} procedure. This method iteratively identifies candidate actions by applying the H1 heuristic to a given state to get a intermediate edge. The procedure explores subsequent potential states in a Breadth-First Search (BFS) manner, accounting for the binary outcomes (ON/OFF) of each queried edge. As shown in the experiments in Section \ref{sec:paper4:experiment}, the LIMIT algorithm significantly outperforms baseline heuristics. We hypothesize that this performance is thank to the \textit{GetRestrictedActionSpace} procedure, which uses the iterative application of the H1 algorithm to effectively localize high-quality solutions within the graph structure.

While LIMIT alone does not guarantee optimality because the RL algorithm only interacts with a restricted action space, utilizing the RL4FT framework to improve the LIMIT policy (RL4FT+LIMIT) ensures convergence to optimality. Furthermore, we experimentally demonstrate that RL4FT+LIMIT improves performance compared to LIMIT. Note that the RL4FT+LIMIT setting runs as follows: First, we extract a policy by training the RL agent with LIMIT. Then, we use this policy as the reference policy for RL4FT and generate the final trained policy. Since the RL4FT procedure allows the agent to interact with the complete action space and the LIMIT policy is only used as a reference, RL4FT+LIMIT still guarantees convergence to optimality. 

\section{Appendix—List of Feature for Network}

Table \ref{tab:ad_features} lists the binary features we extract per edge in the Active Directory (AD) graph. Features are grouped into three families: permission-based, topological, and source-side vulnerability. The source-side vulnerability indicators capture vulnerabilities of the source node that may be abused at different stages of the attack lifecycle, including reconnaissance, credential access, lateral movement, persistence, and cross-tier privilege escalation.

By contrast, Table \ref{tab:misc_features} lists the binary features we extract per edge in the Miscellaneous graph. This graph does not include properties or labels, so we extract only topological features.

\clearpage
\begin{landscape}
\thispagestyle{empty}

% ===== Top table =====
\noindent
\begin{minipage}{\linewidth}
  \centering
  \captionof{table}{Features extracted for the \textbf{AD network}.}
  \label{tab:ad_features}
  \begin{adjustbox}{max totalsize={\linewidth}{0.46\textheight},center}
    \begin{tabular}{@{}|l|l|l|@{}}
      \toprule
      \textbf{Category} & \textbf{Feature Name} & \textbf{Description} \\ \midrule
      \textit{Permission-based} & \texttt{HasSession} & 1 if the edge type is "HasSession" \\
      & \texttt{MemberOf} & 1 if the edge type is "MemberOf" \\
      & \texttt{Contains} & 1 if the edge type is "Contains" \\
      & \texttt{AdminTo} & 1 if the edge type is "AdminTo" \\
      & \texttt{ForceChangePassword} & 1 if the edge type is "ForceChangePassword" \\
      & \texttt{AddMembers} & 1 if the edge type is "AddMembers" \\
      & \texttt{WriteDACL} & 1 if the edge type is "WriteDACL" \\
      & \texttt{WriteOwner} & 1 if the edge type is "WriteOwner" \\
      & \texttt{Owns} & 1 if the edge type is "Owns" \\
      & \texttt{AllExtendedRights} & 1 if the edge type is "AllExtendedRights" \\
      \midrule
      \textit{Topological} & \texttt{HighSourceInDegree} & 1 if the source node's in-degree is above the graph's median \\
      & \texttt{HighSourceOutDegree} & 1 if the source node's out-degree is above the graph's median \\
      & \texttt{HighTargetInDegree} & 1 if the target node's in-degree is above the graph's median \\
      & \texttt{HighTargetOutDegree} & 1 if the target node's out-degree is above the graph's median \\
      \midrule
      \textit{Source Vulnerabilities} & \texttt{SourceExposedForRecon} & 1 if the source has properties abusable for reconnaissance (e.g. hasspn, userpassword)\\
      & \texttt{SourceCredentialVulnerability} & 1 if the source has properties abusable for credential theft (e.g. dontreqpreauth, passwordnotreqd, pwdneverexpires, creddump)\\
      & \texttt{SourceEnablesLateralMovement} & 1 if the source has properties abusable for lateral movement techniques (e.g. unconstraineddelegation) \\
      & \texttt{SourceAllowsPersistence} & 1 if the source has properties abusable for persistence (e.g. haslaps)\\
      & \texttt{CrossTierVulnerability} & 1 if the edge connects a lower to a higher privilege tier \\
      & \texttt{Inactive} & 1 if the source node has been inactive for over 90 days \\
      \bottomrule
    \end{tabular}
  \end{adjustbox}
\end{minipage}

\vspace{0.8\baselineskip}

% ===== Bottom table =====
\noindent
\begin{minipage}{\linewidth}
  \centering
  \captionof{table}{Features extracted for the \textbf{Miscellaneous network}.}
  \label{tab:misc_features}
  \begin{adjustbox}{max totalsize={\linewidth}{0.46\textheight},center}
    \begin{tabular}{|l|l|}
      \toprule
      \textbf{Feature Name} & \textbf{Description} \\ \midrule
      \texttt{HighSourceNodeDegree} & 1 if the source node's degree is above the graph's median. \\
      \texttt{HighTargetNodeDegree} & 1 if the target node's degree is above the graph's median. \\
      \texttt{HighEdgeBetweenness} & 1 if the edge's betweenness is above the median. \\
      \texttt{HighJaccardCoefficient} & 1 if the Jaccard coefficient of the endpoints is above the median. \\
      \texttt{IsBridge} & 1 if the edge is a bridge (its removal increases connected components). \\
      \texttt{IsCommunityBridge} & 1 if the edge connects two distinct graph communities. \\
      \bottomrule
    \end{tabular}
  \end{adjustbox}
\end{minipage}

\end{landscape}
\clearpage

%% file: Chapter7/chapter7.tex
\chapter{Conclusion and Future Work} % Main chapter title
\label{chapter:end} % Change X to a consecutive number; for referencing this chapter elsewhere, use \ref{ChapterX}

\section{Conclusion}

This thesis studies practical models for defending AD attack graphs. Throughout the research, we have presented a suite of four novel models for hardening large Active Directory networks. These contributions fall into two themes. The first is honeypot allocation on time varying AD attack graphs, aiming to find placement strategies that interdict dynamic attack paths and provide early warning. In the second line of work, we introduce a human in the loop edge removal model and focus on making it practical and usable for large Active Directory networks. 
As these models address challenges that have not been previously explored together, they can be effectively combined into a comprehensive framework for enterprise defense. This framework allows IT administrators to implement a two-layered solution: active defenses strategies that remain effective even as the attack graph changes overtime, and an adaptive removal process that respects human operational constraints. 
Because the proposed framework is designed specifically to recommend and identify critical edges and optimal locations for active defenses, it does not interfere with existing security workflows. Rather than imposing a new process, the framework serves as a supportive tool that enhances security without disrupting established administrative practices. This allows IT security teams to leverage optimized insights while maintaining their current operational routines. Because these models align with existing security practices, they can be easily integrated into the standard enterprise workflow. Finally, we summarize the main contributions of each chapter as follows:
\begin{itemize}
    \item Chapter~\ref{chapter:paper1} investigates, for the first time, honeypot placement in Active Directory attack graphs. The chapter formulates the problem as a Stackelberg game between a defender and attackers. The defender commits to a placement first, and the attacker responds optimally. The model considers two attacker types, one that can observe honeypots and one that cannot. In this problem, the defender’s objective is to place honeypots to interdict attack paths and minimize the attacker’s expected probability of reaching DA. First, we show that computing optimal defenses against both attacker types is intractable. To address this, we propose a set of mixed integer programming formulations and extend it to dynamic attack graphs. To scale to many graph snapshots, we cluster the snapshots to identify a small set of representative instances and solve the MIP on those representatives. Experiments show that the approach scales well and produces near optimal blocking plans on graphs with about one hundred thousand nodes and millions of edges.

    \item Chapter~\ref{chapter:paper2} extends honeypot placement for AD by modeling the time-varying attack graph as a temporal directed graph. In a temporal graph, each edge is labeled with an online time (or interval), and a temporal attack path is a sequence of edges with non-decreasing timestamps. This captures both the paths an attacker can take and the timing of when those paths become available. We introduce a metric, response time, defined as the elapsed time from the first decoy trigger to compromise of the DA, and we use it to evaluate decoy placements. We formulate the problem as a Stackelberg game in which the attacker chooses a temporal path that minimizes traversal time and the defender allocates honeypots to maximize response time. We prove the defender’s problem is $\mathcal{NP}$-hard and proposed to use Evolutionary Diversity Optimization (EDO) to address it. Initial experiments show that vanilla EDO struggles to scale and to find feasible solutions because fitness evaluation is too costly. We therefore make two improvements: a Dijkstra-based earliest-arrival algorithm that exploits AD’s largely static structure to speed up fitness computation, and a surrogate-assisted scheme with a penalty fitness that evaluates candidates on important paths. We prove the surrogate procedure converges to a feasible solution after sufficient iterations, and we show experimentally that our method finds feasible solutions up to $108\times$ faster than vanilla EDO and improves performance by $23\%$ over the best baseline.

    \item Chapter \ref{chapter:paper3} investigates human in the loop edge removal for AD. We propose an adaptive model called Adaptive Path Removal. In each round, a wizard selects one current attack path and presents a multiple choice set of candidate edges on that path, and the administrator chooses one edge to cut. The process repeats until all paths from s to t are removed or a query budget is reached. The algorithm design goal is to reduce security team workload by minimizing the expected number of queries. First, we prove that the problem is $\#\mathcal{P}$-hard. To make it tractable, we develop an exact algorithm, an approximation algorithm, and several scalable heuristics. Among these, Dynamic Programming with Restrictions (DPR) achieves the best performance. It is derived from the exact method by restricting subproblems to gain scalability, while the approximation algorithm is used to derive high quality samples that guide the search. Furthermore, we design a reinforcement learning (RL) heuristic to deal with the problem. Although it does not perform as well as DPR, it can train offline and offers fast inference, which makes it promising for larger graphs and time sensitive settings. Finally, we validate the approach on multiple synthetic AD graphs and on an AD attack graph collected from a real organization.

    \item Chapter \ref{chapter:paper4} also investigates human in the loop edge removal for AD. This chapter adopt an adaptive mode in which, at each step, the wizard selects a candidate edge, presents its security context, and the administrator decides whether to remove or keep it. To better model IT administrator decision making, we associated each edge with a security-context feature vector and treat the administrator’s rule book as a hidden decision function that maps features to actions. To make the query process more efficient, we introduce a monotone risk-deduction rule that propagates one decision to edges with similar or higher risk. We proved that this problem is $\#\mathcal{P}$-hard and proposed to use reinforcement learning to solve it. Standard RL performs poorly because rewards are sparse and exploration is difficult, so we develop the RL4FT framework with three components: a policy agnostic self improvement routine, a policy invariant reward shaping scheme, and a custom prioritized experience replay. RL4FT is flexible and can wrap any reference policy to push performance further. In experiments on ten large real world graphs, RL4FT consistently improves on its reference policy and significantly outperforms strong baselines.
        
\end{itemize}

\section{Future Directions}

This thesis has made significant contributions to improving the resilience of Active Directory networks by investigating and proposing several hardening models. These efforts are expected to move optimization-based hardening from theory to practice and enable integration into existing AD defense workflows. Even so, there are many opportunities for further work. We outline several promising directions below:

\begin{itemize}
    \item \textbf{Generalize the current human in the loop model.} As discussed in this thesis, adaptive hardening is a strong candidate for integration into existing AD defense workflows because it incorporates administrator decisions and all critical changes require human verification. The adaptive models in Chapter~\ref{chapter:paper3}, Chapter~\ref{chapter:paper4}, and \citep{guo2024limited} focus on disconnecting all paths from a low privilege source $s$ to a high privilege target $t$ while minimizing queries. In practice, complete disconnection is often infeasible because the cut set may exceed the reserved change budget, leaving insufficient time to implement all removals, or because operational dependencies require some edges to remain. Non adaptive methods address this with "soft" objectives such as increasing the shortest path length \citep{guo2022practical} or reducing the number of sources that can reach privileged assets \citep{zhang2024practical}. Future work may generalize these soft-objective optimizations to an adaptive setting. A mixed-integer programming will be the backbone for optimization, since most soft objective formulations can be expressed as MILPs. On top of this, an adaptive constraint verification loop can reveal implementability constraints that are initially unknown through a limited number of IT administrator queries. The goal is to learn the constraints under a limited query budget and maximize the expected utility of the resulting MILP solution.

    \item \textbf{Generalize defense strategy in time-varying attack graphs.} Prior work has studied defense on time-varying AD attack graphs in both the honeypot placement setting (Chapters~\ref{chapter:paper1}, \ref{chapter:paper2}) and the edge-removal setting \citep{goel2024optimizing}. However, there are two limitations remain which hinder the practical usage. First, many methods assume that all snapshots are known in advance. Real AD environments change unpredictably through sign-ins, credential residue, policy updates, and routine administration, so future snapshots are not available ahead of time. A more realistic models treat edge availability as a prior online probabilities that are learned from past snapshot. Second, even when probabilities are modeled, defenses are usually optimized against a single attacker strategy, for example a weighted shortest path. In practice an adversary can choose any feasible route and may adapt once a path is not available. In the uncertain graph literature, a large body of work studies network reliability~\citep{khan2018uncertain,potamias2010k,jin2011discovering,reliability,valiant1979complexity}. In that framework, $s$–$t$ reliability is the probability that at least one path exists from $s$ to $t$ when edges materialize according to their probabilities. In the AD context with a time-varying attack graph, we aim to minimize the probability that a path exists from low-privileged sources to high-privileged targets. Subsequently, the defense optimization problem is to select a limited set of edges to remove so as to minimize $s$–$t$ reliability. 
    
    \item \textbf{Study hardening models for higher-order attack graphs.} Current AD attack graphs, following the BloodHound definition, represent security permissions as node-to-node relationships. In practice, however, Active Directory often specifies policy at group or container scope, which creates node-to-set or set-to-set relationships. For example, a AD permission granted to a group A applies to every object within an Organizational Unit (OU) B, this create a node-to-set edges from group A to a set of nodes in OU B. Nguyen et al.\ \citep{nguyen2024adsynth} shown that AD is well modeled as a $\alpha$-metagraph that captures these set-level relations as a node-to-set edges. By contrast, BloodHound flattens node-to-set edges into many node-to-node edges, which can inflate the graph and make hardening more expensive. Building on this idea, \citep{nguyen2025thesis} extends the objective of minimizing the number of sources that can reach Domain Admins from the node-to-node setting of Zhang et al.\ \citep{zhang2024practical} to the $\alpha$-metagraph model, and shows that running the corresponding cut algorithms directly on the metagraph is more efficient. Future work can further expand this line by developing adaptive metagraph-based methods and incorporating time-varying graphs, which may improve the efficiency of current algorithms.

\end{itemize}